\documentclass[12pt]{amsart}

\usepackage[letterpaper,margin=1in]{geometry}
\IfFileExists{fourier.sty}{\usepackage{fourier}}{\usepackage{utopia}}
\usepackage{amsmath, amssymb, amsthm, mathtools}
\usepackage{bm}
\IfFileExists{bbm.sty}{\usepackage{bbm}}{}
\usepackage{setspace}
\usepackage{url}
\usepackage{enumitem}
\usepackage{parskip}
\usepackage{xcolor}
\usepackage{booktabs}
\usepackage{array}
\usepackage{graphicx}
\usepackage{tikz}
\usetikzlibrary{positioning,arrows.meta}
\usepackage{placeins}
\extrafloats{100}
\newcommand{\empnotes}[1]{\par\vspace{4pt}\begin{minipage}{\textwidth}\setstretch{1}\footnotesize\textit{Notes:} #1\end{minipage}}

\usepackage[numbers,sort&compress]{natbib}

\newcommand{\cH}{\mathcal{H}}
\newcommand{\cQ}{\mathcal{Q}}
\newcommand{\cE}{\mathcal{E}}
\newcommand{\cF}{\mathcal{F}}
\newcommand{\E}{\mathrm{E}}
\newcommand{\Pp}{\mathrm{P}}
\newcommand{\Pn}{\mathbb{P}_n}
\newcommand{\Gn}{\mathbb{G}_n}
\newcommand{\Var}{\mathrm{Var}}
\newcommand{\R}{\mathbb{R}}
\IfFileExists{bbm.sty}{\newcommand{\1}{\mathbbm{1}}}{\newcommand{\1}{\mathbf{1}}}
\newcommand{\norm}[1]{\left\lVert #1 \right\rVert}

\newcommand{\lessssim}{\mathrel{\lesssim}}
\newcommand{\conv}{\operatorname{conv}}

\theoremstyle{plain}
\newtheorem{theorem}{Theorem}
\newtheorem{proposition}{Proposition}
\newtheorem{lemma}{Lemma}
\newtheorem{corollary}{Corollary}
\theoremstyle{definition}
\newtheorem{assumption}{Assumption}
\newtheorem{definition}{Definition}
\theoremstyle{remark}
\newtheorem{remark}{Remark}

\title{Pragmatic DML with AI-Learned Representations}
\date{First version: January 2026; This version: \today}

\author[Aradillas Fern\'andez et al.]{Andr\'es Aradillas Fern\'andez,
Victor Chernozhukov, Carlos Cinelli, Sven Klaassen, Whitney Newey,
Martin Spindler, Jan Teichert-Kluge, Suhas Vijaykumar}

\begin{document}
\setstretch{1.25}
\setlength{\emergencystretch}{3em}

\begin{abstract}
Text, images, and other rich objects increasingly enter causal analysis through AI-learned representations used as controls.
For a broad class of estimands, replacing the raw object by its representation shifts the target by the inner product of two errors: one in the outcome regression and one in the balancing weight. The bias is therefore bounded by the product of their sizes.
Three consequences follow.
With a fixed representation, the usual cross-fitted debiased machine learning (DML) interval is valid for a representation-specific target, and for the causal target when the approximation bias is negligible relative to sampling error.
When this bias may not be negligible, a sensitivity analysis bounds the causal target and supplies confidence bounds for the endpoints.
When the embedding is tuned on the sample, three approaches preserve valid inference: fixed-menu fine-tuning, joint-loss tuning, and double representation selection.
A multimodal demand application illustrates the workflow.
\end{abstract}

\maketitle

\section{Introduction}

Consider a cardigan sold on an online marketplace. Its photograph shows cut, texture, and color; its description gives fabric, warmth, and intended use; structured fields record brand and category. Color variants of the same cardigan sell at different prices. An economist who wants the response of demand to price must hold such differences fixed, and most of them live in the image and the text, not in a spreadsheet. Figure~\ref{fig:listing} sketches such a listing. Covariates of this kind are now common: a document describing a firm, a product image, a sequence of transactions. To use them, researchers compress them into learned representations---embeddings, model predictions, or features produced by foundation models---and then use those representations as controls~\cite{mikolov2013efficient,pennington2014glove,devlin2019bert,mann2020language,athey2019machine}.
We ask what the resulting estimates mean and how to report them.

Write $Y$ for an outcome, $D$ for a treatment, and $Z$ for the rich object, and let $X_\phi=x(Z,\phi)$ be its learned representation.
The analyst picks a target, such as an average treatment effect (ATE), fits flexible regressions of $Y$ on $(D,X_\phi)$, and reports a debiased machine learning (DML) estimate with its standard error~\cite{chernozhukov2018}.
Replacing $Z$ by $X_\phi$ changes two things at once: the regressions used for adjustment, and the population parameter on which DML centers.
For the ATE of a binary treatment, let $g^\star(d,z)=\E[Y\mid D=d,Z=z]$ and $g_\phi(d,x)=\E[Y\mid D=d,X_\phi=x]$.
Under unconfoundedness given $Z$, the ATE is $\theta^\star=\E[g^\star(1,Z)-g^\star(0,Z)]$; DML with $X_\phi$ instead estimates $\theta_\phi=\E[g_\phi(1,X_\phi)-g_\phi(0,X_\phi)]$, which we call the \emph{average predictive effect}.
If the representation blurs a feature that moves both treatment and outcome, as a color variant may move both price and demand, $\theta_\phi$ and $\theta^\star$ differ.

Let $\alpha^\star$ and $\alpha_\phi$ be the Riesz representers (balancing weights) associated with $g^\star$ and $g_\phi$.
A version of the omitted variable bias formula of~\cite{chernozhukov2024long} gives
\begin{equation}
\theta^\star-\theta_\phi
=
\E\bigl[(g^\star-g_\phi)(\alpha^\star-\alpha_\phi)\bigr],
\label{eq:intro-ovb}
\end{equation}
where the full-information objects $g^\star$ and $\alpha^\star$ are evaluated at $(D,Z)$ and their representation-based counterparts $g_\phi$ and $\alpha_\phi$ at $(D,X_\phi)$.
The approximation bias is nonzero only if the representation discards information that matters both for the outcome and for the representer.
A representation that preserves either the outcome regression or the representer therefore has no approximation bias, even if it discards much of $Z$.

We use~\eqref{eq:intro-ovb} to answer three questions that an applied researcher faces.
\begin{enumerate}[leftmargin=*,label=(\roman*)]
\item \emph{What does DML estimate with a fixed embedding?}
It estimates $\theta_\phi$, with the usual guarantees (Section~\ref{sec:fixed}).
If the approximation bias in~\eqref{eq:intro-ovb} is small relative to the sampling error, the DML interval also covers $\theta^\star$; if both errors vanish, then under further regularity conditions the estimator is also semiparametrically efficient for $\theta^\star$.
\item \emph{What if the approximation bias is not small?}
Then one reports a sensitivity region for $\theta^\star$, built from interpretable bounds on the strength and alignment of the discarded information, together with confidence bounds for its endpoints (Section~\ref{sec:bounds}).
We also show how one can calibrate the plausible magnitude of the approximation errors by deleting observed modalities from the representation and measuring the errors that the deletion induces.
\item \emph{What if the embedding is tuned on the same sample?}
Fine-tuning, prompt search, and model choice make the representation data-dependent.
Section~\ref{sec:tuning} gives three approaches to valid inference that differ in what they ask of the tuning and in how they use the sample: fixed-menu fine-tuning inside the folds, tuning with a joint loss, and double representation selection on a held-out sample.
\end{enumerate}
Equation~\eqref{eq:intro-ovb} also holds when text or images are the treatment itself (Section~\ref{sec:treatment}).
Figure~\ref{fig:workflow} summarizes the resulting workflow, and Section~\ref{sec:empirical} applies it to product demand with text, image, and tabular information.
The appendices contain the regularity conditions and proofs, together with auxiliary results.

\begin{figure}[t]
\centering
\begin{tikzpicture}[
  font=\scriptsize,
  head/.style={font=\footnotesize\bfseries, anchor=south},
  lab/.style={font=\scriptsize\itshape, anchor=south},
  enc/.style={draw, rounded corners=2pt, fill=black!5, align=center, inner sep=3pt, font=\scriptsize},
  arr/.style={-{Stealth[length=4pt]}, semithick}]
% ---------- image panel: x in [0,3.6]
\draw[rounded corners=2pt] (0,0) rectangle (3.6,3.4);
\node[head] at (1.8,3.75) {Image};
\node[lab] at (1.8,3.42) {cut, texture, color};
\begin{scope}[shift={(1.8,2.05)}, line join=round]
  \draw[fill=black!18] (-0.55,0.75) -- (-1.15,-0.55) -- (-0.92,-0.65) -- (-0.45,0.25) -- cycle;
  \draw[fill=black!18] (0.55,0.75) -- (1.15,-0.55) -- (0.92,-0.65) -- (0.45,0.25) -- cycle;
  \draw[fill=black!12] (-0.55,0.75) -- (-0.15,0.85) -- (0,0.35) -- (0.15,0.85) -- (0.55,0.75) -- (0.6,-0.6) -- (-0.6,-0.6) -- cycle;
  \draw (0,0.35) -- (0,-0.6);
  \foreach \y in {0.2,-0.05,-0.3} \fill (0.07,\y) circle (0.025);
  \draw[black!40] (-0.6,-0.48) -- (0.6,-0.48);
\end{scope}
\foreach \c/\p [count=\i] in {black!12/\$15,black!70/\$15,white/\$15,yellow!20/\$29,red!45/\$27} {
  \pgfmathsetmacro{\x}{0.35+(\i-1)*0.62}
  \draw[fill=\c] (\x,0.62) rectangle ++(0.38,0.3);
  \node[font=\tiny] at (\x+0.19,0.38) {\p};
}
% ---------- text panel: x in [3.9,8.1]
\draw[rounded corners=2pt] (3.9,0) rectangle (8.1,3.4);
\node[head] at (6.0,3.75) {Text};
\node[lab] at (6.0,3.42) {fabric, warmth, intended use};
\node[anchor=north west, text width=3.9cm, align=left, inner sep=0pt] at (4.05,3.25)
  {\textbf{Cropped knit cardigan}\\[3pt]
   Lightweight button-down knit in a soft cotton blend. A versatile layering piece for fall, spring, and the office.};
% ---------- tabular panel: x in [8.4,12.4]
\draw[rounded corners=2pt] (8.4,0) rectangle (12.4,3.4);
\node[head] at (10.4,3.75) {Tabular};
\node[lab] at (10.4,3.42) {brand, category, sizes};
\node[anchor=north west, inner sep=0pt] at (8.55,3.25) {%
  \renewcommand{\arraystretch}{1.25}%
  \begin{tabular}{@{}l@{\hspace{6pt}}l@{}}
  Brand & Brand A\\
  Category & Cardigans\\
  Material & Cotton blend\\
  Sizes & XXS--XXL\\
  Colors & 5\\
  \end{tabular}};
% ---------- encoders and embedding
\node[enc] (e1) at (1.8,-0.85) {image encoder};
\node[enc] (e2) at (6.0,-0.85) {text encoder or LLM};
\node[enc] (e3) at (10.4,-0.85) {attribute coding};
\foreach \e/\x in {e1/1.8,e2/6.0,e3/10.4} \draw[arr] (\x,0) -- (\e.north);
\node[draw, rounded corners=2pt, inner sep=4pt, font=\footnotesize, minimum width=9.4cm] (x) at (6.0,-2.0)
  {$X_\phi=x(Z,\phi)=(0.12,\,-0.47,\,0.80,\,\dots,\,0.05)\in\R^p$};
\foreach \e in {e1,e2,e3} \draw[arr] (\e.south) -- (\e.south |- x.north);
\end{tikzpicture}
\caption{A stylized product listing. The object $Z$ combines an image, text, and structured fields. Color variants of the same item sell at different prices, so color, visible mainly in the image, moves the price $D$ and plausibly demand $Y$. Encoders and attribute coding turn $Z$ into one row of controls $X_\phi$; if $X_\phi$ blurs color, it discards information that matters for both the outcome regression and the representer, and $\theta_\phi$ can differ from $\theta^\star$ by the amount in~\eqref{eq:intro-ovb}. The listing and the entries of $X_\phi$ are illustrative.}
\label{fig:listing}
\end{figure}

\begin{figure}[t]
\centering
\begin{tikzpicture}[
  font=\footnotesize,
  col/.style={draw, rounded corners=2pt, inner sep=4pt, anchor=north},
  q/.style={draw, align=center, inner sep=4pt, fill=black!5},
  arr/.style={-{Stealth[length=5pt]}, semithick},
  lab/.style={font=\scriptsize\itshape, fill=white, inner sep=1pt}]
\node[q, text width=7.6cm] (q1) at (0,0) {Is $\phi$ chosen from the analysis sample?};
\node[col] (c1) at (-5.5,-1.7)
  {\parbox[t][3.0cm][t]{4.9cm}{\raggedright
   \textbf{Approach~1: fixed-menu fine-tuning}\\
   One embedding: target $\theta_\phi$ (Theorem~\ref{thm:dml-expansion}).\\
   Fixed menu, fine-tuned inside the folds: $\theta_0=\theta_\phi$ (Theorem~\ref{thm:pipeline-dml}); star aggregation for menus.}};
\node[col] (c2) at (0,-1.7)
  {\parbox[t][3.0cm][t]{4.9cm}{\raggedright
   \textbf{Approach 2: joint-loss tuning}\\
   One shared representation, light and stable tuning. Target $\bar\theta_n$ (Proposition~\ref{prop:joint}, Theorem~\ref{thm:pooled-joint}). For menus, $\bar\theta_n(\hat\jmath)$ with critical value $\hat c_{1-\gamma}$.}};
\node[col] (c3) at (5.5,-1.7)
  {\parbox[t][3.0cm][t]{4.9cm}{\raggedright
   \textbf{Approach 3: double representation selection}\\
   Large $n$, heavy or unstable tuning; tune on $\mathcal S_0$, freeze the union, infer on $\mathcal S_1$. Target $\theta_{\hat\phi}$ (Proposition~\ref{prop:sfu}).}};
\draw[arr] (q1.south) -- ++(0,-3mm) -| node[lab, pos=0.25] {no: fixed in advance} (c1.north);
\draw[arr] (q1.south) -- (c2.north);
\draw[arr] (q1.south) -- ++(0,-3mm) -| node[lab, pos=0.25] {yes: once, on $\mathcal S_0$, then frozen} (c3.north);
\node[lab] at ([yshift=-0.95cm]q1.south) {yes: in each training fold, one shared $\phi$};
\coordinate[below=5mm of c2.south] (bus);
\node[q, text width=9.4cm, below=10mm of c2.south] (q2)
  {Is the approximation bias of the reported target $o(n^{-1/2})$, so that target adaptivity is defensible?};
\draw (c1.south) -- (c1.south |- bus) -- (c3.south |- bus);
\draw (c3.south) -- (c3.south |- bus);
\draw[arr] (c2.south) -- (q2.north);
\node[col] (yes) at (-4.04,-9.4)
  {\parbox[t][3.6cm][t]{7.4cm}{\raggedright
   \textbf{Report the interval for $\theta^\star$}\\
   It covers $\theta^\star$ under target adaptivity (Section~\ref{sec:omitted}; Approach~3: Corollary~\ref{cor:adaptive}); efficient if the embedding is fixed and information-lossless (Theorem~\ref{thm:strong-adaptive}).}};
\node[col] (no) at (4.04,-9.4)
  {\parbox[t][3.6cm][t]{7.4cm}{\raggedright
   \textbf{Report a sensitivity analysis}\\
   Region $[\hat\theta^-,\hat\theta^+]$ around the reported target and one-sided bounds $L,U$ (Proposition~\ref{prop:onesided} for a fixed embedding, Appendix~\ref{app:sens} for Approach~1, Corollary~\ref{cor:joint} for Approach~2, and Proposition~\ref{prop:sfu} for Approach~3) over a grid of $(C_Y^2, 1-R^2_{\alpha^\star\sim\alpha_\phi})$ at fixed $\rho$.}};
\draw[arr] (q2.south) -- ++(0,-3mm) -| node[lab, pos=0.25] {yes} (yes.north);
\draw[arr] (q2.south) -- ++(0,-3mm) -| node[lab, pos=0.25] {no} (no.north);
\end{tikzpicture}
\caption{A workflow for DML with learned representations. The first question asks whether the representation $\phi$ (an embedding, or a menu of encoders with head classes) is chosen from the analysis sample, which selects the approach and its target (Table~\ref{tab:routes}); the second asks whether the approximation bias of that target is $o(n^{-1/2})$, and if it is not, the analyst reports a sensitivity analysis for $\theta^\star$.}
\label{fig:workflow}
\end{figure}

\subsection*{Related work}
We build on the literature on orthogonal scores and cross-fitting~\cite{chernozhukov2018,belloni2014high}, on Riesz regression for linear functionals~\cite{chernozhukov2022riesz}, and on the omitted variable bias framework of~\cite{chernozhukov2024long}, which we apply to information discarded by a representation rather than to unobserved confounders.
Empirical work already uses text and other learned representations in causal analysis~\cite{wooddoughty2018challenges,sridhar2022causal}.
Double representation selection borrows its logic from post-double-selection~\cite{belloni2014inference}.
To combine candidate learners, we use star aggregation~\cite{audibert2007progressive}, fitted to pooled out-of-fold predictions as in stacking~\cite{wolpert1992stacked,breiman1996stacked,vanderlaan2007super}.

Our analysis is closely related to adaptive debiased machine learning (ADML)~\cite{van2023adaptive}.
That paper shows that data-driven model selection can be combined with debiasing to give estimators that are regular, asymptotically linear, and efficient for an oracle projection parameter, and hence superefficient for the original target.
Its framework covers selection from sieves and by the Lasso, as well as learned feature representations for confounding adjustment~\cite{benkeser2020nonparametric,damour2021deconfounding}.
With the working model of all functions of $(D,X_\phi)$, the ADML working parameter is our $\theta_\phi$, and its bound on the error of the working model has the mixed-bias structure of Proposition~\ref{prop:ovb} below; both derive from the omitted variable bias framework of~\cite{chernozhukov2024long}.

Our paper differs from ADML in two respects.
First, for inference on the original target, ADML assumes that the outcome regression is a function of oracle features to which the learned map converges.
We study embeddings produced by foundation models, for which there is no reason to expect this, so we bound the approximation bias $\theta^\star-\theta_\phi$ by a sensitivity analysis when it may matter.
When the outcome regression is a function of oracle features, our analysis simplifies even without convergence, because the approximation bias is then at most the product of the regression error of the tuned representation and the change in the representer when the oracle features are added to that representation (Appendix~\ref{app:containment}).
Second, ADML requires the learned model to approximate a fixed oracle model, the natural condition for sieves and the Lasso and close to the rate condition of our first approach (Section~\ref{sec:fine-tuned}).
Because convergence of a fine-tuned foundation model is hard to verify, our second approach weakens that condition to convergence without a rate, and our third drops it at the price of a data-dependent target.
The approaches are complementary: an ADML estimator could serve as the head learner within any of our approaches, and our sensitivity analysis could accompany an ADML analysis whose oracle features may not determine the outcome regression.

\subsection*{Notation}
We observe i.i.d.\ copies $O_1,\dots,O_n$ of $O=(Y,D,Z)$ with law $P$.
For a function $f$ we write $Pf=\int f\,dP$, $\norm{f}_{P,q}=(P|f|^q)^{1/q}$, and $\E_n f=n^{-1}\sum_{i}f(O_i)$.
Stochastic order symbols are with respect to $P$, and $\Rightarrow$ denotes weak convergence.

\section{Setup and the approximation bias}
\label{sec:warmup}
Replacing $Z$ by $X_\phi$ changes the regression, the representer, and therefore the target.

\subsection{Full information and representations}
\label{sec:setup}
Let $Y\in\R$ be an outcome, $D$ a treatment, and $Z$ a high-dimensional object, and let $X_\phi:=x(Z,\phi)\in\R^p$ be a representation indexed by $\phi$, a possibly very high-dimensional parameter.

\paragraph{\textbf{Sources of representations.}}
Almost any foundation model can supply $X_\phi$, and the theory treats it as a black box.
Encoders such as BERT for text and the Vision Transformer for images pool their final hidden states into a vector~\cite{devlin2019bert,dosovitskiy2021image}, and contrastive models such as CLIP map images and text into a common space~\cite{radford2021learning}.
Generative models supply representations in two further ways.
Their hidden states can serve as the embedding, possibly after a prompt that tells the model what to attend to, such as fabric, fit, and quality.
Alternatively, the model answers a fixed list of questions about each object, with or without intermediate reasoning~\cite{mann2020language,wei2022chain}, and the answers become coordinates with names.
The index $\phi$ collects every such choice: the model, the layer, the prompt, the questions, and the weights.
See~\cite{gentzkow2019text} for a survey of text as data in economics.
Until Section~\ref{sec:tuning} we treat the representation rule as fixed, as when a frozen model is used with a fixed prompt or was trained on independent data.

We compare two information sets,
\begin{equation*}
W:=(D,Z)
\qquad\text{and}\qquad
W_\phi:=(D,X_\phi),
\end{equation*}
with regressions $g^\star(W):=\E[Y\mid W]$ and $g_\phi(W_\phi):=\E[Y\mid W_\phi]$.
Since $W_\phi$ is a function of $W$, $g_\phi(W_\phi)=\E[g^\star(W)\mid W_\phi]$.

We study targets that are linear functionals of the regression.
Let $m(o,h)$ be linear in the function $h$, and define the \emph{full-information} (long) and \emph{representation-based} (short) targets
\begin{equation}
\theta^\star:=\E\bigl[m(O,g^\star)\bigr],
\qquad
\theta_\phi:=\E\bigl[m(O,g_\phi)\bigr].
\label{eq:theta-long}
\end{equation}
We assume that $m(O,h)$ is a function of $W_\phi$ whenever $h$ is.
This compatibility condition holds in the examples below: there $m$ evaluates $h$ at arguments that are functions of $W_\phi$ or, for a change of content distribution (Section~\ref{sec:treatment}), depends on $O$ only through the customer context.

\paragraph{\textbf{Examples.}}
For the ATE with $D\in\{0,1\}$, $m(O,h)=h(1,Z)-h(0,Z)$, and $\theta_\phi$ is the average predictive effect defined in the introduction.
For a continuous treatment, the average derivative has $m(O,h)=\partial_d h(d,Z)|_{d=D}$. The application in Section~\ref{sec:empirical} uses this functional on varying-coefficient projections to estimate a price-response slope. Appendix~\ref{app:avg-deriv} gives the restricted-space interpretation, which permits atoms in treatment and distinguishes this slope from the average derivative of an unrestricted regression.
Weighted average effects and policy effects are covered as well when their weights and policies depend on $Z$ only through $X_\phi$~\cite{chernozhukov2022riesz}.

\subsection{Riesz representers and the identity}
\label{sec:identity}
To express the approximation bias, we need a Riesz representer for each information set.
When $h\mapsto\E[m(O,h)]$ is continuous on $L^2(P_W)$, the Riesz representation theorem supplies a unique $\alpha^\star\in L^2(P_W)$ with
\begin{equation}
\E[m(O,h)]=\E[h(W)\alpha^\star(W)]
\qquad\text{for all } h\in L^2(P_W),
\label{eq:rr-long}
\end{equation}
and, likewise, a unique $\alpha_\phi\in L^2(P_{W_\phi})$ with
\begin{equation}
\E[m(O,h)]=\E[h(W_\phi)\alpha_\phi(W_\phi)]
\qquad\text{for all } h\in L^2(P_{W_\phi}).
\label{eq:rr-short}
\end{equation}
These representers are balancing weights, because averages of $h$ weighted by them reproduce the target functional.
For the ATE, $\alpha^\star=D/e(Z)-(1-D)/\{1-e(Z)\}$ with $e(z)=\Pp(D=1\mid Z=z)$, and $\alpha_\phi$ is the same expression with $e(Z)$ replaced by $e_\phi(X_\phi)=\Pp(D=1\mid X_\phi)$.

\begin{proposition}[{Omitted-information identity}]
\label{prop:ovb}
Assume~\eqref{eq:rr-long} and~\eqref{eq:rr-short} hold and $Y\in L^2(P)$.
Then
\begin{equation}
\theta^\star-\theta_\phi
=
\E\bigl[(g^\star-g_\phi)(\alpha^\star-\alpha_\phi)\bigr],
\label{eq:ovb-identity}
\end{equation}
and consequently
\begin{equation}
|\theta^\star-\theta_\phi|
\le
\norm{g^\star-g_\phi}_{P,2}\,\norm{\alpha^\star-\alpha_\phi}_{P,2}.
\label{eq:cs-bound}
\end{equation}
\end{proposition}

The proposition specializes Theorem~2 of~\cite{chernozhukov2024long} to nested information sets, $\sigma(W_\phi)\subseteq\sigma(W)$; the proof is in Appendix~\ref{app:proofs-sec2}.

A representation has no approximation bias if it preserves either the outcome regression or the representer, and by~\eqref{eq:cs-bound} the approximation bias is small when both are approximately preserved.
For the ATE, a representation that retains the propensity score, $e_\phi=e$, retains $\alpha^\star$ and hence the ATE, however poorly it predicts $Y$.

\subsection{Embeddings that describe the treatment}
\label{sec:treatment}
Until now $X_\phi$ has served as a control.
Proposition~\ref{prop:ovb} also applies when the text or image is the object the analyst changes.
Suppose a platform considers replacing the distribution $\Lambda_0$ of the content $Z$ it displays by an alternative $\Lambda_1$, say a different mix of descriptions and photographs.
Let $g^\star(z)=\E[Y\mid Z=z]$ be the mean response to content $z$ in a fixed customer population, and suppose it is the causal response, as under random delivery of content.
The effect of the change is
\begin{equation}
\theta^\star=\int g^\star(z)\,\{d\Lambda_1(z)-d\Lambda_0(z)\}=\E[m(O,g^\star)],
\qquad
m(O,h):=\int h\,d(\Lambda_1-\Lambda_0),
\label{eq:treatment-theta}
\end{equation}
a linear functional of the regression with $W=Z$.
If $\Lambda_0$ and $\Lambda_1$ are dominated by the law $P_Z$ of content in the data and the density ratios are square integrable, its Riesz representer is the difference of the two density ratios,
\begin{equation}
\alpha^\star(z)=\frac{d\Lambda_1}{dP_Z}(z)-\frac{d\Lambda_0}{dP_Z}(z).
\label{eq:treatment-alpha}
\end{equation}
Content that the new policy makes more likely receives a positive weight, and content it makes less likely a negative weight.
With $W_\phi=X_\phi$, Proposition~\ref{prop:ovb} applies unchanged. Here $\theta_\phi$ is the effect the policy would have if the response depended on content only through $X_\phi$, and $\alpha_\phi=\E[\alpha^\star\mid X_\phi]$ is the difference of the density ratios of the laws of $X_\phi$ under the two policies.

By Proposition~\ref{prop:ovb}, an embedding of the treatment preserves the effect if it keeps what matters for the response or what the policy changes.
If the density ratios depend on $z$ only through $x(z)$, so that the policy acts on content only through the embedding, then $\alpha^\star=\alpha_\phi$ and $\theta^\star=\theta_\phi$, however much else the embedding drops.
If the policy changes features that the embedding discards, $\alpha^\star-\alpha_\phi$ is nonzero, and an embedding that predicts $Y$ well can still give the wrong effect.
The results of Sections~\ref{sec:fixed}--\ref{sec:tuning} apply with $D$ removed, $W=Z$ and $W_\phi=X_\phi$, provided the density ratios have moments of order $q>4$ (Appendix~\ref{app:treatment}).
A deletion benchmark is informative only if the policy changes something the retained representation sees; if it acts only on a deleted modality, the short representer is zero and $C_D$ is undefined.
We treat $\Lambda_0$ and $\Lambda_1$ as known, as when they are the platform's own display rules.
Appendix~\ref{app:treatment} derives~\eqref{eq:treatment-alpha}, adds customer context, and states what changes when the policies are estimated.

\section{Fixed embeddings}
\label{sec:fixed}
With $\phi$ fixed, DML on $(D,X_\phi)$ is ordinary DML for $\theta_\phi$.

\subsection{The estimator}
\label{sec:dml}
For a fixed embedding, the debiased score for $\theta_\phi$ is
\begin{equation}
\psi_\phi(O;g,\alpha):=m(O,g)+\alpha(W_\phi)\{Y-g(W_\phi)\}.
\label{eq:score-short}
\end{equation}
It is Neyman-orthogonal, and its mean at $(g,\alpha)$ is $\theta_\phi-\E[(g-g_\phi)(\alpha-\alpha_\phi)]$.
For square-integrable $g$ and $\alpha$ that are functions of $W_\phi$,
\begin{equation}
\E\bigl[\psi_\phi(O;g,\alpha)\bigr]-\theta_\phi
=
-\E\bigl[(g-g_\phi)(\alpha-\alpha_\phi)\bigr].
\label{eq:mixed-bias}
\end{equation}
The algebra is that of Proposition~\ref{prop:ovb}.
The long score
\begin{equation}
\psi^\star(O;g,\alpha):=m(O,g)+\alpha(W)\{Y-g(W)\},
\label{eq:score-long}
\end{equation}
evaluated at $(g^\star,\alpha^\star)$ and centered at $\theta^\star$, is the efficient influence function for $\theta^\star$~\cite{bickel1993,newey1994,vanderVaart1998}.

Split the sample into $K$ folds $I_1,\dots,I_K$, and for each $k$ fit $\hat g_\phi^{(-k)}$ and $\hat\alpha_\phi^{(-k)}$ on the other folds, the outcome learner by least squares and the representer by minimizing the Riesz loss $\E_n[\alpha(W_\phi)^2-2m(O,\alpha)]$~\cite{chernozhukov2022riesz}.
With $\hat\psi_{\phi,i}:=\psi_\phi(O_i;\hat g_\phi^{(-k)},\hat\alpha_\phi^{(-k)})$ for $i\in I_k$, the estimator and its variance estimator are
\begin{align}
\hat\theta_\phi&:=\frac1n\sum_{i=1}^n\hat\psi_{\phi,i},
\label{eq:dml-est}\\
\hat\sigma_\phi^2&:=\frac1n\sum_{i=1}^n(\hat\psi_{\phi,i}-\hat\theta_\phi)^2.
\label{eq:var-est}
\end{align}
Assumption~\ref{ass:dml} in Appendix~\ref{app:proofs-sec2} includes the rate condition
\begin{equation*}
\max_{k}\norm{\hat g_\phi^{(-k)}-g_\phi}_{P,2}\,\max_{k}\norm{\hat\alpha_\phi^{(-k)}-\alpha_\phi}_{P,2}=o_\Pp(n^{-1/2}).
\end{equation*}

\begin{theorem}[DML for the representation-based target]
\label{thm:dml-expansion}
Suppose Assumption~\ref{ass:dml} in Appendix~\ref{app:proofs-sec2} holds for a possibly drifting sequence $\phi=\phi_n$.
Then
\begin{equation}
\sqrt n\bigl(\hat\theta_\phi-\theta_\phi\bigr)
=
\frac1{\sqrt n}\sum_{i=1}^n\bigl\{\psi_\phi(O_i;g_\phi,\alpha_\phi)-\theta_\phi\bigr\}+o_\Pp(1),
\label{eq:dml-expansion}
\end{equation}
$\sigma_\phi^{-1}\sqrt n(\hat\theta_\phi-\theta_\phi)\Rightarrow N(0,1)$ with $\sigma_\phi^2:=\Var\psi_\phi(O;g_\phi,\alpha_\phi)$, and $\hat\sigma_\phi^2/\sigma_\phi^2\to1$ in probability.
\end{theorem}

We allow $\phi$ to depend on $n$ so that we can apply the result along sequences of representations in Theorem~\ref{thm:strong-adaptive} and Proposition~\ref{prop:sfu}.
By Theorem~\ref{thm:dml-expansion}, the interval
\begin{equation}
\mathrm{CI}_\phi(1-\gamma):=\bigl[\hat\theta_\phi\pm z_{1-\gamma/2}\,\hat\sigma_\phi/\sqrt n\bigr]
\label{eq:ci-short}
\end{equation}
covers $\theta_\phi$ with probability tending to $1-\gamma$.
Without further assumptions, the interval need not cover $\theta^\star$.

\subsection{Coverage of the full-information target}
\label{sec:omitted}
The interval~\eqref{eq:ci-short} covers $\theta^\star$ when $\theta_\phi$ is within $o(n^{-1/2})$ of $\theta^\star$.

\begin{definition}[Target adaptivity]
\label{def:target-adaptivity}
A sequence of representations $X_{\phi_n}$ is \emph{target-adaptive} for $\theta^\star$ if
\begin{equation}
\sqrt n\,(\theta^\star-\theta_{\phi_n})\to0.
\label{eq:target-adaptivity}
\end{equation}
\end{definition}

By~\eqref{eq:cs-bound}, a sufficient condition is
\begin{equation}
\sqrt n\,\norm{g^\star-g_{\phi_n}}_{P,2}\,\norm{\alpha^\star-\alpha_{\phi_n}}_{P,2}\to0.
\label{eq:target-adaptivity-suff}
\end{equation}
Under target adaptivity, $\sqrt n(\hat\theta_{\phi_n}-\theta^\star)$ has the same expansion as in~\eqref{eq:dml-expansion}, so $\mathrm{CI}_{\phi_n}$ covers $\theta^\star$.
The condition is the analogue of the DML rate condition on the product of nuisance errors, because the representation, like a first-stage learner, may err on both nuisances provided the product of its errors is $o(n^{-1/2})$.

Coverage is not efficiency.
The variance $\sigma_{\phi_n}^2$ is that of the short score, which may exceed, or fall below, the efficiency bound $\sigma_\star^2:=\Var\psi^\star(O;g^\star,\alpha^\star)$.

\begin{definition}[Adaptive efficiency]
\label{def:adaptive-efficiency}
A sequence of representations $X_{\phi_n}$ is \emph{adaptively efficient} for $\theta^\star$ if
\begin{equation}
\norm{g^\star-g_{\phi_n}}_{P,2}\to0,
\qquad
\norm{\alpha^\star-\alpha_{\phi_n}}_{P,2}\to0,
\qquad
\sup_n\norm{\alpha_{\phi_n}}_{P,2}<\infty.
\label{eq:adaptive-efficiency}
\end{equation}
A representation that is both target-adaptive and adaptively efficient is \emph{information-lossless} for the target.
\end{definition}

The last condition in~\eqref{eq:adaptive-efficiency} is automatic, since $\norm{\alpha_{\phi_n}}_{P,2}\le\norm{\alpha^\star}_{P,2}$.

\begin{theorem}[Efficiency under information-losslessness]
\label{thm:strong-adaptive}
Suppose the conditions of Theorem~\ref{thm:dml-expansion} hold, the representation is information-lossless, and the moment conditions of Lemma~\ref{lem:score-conv} in Appendix~\ref{app:proofs-sec2} hold.
Then
\begin{equation}
\sqrt n\bigl(\hat\theta_{\phi_n}-\theta^\star\bigr)
=
\frac1{\sqrt n}\sum_{i=1}^n\bigl\{\psi^\star(O_i;g^\star,\alpha^\star)-\theta^\star\bigr\}+o_\Pp(1),
\label{eq:efficient-expansion}
\end{equation}
and $\hat\sigma_{\phi_n}^2\to\sigma_\star^2$ in probability.
\end{theorem}

Thus $\hat\theta_{\phi_n}$ attains the semiparametric efficiency bound for $\theta^\star$.
The two properties answer different questions. Target adaptivity decides which parameter the interval covers; adaptive efficiency decides whether the interval is as short as possible.
For a fixed embedding, neither can be checked reliably from $(Y,D,X_\phi)$: both depend on $g^\star$ and $\alpha^\star$, which must be learned over the rich object $Z$, the task the representation was meant to avoid.
Whether they hold is therefore a judgment about the representation, and when the analyst cannot defend target adaptivity, the right report is a sensitivity analysis.

\section{Sensitivity analysis}
\label{sec:bounds}
In this section, we bound $\theta^\star$ when the approximation bias may not be small. We write the bias as the product of an estimable scale and three interpretable sensitivity parameters, which measure the magnitudes of the two errors introduced by the representation and how closely these errors move together. Plausibility judgments on these parameters then give an interval for $\theta^\star$.

\subsection{Bounds on the approximation bias}
Following~\cite{chernozhukov2024long}, write
\begin{equation}
\theta^\star-\theta_\phi=\rho\,C_Y\,C_D\,S_\phi,
\label{eq:bias-factorization}
\end{equation}
where
\begin{equation*}
S_\phi^2:=V_{Y,\phi}V_{\alpha,\phi},
\qquad
V_{Y,\phi}:=\E\bigl[(Y-g_\phi)^2\bigr],
\qquad
V_{\alpha,\phi}:=\E\bigl[\alpha_\phi^2\bigr],
\end{equation*}
and
\begin{equation*}
C_Y^2:=\frac{\E[(g^\star-g_\phi)^2]}{\E[(Y-g_\phi)^2]},
\qquad
C_D^2:=\frac{\E[(\alpha^\star-\alpha_\phi)^2]}{\E[\alpha_\phi^2]};
\end{equation*}
finally, $\rho\in[-1,1]$ is the $L^2(P)$ alignment of $g^\star-g_\phi$ and $\alpha^\star-\alpha_\phi$, set to zero if either vanishes.
The scale $S_\phi$ is estimable at root-$n$ from $(Y,D,X_\phi)$; the other three factors would require learning the full-information nuisances.
The parameter $C_Y^2$ is the share of the residual outcome variation that the discarded information would explain.
The parameter $C_D^2$ is the relative gain in the second moment of the representer, $(1-R^2_{\alpha^\star\sim\alpha_\phi})/R^2_{\alpha^\star\sim\alpha_\phi}$, where $R^2_{\alpha^\star\sim\alpha_\phi}:=\E[\alpha_\phi^2]/\E[\alpha^{\star2}]$ is the share of the second moment of $\alpha^\star$ that $\alpha_\phi$ retains.
Explicitly, $\rho=\E[(g^\star-g_\phi)(\alpha^\star-\alpha_\phi)]/(\norm{g^\star-g_\phi}_{P,2}\,\norm{\alpha^\star-\alpha_\phi}_{P,2})$ when both differences are nonzero.
Given bounds $|\rho|\le\bar\rho$, $C_Y\le\bar C_Y$, and $C_D\le\bar C_D$, and with $c:=\bar\rho\,\bar C_Y\bar C_D$,
\begin{equation}
\theta^\star\in\bigl[\theta_\phi-cS_\phi,\;\theta_\phi+cS_\phi\bigr]=:[\theta_\phi^-,\theta_\phi^+].
\label{eq:bounds}
\end{equation}
We call $(\bar\rho,\bar C_Y,\bar C_D)$ a scenario, and the scenario is valid when the three bounds hold.

\subsection{Inference for the endpoints}
The endpoints depend on three quantities estimable at root-$n$: $\theta_\phi$, $V_{Y,\phi}$, and $V_{\alpha,\phi}=\E[2m(O,\alpha_\phi)-\alpha_\phi^2]$ (by~\eqref{eq:rr-short} with $h=\alpha_\phi$).
The cross-fitted estimators of the last two, with $k(i)$ the fold that contains observation $i$,
\begin{equation*}
\hat V_{Y,\phi}:=\frac1n\sum_{i}\bigl\{Y_i-\hat g_\phi^{(-k(i))}(W_{\phi,i})\bigr\}^2,
\qquad
\hat V_{\alpha,\phi}:=\frac1n\sum_{i}\bigl\{2m(O_i,\hat\alpha_\phi^{(-k(i))})-\hat\alpha_\phi^{(-k(i))}(W_{\phi,i})^2\bigr\},
\end{equation*}
use orthogonal scores, so their first-stage errors enter only at second order.
Set $\hat S_\phi:=(\hat V_{Y,\phi}\hat V_{\alpha,\phi})^{1/2}$ and $\hat\theta_\phi^\pm:=\hat\theta_\phi\pm c\hat S_\phi$.
Theorem~\ref{thm:bounds-pipeline} in Appendix~\ref{app:sens} shows that, under regularity conditions that include moments of order $q>4$, nuisance errors $o_\Pp(n^{-1/4})$, and $S_\phi>0$, each endpoint estimator is asymptotically linear,
\begin{equation*}
\sqrt n\bigl(\hat\theta_\phi^\pm-\theta_\phi^\pm\bigr)\Rightarrow N(0,s_\pm^2),
\end{equation*}
with an explicit influence function whose variance $s_\pm^2$ is estimated consistently by $\hat s_\pm^2$.
Following~\cite{chernozhukov2024long}, we report one-sided bounds,
\begin{equation}
L:=\hat\theta_\phi^- - z_{1-\gamma}\,\frac{\hat s_-}{\sqrt n},
\qquad
U:=\hat\theta_\phi^+ + z_{1-\gamma}\,\frac{\hat s_+}{\sqrt n}.
\label{eq:one-sided}
\end{equation}
Each bound makes a one-sided claim, so the quantile is $z_{1-\gamma}$, not $z_{1-\gamma/2}$.

\begin{proposition}[One-sided coverage]
\label{prop:onesided}
Suppose the conditions of Theorem~\ref{thm:bounds-pipeline} hold with $s_\pm^2>0$.
Then $\Pp\{\theta_\phi^-<L\}\to\gamma$ and $\Pp\{\theta_\phi^+>U\}\to\gamma$.
If the scenario is valid, so that $|\rho|\le\bar\rho$, $C_Y\le\bar C_Y$, and $C_D\le\bar C_D$, then
\begin{equation}
\limsup_{n\to\infty}\Pp\{\theta^\star<L\}\le\gamma,
\qquad
\limsup_{n\to\infty}\Pp\{\theta^\star>U\}\le\gamma.
\label{eq:onesided-target}
\end{equation}
If, moreover, $c>0$ is fixed and $S_\phi$ is bounded away from zero, then $\liminf_{n}\Pp\{L\le\theta^\star\le U\}\ge1-\gamma$.
\end{proposition}

Each claim, ``$\theta^\star$ is not below $L$'' and ``$\theta^\star$ is not above $U$,'' fails with probability at most $\gamma$.
When the region is wide relative to $n^{-1/2}$, the target cannot lie near both endpoints, so only one claim can fail, and $[L,U]$ is a $1-\gamma$ confidence interval for $\theta^\star$, as in~\cite{imbens2004confidence}.
For scenarios that shrink with $n$, the Bonferroni envelope $[\hat\theta_\phi^- - z_{1-\gamma/2}\hat s_-/\sqrt n,\ \hat\theta_\phi^+ + z_{1-\gamma/2}\hat s_+/\sqrt n]$ of~\eqref{eq:confidence-envelope-bounds} in Appendix~\ref{app:sens} is a valid two-sided confidence set.

\subsection{Calibration and reporting}
\label{sec:calibration}
Bounds on $(\rho,C_Y,C_D)$ are plausibility judgments that must ultimately rest on domain knowledge. Analysts can nevertheless use the observed data to aid these judgments, by comparing what $X_\phi$ discards from $Z$ with what is lost when an observed modality is deleted.
Treat a rich observed representation as long and a deliberately poorer one as short, for example all modalities against text alone.
The premise is that an observed deletion is a benchmark for what $X_\phi$ omits from $Z$.
Both pairs of nuisances are now estimable, and by the projection formulas~\eqref{eq:profile-reg} and~\eqref{eq:profile-riesz} of Section~\ref{sec:subordinated-geometry}, their validation losses give estimates of the gains of the deletion,
\begin{equation*}
\widehat G_Y=\frac{\widehat R^2_{\mathrm{long}}-\widehat R^2_{\mathrm{short}}}{1-\widehat R^2_{\mathrm{long}}},
\qquad
\widehat G_D=\frac{\widehat{\mathcal L}_{\mathrm{short}}-\widehat{\mathcal L}_{\mathrm{long}}}{-\widehat{\mathcal L}_{\mathrm{short}}},
\end{equation*}
where $\widehat R^2$ is the validation $R^2$ of the outcome learner and $\widehat{\mathcal L}$ the validation Riesz loss.
By~\eqref{eq:profile-reg}, the numerator of $\widehat G_Y$ estimates $\norm{g_{\mathrm{long}}-g_{\mathrm{short}}}_{P,2}^2$ and the denominator the residual variance of the long model.
By~\eqref{eq:profile-riesz}, the minimal Riesz loss with representation $\phi$ is $-\E[\alpha^{\star2}]+\norm{\Delta_\alpha(\phi)}_{P,2}^2$; differencing the short and long values gives $\norm{\alpha_{\mathrm{long}}-\alpha_{\mathrm{short}}}_{P,2}^2$, and $-\widehat{\mathcal L}_{\mathrm{short}}$ estimates $\E[\alpha_{\mathrm{short}}^2]$.
Following the benchmarking analysis of~\cite[Appendix~E]{chernozhukov2024long}, suppose that the information the long representation omits is $k$ times as strong as the deleted modality, in the sense that $C_Y^2=kG_Y$ and $1-R^2_{\alpha^\star\sim\alpha_\phi}=kG_D$; the definition of $C_D^2$ in Section~\ref{sec:bounds} then gives $C_D^2=kG_D/(1-kG_D)$.
For $k=1$ the estimates are $\widehat C_Y^2=\widehat G_Y$ and $\widehat C_D^2=\widehat G_D/(1-\widehat G_D)$.
Both normalizations are the more conservative ones: $\widehat G_Y$ divides by the smaller residual variance, that of the long model, and $\widehat C_D^2\ge\widehat G_D$.
By Proposition~\ref{prop:ovb} with the long representation in place of $Z$, $\theta_{\mathrm{long}}-\theta_{\mathrm{short}}=\E[(g_{\mathrm{long}}-g_{\mathrm{short}})(\alpha_{\mathrm{long}}-\alpha_{\mathrm{short}})]$, the alignment of the two differences times the product of their norms. Following~\cite[Appendix~E]{chernozhukov2024long}, we therefore estimate the alignment by $\hat\rho=(\hat\theta_{\mathrm{long}}-\hat\theta_{\mathrm{short}})/\{(\hat V_{Y,\mathrm{short}}-\hat V_{Y,\mathrm{long}})(\widehat{\mathcal L}_{\mathrm{short}}-\widehat{\mathcal L}_{\mathrm{long}})\}^{1/2}$, where $\hat V_{Y,\mathrm{long}}$ and $\hat V_{Y,\mathrm{short}}$ are the validation mean squared errors of the two outcome learners; $\hat\rho$ is defined when both gains are positive.
The analyst can then ask whether the information missing from the richest observed representation is plausibly as strong as a whole observed modality, or several times stronger.
We recommend a grid of scenarios for the two shares $C_Y^2$ and $1-R^2_{\alpha^\star\sim\alpha_\phi}$, which both lie in $[0,1]$, with $\rho$ fixed, and for each scenario a report of the estimate $\hat\theta_\phi$, the region $[\hat\theta_\phi^-,\hat\theta_\phi^+]$, and the bounds $L$ and $U$, with the points $C_Y^2=k\widehat G_Y$ and $1-R^2_{\alpha^\star\sim\alpha_\phi}=k\widehat G_D$ marked for several values of $k$.
The decision-relevant question is usually whether $U$ stays on the right side of a threshold set by the application, such as a demand elasticity of $-1$, rather than whether the target is nonzero. The robustness value of~\cite{chernozhukov2024long} is the common value of the two shares at which $U$ reaches the threshold for a fixed $\rho$, and the benchmark multiple is the smallest $k$ for which $U$ reaches the threshold at the point $C_Y^2=k\widehat G_Y$, $1-R^2_{\alpha^\star\sim\alpha_\phi}=k\widehat G_D$. Both follow from the value of $c$ at which $U$ reaches the threshold, because $U$ depends on the scenario only through $c$, in the region and in the standard error $\hat s_+$ of~\eqref{eq:one-sided} alike.

\section{Tuning the representation}
\label{sec:tuning}

Analysts often adapt the representation to the task: they fine-tune an encoder, search prompts, or choose among models and modalities.
The representation then depends on the sample, and Theorem~\ref{thm:dml-expansion}, which assumes a fixed representation, no longer applies.
Tuning creates two difficulties: the losses used must control both factors of the bias, and the target $\theta_{\hat\phi}$ moves with $\hat\phi$ at first order.
The three approaches answer one question: is $\phi$ chosen from the inference sample? Approach~1 chooses nothing: $\phi$ is the whole pipeline, fixed in advance, so the target does not move. Approach~2 chooses $\phi$ in each training fold and centers on the moving target. Approach~3 chooses $\phi$ once, on a separate sample $\mathcal S_0$, and freezes it.

\subsection{Tuning losses and target drift}
\label{sec:subordinated-geometry}
In this section $\phi$ may denote a whole menu: encoders $\{x_j(Z)\}_{j\in J(\phi)}$ with head classes $\mathcal N_g(\phi)$ and $\mathcal N_\alpha(\phi)$, which generate the closed subspace $\cH_\phi:=\overline{\mathrm{span}}\{h(D,x_j(Z)):\ j\in J(\phi),\ h\in\mathcal N_g(\phi)\cup\mathcal N_\alpha(\phi)\}\subseteq L^2(P_W)$.
Now $g_\phi$ and $\alpha_\phi$ are the $L^2(P)$ projections of $g^\star$ and $\alpha^\star$ onto $\cH_\phi$ and $\theta_\phi:=\E[m(O,g_\phi)]$; one encoder with rich heads, $\cH_\phi=L^2(P_{W_\phi})$, is a fixed embedding.
Both nuisance losses are quadratic, and minimizing them over $\cH_\phi$ leaves exactly the two errors in~\eqref{eq:cs-bound}.
Write $\Delta_g(\phi):=g^\star-g_\phi$ and $\Delta_\alpha(\phi):=\alpha^\star-\alpha_\phi$.
By the Pythagorean identity and the Riesz property,
\begin{align}
\min_{h\in\cH_\phi}\E\bigl[\{Y-h(W)\}^2\bigr]
&=\E\bigl[(Y-g^\star)^2\bigr]+\norm{\Delta_g(\phi)}_{P,2}^2,
\label{eq:profile-reg}\\
\min_{a\in\cH_\phi}\E\bigl[a(W)^2-2m(O,a)\bigr]
&=-\E\bigl[\alpha^{\star2}\bigr]+\norm{\Delta_\alpha(\phi)}_{P,2}^2.
\label{eq:profile-riesz}
\end{align}
Prediction risk therefore tunes the representation toward a small regression loss $\norm{\Delta_g(\phi)}_{P,2}$, and Riesz loss toward a small representer loss $\norm{\Delta_\alpha(\phi)}_{P,2}$; prediction risk alone guards one factor of the bias and ignores the other.

\paragraph{\textbf{Richness versus learnability.}}
A richer representation shrinks both losses, but its heads are harder to learn at a given sample size, and the product-rate condition of Theorem~\ref{thm:dml-expansion} may fail.
For the ATE, finer information can also push propensity scores toward zero and one and inflate the variance $\sigma_\phi^2$, so a coarser representation can give a better estimate even though its approximation bias is larger.
Losses evaluated out of sample, such as validation losses, estimate the population loss plus the estimation error of the fitted heads, so tuning by them already penalizes representations whose heads are hard to learn at the given $n$.
They do not penalize a large $\sigma_\phi^2$, which the analyst should report alongside the estimate.

\paragraph{\textbf{Target drift.}}

For any two representations $\phi$ and $\phi'$, the algebra of Proposition~\ref{prop:ovb} gives $\theta_\phi-\theta_{\phi'}=\E[(g_\phi-g_{\phi'})\Delta_\alpha(\phi)]+\E[\Delta_g(\phi')(\alpha_\phi-\alpha_{\phi'})]$, and by Cauchy--Schwarz
\begin{equation}
|\theta_\phi-\theta_{\phi'}|
\le
\norm{g_\phi-g_{\phi'}}_{P,2}\norm{\Delta_\alpha(\phi)}_{P,2}
+\norm{\Delta_g(\phi')}_{P,2}\norm{\alpha_\phi-\alpha_{\phi'}}_{P,2}.
\label{eq:drift}
\end{equation}
In~\eqref{eq:drift} the coefficients are the representer loss $\norm{\Delta_\alpha(\phi)}_{P,2}$ and the regression loss $\norm{\Delta_g(\phi')}_{P,2}$.
If the tuned $\hat\phi$ converges to a limit $\bar\phi$ at rate $r_n$, in the sense that $\norm{g_{\hat\phi}-g_{\bar\phi}}_{P,2}+\norm{\alpha_{\hat\phi}-\alpha_{\bar\phi}}_{P,2}=O_\Pp(r_n)$, the drift $\theta_{\hat\phi}-\theta_{\bar\phi}$ is at most of order $r_n$ times the losses.
Thus, by~\eqref{eq:drift}, the drift is negligible at the $\sqrt n$ scale when the losses are small, as under target adaptivity, or when tuning converges faster than $n^{-1/2}$, which is rarely plausible for a fine-tuned foundation model.

\subsection{Approach 1: fixed-menu fine-tuning}
\label{sec:fine-tuned}
\label{sec:swap}
\label{sec:pipeline-inference}
In the first approach, fixed-menu fine-tuning, the menu $\phi$ is fixed in advance and only the heads, possibly with fine-tuned encoders, are fitted inside each training fold, each nuisance being tuned by its own loss.
On the training folds $I_k^c$, the analyst fine-tunes an encoder for the outcome by prediction risk and, possibly separately, an encoder for the representer by Riesz loss, fits the heads, and evaluates only $\hat g^{(-k)}$ and $\hat\alpha^{(-k)}$ on the held-out fold.
Given the training folds, those functions are fixed, as cross-fitting requires.
The estimate is $\hat\theta_0:=\E_n\hat\psi$ with $\hat\psi_i:=\psi^\star(O_i;\hat g^{(-k(i))},\hat\alpha^{(-k(i))})$, and $\hat\sigma_0^2$ is the empirical variance of $\hat\psi_i$.

The conditions concern the fitted functions $\hat g^{(-k)}$ and $\hat\alpha^{(-k)}$, not the tuned parameter $\hat\phi$.
Here $\phi$ is the whole pipeline: its encoders are the representation rules $\Phi$ it may reach, with head classes $\mathcal N_g$ and $\mathcal N_\alpha$.
Write $\cH:=\cH_\phi$ for the closed span of everything the pipeline can express.

\begin{definition}[Pipeline target]
\label{def:pipeline-target}
Let $g_0$ and $\alpha_0$ be the $L^2(P)$ projections of $g^\star$ and $\alpha^\star$ onto $\cH$; equivalently,
\begin{equation*}
g_0:=\arg\min_{g\in\cH}\E\bigl[\{Y-g(W)\}^2\bigr],
\qquad
\alpha_0:=\arg\min_{\alpha\in\cH}\E\bigl[\alpha(W)^2-2m(O,\alpha)\bigr].
\end{equation*}
The \emph{pipeline target} is $\theta_0:=\E[m(O,g_0)]$. It is the target $\theta_\phi$ of Section~\ref{sec:subordinated-geometry} for the pipeline's own menu $\phi$.
\end{definition}

Because both nuisances are projections onto one space, Proposition~\ref{prop:pipeline-mixed-bias} in Appendix~\ref{app:pipeline} shows that, for $g,\alpha\in\cH$, the mean of the long score differs from $\theta_0$ by $-\E[(g-g_0)(\alpha-\alpha_0)]$, and that $\theta^\star-\theta_0=\E[(g^\star-g_0)(\alpha^\star-\alpha_0)]$.

\begin{theorem}[DML for the pipeline target]
\label{thm:pipeline-dml}
Suppose Assumption~\ref{ass:pipeline-dml} in Appendix~\ref{app:pipeline} holds; in particular, $\hat g^{(-k)},\hat\alpha^{(-k)}\in\cH$ and
$\sqrt n\,\norm{\hat g^{(-k)}-g_0}_{P,2}\,\norm{\hat\alpha^{(-k)}-\alpha_0}_{P,2}=o_\Pp(1)$ for every $k$.
Then
\begin{equation}
\sqrt n\bigl(\hat\theta_0-\theta_0\bigr)=\frac1{\sqrt n}\sum_{i=1}^n\bigl\{\psi^\star(O_i;g_0,\alpha_0)-\theta_0\bigr\}+o_\Pp(1),
\label{eq:pipeline-expansion}
\end{equation}
$\hat\theta_0$ is asymptotically normal, and $\hat\sigma_0^2$ is consistent.
\end{theorem}

No stability or differentiability is required of $\hat\phi$, and the fold-wise representations may differ.
The sensitivity analysis of Section~\ref{sec:bounds} applies to $\theta_0$ (Appendix~\ref{app:sens}).

The price is the rate condition: the learners must converge to the projections onto the span $\cH=\cH_\phi$ of all reachable representations.
A pipeline that ends by selecting a single representation converges, at best, to the projection onto that representation's information, which is generally smaller.
Approach~1 therefore suits pipelines that learn over a rich class, such as an encoder trained end to end with its heads, or that aggregate a menu of candidates.

\paragraph{\textbf{Menus and star aggregation.}}
Applied pipelines often produce a finite menu of modalities, checkpoints, prompts, or heads, which star aggregation~\cite{audibert2007progressive} combines.

\begin{definition}[Star estimator]
\label{def:star-estimator}
Let $\mathcal C\subseteq\cH$ be a finite menu and $Q_n(\kappa)=\E_n\ell(O,\kappa)$ an empirical quadratic loss.
Choose the empirical best $\tilde\kappa\in\arg\min_{\kappa\in\mathcal C}Q_n(\kappa)$ and minimize $Q_n$ over the star hull $\{\lambda\tilde\kappa+(1-\lambda)\kappa:\ \kappa\in\mathcal C,\ \lambda\in[0,1]\}$.
\end{definition}

The second step is a one-dimensional quadratic problem along each segment, with the closed forms~\eqref{eq:lambda-reg} and~\eqref{eq:lambda-riesz} of Appendix~\ref{sec:learn-H} for the square and Riesz losses.
The star aggregate stays in $\cH$, and for a menu of size $M$ its excess loss is within $\log M/n$ of the best candidate's (Propositions~\ref{prop:star-finite-reg} and~\ref{prop:star-finite-riesz} in Appendix~\ref{sec:learn-H}), against $\sqrt{\log M/n}$ for selecting that candidate.
Aggregation improves the remainder, not the limit: Theorem~\ref{thm:pooled-star} requires the errors $e_g$ and $e_\alpha$ of the best candidates to vanish at the stated rate, so the outcome menu must contain a candidate approaching $g_\phi=g_0$ and the representer menu one approaching $\alpha_\phi=\alpha_0$.
A convex mixture of an image learner and a text learner is a weighted average of the two predictions, whereas the projection onto $\cH_\phi$ is the sum of a function of the image and a function of the text.
For example, a candidate whose input stacks all the encoders can approximate the projection.
We aggregate the outcome and representer menus separately, so the two aggregates may use different candidates (Section~\ref{sec:empirical}).

\paragraph{\textbf{Cross-fitting the star weights.}}
The candidates must be independent of the data used to aggregate them.
Nested splits satisfy this requirement but multiply the number of fine-tuning runs.
Train each candidate once per training fold, as cross-fitting already requires, and fit one set of star weights to the pooled out-of-fold predictions.
The weights then depend on the evaluation fold, but only through the choice, for each nuisance, of a center, a second candidate, and a scalar weight $\lambda$.
Theorem~\ref{thm:pooled-star} in Appendix~\ref{app:cross-fold} shows that this costs a factor $\sqrt{\log M}$ in one remainder term, and the same holds for convex weights with $\sqrt{M\log n}$.
Pooling also gives every fold the same star weights, so the fold-wise estimates share one target.

\subsection{Approach 2: joint-loss tuning}
\label{sec:joint}
A natural practice is to tune a single representation for both nuisances by minimizing the prediction risk plus a weighted Riesz loss, and in the second approach we give valid inference for this practice.
Let $R(h):=\E[\{Y-h(W)\}^2]$ and $Q(a):=\E[a(W)^2-2m(O,a)]$, with empirical versions $R_{-k}$ and $Q_{-k}$ on $I_k^c$, and take folds of equal size.

\begin{definition}[Joint-loss tuning]
\label{def:joint}
Fix normalizing constants $V_Y,V_\alpha>0$.
For each $k$, compute on $I_k^c$ an approximate minimizer
\begin{equation}
(\hat\phi_k,\hat\chi_k,\hat a_k)\in\arg\min_{\phi,\chi,a}\;
\frac{R_{-k}(\chi\circ W_\phi)}{V_Y}+\frac{Q_{-k}(a\circ W_\phi)}{V_\alpha},
\label{eq:joint-loss}
\end{equation}
set $\hat g_k:=\hat\chi_k\circ W_{\hat\phi_k}$ and $\hat\alpha_k:=\hat a_k\circ W_{\hat\phi_k}$, and let $\hat\psi_k(O):=m(O,\hat g_k)+\hat\alpha_k(W)\{Y-\hat g_k(W)\}$.
The estimate is $\hat\theta:=K^{-1}\sum_k\E_{n,k}\hat\psi_k$, where $\E_{n,k}$ averages over $I_k$, and its variance estimator is $\hat\sigma^2:=K^{-1}\sum_k\E_{n,k}(\hat\psi_k-\hat\theta)^2$.
\end{definition}

\paragraph{\textbf{Normalizing the joint loss.}}
The normalization matters.
By~\eqref{eq:profile-reg} and~\eqref{eq:profile-riesz}, the profiled population criterion equals a constant plus $\mathcal J(\phi):=\norm{\Delta_g(\phi)}_{P,2}^2/V_Y+\norm{\Delta_\alpha(\phi)}_{P,2}^2/V_\alpha$, and
\begin{equation}
|\theta^\star-\theta_\phi|\le\tfrac12\sqrt{V_YV_\alpha}\;\mathcal J(\phi)
\label{eq:profile-bound}
\end{equation}
(Lemma~\ref{lem:profile}).
If $V_Y=V_{Y,\phi}$ and $V_\alpha=V_{\alpha,\phi}$ are the scales of Section~\ref{sec:bounds} at the representation being evaluated, the criterion is $C_Y^2+C_D^2$, which bounds $2C_YC_D$, the factor that scales the bias in~\eqref{eq:bias-factorization}.
In practice, compute $V_Y$ and $V_\alpha$ on the training folds from the untuned representation and hold them fixed; the bound~\eqref{eq:profile-bound} still holds.
Equal weights on the raw losses make the criterion depend on the units of measurement.

Because $\hat g_k$ and $\hat\alpha_k$ are functions of the same $W_{\hat\phi_k}$, the mixed-bias identity~\eqref{eq:mixed-bias} holds exactly at $\theta_{\hat\phi_k}$, and the natural center of $\hat\theta$ is the fold-averaged target
\begin{equation*}
\bar\theta_n:=\frac1K\sum_{k=1}^K\theta_{\hat\phi_k}.
\end{equation*}
Centering there removes the drift bounded in~\eqref{eq:drift}, so the tuning needs no rate.

\begin{proposition}[DML around the fold-averaged target]
\label{prop:joint}
Suppose Assumption~\ref{ass:joint} in Appendix~\ref{app:joint} holds; in particular, each head attains the product rate around the best predictor given its own fold's representation, and the fitted scores converge in $L^2(P)$, at any rate, to a nonrandom $\psi_n^\ast$ with $\sigma_n^2:=\Var\psi_n^\ast$ bounded away from zero.
Then $\sqrt n(\hat\theta-\bar\theta_n)/\sigma_n\Rightarrow N(0,1)$, and the cross-fitted variance estimator $\hat\sigma^2$ satisfies $\hat\sigma^2/\sigma_n^2\to1$ in probability.
\end{proposition}

The target $\bar\theta_n$ is data-dependent but interpretable.
If $\sqrt n\,\norm{\Delta_g(\hat\phi_k)}_{P,2}\norm{\Delta_\alpha(\hat\phi_k)}_{P,2}=o_\Pp(1)$ for every $k$, it is within $o_\Pp(n^{-1/2})$ of $\theta^\star$ by~\eqref{eq:cs-bound}; otherwise we bound $\theta^\star$ by a sensitivity analysis with the scale averaged over folds (Corollary~\ref{cor:joint}).
The convergence condition fails when the criterion is flat, because the folds may then select representations that predict differently.
To check it, compare the fitted functions and the estimates across folds.

\paragraph{\textbf{Pooled selection across folds.}}
With a menu of recipes $j=1,\dots,M$, where a recipe is a candidate tuning procedure (for example, an encoder or a checkpoint) applied fold by fold, we remove the disagreement across folds by selecting one recipe for all folds with the joint loss.
Train every recipe on every $I_k^c$ by~\eqref{eq:joint-loss}, which gives $(\hat\phi_{jk},\hat g_{jk},\hat\alpha_{jk})$, and select once, by pooled out-of-fold joint loss,
\begin{equation}
\hat\jmath\in\arg\min_{1\le j\le M}\bar J_n(j),
\qquad
\bar J_n(j):=\frac1n\sum_{k=1}^K\sum_{i\in I_k}\Bigl[\frac{\{Y_i-\hat g_{jk}(W_i)\}^2}{V_Y}+\frac{\hat\alpha_{jk}(W_i)^2-2m(O_i,\hat\alpha_{jk})}{V_\alpha}\Bigr].
\label{eq:pooled-joint}
\end{equation}
Every fold then uses recipe $\hat\jmath$.
Write $\hat\theta_j$, $\hat\sigma_j^2$, and $\bar\theta_n(j):=K^{-1}\sum_k\theta_{\hat\phi_{jk}}$ for the estimate, variance estimate, and fold-averaged target of recipe $j$; for a frozen recipe, $\bar\theta_n(j)=\theta_{\phi_j}$.
The selection uses every fold, so the reported $\hat\theta_{\hat\jmath}$ is not cross-fitted in the strict sense.
Let $\hat\Omega$ be the empirical correlation matrix of the $M$ cross-fitted score vectors, and let $\hat c_{1-\gamma}$ be the $(1-\gamma)$-quantile of $\max_j|\xi_j|$ for $\xi\sim N(0,\hat\Omega)$.

\begin{theorem}[Pooled joint-loss selection]
\label{thm:pooled-joint}
Let $M$ be fixed, and suppose Assumption~\ref{ass:joint} holds for every recipe $j$, with limits $\psi_{n,j}^\ast$ whose variances are bounded away from zero.
\begin{enumerate}[label=(\roman*),leftmargin=*]
\item For any selection rule $\hat\jmath$, including~\eqref{eq:pooled-joint},
\begin{equation*}
\liminf_{n\to\infty}\Pp\Bigl\{\bar\theta_n(\hat\jmath)\in\bigl[\hat\theta_{\hat\jmath}\pm\hat c_{1-\gamma}\,\hat\sigma_{\hat\jmath}/\sqrt n\bigr]\Bigr\}\ge1-\gamma.
\end{equation*}
\item If $\Pp\{\hat\jmath=j_n^\ast\}\to1$ for nonrandom $j_n^\ast$, the usual interval $[\hat\theta_{\hat\jmath}\pm z_{1-\gamma/2}\hat\sigma_{\hat\jmath}/\sqrt n]$ covers $\bar\theta_n(\hat\jmath)$ with probability tending to $1-\gamma$.
\item Under the conditions of Lemma~\ref{lem:pooled-margin} in Appendix~\ref{app:pooled-joint}, the rule~\eqref{eq:pooled-joint} satisfies the condition in (ii). These conditions require that the fitted heads of each recipe converge in $L^2(P)$ to nonrandom limits and stay bounded in $L^4(P)$, that $Y\in L^4(P)$ and $m$ is mean-square continuous, and that the population joint loss of the limits has a unique minimizer $j_n^\ast$, separated from the other recipes by a margin bounded away from zero.
\end{enumerate}
\end{theorem}

Part~(i) needs no condition on the selection; its critical value is at most the Bonferroni value $z_{1-\gamma/(2M)}$, which grows like $\sqrt{2\log M}$, the analogue of the $\sqrt{\log M}$ that pooling costs in Approach~1.
When one recipe is selected with probability tending to one, as in part~(ii), the usual critical value $z_{1-\gamma/2}$ suffices.
Report $\hat c_{1-\gamma}$ alongside $z_{1-\gamma/2}$; if they agree, the selection does not matter.

\paragraph{\textbf{Star aggregation and Approach 2.}}
We do not use star aggregation in Approach~2, because it mixes predictions from different representations, separately for the two nuisances.
The mixed heads are functions of the union of the representations, but they are not close to the regression and representer of that union, so the mixed-bias identity at $\theta_{\hat\phi_k}$ fails and inference again needs the rate condition of Approach~1.
Approach~2 centers at the target of the selected representation and does not use $\cH$.
To combine candidates, combine representations rather than predictions, as Approach~3 does.

\subsection{Approach 3: double representation selection}
\label{sec:design}
The third approach imposes no condition on the tuning, at the cost of a held-out sample.
Its logic is that of post-double-selection~\cite{belloni2014inference}: there, one selects controls that predict the outcome, then controls that predict the treatment, and regresses on the union.
Here one selects a representation for the outcome regression and another for the representer, and runs DML on their union.
In both cases the bias is a product of two errors, and each selection controls one factor.
Post-double-selection needs no split, because sparsity bounds the complexity of the selected model; for fine-tuned encoders we have no complexity bound, so we tune them on a separate sample.

For $\phi=(\phi_g,\phi_\alpha)$, set $X_\phi:=(x(Z,\phi_g),x(Z,\phi_\alpha))$, so $g_\phi,\alpha_\phi,\theta_\phi,S_\phi$ refer to the union.

\begin{definition}[Double representation selection]
\label{def:sfu}
Split the sample at random into a tuning sample $\mathcal S_0$ of size $n_0$ and an inference sample $\mathcal S_1$ of size $n_1$.
\begin{enumerate}[label=(\roman*),leftmargin=*]
\item \emph{Tune.} On $\mathcal S_0$ alone, choose $\hat\phi_g$ for the outcome regression and $\hat\phi_\alpha$ for the representer, by any procedure: fine-tuning, prompt search, a menu with star aggregation, or inspection.
\item \emph{Freeze the union.} Set $\hat\phi:=(\hat\phi_g,\hat\phi_\alpha)$.
\item \emph{Infer.} On $\mathcal S_1$ alone, run cross-fitted DML and the sensitivity analysis of Sections~\ref{sec:fixed}--\ref{sec:bounds} with the frozen representation $X_{\hat\phi}$; this gives the estimate $\hat\theta$ and variance estimate $\hat\sigma^2$ of~\eqref{eq:dml-est} and~\eqref{eq:var-est} computed on $\mathcal S_1$.
\end{enumerate}
\end{definition}

\begin{lemma}[{Accuracy of the union}]
\label{lem:union}
For $\phi=(\phi_g,\phi_\alpha)$,
\begin{equation*}
\norm{g^\star-g_\phi}_{P,2}\le\norm{g^\star-g_{\phi_g}}_{P,2},
\qquad
\norm{\alpha^\star-\alpha_\phi}_{P,2}\le\norm{\alpha^\star-\alpha_{\phi_\alpha}}_{P,2},
\end{equation*}
so $|\theta^\star-\theta_\phi|\le\norm{g^\star-g_{\phi_g}}_{P,2}\norm{\alpha^\star-\alpha_{\phi_\alpha}}_{P,2}$; moreover $C_Y(\phi)\le C_Y(\phi_g)$ and $C_D(\phi)\le C_D(\phi_\alpha)$.
\end{lemma}

By Lemma~\ref{lem:union}, each embedding can be tuned for its own nuisance.
By~\eqref{eq:profile-reg} and~\eqref{eq:profile-riesz}, prediction risk is the right criterion for $\hat\phi_g$ and Riesz loss the right one for $\hat\phi_\alpha$, but nothing in the inference relies on how the tuning was done.

\begin{proposition}[Validity conditional on the tuning sample]
\label{prop:sfu}
Suppose Assumption~\ref{ass:sfu} in Appendix~\ref{app:sfu} holds, which requires that, on a set of tuning samples whose probability tends to one, the frozen representation and the learners on $\mathcal S_1$ satisfy the conditions of Theorem~\ref{thm:dml-expansion} with $n_1$ in place of $n$.
Then $[\hat\theta\pm z_{1-\gamma/2}\hat\sigma/\sqrt{n_1}]$ covers $\theta_{\hat\phi}$ with probability tending to $1-\gamma$, conditionally on $\mathcal S_0$ uniformly over those samples and unconditionally.
The one-sided bounds of Proposition~\ref{prop:onesided}, computed on $\mathcal S_1$, satisfy~\eqref{eq:onesided-target} for every scenario that is valid for the union.
\end{proposition}

The target $\theta_{\hat\phi}$ depends on $\mathcal S_0$ and is fixed once $\mathcal S_0$ is.
The interval also covers $\theta^\star$ if $\sqrt{n_1}\norm{g^\star-g_{\hat\phi_g}}_{P,2}\norm{\alpha^\star-\alpha_{\hat\phi_\alpha}}_{P,2}=o_\Pp(1)$ (Corollary~\ref{cor:adaptive}).
The standard error grows by a factor of roughly $\sqrt{n/n_1}$, which stays between $1.1$ and $1.4$ for a tuning share between one fifth and one half.
Swapping the roles of $\mathcal S_0$ and $\mathcal S_1$ and averaging would use all the data, but the average is valid only if the scores fitted on the two halves share a limit, as in Assumption~\ref{ass:joint}; we therefore recommend one split for inference and repeated splits only as a diagnostic of unstable tuning.
To adapt the deletion benchmark of Section~\ref{sec:calibration}, drop either half of the union within $\mathcal S_1$; the resulting changes in the losses show how much each nuisance relies on information that the other embedding lacks.

\subsection{Choosing an approach}
\label{sec:regimes}
Table~\ref{tab:routes} sets the three approaches side by side.
Approach~1 suits end-to-end learners and menus, where the pipeline can approach the span of its candidates.
Approach~2 suits moderate samples and light, stable fine-tuning from a common starting point. Approach~3 suits large samples and heavy or unstable tuning.

\begin{table}[tbp]
\centering
\caption{Three approaches to inference with a tuned representation}
\label{tab:routes}
\small
\begin{tabular}{@{}>{\raggedright\arraybackslash}p{0.14\textwidth}>{\raggedright\arraybackslash}p{0.25\textwidth}>{\raggedright\arraybackslash}p{0.25\textwidth}>{\raggedright\arraybackslash}p{0.25\textwidth}@{}}
\toprule
& Approach 1: fixed-menu fine-tuning & Approach 2: joint-loss tuning & Approach 3: double representation selection \\
\midrule
Choice of $\phi$ & Fixed in advance & In each training fold & Once on $\mathcal S_0$, then frozen \\
Target & $\theta_0=\theta_\phi$, defined by $\cH$ & $\bar\theta_n$, the fold-averaged target & $\theta_{\hat\phi}$, defined by the frozen union \\
Tuning & Own loss per nuisance & Shared across nuisances & Any procedure \\
Menus & Star aggregation, pooled across folds & Selection by pooled joint loss & Any rule, including star aggregation \\
Condition on tuning & Learners reach the projections onto $\cH$ at product rates & Fitted scores converge, at any rate & None \\
Inference sample & All $n$ & All $n$ & $n_1$ \\
\bottomrule
\end{tabular}
\end{table}

\section{Illustration: multimodal demand}
\label{sec:empirical}

We apply the workflow to product demand, where text, images, and conventional attributes describe the same good in different ways.
The data, scraped from Keepa.com,\footnote{\url{https://keepa.com/}} cover clothing and shoes sold on Amazon and combine price histories with three sources of product information: text, product images, and tabular attributes.
Seven combinations of these modalities give seven information sets, which in the language of Section~\ref{sec:tuning} form a menu of representations.
We first treat the menu as frozen and combine it by star aggregation as in Approach~1, the baseline, and then fine-tune the text and image encoders on the sample and run all three approaches.
The estimated elasticity of the demand proxy is about $-1.1$ and stable across the seven information sets and the three approaches, and the sensitivity analysis asks whether omitted information could move it above the threshold $-0.605$ at which demand ceases to be elastic (Section~\ref{sec:empirical-specification}).

\subsection{Data, design, and estimand}
\label{sec:empirical-specification}

The panel has $626{,}919$ weekly price and rank changes for $8{,}642$ products, from 6 February 2023 to 28 July 2025. The estimation sample consists of the $186{,}335$ changes with a nonzero price change, which belong to $7{,}725$ products.
Let $R_{it}$ be the sales rank of product $i$ at date $t$.
A smaller rank means stronger sales, so we take $Y_{it}:=-\log R_{it}$ as the demand proxy, and we let $D_{it}$ be log price.
Write
\begin{equation*}
\Delta Y_i:=Y_{it}-Y_{i,t-1},
\qquad
\Delta D_i:=D_{it}-D_{i,t-1}
\end{equation*}
for adjacent changes, indexing product-period pairs by $i$. We exclude weeks with $\Delta D_i=0$ because the population Riesz loss has no lower bound when the features identify products with structurally constant prices (Section~\ref{sec:empirical-tuned}); Appendix~\ref{app:empirical-allrows} reports the estimates on all weekly changes.

We project the first-difference outcome onto functions that are linear in the price change, with coefficients varying with the information set:
\begin{equation}
\Delta Y_i=g_j(X_i^{(j)},\Delta D_i)+\varepsilon_i,
\qquad
\E[\varepsilon_i\mid X_i^{(j)}]=\E[\Delta D_i\varepsilon_i\mid X_i^{(j)}]=0,
\label{eq:empirical-fd-model}
\end{equation}
where $X_i^{(j)}$ is one of seven information sets, $j\in\{\mathrm{txt},\allowbreak\mathrm{img},\allowbreak\mathrm{tab},\allowbreak\mathrm{txt{+}img},\allowbreak\mathrm{txt{+}tab},\allowbreak\mathrm{img{+}tab},\allowbreak\mathrm{full}\}$.
Text and image features are embeddings from EmbeddingGemma~2~\cite{vera2026embeddinggemma}, a single model for text and images. For the text-plus-image set it embeds the text and the image together as one input, and ``full'' adds the tabular attributes to that embedding.

In the notation of Section~\ref{sec:warmup}, the observation is $O_i=(\Delta Y_i,\Delta D_i,Z_i)$, where $Z_i$ collects the raw product information, and $X_i^{(j)}=x_j(Z_i)$.
The functional is the average derivative, $m(O,h)=\partial_d h(d,Z)|_{d=\Delta D}$, so the target for information set $j$ is
\begin{equation}
\theta_j
:=
\E\Bigl[\partial_d\, g_j(X_i^{(j)},d)\big|_{d=\Delta D_i}\Bigr],
\label{eq:empirical-apd}
\end{equation}
the restricted representation-based target $\theta_{\phi_j}$.
When $\Var(\Delta D\mid X^{(j)})$ is bounded away from zero, Appendix~\ref{app:avg-deriv} shows that
\[
\theta_j=\E\!\left[\frac{\mathrm{Cov}(\Delta D,\Delta Y\mid X^{(j)})}{\Var(\Delta D\mid X^{(j)})}\right].
\]
This is the average conditional least-squares slope. It equals the average derivative of the conditional mean when that mean is linear in the price change; otherwise it has the weighted-derivative interpretation only under the additional regularity conditions in the appendix.
Both variables are in logarithms, so $\theta_j$ is an elasticity of the demand proxy.
A conventional demand elasticity follows from a power law between sales rank and sales, $\log R=c-\kappa\log(\mathrm{sales})$, which Chevalier and Goolsbee~\cite{chevalier2003measuring} use to recover demand elasticities from the sales ranks of online books.
Under it $Y=-\log R=\kappa\log(\mathrm{sales})-c$, so $\theta_j$ is $\kappa$ times the elasticity of demand, and demand is elastic, with elasticity below $-1$, exactly when $\theta_j<-\kappa$.
He and Hollenbeck~\cite{he2020sales} estimate $\kappa$ for seventeen Amazon categories from inventories of listed products; for Clothing, Shoes, and Jewelry they find $\hat\kappa=0.605$ (standard error $0.091$), with a 5 to 95\% band from $0.45$ to $0.75$.
We use $\kappa=0.605$, so that demand elasticities are $1/0.605\approx1.65$ times the estimates and the elastic threshold is $\theta_j<-0.605$.
% TODO(kappa): He and Hollenbeck (2020, SSRN 3728281) is a preliminary draft; its sales come from inventory differences of products with fewer than 1000 units in stock (233 clothing observations, R^2 = 0.44); Chevalier and Goolsbee (2003) estimate kappa of about 1.2 for online books (range 0.9 to 1.3). Replace by the published version when there is one; about 78 percent of our product-weeks lie inside the 5th to 95th percentile range of ranks in their data.
The full information set is the richest \emph{observed} one; $\theta^\star$ still denotes the target that uses all relevant information.

Suppose that price changes are mean-independent of unobserved demand changes given $X_i^{(j)}$.
If the conditional causal response is also linear in the price change, $\theta_j$ is its average slope. For nonlinear responses, a causal interpretation requires the corresponding weighted-response argument of Appendix~\ref{app:avg-deriv}.
This is the continuous-treatment analogue of conditional parallel trends: differencing absorbs fixed product quality, and the multimodal controls absorb observed sources of heterogeneous demand trends.
Section~\ref{sec:empirical-sensitivity} asks how much time-varying shocks or promotions that the controls miss could move the conclusion.

\subsection{Estimation and comparison of approaches}
\label{sec:empirical-learners}

Let $b_j(X_i^{(j)})$ be the learned feature vector for information set $j$, the five-dimensional output of a ReLU network with two hidden layers applied to the text and image representations and the tabular attributes in the set.
Both nuisances use the same varying-coefficient architecture,
\begin{align}
g_j(X_i^{(j)},\Delta D_i)
&=
\Delta D_i\,\beta_j^\top b_j(X_i^{(j)})
+\delta_j^\top b_j(X_i^{(j)}),
\label{eq:empirical-outcome-learner}
\\
\alpha_j(X_i^{(j)},\Delta D_i)
&=
\Delta D_i\,\beta_j^\top b_j(X_i^{(j)})
+\eta_j^\top b_j(X_i^{(j)}),
\label{eq:empirical-riesz-learner}
\end{align}
but with separate coefficients; each learner fits its own network. These heads approximate the restricted population projections above. Consistency for those projections is an assumption; if both learners instead target a smaller common closed linear space, the estimand is its pipeline target as in Appendix~\ref{app:avg-deriv}.
The outcome learner minimizes squared error, and the Riesz learner minimizes
\begin{equation}
\E\bigl[\alpha_j(X_i^{(j)},\Delta D_i)^2\bigr]
-2\,\E\Bigl[\partial_d\,\alpha_j(X_i^{(j)},d)\big|_{d=\Delta D_i}\Bigr].
\label{eq:empirical-riesz-loss}
\end{equation}
We add the penalty $\lambda\,\E[\{\partial_d\,\alpha_j(X_i^{(j)},d)|_{d=\Delta D_i}\}^2]$ to this loss, with $\lambda=3\times10^{-4}$ fixed before estimation. We use the penalty for training and selection only, and we compute the scale $\hat S_\phi$ and the gains of Section~\ref{sec:calibration} from the loss without it, so that both describe the same representer. The debiased score is
\begin{equation}
\psi_{ij}
=
\partial_d\,\hat g_j(X_i^{(j)},d)\big|_{d=\Delta D_i}
+\hat\alpha_j(X_i^{(j)},\Delta D_i)
\bigl\{\Delta Y_i-\hat g_j(X_i^{(j)},\Delta D_i)\bigr\}.
\label{eq:empirical-dr-score}
\end{equation}
We cross-fit with $K=5$ folds grouped by product, so that all changes of a product fall in the same fold; on each training sample, a fifth of the products is held out to stop training early.
Every observation is evaluated once, by learners that never saw its product, and $\hat\theta_j$ averages the pooled out-of-fold scores.
We fit the star centers and weights once to the pooled out-of-fold predictions, following the construction of Theorem~\ref{thm:pooled-star}.
Standard errors are clustered by product.

The frozen baseline holds the embeddings fixed. We then compare three ways of tuning them: separate nuisance pipelines within each training fold (Approach~1), a shared encoder fitted by the joint loss (Approach~2), and selection on an independent tuning sample followed by inference on a frozen union (Approach~3).
The last approach reserves 30\% of products for tuning and uses the remaining 70\% for inference.
Appendix~\ref{app:empirical-details} records the encoders, training settings, aggregation weights, and sample-splitting details. It also reports estimates on all weekly changes and compares the main estimates with those from RoBERTa and the Vision Transformer with $K=2$ folds.

\FloatBarrier
\subsection{Baseline estimates and aggregation}
\label{sec:empirical-results}

Panel A of Table~\ref{tab:empirical-modalities} reports the frozen-encoder estimates. All seven are negative and precise, and every 95\% interval lies below $-1$.
The estimates range from $-1.167$ for tabular attributes alone to $-1.137$ for text alone; the full information set gives $-1.143$ with a standard error of $0.018$.
Thus, the choice among the observed information sets has little effect on the estimated response.

\begin{table}[!htbp]
\centering
\setstretch{1}
\caption{Price-response estimates by information set}
\label{tab:empirical-modalities}
\small
\renewcommand{\arraystretch}{1.13}
\begin{tabular*}{\textwidth}{@{\extracolsep{\fill}}lrrrr@{}}
\toprule
\multicolumn{5}{@{}l}{\textit{Panel A. Frozen encoders}}\\[3pt]
& & & \multicolumn{2}{c}{95\% confidence interval}\\
\cmidrule(lr){4-5}
Information set & Estimate & SE & Lower & Upper\\
\midrule
Text & $-1.137$ & $0.018$ & $-1.172$ & $-1.102$\\
Image & $-1.163$ & $0.019$ & $-1.200$ & $-1.127$\\
Tabular & $-1.167$ & $0.019$ & $-1.203$ & $-1.131$\\
Text $+$ image & $-1.148$ & $0.018$ & $-1.183$ & $-1.113$\\
Text $+$ tabular & $-1.142$ & $0.018$ & $-1.179$ & $-1.106$\\
Image $+$ tabular & $-1.162$ & $0.018$ & $-1.198$ & $-1.126$\\
Full multimodal & $-1.143$ & $0.018$ & $-1.179$ & $-1.107$\\
\addlinespace[4pt]
Star aggregate & $-1.142$ & $0.017$ & $-1.176$ & $-1.108$\\
\midrule
\multicolumn{5}{@{}l}{\textit{Panel B. Tuned encoders}}\\[3pt]
& \multicolumn{2}{c}{Separate pipelines (1)} & \multicolumn{2}{c}{Joint loss (2)}\\
\cmidrule(lr){2-3}\cmidrule(lr){4-5}
Information set & Estimate & SE & Estimate & SE\\
\midrule
Text & $-1.128$ & $0.018$ & $-1.144$ & $0.018$\\
Image & $-1.147$ & $0.018$ & $-1.164$ & $0.019$\\
Tabular & $-1.167$ & $0.019$ & $-1.169$ & $0.019$\\
Text $+$ image & $-1.133$ & $0.018$ & $-1.140$ & $0.019$\\
Text $+$ tabular & $-1.123$ & $0.019$ & $-1.129$ & $0.019$\\
Image $+$ tabular & $-1.142$ & $0.018$ & $-1.152$ & $0.019$\\
Full multimodal & $-1.129$ & $0.018$ & $-1.141$ & $0.019$\\
\addlinespace[4pt]
Star aggregate & $-1.129$ & $0.018$ & \multicolumn{2}{c}{---}\\
\bottomrule
\end{tabular*}
\empnotes{The outcome is the change in minus log sales rank; the treatment is the change in log price. Standard errors (SE) are clustered by product. The sample contains 186,335 weekly changes with a nonzero price change for 7,725 products. Numbers in parentheses in the panel-B headings identify the approaches. Star aggregation is used in Approach~1. Frozen convex aggregation gives $-1.141$ (SE $0.017$). The demand-elasticity threshold on this scale is $-0.605$.}
\end{table}

Star aggregation chooses separate combinations for the outcome regression and the representer. In the frozen baseline it combines text-plus-tabular and text outcome learners, and text-plus-image and tabular representer learners.
Both aggregates lie in the span $\cH$ of the seven pipelines, so their debiased estimator targets $\theta_0$ (Theorem~\ref{thm:pipeline-dml}).
The estimate is $-1.142$, with standard error $0.017$ and 95\% interval $[-1.176,-1.108]$.
The exact weights are reported in Appendix~\ref{app:empirical-details}.

A 1\% price increase lowers the inverse-rank demand index by about 1.1\%.
Under the power-law conversion of Section~\ref{sec:empirical-specification}, this corresponds to a demand elasticity of about $-1.9$, or $-1.5$ to $-2.5$ over the reported band for $\kappa$.

\subsection{Tuned representations}
\label{sec:empirical-tuned}

Panel B of Table~\ref{tab:empirical-modalities} compares the tuned information sets.
With tuning, the clustered standard errors stay between $0.018$ and $0.019$, as in the frozen baseline.
Every interval remains far below $-0.605$ and below $-1$, with the highest upper endpoint, $-1.085$, for text plus tabular attributes in Approach~1.
For Approach~1, the text-plus-tabular set remains the star center for the outcome regression and the full set becomes the center for the representer; the aggregate estimate is $-1.129$, compared with $-1.142$ in the frozen baseline.
The outcome learners' out-of-fold $R^2$ rises from $0.120$--$0.129$ to $0.121$--$0.135$, reaching $0.1348$ for text plus tabular attributes and $0.1341$ for the full set.

In Approach~2, text plus tabular attributes minimizes the pooled joint loss~\eqref{eq:pooled-joint}, giving $\hat\theta=-1.129$ with standard error $0.019$, but the next recipe, text alone, trails by only $0.007$ (clustered standard error $0.018$, $t=0.40$), too small a difference to rely on the margin condition of Theorem~\ref{thm:pooled-joint}(iii).
We therefore report the max-$t$ interval of part~(i), $[-1.171,-1.086]$, for the fold-averaged target of the selected recipe; the usual interval is $[-1.165,-1.092]$ and the Bonferroni interval $[-1.179,-1.078]$.
The five fold models agree on $\hat g$ (mean pairwise correlation at least $0.979$ for every recipe).

In Approach~3, we select on $\mathcal S_0$ the text encoder for the outcome regression (validation $R^2=0.1196$) and the image encoder with the tabular attributes for the representer (validation Riesz loss $-48.94$); their frozen union, with the tabular attributes entering once, has $1{,}564$ features.
Cross-fitting the heads on $\mathcal S_1$ gives $\hat\theta=-1.146$, with clustered standard error $0.022$ and 95\% interval $[-1.189,-1.103]$, for the target of the frozen union.
Inference on $\mathcal S_1$ alone inflates the standard error by about $\sqrt{n/n_1}=1.20$.
Table~\ref{tab:empirical-route3} in Appendix~\ref{app:empirical-details} reports the sample split and validation results.

\FloatBarrier
\paragraph{\textit{Riesz-learner diagnostics.}}
If a set of products has structurally constant prices and the features identify that set, then $\Var(\Delta D\mid X)=0$ there and the population Riesz loss has no lower bound. The estimation sample excludes weeks without a price change, and with the penalty of Section~\ref{sec:empirical-learners} the Riesz loss is bounded below by $-1/\lambda$; no out-of-fold Riesz loss in Approaches~1 and~2 is positive.
In Approach~1, $\E[\beta^2]$, the mean square of the representer's derivative in the price, is at most $1.24$ times its value for the frozen learner of the same information set, and the largest fitted value of $\lambda\beta$ is $0.14$.
In Approach~3, the image encoder with the tabular attributes has the lowest validation Riesz loss on $\mathcal S_0$ and is selected for the representer; its $\E[\beta^2]=4{,}830$ is the largest among the seven candidates. The Approach~2 fold models agree less on $\hat\alpha$ for recipes with text ($0.924$ to $0.948$, against $0.953$ to $0.992$ for the others).
The point estimates move little: every tuned candidate lies within $0.027$ of the frozen star estimate.

\subsection{Sensitivity to omitted information}
\label{sec:empirical-sensitivity}

We calibrate omitted-information strength by deleting images and tabular attributes from the full information set, leaving text alone. Table~\ref{tab:empirical-calibration} separates the change in predictive fit from the implied sensitivity parameters.
In the frozen baseline, the deletion barely changes outcome fit ($\widehat C_Y^2=0.0005$), and the text-only representer attains a lower out-of-fold Riesz loss than the full one, so that $\widehat G_D<0$ and $\widehat C_D^2$ is undefined.
In Approach~1 the full set improves both fits; the implied shares are $\widehat C_Y^2=0.0017$ and $1-\widehat R^2_{\alpha^\star\sim\alpha_\phi}=0.018$, so that $\widehat C_D^2=0.019$, with alignment $\hat\rho=-0.01$, and the small outcome contrast limits the implied approximation bias.

Joint-loss tuning raises the outcome contrast to $\widehat C_Y^2=0.0038$, but the text-only representer again has the lower Riesz loss, so the benchmark is undefined.
In the frozen baseline and in Approach~2 the alignment $\hat\rho$ is undefined as well, so we calibrate the sensitivity analysis with the Approach~1 benchmark.
The Approach~3 deletion results appear in Appendix Table~\ref{tab:empirical-route3}; Appendix~\ref{app:empirical-details} also gives the calibration calculations.

\begin{table}[!htbp]
\centering
\setstretch{1}
\caption{Sensitivity calibration from omitted modalities}
\label{tab:empirical-calibration}
\small
\renewcommand{\arraystretch}{1.13}
\begin{tabular*}{\textwidth}{@{\extracolsep{\fill}}lrrr@{}}
\toprule
& Frozen & Separate pipelines & Joint loss\\
& & Approach 1 & Approach 2\\
\midrule
\multicolumn{4}{@{}l}{\textit{Panel A. Out-of-fold fit}}\\[3pt]
Outcome $R^2$, full & $0.1288$ & $0.1341$ & $0.1308$\\
Outcome $R^2$, text only & $0.1284$ & $0.1327$ & $0.1275$\\
Riesz loss, full & $-41.73$ & $-42.71$ & $-42.27$\\
Riesz loss, text only & $-42.13$ & $-41.95$ & $-42.45$\\
\midrule
\multicolumn{4}{@{}l}{\textit{Panel B. Benchmark with $k=1$}}\\[3pt]
Outcome gain, $\widehat G_Y$ & $0.0005$ & $0.0017$ & $0.0038$\\
Representer gain, $\widehat G_D$ & $-0.009$ & $0.018$ & $-0.004$\\
Representer strength, $\widehat C_D^2$ & --- & $0.019$ & ---\\
Alignment, $\hat\rho$ & --- & $-0.01$ & ---\\
\bottomrule
\end{tabular*}
\empnotes{The deletion removes images and tabular attributes from the full information set, leaving text alone. The Riesz losses exclude the penalty of Section~\ref{sec:empirical-learners}. Panel B uses $k=1$ in Section~\ref{sec:calibration}, so that $\widehat C_Y^2=\widehat G_Y$, $1-\widehat R^2_{\alpha^\star\sim\alpha_\phi}=\widehat G_D$, and $\widehat C_D^2=\widehat G_D/(1-\widehat G_D)$. Dashes indicate undefined quantities because the full-set Riesz loss exceeds the text-only loss, so that $\widehat G_D<0$.}
\end{table}

The deletion exercise treats the full information set as long; the sensitivity analysis treats it as short.
In Figure~\ref{fig:empirical-sensitivity}, we center the region at the Approach~1 estimate for the full set, $\hat\theta_{\mathrm{full}}=-1.129$, whose target is the pipeline target of the fine-tuned full set alone; the same analysis applies to the star aggregate, whose target is the pipeline target $\theta_0$ of all seven sets (Appendix~\ref{app:sens}).
The scenarios set both shares, $C_Y^2$ and $1-R^2_{\alpha^\star\sim\alpha_\phi}$, to $r^2$ at $1\%$, $2.5\%$, and $5\%$, with $\rho\in\{1,0.5,0.2\}$.
Table~\ref{tab:empirical-sensitivity} in Appendix~\ref{app:sens-grid} lists the nine rows, with the estimated endpoints $\hat\theta^\pm$ of~\eqref{eq:bounds} and the one-sided bounds $L$ and $U$ of~\eqref{eq:one-sided} with $\gamma=0.10$. The standard error of each endpoint is evaluated at the value of $c$ of its scenario.

\begin{figure}[!htbp]
\centering
\setstretch{1}
\vspace{8pt}
\includegraphics[width=\textwidth]{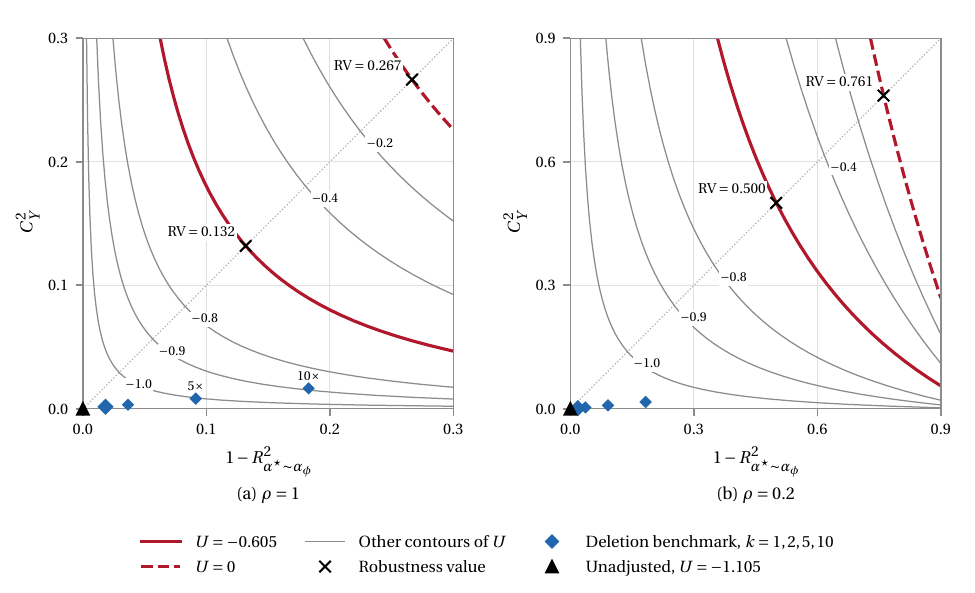}
\caption[Sensitivity of the Approach~1 full-model estimate]{Sensitivity of the Approach~1 full-model estimate}
\label{fig:empirical-sensitivity}
\empnotes{The estimate is $-1.129$, with $\hat S_\phi=3.51$. Each panel fixes $\rho$ and shows contours of the upper bound $U$ of~\eqref{eq:one-sided} with $\gamma=0.10$ over the two shares, with the standard error evaluated in each scenario; at the origin, without omitted information, $U=-1.105$. The benchmark points lie at $C_Y^2=k\widehat G_Y$ and $1-R^2_{\alpha^\star\sim\alpha_\phi}=k\widehat G_D$ for $k=1$, $2$, $5$, and $10$.}
\end{figure}

\FloatBarrier
We assess robustness by asking whether the upper confidence bound remains below $-0.605$, the elasticity threshold on the scale of Section~\ref{sec:empirical-specification}. Every interval in Table~\ref{tab:empirical-modalities} and every one-sided bound $U$ in Table~\ref{tab:empirical-sensitivity} lies below $-0.605$; even the strong scenario with perfect alignment gives $U=-0.924$.
The contour $U=-0.605$ is far from the benchmark, since with $\rho=1$ and equal shares the upper bound reaches the threshold only when the omitted information accounts for $13.2\%$ of the residual variation of the outcome and of the second moment of $\alpha^\star$, which is the robustness value, and on the ray through the benchmark only when both gains of the deletion are multiplied by $k=20.4$; with $\rho=0.2$ the robustness value is $0.500$ and the multiple is $47.3$.

The magnitude is less robust, since with $\rho=1$ moving from the mild to the strong scenario widens the region from $[-1.164,-1.093]$ to $[-1.309,-0.949]$.
We report the robustness values for the other targets in Table~\ref{tab:empirical-frontier}.

The conclusion also holds for the other targets, where the robustness values at perfect alignment lie between $0.131$ and $0.137$, and every one-sided bound in their nine-scenario grids stays below $-0.605$.
For the Approach~1 star aggregate, the benchmark multiple is $20.3$, against $20.4$ for the full set.

\begin{table}[!htbp]
\centering
\setstretch{1}
\caption{Sensitivity required to cross the elasticity threshold}
\label{tab:empirical-frontier}
\small
\renewcommand{\arraystretch}{1.16}
\begin{tabular*}{\textwidth}{@{\extracolsep{\fill}}lrrrr@{}}
\toprule
& \multicolumn{2}{c}{Robustness value at $U=-0.605$} & Benchmark & Strong\\
\cmidrule(lr){2-3}
Target & $\rho=1$ & $\rho=0.2$ & multiple$^{a}$ & scenario $U^{b}$\\
\midrule
Frozen, full & $0.136$ & $0.513$ & --- & $-0.940$\\
Frozen, star & $0.135$ & $0.509$ & --- & $-0.938$\\
\addlinespace[4pt]
Approach 1, star & $0.131$ & $0.498$ & $20.3$ & $-0.924$\\
Approach 1, full & $0.132$ & $0.500$ & $20.4$ & $-0.924$\\
Approach 2, selected text $+$ tabular & $0.132$ & $0.500$ & --- & $-0.924$\\
Approach 3, union & $0.137$ & $0.515$ & --- & $-0.940$\\
\bottomrule
\end{tabular*}
\empnotes{The threshold $-0.605$ corresponds to unit demand elasticity under the rank-to-sales conversion. The robustness value is the common value of $C_Y^2$ and $1-R^2_{\alpha^\star\sim\alpha_\phi}$ at which $U$ reaches the threshold. $^{a}$The multiple is the smallest $k$ for which $U$ reaches the threshold at $C_Y^2=k\widehat G_Y$ and $1-R^2_{\alpha^\star\sim\alpha_\phi}=k\widehat G_D$, with $\rho=1$, and is shown only where the deletion benchmark is defined. At $\rho=0.2$ the multiples are $47.2$ and $47.3$ for the Approach~1 star and full targets, and at the alignment of the deletion, $|\hat\rho|=0.01$, both are $54.7$, where $k\widehat G_D$ is close to one. $^{b}$The strong scenario sets $C_Y^2=1-R^2_{\alpha^\star\sim\alpha_\phi}=0.050$ and $\rho=1$; the selection-adjusted max-$t$ bound for Approach~2 is $-0.918$.}
\end{table}

The data therefore give strong evidence that demand is elastic under the differenced design, provided the power law of Section~\ref{sec:empirical-specification} holds with $\kappa\le0.91$, since every bound in the strong scenario, including the max-$t$ bound for Approach~2, lies below $-0.91$; this range includes the whole band from $0.45$ to $0.75$ around the clothing estimate $\hat\kappa=0.605$.
The full information set is only the richest observed one, so our conclusions about $\theta^\star$ come from the sensitivity grid.

\FloatBarrier

\section{Conclusion}
\label{sec:conclusion}

Replacing a rich object by an AI-learned representation shifts a linear functional of the regression by the inner product of the discarded parts of the outcome regression and the Riesz representer.
By this formula, the analyst should report the DML interval for the representation-based target and, when the approximation bias may not be negligible, a calibrated sensitivity region; for a tuned embedding, the analyst should use one of the three approaches.

\subsection*{Acknowledgments}
We are enormously grateful to Claude Opus~5.5 (Anthropic) for its help with this paper, and particularly with Appendix~\ref{sec:cross-fold} and the accompanying proofs in Appendix~\ref{app:proofs-cross-fold}. GPT-6 Astra was also used to check the proofs.

\newpage
% Compile with pdfLaTeX, BibTeX, pdfLaTeX, pdfLaTeX.
\begingroup
\setstretch{1.1}
\bibliographystyle{plainnat}
\bibliography{AI_embeddings_2026-10-12}
\endgroup

\newpage
\appendix

\section{Embeddings as the treatment}
\label{app:treatment}

This appendix derives the representers of Section~\ref{sec:treatment}, adds customer context, and states which results of the paper carry over and what changes when the policies are estimated.

\paragraph{Setup.}
Let $O=(Y,Z,C)\sim P$, where $Z$ is the displayed content and $C$ is pre-treatment customer context.
A policy $j\in\{0,1\}$ is a conditional law $\Lambda_j(\cdot\mid c)$ of content given context, held fixed, and the customer distribution is the same under both.
Write $W=(Z,C)$ and $W_\phi=(X_\phi,C)$ with $X_\phi=x(Z,\phi)$, and let $g^\star(z,c)=\E[Y\mid Z=z,C=c]$.
Two conditions connect the regression to the causal response.
\begin{enumerate}[label=(T\arabic*),leftmargin=*]
\item \emph{Identification.} Content is ignorable given $C$, and the response to a given content item is the same under both policies; so $g^\star(z,c)$ is the mean response to content $z$ of customers with context $c$ under either policy.
\item \emph{Overlap.} For $P_C$-almost every $c$, $\Lambda_j(\cdot\mid c)$ is absolutely continuous with respect to the law $P_{Z\mid C=c}$ of content in the data, and the density ratios $r_j(z,c):=d\Lambda_j(\cdot\mid c)/dP_{Z\mid C=c}\,(z)$ lie in $L^q(P_W)$ for some $q>4$; if $C$ is present, also $\operatorname*{ess\,sup}_c\int(r_1-r_0)^2\,dP_{Z\mid C=c}<\infty$.
\end{enumerate}
The effect of replacing $\Lambda_0$ by $\Lambda_1$ is
\begin{align*}
\theta^\star&=\E\Bigl[\int g^\star(z,C)\,\{d\Lambda_1(z\mid C)-d\Lambda_0(z\mid C)\}\Bigr]=\E[m(O,g^\star)],\\
m(O,h)&:=\int h(z,C)\,d(\Lambda_1-\Lambda_0)(z\mid C).
\end{align*}
With $C$ absent this is~\eqref{eq:treatment-theta}.

\begin{lemma}[Representers for a change of content distribution]
\label{lem:treatment-representers}
Under (T2):
\begin{enumerate}[label=(\alph*),leftmargin=*]
\item $\alpha^\star:=r_1-r_0$ satisfies $\E[m(O,h)]=\E[h(W)\alpha^\star(W)]$ for all $h\in L^2(P_W)$.
\item Let $\Lambda_j^\phi(\cdot\mid c)$ and $P^\phi_{X\mid C=c}$ be the laws of $X_\phi$ given $C=c$ under $\Lambda_j(\cdot\mid c)$ and under $P$. Then $\Lambda_j^\phi\ll P^\phi$, and $\alpha_\phi=\E[\alpha^\star\mid W_\phi]=r_1^\phi-r_0^\phi$ with $r_j^\phi:=d\Lambda_j^\phi/dP^\phi$.
\item If $r_j(z,c)=\varrho_j(x(z,\phi),c)$ for some functions $\varrho_j$, so that the policy acts on content only through $X_\phi$, then $\alpha^\star=\alpha_\phi$ and $\theta^\star=\theta_\phi$.
\end{enumerate}
\end{lemma}

\begin{proof}
(a) For $h\in L^2(P_W)$, by Cauchy--Schwarz and (T2), $hr_j$ is integrable, and
\[
\E[h(W)r_j(W)]=\E\Bigl[\int h(z,C)\,r_j(z,C)\,dP_{Z\mid C}(z)\Bigr]=\E\Bigl[\int h(z,C)\,d\Lambda_j(z\mid C)\Bigr].
\]
Subtracting the case $j=0$ from $j=1$ gives the claim.
(b) Absolute continuity of the pushforwards follows from $\Lambda_j\ll P_{Z\mid C}$.
For bounded $f$,
\begin{align*}
\E[f(W_\phi)r_j(W)]&=\E\Bigl[\int f(x(z,\phi),C)\,d\Lambda_j(z\mid C)\Bigr]=\E\Bigl[\int f(x,C)\,d\Lambda_j^\phi(x\mid C)\Bigr]\\
&=\E\Bigl[\int f(x,C)\,r_j^\phi(x,C)\,dP^\phi(x\mid C)\Bigr]=\E[f(W_\phi)r_j^\phi(W_\phi)],
\end{align*}
so $\E[r_j\mid W_\phi]=r_j^\phi$.
The identity $\alpha_\phi=\E[\alpha^\star\mid W_\phi]$ is the one used in the proof of Proposition~\ref{prop:ovb}.
(c) Under the hypothesis, $\alpha^\star=\varrho_1(X_\phi,C)-\varrho_0(X_\phi,C)$ is a function of $W_\phi$, hence equals its conditional expectation $\alpha_\phi$, and~\eqref{eq:ovb-identity} gives $\theta^\star=\theta_\phi$.
\end{proof}

\begin{remark}
Square integrability of $r_j$ is all that Lemma~\ref{lem:treatment-representers} and Proposition~\ref{prop:ovb} use.
Assumptions~\ref{ass:dml}, \ref{ass:pipeline-dml}, \ref{ass:pipeline-bounds-dml} and~\ref{ass:sfu} need $\alpha^\star\in L^q$ with $q>4$, and since $\alpha_\phi=\E[\alpha^\star\mid W_\phi]$ the moment passes to $\alpha_\phi$ by Jensen's inequality.
The essential-supremum condition is equivalent to $\norm{m(O,h)}_{P,2}\le C_m\norm{h}_{P,2}$ on $L^2(P_W)$, with $C_m^2$ equal to the essential supremum; without $C$ it holds with $C_m=\norm{\alpha^\star}_{P,2}$.
The proofs of Theorems~\ref{thm:dml-expansion} and~\ref{thm:strong-adaptive} use that bound only for errors that tend to zero in $L^4$, where $\norm{m(O,h)}_{P,2}\le\norm{\alpha^\star}_{P,4}\norm{h}_{P,4}$ suffices.
The analysis of Approach~1 with menus (Appendices~\ref{sec:learn-H} and~\ref{app:proofs-cross-fold}) additionally needs bounded density ratios and the essential-supremum condition.
\end{remark}

\paragraph{What carries over.}
The map $m(O,h)$ depends on $O$ only through $C$ and is a function of $W_\phi$ whenever $h$ is, so the compatibility condition of Section~\ref{sec:setup} holds.
Proposition~\ref{prop:ovb} and its bound~\eqref{eq:cs-bound} therefore give $\theta^\star-\theta_\phi=\E[(g^\star-g_\phi)(\alpha^\star-\alpha_\phi)]$, now with $\alpha^\star-\alpha_\phi=(r_1-r_1^\phi)-(r_0-r_0^\phi)$, the part of the policy shift orthogonal to the functions of $W_\phi$.
The orthogonal score of Section~\ref{sec:fixed} uses this $m$, and the Riesz loss $\E[\alpha(W)^2-2m(O,\alpha)]$ can be evaluated because $m(O,\alpha)=\int\alpha(z,C)\,d(\Lambda_1-\Lambda_0)(z\mid C)$ is computable from the policies, by summation over a finite content library or by simulation; the encoder and the heads must then be evaluated on content drawn from $\Lambda_j$, not only on observed $Z$.
The sensitivity region of Section~\ref{sec:bounds} and its calibration depend on $g^\star$ and $\alpha^\star$ only through $L^2$ norms, so they apply as stated.
The inference for the endpoints, Propositions~\ref{prop:onesided}, \ref{prop:joint} and~\ref{prop:sfu}, and Theorems~\ref{thm:pipeline-dml} and~\ref{thm:pooled-joint} hold under (T2), with $h(x_j(Z),C)$ in place of $h(D,x_j(Z))$ in the definition of $\cH_\phi$.
They require $\E[\alpha_\phi^2]>0$, that is, that the two policies induce different laws of $X_\phi$ given $C$; otherwise $\theta_\phi=0$, $\sigma_\phi=0$ and $C_D$ is undefined.

\paragraph{Estimated policies.}
If $\Lambda_0$ and $\Lambda_1$ are known only through samples or are themselves fitted, $m(O,\cdot)$ is replaced by a plug-in version and the score acquires an extra error term and the representers are those of the plug-in functional.
Then the guarantees of Section~\ref{sec:fixed} require that this term be $o_\Pp(n^{-1/2})$.
For policies that can be simulated, Monte Carlo integration adds mean-zero noise: with independent draws for each of the $n$ customers the error in $\hat\theta$ is $O_\Pp\{(nB)^{-1/2}\}$, negligible as soon as $B\to\infty$, and if one set of $B$ draws is shared across customers it is $O_\Pp(B^{-1/2})$, negligible when $B$ grows faster than $n$.
We do not prove a result for estimated policies here.

\section{Oracle containment}
\label{app:containment}

We derive the consequences of the oracle-containment assumption of ADML~\cite{van2023adaptive} in our setting.
Suppose there is a representation $\phi_0$ such that
\begin{equation}
g^\star(W)=\E[Y\mid W]\ \text{is a function of}\ W_{\phi_0}=(D,X_{\phi_0}).
\label{eq:containment}
\end{equation}
The map $\phi_0$ is not identified: any bimeasurable bijective transformation of $X_{\phi_0}$ satisfies~\eqref{eq:containment} as well.
The results below depend on representations only through the subspaces $L^2(P_{W_\phi})$ they generate, and so are invariant to such transformations.
They do not require $\phi_0$ to belong to the class over which the analyst tunes.

Let $\hat\phi$ be a tuned representation; we argue conditionally on the data that produced $\hat\phi$, so that $\hat\phi$ is fixed. Stack the two representations,
\[
W_L:=(D,X_{\hat\phi},X_{\phi_0}),
\]
so that $\sigma(W_{\hat\phi})\subseteq\sigma(W_L)\subseteq\sigma(W)$.
The stacked information is a representation in the sense of Section~\ref{sec:subordinated-geometry}: one encoder $(X_{\hat\phi},X_{\phi_0})$ with head classes rich enough that $\cH_L=L^2(P_{W_L})$.
Write $g_L:=\E[Y\mid W_L]$ and let $\alpha_L$ be the Riesz representer on $L^2(P_{W_L})$.
We assume the compatibility condition of Section~\ref{sec:setup} for $W_L$ and $W_{\hat\phi}$.

\begin{proposition}[{Approximation bias under oracle containment}]
\label{prop:containment}
Under~\eqref{eq:containment},
\begin{equation}
\theta^\star-\theta_{\hat\phi}
=
\E\bigl[(g^\star-g_{\hat\phi})(\alpha_L-\alpha_{\hat\phi})\bigr],
\label{eq:containment-identity}
\end{equation}
and therefore
\begin{equation}
|\theta^\star-\theta_{\hat\phi}|
\le
\norm{g^\star-g_{\hat\phi}}_{P,2}\,\norm{\alpha_L-\alpha_{\hat\phi}}_{P,2}
\le
\norm{g^\star-g_{\hat\phi}}_{P,2}\,\norm{\alpha^\star-\alpha_{\hat\phi}}_{P,2}.
\label{eq:containment-bound}
\end{equation}
\end{proposition}

\begin{proof}
By~\eqref{eq:containment}, $g^\star$ is $\sigma(W_L)$-measurable, so $g_L=\E[g^\star\mid W_L]=g^\star$ and $\theta_L:=\E[m(O,g_L)]=\theta^\star$.
The functional $h\mapsto\E[m(O,h)]$ is continuous on $L^2(P_W)$, hence on its subspace $L^2(P_{W_L})$, so $\alpha_L$ exists.
Proposition~\ref{prop:ovb}, applied with $W_L$ as the long information and $W_{\hat\phi}$ as the short, gives~\eqref{eq:containment-identity}.
The first inequality in~\eqref{eq:containment-bound} is Cauchy--Schwarz.
For the second, $\alpha_L=\E[\alpha^\star\mid W_L]$ and $\alpha_{\hat\phi}=\E[\alpha^\star\mid W_{\hat\phi}]$, so $\alpha_L-\alpha_{\hat\phi}$ is the projection of $\alpha^\star-\alpha_{\hat\phi}$ onto $L^2(P_{W_L})$ and has no larger norm.
\end{proof}

The first factor is the regression loss of the tuned representation.
The second is never larger than the representer loss of Proposition~\ref{prop:ovb}, and it is the norm of the change in the representer when the oracle features are added to $X_{\hat\phi}$.
If $X_{\hat\phi}$ refines the oracle, $\sigma(X_{\phi_0})\subseteq\sigma(X_{\hat\phi})$, the second factor vanishes and $\theta_{\hat\phi}=\theta^\star$, however far $\hat\phi$ is from $\phi_0$.
Containment always holds with $X_{\phi_0}=Z$, and then $W_L=W$ and the proposition reduces to Proposition~\ref{prop:ovb}.
The assumption therefore has content only when the oracle features are simple enough for $X_{\hat\phi}$ to capture most of their information about $\alpha^\star$.

\begin{corollary}[Target adaptivity under oracle containment]
\label{cor:containment}
Under~\eqref{eq:containment}, if
\[
\sqrt n\,\norm{g^\star-g_{\hat\phi}}_{P,2}\,\norm{\alpha_L-\alpha_{\hat\phi}}_{P,2}=o_\Pp(1),
\]
then $\sqrt n(\theta^\star-\theta_{\hat\phi})=o_\Pp(1)$.
It suffices that both factors be $o_\Pp(n^{-1/4})$.
\end{corollary}

\begin{remark}[The symmetric case]
If instead $\alpha^\star$ is a function of $W_{\phi_0}$, as when the propensity score depends on $Z$ only through $X_{\phi_0}$, then $\alpha_L=\alpha^\star$, $\theta_L=\E[g_L\alpha^\star]=\theta^\star$, and Proposition~\ref{prop:ovb}, applied with $W_L$ as the long information and $W_{\hat\phi}$ as the short, gives $\theta^\star-\theta_{\hat\phi}=\E[(g_L-g_{\hat\phi})(\alpha^\star-\alpha_{\hat\phi})]$.
\end{remark}

\emph{The three approaches.}
In Approach~1, if every square-integrable function of $W_{\phi_0}$ lies in $\cH$, then $g_0=g^\star$, and $\theta^\star-\theta_0=\E[(g^\star-g_0)(\alpha^\star-\alpha_0)]=0$, so under containment the pipeline target equals $\theta^\star$, although the rate condition of Theorem~\ref{thm:pipeline-dml} must still hold.
In Approach~2, the proposition applies fold by fold with $W_{L,k}:=(D,X_{\hat\phi_k},X_{\phi_0})$, so the center $\bar\theta_n$ of Proposition~\ref{prop:joint} is within $o_\Pp(n^{-1/2})$ of $\theta^\star$ when the condition of Corollary~\ref{cor:containment} holds in every fold.
In Approach~3, apply the proposition conditionally on $\mathcal S_0$ to the frozen union; by Lemma~\ref{lem:union}, its outcome factor is at most $\norm{g^\star-g_{\hat\phi_g}}_{P,2}$, and the condition with $n_1$ in place of $n$ sharpens Corollary~\ref{cor:adaptive}.

\section{Fixed embeddings: conditions and proofs}
\label{app:proofs-sec2}

This appendix states the regularity conditions behind Theorems~\ref{thm:dml-expansion} and~\ref{thm:strong-adaptive} and collects the proofs for Sections~\ref{sec:warmup} and~\ref{sec:fixed}.
We use empirical process notation in the style of \cite{vanderVaart1998,vanderVaartWellner1996}.

\subsection{Regularity conditions}

To interpret~\eqref{eq:dml-est} as a root-$n$ estimator, we need two conditions: (i) the score must be orthogonal, and (ii) the nuisance learners must converge fast enough in $L^2(P)$ so that the product of their errors is $o_\Pp(n^{-1/2})$.
The next assumption records a standard set of sufficient conditions.

\begin{assumption}[{DML regularity for the representation-based target}]
\label{ass:dml}
Fix $K\ge 2$ and consider a (possibly drifting) sequence $\phi=\phi_n$.
Let $\hat g_{\phi}^{(-k)}$ and $\hat\alpha_{\phi}^{(-k)}$ be the cross-fitted nuisance estimators defined in Section~\ref{sec:dml}.

\begin{enumerate}[label=(\roman*)]
\item \emph{Moments and bounded linearity.}
There exist $q>4$ and a constant $C_m<\infty$ such that $Y\in L^q(P)$, $\sup_n \norm{\alpha_{\phi_n}}_{P,q}<\infty$, and
\begin{equation}
\norm{m(O,h)}_{P,2}\le C_m \norm{h}_{P,2}
\qquad\text{for all } h \in L^2(P_{W_{\phi_n}}).
\label{eq:m-L2-bounded}
\end{equation}
(When $\phi_n\equiv\phi$ is fixed, the supremum over $n$ is unnecessary.)

\item \emph{Nuisance consistency and moment control.}
Uniformly over folds,
\[
\max_{1\le k\le K}\norm{\hat g_{\phi_n}^{(-k)}-g_{\phi_n}}_{P,2}=o_\Pp(1),
\qquad
\max_{1\le k\le K}\norm{\hat\alpha_{\phi_n}^{(-k)}-\alpha_{\phi_n}}_{P,2}=o_\Pp(1),
\]
and the nuisance estimates have uniformly bounded $q$th moments,
\[
\max_{1\le k\le K}\norm{\hat g_{\phi_n}^{(-k)}}_{P,q}
+
\max_{1\le k\le K}\norm{\hat\alpha_{\phi_n}^{(-k)}}_{P,q}
=
O_\Pp(1).
\]

\item \emph{Rate condition.} Uniformly over folds,
\begin{equation}
\max_{1\le k\le K}\norm{\hat g_{\phi_n}^{(-k)}-g_{\phi_n}}_{P,2}\;
\max_{1\le k\le K}\norm{\hat\alpha_{\phi_n}^{(-k)}-\alpha_{\phi_n}}_{P,2}
=
o_\Pp(n^{-1/2}).
\label{eq:dml-rate}
\end{equation}

\item \emph{Nondegeneracy.}
Let $\sigma_{\phi_n}^2:=\Var\!\left(\psi_{\phi_n}(O;g_{\phi_n},\alpha_{\phi_n})\right)$.
Assume $0<\inf_n\sigma_{\phi_n}\le \sup_n\sigma_{\phi_n}<\infty$.

\item \emph{Extra moment when drifting.} For a possibly drifting sequence $\phi=\phi_n$, assume moreover that there exists $\delta>0$ such that
\[\sup_n\norm{\psi_{\phi_n}(O;g_{\phi_n},\alpha_{\phi_n})-\theta_{\phi_n}}_{P,2+\delta}<\infty.\]
When $\phi_n\equiv\phi$ is fixed, this additional uniform moment condition is unnecessary.
\end{enumerate}
\end{assumption}

\subsection{Convergence of the short score}

We prove Theorem~\ref{thm:strong-adaptive} with the following lemma, which gives convergence of the short score to the efficient long score under adaptive efficiency.

\begin{lemma}[Short score convergence]
\label{lem:score-conv}
Assume $Y\in L^2(P)$ and~\eqref{eq:rr-long}--\eqref{eq:rr-short} hold.
Assume furthermore that $m$ is $L^2$-bounded in the sense that there exists $C_m<\infty$ such that
$\norm{m(O,h)}_{P,2}\le C_m \norm{h}_{P,2}$ for all $h\in L^2(P_W)$.
Suppose adaptive efficiency~\eqref{eq:adaptive-efficiency} holds, and, in addition,
\begin{equation}
\norm{(\alpha^\star-\alpha_{\phi_n})(Y-g^\star)}_{P,2}=o(1),
\qquad
\norm{\alpha_{\phi_n}(g_{\phi_n}-g^\star)}_{P,2}=o(1).
\label{eq:adaptive-efficiency-score-products}
\end{equation}
Define the centered long score
$
\varphi^\star(O):=\psi^\star(O;g^\star,\alpha^\star)-\theta^\star
$
and the centered short score
$
\varphi_{\phi_n}(O):=\psi_{\phi_n}(O;g_{\phi_n},\alpha_{\phi_n})-\theta_{\phi_n}.
$
Then
\begin{equation}
\norm{\varphi_{\phi_n}-\varphi^\star}_{P,2}=o(1).
\label{eq:score-conv}
\end{equation}
In particular, $\sigma_{\phi_n}^2\to \sigma_\star^2$, where $\sigma_\star^2:=\Var(\psi^\star(O;g^\star,\alpha^\star))$.
\end{lemma}

\subsection{Empirical process notation for proofs}

For a measurable function $f:\mathcal O\to\R$ with $P|f|<\infty$, write
\[
\Pn f := \frac{1}{n}\sum_{i=1}^n f(O_i),
\qquad
\Gn f := \sqrt{n}\,(\Pn-P)f.
\]
For each fold $I_k$ of size $n_k:=|I_k|$, we also write
\[
\Pn^{(k)} f := \frac{1}{n_k}\sum_{i\in I_k} f(O_i),
\qquad
\Gn^{(k)} f := \sqrt{n_k}\,(\Pn^{(k)}-P)f.
\]
Appendices~\ref{app:sens} and~\ref{app:cross-fold} use the same symbols; in Appendix~\ref{sec:learn-H}, $\Pn$ denotes the empirical measure of the training sample, and in Appendix~\ref{app:joint}, $\E_{n,k}$ denotes $\Pn^{(k)}$.
\subsection{Proof of Proposition~\ref{prop:ovb}}

\begin{proof}
By~\eqref{eq:rr-long} and~\eqref{eq:rr-short},
\[
\theta^\star=\E[g^\star\alpha^\star]
\qquad\text{and}\qquad
\theta_\phi=\E[g_\phi\alpha_\phi].
\]
Expand
\[
\E[g^\star\alpha^\star]-\E[g_\phi\alpha_\phi]
=
\E[(g^\star-g_\phi)(\alpha^\star-\alpha_\phi)]
+\E[g_\phi(\alpha^\star-\alpha_\phi)]
+\E[\alpha_\phi(g^\star-g_\phi)].
\]
We now argue that the last two terms vanish.
First, $g_\phi(W_\phi)=\E[g^\star(W)\mid W_\phi]$, so $g^\star-g_\phi$ is orthogonal in $L^2(P)$ to every square-integrable function of $W_\phi$; in particular, $\E[\alpha_\phi(g^\star-g_\phi)]=0$ because $\alpha_\phi$ is $\sigma(W_\phi)$-measurable.
Second, the compatibility condition from Section~\ref{sec:setup} implies that the short functional is the restriction of the long functional to $\sigma(W_\phi)$-measurable functions.
Hence, for every $h\in L^2(P_{W_\phi})$ we have $\E[h(W_\phi)\alpha^\star]=\E[h(W_\phi)\alpha_\phi]$, which implies $\alpha_\phi(W_\phi)=\E[\alpha^\star(W)\mid W_\phi]$.
Therefore $\alpha^\star-\alpha_\phi$ is orthogonal to every square-integrable function of $W_\phi$, and in particular $\E[g_\phi(\alpha^\star-\alpha_\phi)]=0$.
This proves~\eqref{eq:ovb-identity}, and~\eqref{eq:cs-bound} follows from Cauchy--Schwarz.
\end{proof}

\subsection{Proof of Theorem~\ref{thm:dml-expansion}}

The proof shows how cross-fitting and orthogonality reduce the effect of nuisance estimation to a second-order remainder.
Because our score~\eqref{eq:score-short} is linear and of Riesz form, we can carry out the argument by direct algebra, without appealing to abstract differentiability.

\begin{proof}
Fix $n$ and write $\phi=\phi_n$ to lighten notation.
Let $n_k:=|I_k|$.
Steps~1--3 split $\hat\theta_\phi-\theta_\phi$ into the oracle average, a bias term, which is a product of errors by orthogonality, and an empirical-process term, which is negligible by cross-fitting; Step~4 and the last two paragraphs give the normal limit and variance consistency.
For each fold $k$, define the nuisance errors
\[
\delta g^{(-k)} := \hat g_{\phi}^{(-k)}-g_{\phi},
\qquad
\delta\alpha^{(-k)} := \hat\alpha_{\phi}^{(-k)}-\alpha_{\phi}.
\]

\noindent\emph{Step 1: A decomposition of the cross-fitted estimator.}
Using~\eqref{eq:dml-est} and fold-specific empirical measures, we can write
\[
\hat\theta_{\phi}
=
\sum_{k=1}^K \frac{n_k}{n}\,\Pn^{(k)}\psi_{\phi}\!\left(\cdot;\hat g_{\phi}^{(-k)},\hat\alpha_{\phi}^{(-k)}\right).
\]
Add and subtract $P\psi_{\phi}(\cdot;g_{\phi},\alpha_{\phi})=\theta_{\phi}$ to obtain
\begin{align*}
\hat\theta_{\phi}-\theta_{\phi}
&=
(\Pn-P)\psi_{\phi}(\cdot;g_{\phi},\alpha_{\phi})
+ R_{1n}+R_{2n},
\\[3pt]
R_{1n}
&:=
\sum_{k=1}^K \frac{n_k}{n}\,
(\Pn^{(k)}-P)\Bigl\{
\psi_{\phi}(\cdot;\hat g_{\phi}^{(-k)},\hat\alpha_{\phi}^{(-k)})
-\psi_{\phi}(\cdot;g_{\phi},\alpha_{\phi})
\Bigr\},
\\[3pt]
R_{2n}
&:=
\sum_{k=1}^K \frac{n_k}{n}\,
P\Bigl\{
\psi_{\phi}(\cdot;\hat g_{\phi}^{(-k)},\hat\alpha_{\phi}^{(-k)})
-\psi_{\phi}(\cdot;g_{\phi},\alpha_{\phi})
\Bigr\}.
\end{align*}
Thus, to prove~\eqref{eq:dml-expansion}, it suffices to show $\sqrt{n}R_{1n}=o_\Pp(1)$ and $\sqrt{n}R_{2n}=o_\Pp(1)$.

\noindent\emph{Step 2: The bias term $R_{2n}$ is second order.}
Fix a fold $k$ and abbreviate $(\hat g,\hat\alpha):=(\hat g_{\phi}^{(-k)},\hat\alpha_{\phi}^{(-k)})$.
By linearity of $m$ and the score definition~\eqref{eq:score-short},
\begin{align*}
\psi_{\phi}(O;\hat g,\hat\alpha)-\psi_{\phi}(O;g_{\phi},\alpha_{\phi})
&=
m(O,\delta g^{(-k)})
+\delta\alpha^{(-k)}(W_{\phi})\{Y-g_{\phi}(W_{\phi})\}
\\
&\qquad
-\alpha_{\phi}(W_{\phi})\,\delta g^{(-k)}(W_{\phi})
-\delta\alpha^{(-k)}(W_{\phi})\,\delta g^{(-k)}(W_{\phi}).
\end{align*}
Taking expectations and using the defining properties of $g_{\phi}$ and $\alpha_{\phi}$ gives two cancellations.
First, $\E[Y-g_{\phi}(W_{\phi})\mid W_{\phi}]=0$, so
\[
P\bigl[\delta\alpha^{(-k)}(W_{\phi})\{Y-g_{\phi}(W_{\phi})\}\bigr]=0.
\]
Second, by the Riesz representation~\eqref{eq:rr-short}, for any square-integrable function $h(W_{\phi})$ we have
\[
P[m(O,h)] = P[h(W_{\phi})\alpha_{\phi}(W_{\phi})].
\]

Applying this to $h=\delta g^{(-k)}$ gives $P[m(O,\delta g^{(-k)})]=P[\alpha_{\phi}\delta g^{(-k)}]$, so the first and third terms cancel in expectation.
Therefore,
\[
P\Bigl\{
\psi_{\phi}(\cdot;\hat g_{\phi}^{(-k)},\hat\alpha_{\phi}^{(-k)})
-\psi_{\phi}(\cdot;g_{\phi},\alpha_{\phi})
\Bigr\}
=
-\,P\bigl[\delta g^{(-k)}\,\delta\alpha^{(-k)}\bigr].
\]
This is exactly the mixed-bias identity~\eqref{eq:mixed-bias} applied with fold-$k$ nuisance errors.
It follows that
\[
|R_{2n}|
\le
\max_{1\le k\le K}\norm{\delta g^{(-k)}}_{P,2}\;
\max_{1\le k\le K}\norm{\delta\alpha^{(-k)}}_{P,2}
\]
by Cauchy--Schwarz, and hence $\sqrt{n}R_{2n}=o_\Pp(1)$ by the rate condition~\eqref{eq:dml-rate}.

\noindent\emph{Step 3: The empirical process term $R_{1n}$ is negligible under $L^2$ consistency.}
For each fold $k$ define the score difference
\[
\Delta_k(O)
:=
\psi_{\phi}(O;\hat g_{\phi}^{(-k)},\hat\alpha_{\phi}^{(-k)})
-\psi_{\phi}(O;g_{\phi},\alpha_{\phi}).
\]
Conditional on the training sample $\{O_i:i\in I_k^c\}$, the function $\Delta_k$ is fixed and the observations $\{O_i:i\in I_k\}$ are i.i.d.\ draws from $P$.
Therefore,
\[
\E\!\left[\Bigl((\Pn^{(k)}-P)\Delta_k\Bigr)^2\;\middle|\;\{O_i:i\in I_k^c\}\right]
=
\frac{1}{n_k}\Var(\Delta_k(O))
\le
\frac{1}{n_k}\norm{\Delta_k}_{P,2}^2.
\]
By conditional Chebyshev's inequality,
\[
(\Pn^{(k)}-P)\Delta_k
=
O_\Pp\!\left(n_k^{-1/2}\norm{\Delta_k}_{P,2}\right)
\qquad\text{uniformly over }k.
\]
Using $n_k/n\le 1$ and $\min_k n_k\asymp n$ (since $K$ is fixed), we obtain
\[
\sqrt{n}\,|R_{1n}|
\le
\sum_{k=1}^K \frac{n_k}{n}\,\sqrt{n}\,\bigl|(\Pn^{(k)}-P)\Delta_k\bigr|
=
O_\Pp\!\left(\max_{1\le k\le K}\norm{\Delta_k}_{P,2}\right).
\]
Thus it remains to show that $\max_k \norm{\Delta_k}_{P,2}=o_\Pp(1)$.

To bound $\norm{\Delta_k}_{P,2}$, expand $\Delta_k$ using the same algebra as in Step~2:
\[
\begin{aligned}
\Delta_k(O)
&=m(O,\delta g^{(-k)})
+\delta\alpha^{(-k)}(W_{\phi})\{Y-g_{\phi}(W_{\phi})\}\\
&\quad-\alpha_{\phi}(W_{\phi})\,\delta g^{(-k)}(W_{\phi})
-\delta\alpha^{(-k)}(W_{\phi})\,\delta g^{(-k)}(W_{\phi}).
\end{aligned}
\]
By the triangle inequality and H\"older,
\begin{align*}
\norm{\Delta_k}_{P,2}
&\le
\norm{m(O,\delta g^{(-k)})}_{P,2}
+ \norm{\alpha_{\phi}\,\delta g^{(-k)}}_{P,2}
+ \norm{\delta\alpha^{(-k)}(Y-g_{\phi})}_{P,2}
+ \norm{\delta\alpha^{(-k)}\,\delta g^{(-k)}}_{P,2}
\\
&\le
C_m\norm{\delta g^{(-k)}}_{P,2}
+ \norm{\alpha_{\phi}}_{P,4}\,\norm{\delta g^{(-k)}}_{P,4}
+ \norm{Y-g_{\phi}}_{P,4}\,\norm{\delta\alpha^{(-k)}}_{P,4}
\\
&\quad
+ \norm{\delta\alpha^{(-k)}}_{P,4}\,\norm{\delta g^{(-k)}}_{P,4},
\end{align*}
where we used~\eqref{eq:m-L2-bounded} for the first term.
Assumption~\ref{ass:dml}(i) gives $\sup_n\norm{\alpha_{\phi_n}}_{P,4}<\infty$, and, since $g_\phi=\E[Y\mid W_\phi]$, by conditional Jensen $\norm{Y-g_{\phi}}_{P,4}\le2\norm{Y}_{P,4}<\infty$.
Moreover, Assumption~\ref{ass:dml}(ii) implies $\norm{\delta g^{(-k)}}_{P,2}=o_\Pp(1)$ and $\norm{\delta\alpha^{(-k)}}_{P,2}=o_\Pp(1)$ uniformly in $k$, while the $q$th-moment bounds ensure $\max_k \norm{\delta g^{(-k)}}_{P,q}=O_\Pp(1)$ and $\max_k \norm{\delta\alpha^{(-k)}}_{P,q}=O_\Pp(1)$.
For any $p\in[2,q]$, the interpolation inequality
\[
\norm{f}_{P,p}\le \norm{f}_{P,2}^{\lambda}\norm{f}_{P,q}^{1-\lambda},
\qquad
\frac{1}{p}=\frac{\lambda}{2}+\frac{1-\lambda}{q},
\]
implies $\norm{f}_{P,p}=o_\Pp(1)$ whenever $\norm{f}_{P,2}=o_\Pp(1)$ and $\norm{f}_{P,q}=O_\Pp(1)$.
Taking $p=4$ (permitted since $q>4$) yields $\norm{\delta g^{(-k)}}_{P,4}=o_\Pp(1)$ and $\norm{\delta\alpha^{(-k)}}_{P,4}=o_\Pp(1)$ uniformly in $k$.
Hence $\max_k\norm{\Delta_k}_{P,2}=o_\Pp(1)$, so $\sqrt{n}R_{1n}=o_\Pp(1)$.

\noindent\emph{Step 4: Conclude the expansion and the CLT.}
Combining Steps~1--3 gives
\[
\sqrt{n}(\hat\theta_{\phi}-\theta_{\phi})
=
\Gn\!\left(\psi_{\phi}(\cdot;g_{\phi},\alpha_{\phi})\right)+o_\Pp(1)
=
\frac{1}{\sqrt{n}}\sum_{i=1}^n\Bigl(\psi_{\phi}(O_i;g_{\phi},\alpha_{\phi})-\theta_{\phi}\Bigr)
+o_\Pp(1),
\]
which is~\eqref{eq:dml-expansion}.

\noindent\emph{Asymptotic normality.}
If $\phi_n\equiv \phi$ is fixed, then the summands in~\eqref{eq:dml-expansion} are i.i.d.\ with variance $\sigma_\phi^2\in(0,\infty)$, so the classical CLT yields
$\sqrt{n}(\hat\theta_\phi-\theta_\phi)\Rightarrow N(0,\sigma_\phi^2)$.
If $\phi=\phi_n$ drifts with $n$, then for each $n$ the summands remain i.i.d.\ within row, with variance $\sigma_{\phi_n}^2$. Let
\[\varphi_{\phi_n}(O):=\psi_{\phi_n}(O;g_{\phi_n},\alpha_{\phi_n})-\theta_{\phi_n}.\]
By Assumption~\ref{ass:dml}, for some $\delta>0$,
\[\sup_n\norm{\varphi_{\phi_n}}_{P,2+\delta}<\infty,\qquad\inf_n\sigma_{\phi_n}>0.\]
Hence the Lyapunov ratio satisfies
\[\frac{nP|\varphi_{\phi_n}|^{2+\delta}}{(n\sigma_{\phi_n}^2)^{1+\delta/2}}=\frac{P|\varphi_{\phi_n}|^{2+\delta}}{n^{\delta/2}\sigma_{\phi_n}^{2+\delta}}\longrightarrow 0.\]
Lyapunov's central limit theorem for triangular arrays therefore gives
\[\sigma_{\phi_n}^{-1}\frac{1}{\sqrt n}
\sum_{i=1}^n\varphi_{\phi_n}(O_i)\Rightarrow N(0,1).\]

Combining this with~\eqref{eq:dml-expansion} and Slutsky's theorem yields
\[
\sigma_{\phi_n}^{-1}\sqrt{n}(\hat\theta_{\phi_n}-\theta_{\phi_n})\Rightarrow N(0,1).
\]

\noindent\emph{Consistency of the cross-fitted variance estimator.}
Write $k(i)$ for the fold containing observation $i$.
For notational convenience, fix $n$ and abbreviate $\phi=\phi_n$.
Define the oracle score
\[
\psi_0(O):=\psi_{\phi}(O;g_{\phi},\alpha_{\phi}),
\qquad
\theta:=\theta_{\phi}=P\psi_0,
\qquad
\sigma^2:=\sigma_{\phi}^2=\Var(\psi_0)=P\psi_0^2-\theta^2.
\]
For each fold $k$, define the foldwise cross-fitted score function
\[
\hat\psi^{(-k)}(O):=\psi_{\phi}\bigl(O;\hat g_{\phi}^{(-k)},\hat\alpha_{\phi}^{(-k)}\bigr).
\]
Then $\hat\psi_{\phi,i}=\hat\psi^{(-k)}(O_i)$ for $i\in I_k$, and the variance estimator~\eqref{eq:var-est} can be written as
\begin{equation}
\hat\sigma_{\phi}^2
=
\Pn(\hat\psi_{\phi}^2)-\hat\theta_{\phi}^2.
\label{eq:var-est-alt}
\end{equation}

We first compare the fitted and oracle scores on the validation observations. Step~3 established
$\max_k\norm{\hat\psi^{(-k)}-\psi_0}_{P,2}=o_\Pp(1)$.
Conditional on the training sample, the score difference is fixed. Conditional Markov's inequality therefore gives
\begin{equation}
\max_k\Pn^{(k)}(\hat\psi^{(-k)}-\psi_0)^2=o_\Pp(1).
\label{eq:var-lln-foldwise}
\end{equation}
Indeed, the conditional probability of exceeding any fixed positive threshold is bounded by the population squared norm divided by that threshold, truncated at one; its expectation tends to zero. This uses convergence of the score difference, rather than a law of large numbers for an arbitrary fitted squared score.

For a fixed representation, the ordinary law of large numbers gives $\Pn\psi_0^2-P\psi_0^2=o_\Pp(1)$. For a drifting representation, Assumption~\ref{ass:dml}(v) and the bounded means $|\theta_{\phi_n}|\le C_m\norm{Y}_{P,2}$ give uniform integrability of the oracle squares. Indeed, write $X_{ni}:=\psi_0(O_i)^2$ and fix $T>0$. The truncated part satisfies $\Var\{\Pn(X_{n\cdot}\wedge T)\}\le T^2/n$, while $\E|(\Pn-P)\{X_{n\cdot}-X_{n\cdot}\wedge T\}|\le2\sup_n P\{\psi_0^2\mathbf 1(\psi_0^2>T)\}\le2T^{-\delta/2}\sup_n P|\psi_0|^{2+\delta}$, which is finite by Assumption~\ref{ass:dml}(v), the bounded means, and Minkowski's inequality, and so tends to zero as $T\to\infty$. Letting $n\to\infty$ and then $T\to\infty$ gives $(\Pn-P)\psi_0^2=o_\Pp(1)$. In either case $\Pn\psi_0^2=O_\Pp(1)$.
Averaging~\eqref{eq:var-lln-foldwise} across the fixed number of folds and applying Cauchy--Schwarz now gives
\begin{equation}
\begin{split}
\bigl|\Pn(\hat\psi_{\phi}^2)-\Pn\psi_0^2\bigr|
&\le \{\Pn(\hat\psi_\phi-\psi_0)^2\}^{1/2}
       \{\Pn(\hat\psi_\phi+\psi_0)^2\}^{1/2}
=o_\Pp(1),
\end{split}
\end{equation}
where the second factor is $O_\Pp(1)$ by the oracle bound and~\eqref{eq:var-lln-foldwise}. Consequently,
\begin{equation}
\Pn(\hat\psi_{\phi}^2)=P\psi_0^2+o_\Pp(1).
\end{equation}

Finally, by~\eqref{eq:dml-expansion} we have $\hat\theta_{\phi}-\theta=O_\Pp(n^{-1/2})=o_\Pp(1)$, and hence $\hat\theta_{\phi}^2-\theta^2=o_\Pp(1)$.
Substituting into~\eqref{eq:var-est-alt} and using $\sigma^2=P\psi_0^2-\theta^2$ gives
\[
\hat\sigma_{\phi}^2-\sigma^2
=
\bigl\{\Pn(\hat\psi_{\phi}^2)-P\psi_0^2\bigr\}
-
\bigl\{\hat\theta_{\phi}^2-\theta^2\bigr\}
=o_\Pp(1).
\]
Since $\inf_n\sigma_{\phi_n}>0$ by Assumption~\ref{ass:dml}(iv), $\hat\sigma_{\phi_n}^2/\sigma_{\phi_n}^2\to1$ in probability.

\end{proof}

\subsection{Proof of Lemma~\ref{lem:score-conv}}

\begin{proof}
Fix $n$.
We use the shorthand $g_{\phi_n}=g_{\phi_n}(W_{\phi_n})$ and $\alpha_{\phi_n}=\alpha_{\phi_n}(W_{\phi_n})$.
Since $m(O,g_{\phi_n})$ is $\sigma(W_{\phi_n})$-measurable, we can write $m(O,g_{\phi_n})=m(W_{\phi_n},g_{\phi_n})$.

A short algebraic expansion yields
\[
\varphi^\star(O)-\varphi_{\phi_n}(O)
=
\Bigl\{m(O,g^\star-g_{\phi_n})-P\,m(O,g^\star-g_{\phi_n})\Bigr\}
+(\alpha^\star-\alpha_{\phi_n})(Y-g^\star)
+\alpha_{\phi_n}(g_{\phi_n}-g^\star).
\]
By the triangle inequality,
\[
\begin{aligned}
\norm{\varphi_{\phi_n}-\varphi^\star}_{P,2}
&\le
\norm{m(O,g^\star-g_{\phi_n})-P\,m(O,g^\star-g_{\phi_n})}_{P,2}
+\norm{(\alpha^\star-\alpha_{\phi_n})(Y-g^\star)}_{P,2}
\\
&\quad
+\norm{\alpha_{\phi_n}(g_{\phi_n}-g^\star)}_{P,2}.
\end{aligned}
\]

For the first term,
\[
\norm{m(O,g^\star-g_{\phi_n})-P\,m(O,g^\star-g_{\phi_n})}_{P,2}
\le
2\norm{m(O,g^\star-g_{\phi_n})}_{P,2}
\le
2C_m \norm{g^\star-g_{\phi_n}}_{P,2}.
\]
This bound tends to zero by~\eqref{eq:adaptive-efficiency}.
For the second and third terms, condition~\eqref{eq:adaptive-efficiency-score-products} implies that each bound is $o(1)$ as $n\to\infty$, which proves~\eqref{eq:score-conv}.
The variance convergence follows because $L^2(P)$ convergence implies convergence of second moments.
\end{proof}

\subsection{Proof of Theorem~\ref{thm:strong-adaptive}}

\begin{proof}
By Theorem~\ref{thm:dml-expansion} and target adaptivity~\eqref{eq:target-adaptivity},
\[
\sqrt{n}\bigl(\hat\theta_{\phi_n}-\theta^\star\bigr)
=
\Gn(\varphi_{\phi_n})+o_\Pp(1).
\]

The long-space operator bound and the score-product conditions supply the hypotheses of Lemma~\ref{lem:score-conv}, so $\norm{\varphi_{\phi_n}-\varphi^\star}_{P,2}=o(1)$.
Since $\E[\Gn(f)^2]=\Var(f)$ for each (nonrandom) $f\in L^2(P)$, we have
\[
\E\bigl[\Gn(\varphi_{\phi_n}-\varphi^\star)^2\bigr]
=
\Var(\varphi_{\phi_n}-\varphi^\star)
\le
\norm{\varphi_{\phi_n}-\varphi^\star}_{P,2}^2
\to 0,
\]
and hence $\Gn(\varphi_{\phi_n}-\varphi^\star)=o_\Pp(1)$.
Therefore $\Gn(\varphi_{\phi_n})=\Gn(\varphi^\star)+o_\Pp(1)$, which gives~\eqref{eq:efficient-expansion}.

The asymptotic normality statement follows from the classical CLT applied to the i.i.d.\ summands $\psi^\star(O_i;g^\star,\alpha^\star)-\theta^\star$.
Finally, since $\varphi^\star$ is the efficient influence function for $\theta^\star$, \eqref{eq:efficient-expansion} implies asymptotic efficiency~\cite{bickel1993,vanderVaart1998}.

\noindent\emph{Consistency of the cross-fitted variance estimator.}
Theorem~\ref{thm:dml-expansion} (proved above) shows that the variance estimator computed from the cross-fitted short scores satisfies
$\hat\sigma_{\phi_n}^2-\sigma_{\phi_n}^2=o_\Pp(1)$.
Lemma~\ref{lem:score-conv} implies $\sigma_{\phi_n}^2\to\sigma_\star^2$ under adaptive efficiency.
Therefore,
\[
\bigl|\hat\sigma_{\phi_n}^2-\sigma_\star^2\bigr|
\le
\bigl|\hat\sigma_{\phi_n}^2-\sigma_{\phi_n}^2\bigr|
+
\bigl|\sigma_{\phi_n}^2-\sigma_\star^2\bigr|
=o_\Pp(1)+o(1)=o_\Pp(1),
\]
so $\hat\sigma_{\phi_n}^2\overset{p}{\to}\sigma_\star^2$.
\end{proof}

\section{Sensitivity analysis: the general case and proofs}
\label{app:sens}

Section~\ref{sec:bounds} states the sensitivity analysis for a fixed embedding.
Here we state it for the pipeline target of Approach~1, which contains the fixed case.
A representation $\phi$ is a collection of encoders $\{x_j(Z)\}_{j\in J(\phi)}$ with head classes $\mathcal N_g(\phi)$ and $\mathcal N_\alpha(\phi)$; it generates the closed subspace $\cH_\phi:=\overline{\mathrm{span}}\{h(D,x_j(Z)):\ j\in J(\phi),\ h\in\mathcal N_g(\phi)\cup\mathcal N_\alpha(\phi)\}\subseteq L^2(P_W)$.
The functions $g_\phi$ and $\alpha_\phi$ are the $L^2(P)$ projections of $g^\star$ and $\alpha^\star$ onto $\cH_\phi$, so $\{g_\phi,\alpha_\phi\}\subseteq\cH_\phi$ by construction, and $\theta_\phi:=\E[m(O,g_\phi)]$ (Section~\ref{sec:subordinated-geometry}).
The fixed-embedding target of Section~\ref{sec:setup} is the case $\cH_\phi=L^2(P_{W_\phi})$, and for the pipeline's $\phi$ we have $\cH=\cH_\phi$ and $\theta_0=\theta_\phi$.
As in Proposition~\ref{prop:pipeline-mixed-bias}, we assume the compatibility condition of Section~\ref{sec:setup} on $\cH_\phi$.
Replace $(\theta_\phi,g_\phi,\alpha_\phi)$ by $(\theta_0,g_0,\alpha_0)$ from Definition~\ref{def:pipeline-target}.
Proposition~\ref{prop:pipeline-mixed-bias}(iii) gives $\theta^\star-\theta_0=\E[(g^\star-g_0)(\alpha^\star-\alpha_0)]$, so the factorization \eqref{eq:bias-factorization} and the region \eqref{eq:bounds} hold verbatim with $S_0:=(\E[(Y-g_0)^2]\,\E[\alpha_0^2])^{1/2}$ in place of $S_\phi$.
For a fixed embedding, take $\Phi=\{\phi\}$ and require $\{g_\phi,\alpha_\phi\}\subseteq\cH$; then $(\theta_0,g_0,\alpha_0)=(\theta_\phi,g_\phi,\alpha_\phi)$.
We state and prove Lemma~\ref{lem:dml-pipeline-bound-components} and Theorem~\ref{thm:bounds-pipeline} (used in Section~\ref{sec:bounds} with $\Phi=\{\phi\}$), and we prove Proposition~\ref{prop:onesided}.

\subsection{DML inference for sensitivity bounds in the fine-tuned case}\label{sec:bounds-components}

Once a sensitivity scenario $(\bar\rho,\bar C_Y,\bar C_D)$ is fixed, by~\eqref{eq:bounds} the only remaining sampling uncertainty is in the estimable components $(\theta_0,S_0)$.

Define the estimable bound components:
\[
V_{Y,0} := \E\!\left[(Y-g_0(W))^2\right],
\qquad
V_{\alpha,0} := \E\!\left[\alpha_0(W)^2\right],
\qquad
S_0 := (V_{Y,0}V_{\alpha,0})^{1/2}.
\]
These are the fine-tuned analogues of the $(\sigma_s^2,\nu_s^2,S)$ components in~\cite{chernozhukov2024long} ($V_{Y,0}$ and $V_{\alpha,0}$ correspond to their $\sigma_s^2$ and $\nu_s^2$).

In addition to the score used to estimate $\theta_0$ in Section~\ref{sec:fine-tuned}, we use the auxiliary scores $\psi_{\theta_0}$, $\psi_{V_{Y,0}}$ and $\psi_{V_{\alpha,0}}$ displayed in Lemma~\ref{lem:dml-pipeline-bound-components}.
As in~\cite{chernozhukov2024long}, these scores are Neyman-orthogonal with respect to perturbations of $(g,\alpha)$ at $(g_0,\alpha_0)$.
Orthogonality for $\psi_{V_{Y,0}}$ uses that $g_0$ is an $L^2(P)$ projection, and orthogonality for $\psi_{V_{\alpha,0}}$ uses linearity of $m$ together with the restricted Riesz property in Proposition~\ref{prop:pipeline-mixed-bias}(i).

Using the same fold partition $\{I_k\}_{k=1}^K$ as in Section~\ref{sec:fine-tuned}, run the end-to-end pipeline on each training sample $I_k^c$ to obtain fitted functions
$\hat g^{(-k)}$ and $\hat\alpha^{(-k)}$, intended to approximate the population projections $g_0$ and $\alpha_0$,
where the pipeline may internally fine-tune a fold-specific representation $\hat\phi^{(-k)}$.
Define the cross-fitted score evaluations
\[
\hat\psi_i:=\psi^\star\bigl(O_i;\hat g^{(-k(i))},\hat\alpha^{(-k(i))}\bigr),
\qquad
\hat\theta_0:=\frac{1}{n}\sum_{i=1}^n\hat\psi_i.
\]
To estimate $(V_{Y,0},V_{\alpha,0},S_0)$, set
\begin{align}
\hat V_{Y,0}
&:= \frac{1}{n}\sum_{k=1}^K\sum_{i\in I_k}\Bigl\{Y_i-\hat g^{(-k)}(W_i)\Bigr\}^2,
\label{eq:sigmahat-pipeline}
\\
\hat V_{\alpha,0}
&:= \frac{1}{n}\sum_{k=1}^K\sum_{i\in I_k}\Bigl\{2m(O_i,\hat\alpha^{(-k)})-\hat\alpha^{(-k)}(W_i)^2\Bigr\},
\label{eq:nuhat-pipeline}
\\
\hat S_0
&:= (\hat V_{Y,0}\hat V_{\alpha,0})^{1/2}.
\end{align}

Fix a sensitivity scenario $(\bar\rho,\bar C_Y,\bar C_D)$ and write
$
c:=\bar\rho\,\bar C_Y\,\bar C_D.$
The population endpoints in~\eqref{eq:bounds} can then be written as
$\theta_0^{\pm}(c):=\theta_0\pm cS_0,$
with plug-in estimators
\begin{equation}
\hat\theta_0^{\pm}(c):=\hat\theta_0\pm c\hat S_0.
\end{equation}

To obtain root-$n$ inference for these endpoints, we require $L^2(P)$ rates for the fitted pipeline nuisances.
As in Theorem~\ref{thm:pipeline-dml}, the condition is placed directly on the end-to-end fitted functions $(\hat g^{(-k)},\hat\alpha^{(-k)})$, allowing the intermediate tuning parameters $\hat\phi^{(-k)}$ to remain fold-specific and nonsmooth.

\begin{assumption}[Regularity for sensitivity bounds with a fine-tuned pipeline]
\label{ass:pipeline-bounds-dml}
Let $\mathcal H$ and $(g_0,\alpha_0,\theta_0)$ be as in Definition~\ref{def:pipeline-target}.
Assume that, for some $q>4$, $Y,g_0,\alpha_0\in L^q(P)$.
Assume that $m$ is $L^2$-bounded on $\mathcal H$ in the sense of Lemma~\ref{lem:score-conv}.

For each fold $k$, require that both $\hat g^{(-k)}$ and $\hat\alpha^{(-k)}$ belong to $\mathcal H$ and are constructed using only $I_k^c$. Also assume
\[
\max_{1\le k\le K}\norm{\hat g^{(-k)}}_{P,q}
+\max_{1\le k\le K}\norm{\hat\alpha^{(-k)}}_{P,q}=O_\Pp(1).
\]
These moment bounds prevent rare extreme predictions from dominating the score and allow $L^2(P)$ consistency to imply $L^4(P)$ consistency. The moment requirement on $g_0$ is explicit because projection onto a general linear space need not preserve the moments of $Y$.
Uniformly over folds,
\begin{equation}
\max_{1\le k\le K}\norm{\hat g^{(-k)}-g_0}_{P,2}=o_\Pp(n^{-1/4}),
\qquad
\max_{1\le k\le K}\norm{\hat\alpha^{(-k)}-\alpha_0}_{P,2}=o_\Pp(n^{-1/4}).
\label{eq:pipeline-bounds-rate}
\end{equation}
Finally, assume nondegeneracy: $V_{Y,0}:=\E[(Y-g_0(W))^2]$ and $V_{\alpha,0}:=\E[\alpha_0(W)^2]$ are bounded away from zero.
\end{assumption}

\begin{lemma}[DML for the bound components]
\label{lem:dml-pipeline-bound-components}
Suppose Assumption~\ref{ass:pipeline-bounds-dml} holds.
Define the auxiliary scores
\begin{align*}
\psi_{\theta_0}(O;\theta_0,g,\alpha)
&:= m(O,g)+\alpha(W)\{Y-g(W)\}-\theta_0,\\
\psi_{V_{Y,0}}(O;V_{Y,0},g)
&:= \{Y-g(W)\}^2-V_{Y,0},\\
\psi_{V_{\alpha,0}}(O;V_{\alpha,0},\alpha)
&:= \{2m(O,\alpha)-\alpha(W)^2\}-V_{\alpha,0}.
\end{align*}
Then
\begin{align}
\sqrt{n}(\hat\theta_0-\theta_0)
&=
\frac{1}{\sqrt{n}}\sum_{i=1}^n\psi_{\theta_0}(O_i;\theta_0,g_0,\alpha_0)+o_\Pp(1),
\label{eq:asylin-theta-pipeline}
\\
\sqrt{n}(\hat V_{Y,0}-V_{Y,0})
&=
\frac{1}{\sqrt{n}}\sum_{i=1}^n\psi_{V_{Y,0}}(O_i;V_{Y,0},g_0)+o_\Pp(1),
\label{eq:asylin-sigma-pipeline}
\\
\sqrt{n}(\hat V_{\alpha,0}-V_{\alpha,0})
&=
\frac{1}{\sqrt{n}}\sum_{i=1}^n\psi_{V_{\alpha,0}}(O_i;V_{\alpha,0},\alpha_0)+o_\Pp(1).
\label{eq:asylin-nu-pipeline}
\end{align}
Consequently, $(\hat\theta_0,\hat V_{Y,0},\hat V_{\alpha,0})$ is jointly asymptotically normal, and its asymptotic covariance matrix is consistently estimated by the empirical covariance of the corresponding cross-fitted score vector.
\end{lemma}

\begin{theorem}[DML inference for sensitivity bounds in the fine-tuned case]
\label{thm:bounds-pipeline}
Suppose Assumption~\ref{ass:pipeline-bounds-dml} holds and fix $c\ge 0$.
Let $S_0:=(V_{Y,0}V_{\alpha,0})^{1/2}$, which is positive by the nondegeneracy condition, and define the population endpoints $\theta_0^{\pm}(c):=\theta_0\pm cS_0$.
Then
\begin{equation}
\sqrt{n}\bigl(\hat\theta_0^{\pm}(c)-\theta_0^{\pm}(c)\bigr)
=
\frac{1}{\sqrt{n}}\sum_{i=1}^n\varphi_{0,\pm}(O_i)+o_\Pp(1)
\quad\Rightarrow\quad
N\bigl(0,\;\E[\varphi_{0,\pm}(O)^2]\bigr),
\label{eq:asylin-bounds-pipeline}
\end{equation}
where
\begin{equation}
\varphi_{0,\pm}(O)
:=
\psi_{\theta_0}(O;\theta_0,g_0,\alpha_0)
\pm
\frac{c}{2S_0}\Bigl\{V_{\alpha,0}\,\psi_{V_{Y,0}}(O;V_{Y,0},g_0)+V_{Y,0}\,\psi_{V_{\alpha,0}}(O;V_{\alpha,0},\alpha_0)\Bigr\}.
\label{eq:if-bounds-pipeline}
\end{equation}
Moreover, a consistent variance estimator is obtained by replacing unknown objects in~\eqref{eq:if-bounds-pipeline} with their cross-fitted estimates and taking the empirical variance.
\end{theorem}

For the Wald intervals and confidence envelope below, additionally require
\[
\E[\varphi_{0,-}(O)^2]>0,
\qquad
\E[\varphi_{0,+}(O)^2]>0.
\]
This requirement means that each endpoint has nonzero first-order sampling variation, so division by its estimated standard error is valid. Positivity of $S_0$ alone does not ensure this. Under these conditions, Theorem~\ref{thm:bounds-pipeline} gives feasible inference for the sensitivity region of $\theta^\star$.
Let $\hat s_-^2$ and $\hat s_+^2$ be the empirical variances of the two estimated influence functions in~\eqref{eq:if-bounds-pipeline}.
Besides pointwise Wald intervals for each endpoint, an asymptotic $(1-\gamma)$ confidence envelope for the entire sensitivity region is
\begin{equation}
\mathcal C_{1-\gamma}(c)
:=
\left[
\hat\theta_0^-(c)-z_{1-\gamma/2}\frac{\hat s_-}{\sqrt n},
\quad
\hat\theta_0^+(c)+z_{1-\gamma/2}\frac{\hat s_+}{\sqrt n}
\right].
\label{eq:confidence-envelope-bounds}
\end{equation}
Bonferroni's inequality and the two endpoint central limit theorems imply
\[
\Pp\!\left\{
[\theta_0^-(c),\theta_0^+(c)]
\subseteq \mathcal C_{1-\gamma}(c)
\right\}
\ge 1-\gamma+o(1).
\]
Consequently, whenever $\theta^\star$ lies in the sensitivity region under the maintained restrictions, \eqref{eq:confidence-envelope-bounds} is also an asymptotically valid confidence set for $\theta^\star$.
A critical value computed from the estimated joint Gaussian law of the two endpoint scores can replace the Bonferroni value for a shorter simultaneous envelope.

\subsection{Proofs}
\label{app:proofs-sec4}

The proofs follow the same logic as in Appendix~\ref{app:proofs-sec2}: cross-fitting lets us treat fold-specific nuisance estimates as fixed when analyzing the hold-out folds, and Neyman orthogonality reduces the impact of first-stage estimation to second order.

\subsubsection{Proof of Lemma~\ref{lem:dml-pipeline-bound-components}}

Theorem~\ref{thm:bounds-pipeline} follows from the lemma by the delta method, because the lemma gives Neyman-orthogonal, cross-fitted expansions of the estimable components of the sensitivity region.

\begin{proof}
Throughout this proof we work in the setup of Section~\ref{sec:fine-tuned}, with $(\theta_0,g_0,\alpha_0)$ from Definition~\ref{def:pipeline-target}.
Let $n_k:=|I_k|$ and define the foldwise nuisance errors
\[
\delta g^{(-k)}:=\hat g^{(-k)}-g_0,
\qquad
\delta\alpha^{(-k)}:=\hat\alpha^{(-k)}-\alpha_0.
\]
Steps~1--3 expand $\hat\theta_0$, $\hat V_{Y,0}$ and $\hat V_{\alpha,0}$ separately; each bias term is a squared or product error because $g_0$ and $\alpha_0$ are projections onto $\cH$. Step~4 stacks the expansions and proves $\hat\Sigma\overset{p}{\to}\Sigma$.

\noindent\emph{Step 1: Expansion for $\hat\theta_0$.}
Assumption~\ref{ass:pipeline-bounds-dml} implies the product-rate condition in Assumption~\ref{ass:pipeline-dml}, since
\begin{equation*}
\sqrt{n}\max_k\norm{\delta g^{(-k)}}_{P,2}\max_k\norm{\delta\alpha^{(-k)}}_{P,2}
=\bigl(n^{1/4}\max_k\norm{\delta g^{(-k)}}_{P,2}\bigr)\bigl(n^{1/4}\max_k\norm{\delta\alpha^{(-k)}}_{P,2}\bigr)=o_\Pp(1)
\end{equation*}
by~\eqref{eq:pipeline-bounds-rate}.
Therefore the expansion argument in the proof of Theorem~\ref{thm:pipeline-dml} yields~\eqref{eq:asylin-theta-pipeline}; that argument does not require a positive variance for this component.

\noindent\emph{Step 2: Expansion for $\hat V_{Y,0}$.}
For each fold $k$ define
\[
\Delta_k(O):=\{Y-\hat g^{(-k)}(W)\}^2-\{Y-g_0(W)\}^2.
\]
Using the fold decomposition~\eqref{eq:sigmahat-pipeline} and adding and subtracting $P\{(Y-g_0(W))^2\}=V_{Y,0}$ gives
\[
\hat V_{Y,0}-V_{Y,0}
=
(\Pn-P)\{(Y-g_0(W))^2\}+R^{Y}_{1n}+R^{Y}_{2n},
\]
where
\begin{align*}
R^{Y}_{1n}
&:=
\sum_{k=1}^K\frac{n_k}{n}\,(\Pn^{(k)}-P)\Delta_k,
\\
R^{Y}_{2n}
&:=
\sum_{k=1}^K\frac{n_k}{n}\,P\Delta_k.
\end{align*}

We now bound $R^{Y}_{2n}$ and $R^{Y}_{1n}$.
Expanding $\Delta_k$ gives
\[
\Delta_k
=
\{(Y-g_0)-\delta g^{(-k)}\}^2-(Y-g_0)^2
=
\{\delta g^{(-k)}\}^2-2\,\delta g^{(-k)}(Y-g_0).
\]

Because $g_0$ is the $L^2(P)$ projection of $Y$ onto $\mathcal H$ (Definition~\ref{def:pipeline-target}), the residual $Y-g_0(W)$ is orthogonal to every element of $\mathcal H$.
Conditional on the training sample, both $\hat g^{(-k)}$ and $g_0$ lie in $\mathcal H$, hence $\delta g^{(-k)}\in\mathcal H$ and
\[
P\bigl[\delta g^{(-k)}(Y-g_0)\bigr]=0.
\]
Therefore $P\Delta_k=P\bigl[\{\delta g^{(-k)}\}^2\bigr]=\norm{\delta g^{(-k)}}_{P,2}^2$, so
\[
|R^{Y}_{2n}|
\le
\max_{1\le k\le K}\norm{\delta g^{(-k)}}_{P,2}^2,
\qquad
\sqrt{n}|R^{Y}_{2n}|
\le
\sqrt{n}\,\max_k\norm{\delta g^{(-k)}}_{P,2}^2=o_\Pp(1)
\]
by~\eqref{eq:pipeline-bounds-rate}.

For the empirical process term, conditional on the training sample $\{O_i:i\in I_k^c\}$ the function $\Delta_k$ is fixed and the observations in $I_k$ are i.i.d.\ draws from $P$.
The same conditional Chebyshev argument as in Step~3 of the proof of Theorem~\ref{thm:dml-expansion} yields
\[
\sqrt{n}\,|R^{Y}_{1n}|=O_\Pp\!\Bigl(\max_{1\le k\le K}\norm{\Delta_k}_{P,2}\Bigr).
\]
Using H\"older,
\[
\norm{\Delta_k}_{P,2}
\le
\norm{\{\delta g^{(-k)}\}^2}_{P,2}
+2\norm{\delta g^{(-k)}(Y-g_0)}_{P,2}
\le
\norm{\delta g^{(-k)}}_{P,4}^2
+2\norm{\delta g^{(-k)}}_{P,4}\,\norm{Y-g_0}_{P,4}.
\]
Assumption~\ref{ass:pipeline-bounds-dml} gives $Y,g_0\in L^q(P)$ for $q>4$, so $\norm{Y-g_0}_{P,4}<\infty$.
The same assumption and~\eqref{eq:pipeline-bounds-rate} give
\[
\max_k\norm{\delta g^{(-k)}}_{P,2}=o_\Pp(1),
\qquad
\max_k\norm{\hat g^{(-k)}}_{P,q}=O_\Pp(1).
\]
Interpolation as in Appendix~\ref{app:proofs-sec2} gives $L^4(P)$ consistency of $\delta g^{(-k)}$ uniformly over folds. The displayed bound on $\norm{\Delta_k}_{P,2}$ therefore tends to zero, implying $\sqrt n R^{Y}_{1n}=o_\Pp(1)$.

Combining the pieces gives
\[
\sqrt{n}(\hat V_{Y,0}-V_{Y,0})
=
\Gn\bigl(\{Y-g_0(W)\}^2\bigr)+o_\Pp(1)
=
\frac{1}{\sqrt{n}}\sum_{i=1}^n\psi_{V_{Y,0}}(O_i;V_{Y,0},g_0)+o_\Pp(1),
\]
which is~\eqref{eq:asylin-sigma-pipeline}.

\noindent\emph{Step 3: Expansion for $\hat V_{\alpha,0}$.}
For each fold $k$ define
\[
\widetilde\Delta_k(O)
:=
\Bigl\{2m(O,\hat\alpha^{(-k)})-\hat\alpha^{(-k)}(W)^2\Bigr\}
-
\Bigl\{2m(O,\alpha_0)-\alpha_0(W)^2\Bigr\}.
\]
Using~\eqref{eq:nuhat-pipeline} and adding and subtracting $P\{2m(O,\alpha_0)-\alpha_0(W)^2\}=V_{\alpha,0}$ yields
\[
\hat V_{\alpha,0}-V_{\alpha,0}
=
(\Pn-P)\Bigl\{2m(O,\alpha_0)-\alpha_0(W)^2\Bigr\}+R^{\alpha}_{1n}+R^{\alpha}_{2n},
\]
where
\begin{align*}
R^{\alpha}_{1n}
&:=
\sum_{k=1}^K\frac{n_k}{n}\,(\Pn^{(k)}-P)\widetilde\Delta_k,
\\
R^{\alpha}_{2n}
&:=
\sum_{k=1}^K\frac{n_k}{n}\,P\widetilde\Delta_k.
\end{align*}

By linearity of $m$ and the definition of $\widetilde\Delta_k$,
\[
\widetilde\Delta_k
=
2m(O,\delta\alpha^{(-k)})-2\alpha_0(W)\,\delta\alpha^{(-k)}(W)-\{\delta\alpha^{(-k)}(W)\}^2.
\]
Conditional on the training sample, $\delta\alpha^{(-k)}\in\mathcal H$.
Proposition~\ref{prop:pipeline-mixed-bias}(i) therefore implies
\[
P[m(O,\delta\alpha^{(-k)})]
=
P[\delta\alpha^{(-k)}(W)\alpha_0(W)],
\]
so the first two terms cancel in expectation and
\[
P\widetilde\Delta_k=-P\bigl[\{\delta\alpha^{(-k)}\}^2\bigr]=-\norm{\delta\alpha^{(-k)}}_{P,2}^2.
\]
Thus
\[
|R^{\alpha}_{2n}|
\le
\max_{1\le k\le K}\norm{\delta\alpha^{(-k)}}_{P,2}^2,
\qquad
\sqrt{n}|R^{\alpha}_{2n}|=o_\Pp(1)
\]
by~\eqref{eq:pipeline-bounds-rate}.

For $R^{\alpha}_{1n}$, the same conditional Chebyshev argument yields
\[
\sqrt{n}|R^{\alpha}_{1n}|=O_\Pp\!\Bigl(\max_{1\le k\le K}\norm{\widetilde\Delta_k}_{P,2}\Bigr).
\]
Using the $L^2$-boundedness of $m$ (Lemma~\ref{lem:score-conv}) and H\"older,
\[
\begin{aligned}
\norm{\widetilde\Delta_k}_{P,2}
&\le
2\norm{m(O,\delta\alpha^{(-k)})}_{P,2}+2\norm{\alpha_0\,\delta\alpha^{(-k)}}_{P,2}+\norm{\{\delta\alpha^{(-k)}\}^2}_{P,2}
\\
&\le
2C_m\norm{\delta\alpha^{(-k)}}_{P,2}+2\norm{\alpha_0}_{P,4}\,\norm{\delta\alpha^{(-k)}}_{P,4}+\norm{\delta\alpha^{(-k)}}_{P,4}^2.
\end{aligned}
\]
Assumption~\ref{ass:pipeline-bounds-dml} gives $\alpha_0\in L^q(P)$ with $q>4$, hence $\norm{\alpha_0}_{P,4}<\infty$.
The same interpolation argument as in Step~2 yields $\max_k\norm{\delta\alpha^{(-k)}}_{P,4}=o_\Pp(1)$.
Therefore $\max_k\norm{\widetilde\Delta_k}_{P,2}=o_\Pp(1)$ and $\sqrt{n}R^{\alpha}_{1n}=o_\Pp(1)$.

Combining the pieces gives
\[
\sqrt{n}(\hat V_{\alpha,0}-V_{\alpha,0})
=
\Gn\Bigl(2m(O,\alpha_0)-\alpha_0(W)^2\Bigr)+o_\Pp(1)
=
\frac{1}{\sqrt{n}}\sum_{i=1}^n\psi_{V_{\alpha,0}}(O_i;V_{\alpha,0},\alpha_0)+o_\Pp(1),
\]
which is~\eqref{eq:asylin-nu-pipeline}.

\noindent\emph{Step 4: Joint asymptotic normality and covariance estimation.}
Stacking~\eqref{eq:asylin-theta-pipeline}--\eqref{eq:asylin-nu-pipeline} yields
\[
\sqrt{n}
\begin{pmatrix}
\hat\theta_0-\theta_0\\
\hat V_{Y,0}-V_{Y,0}\\
\hat V_{\alpha,0}-V_{\alpha,0}
\end{pmatrix}
=
\frac{1}{\sqrt{n}}\sum_{i=1}^n
\begin{pmatrix}
\psi_{\theta_0}(O_i;\theta_0,g_0,\alpha_0)\\
\psi_{V_{Y,0}}(O_i;V_{Y,0},g_0)\\
\psi_{V_{\alpha,0}}(O_i;V_{\alpha,0},\alpha_0)
\end{pmatrix}
+o_\Pp(1).
\]
The summands are i.i.d.\ with finite second moments, so the multivariate CLT gives joint asymptotic normality.

For covariance estimation, define the (centered) score vector
\[
\Psi(O)
:=
\begin{pmatrix}
\psi_{\theta_0}(O;\theta_0,g_0,\alpha_0)\\
\psi_{V_{Y,0}}(O;V_{Y,0},g_0)\\
\psi_{V_{\alpha,0}}(O;V_{\alpha,0},\alpha_0)
\end{pmatrix},
\qquad
\Sigma:=\E\bigl[\Psi(O)\Psi(O)^{\top}\bigr].
\]
For each observation $i\in I_k$, define the cross-fitted plug-in version
\[
\hat\Psi_i
:=
\begin{pmatrix}
\psi_{\theta_0}(O_i;\hat\theta_0,\hat g^{(-k)},\hat\alpha^{(-k)})\\
\psi_{V_{Y,0}}(O_i;\hat V_{Y,0},\hat g^{(-k)})\\
\psi_{V_{\alpha,0}}(O_i;\hat V_{\alpha,0},\hat\alpha^{(-k)})
\end{pmatrix}.
\]
By the definitions of $(\hat\theta_0,\hat V_{Y,0},\hat V_{\alpha,0})$, each component is empirically centered, so $n^{-1}\sum_{i=1}^n\hat\Psi_i=0$.
We therefore take as covariance estimator
\[
\hat\Sigma:=\frac{1}{n}\sum_{i=1}^n \hat\Psi_i\hat\Psi_i^{\top}.
\]

We show that $\hat\Sigma\overset{p}{\to}\Sigma$.
Decompose
\[
\hat\Sigma-\Sigma
=
\underbrace{\frac{1}{n}\sum_{i=1}^n\Bigl(\hat\Psi_i\hat\Psi_i^{\top}-\Psi_i\Psi_i^{\top}\Bigr)}_{A_n}
+
\underbrace{\Bigl\{\frac{1}{n}\sum_{i=1}^n\Psi_i\Psi_i^{\top}-\Sigma\Bigr\}}_{B_n},
\qquad
\Psi_i:=\Psi(O_i).
\]
Since $\E\norm{\Psi(O)}^2<\infty$ under Assumption~\ref{ass:pipeline-bounds-dml}, the (matrix) law of large numbers gives $B_n\overset{p}{\to}0$.

For $A_n$, note the identity
\[
\hat\Psi_i\hat\Psi_i^{\top}-\Psi_i\Psi_i^{\top}
=
(\hat\Psi_i-\Psi_i)\hat\Psi_i^{\top}
+
\Psi_i(\hat\Psi_i-\Psi_i)^{\top},
\]
hence, for the Frobenius norm $\norm{\cdot}_F$,
\[
\norm{\hat\Psi_i\hat\Psi_i^{\top}-\Psi_i\Psi_i^{\top}}_F
\le
\norm{\hat\Psi_i-\Psi_i}\,\bigl(\norm{\hat\Psi_i}+\norm{\Psi_i}\bigr).
\]
By Cauchy--Schwarz,
\begin{equation}
\norm{A_n}_F
\le
\Bigl(\frac{1}{n}\sum_{i=1}^n\norm{\hat\Psi_i-\Psi_i}^2\Bigr)^{1/2}
\Bigl(\frac{1}{n}\sum_{i=1}^n(\norm{\hat\Psi_i}+\norm{\Psi_i})^2\Bigr)^{1/2}.
\label{eq:cov-bound}
\end{equation}

We now control the two factors in~\eqref{eq:cov-bound}. Define the score with the fitted nuisances and the true values $(\theta_0,V_{Y,0},V_{\alpha,0})$, which is fixed given the training sample,
\[
\hat\Psi^{(-k)}(O):=
\begin{pmatrix}
\psi_{\theta_0}(O;\theta_0,\hat g^{(-k)},\hat\alpha^{(-k)})\\
\psi_{V_{Y,0}}(O;V_{Y,0},\hat g^{(-k)})\\
\psi_{V_{\alpha,0}}(O;V_{\alpha,0},\hat\alpha^{(-k)})
\end{pmatrix}.
\]
We first show $\max_k\norm{\hat\Psi^{(-k)}-\Psi}_{P,2}=o_\Pp(1)$.
For the first component, bounded linearity and H\"older give
\begin{align*}
&\norm{\psi_{\theta_0}(\cdot;\theta_0,\hat g^{(-k)},\hat\alpha^{(-k)})
-\psi_{\theta_0}(\cdot;\theta_0,g_0,\alpha_0)}_{P,2}\\
&\quad\le C_m\norm{\delta g^{(-k)}}_{P,2}
+\norm{\delta\alpha^{(-k)}}_{P,4}\norm{Y-g_0}_{P,4}\\
&\qquad+\norm{\alpha_0}_{P,4}\norm{\delta g^{(-k)}}_{P,4}
+\norm{\delta\alpha^{(-k)}}_{P,4}\norm{\delta g^{(-k)}}_{P,4}
=o_\Pp(1).
\end{align*}
Here interpolation between $L^2(P)$ consistency and the $L^q(P)$ bounds gives $L^4(P)$ consistency uniformly over folds. The second and third component differences are $\Delta_k$ and $\widetilde\Delta_k$ from Steps~2--3, whose $L^2(P)$ norms also tend to zero.
Since these differences are fixed conditional on the training sample, conditional Markov's inequality gives
\[
\max_k\Pn^{(k)}\norm{\hat\Psi^{(-k)}-\Psi}^2=o_\Pp(1).
\]
For $i\in I_k$, the vector $\hat\Psi_i$ differs from $\hat\Psi^{(-k)}(O_i)$ by the constant centering vector
\[
-\begin{pmatrix}
\hat\theta_0-\theta_0\\
\hat V_{Y,0}-V_{Y,0}\\
\hat V_{\alpha,0}-V_{\alpha,0}
\end{pmatrix},
\]
which is $O_\Pp(n^{-1/2})$ by Steps~1--3. Thus averaging across folds yields
\[
\frac1n\sum_{i=1}^n\norm{\hat\Psi_i-\Psi_i}^2=o_\Pp(1).
\]
Separating this final centering step avoids conditioning on full-sample estimates as though they were training-measurable.

Second, since
\[
\frac{1}{n}\sum_{i=1}^n\norm{\hat\Psi_i}^2
\le
2\frac{1}{n}\sum_{i=1}^n\norm{\Psi_i}^2
+2\frac{1}{n}\sum_{i=1}^n\norm{\hat\Psi_i-\Psi_i}^2,
\]
the second factor in~\eqref{eq:cov-bound} is $O_\Pp(1)$.
Hence~\eqref{eq:cov-bound} implies $\norm{A_n}_F=o_\Pp(1)$.
\end{proof}

\subsubsection{Proof of Theorem~\ref{thm:bounds-pipeline}}
\begin{proof}
Fix $c\ge 0$.
Write $S_0=(V_{Y,0}V_{\alpha,0})^{1/2}$ and recall $\hat\theta_0^{\pm}(c)=\hat\theta_0\pm c\hat S_0$ with $\hat S_0=(\hat V_{Y,0}\hat V_{\alpha,0})^{1/2}$.

\noindent\emph{Step 1: A first-order expansion of $\hat S_0$.}
Because $S_0>0$, the map $(u,v)\mapsto (uv)^{1/2}$ is continuously differentiable in a neighborhood of $(V_{Y,0},V_{\alpha,0})$.
A Taylor expansion yields
\begin{equation}
\hat S_0-S_0
=
\frac{1}{2S_0}\Bigl\{V_{\alpha,0}(\hat V_{Y,0}-V_{Y,0})+V_{Y,0}(\hat V_{\alpha,0}-V_{\alpha,0})\Bigr\}
+\varrho_n,
\label{eq:delta-S0}
\end{equation}
where $\varrho_n=o_\Pp(n^{-1/2})$ by the delta method, since Lemma~\ref{lem:dml-pipeline-bound-components} gives root-$n$ consistency of $(\hat V_{Y,0},\hat V_{\alpha,0})$.

\noindent\emph{Step 2: Combine with the linear representations for the components.}
Using~\eqref{eq:delta-S0} and Lemma~\ref{lem:dml-pipeline-bound-components},
\begin{align*}
\sqrt{n}\bigl(\hat\theta_0^{\pm}(c)-\theta_0^{\pm}(c)\bigr)
&=
\sqrt{n}(\hat\theta_0-\theta_0)\pm c\sqrt{n}(\hat S_0-S_0)
\\
&=
\frac{1}{\sqrt{n}}\sum_{i=1}^n\psi_{\theta_0}(O_i;\theta_0,g_0,\alpha_0)
\\
&\quad\pm
\frac{c}{2S_0}\frac{1}{\sqrt{n}}\sum_{i=1}^n
\Bigl\{V_{\alpha,0}\,\psi_{V_{Y,0}}(O_i;V_{Y,0},g_0)+V_{Y,0}\,\psi_{V_{\alpha,0}}(O_i;V_{\alpha,0},\alpha_0)\Bigr\}
\\
&\quad+o_\Pp(1).
\end{align*}
This is exactly~\eqref{eq:asylin-bounds-pipeline} with influence function~\eqref{eq:if-bounds-pipeline}.
The asymptotic normality statement follows from the CLT applied to the i.i.d.\ summands $\varphi_{0,\pm}(O_i)$.

\noindent\emph{Step 3: Variance estimation.}
For each $i\in I_k$, define the plug-in influence value
\begin{align*}
\hat\varphi_{0,\pm,i}
:=
&\;\psi_{\theta_0}\!\bigl(O_i;\hat\theta_0,\hat g^{(-k)},\hat\alpha^{(-k)}\bigr)
\\
&\;\pm
\frac{c}{2\hat S_0}
\Bigl\{
\hat V_{\alpha,0}\,\psi_{V_{Y,0}}\!\bigl(O_i;\hat V_{Y,0},\hat g^{(-k)}\bigr)
+
\hat V_{Y,0}\,\psi_{V_{\alpha,0}}\!\bigl(O_i;\hat V_{\alpha,0},\hat\alpha^{(-k)}\bigr)
\Bigr\},
\end{align*}
where $\hat S_0=(\hat V_{Y,0}\hat V_{\alpha,0})^{1/2}$.
Let
\[
\hat V_{0,\pm}
:=
\frac{1}{n}\sum_{i=1}^n
\bigl(\hat\varphi_{0,\pm,i}-\bar{\hat\varphi}_{0,\pm}\bigr)^2,
\qquad
\bar{\hat\varphi}_{0,\pm}:=\frac{1}{n}\sum_{i=1}^n\hat\varphi_{0,\pm,i}.
\]
Because each of the component scores is empirically centered and the multipliers are constant in $i$, we have $\bar{\hat\varphi}_{0,\pm}=0$ exactly; we keep the centering in the definition to match standard software output.

Since $\E[\varphi_{0,\pm}(O)]=0$, the asymptotic variance in~\eqref{eq:asylin-bounds-pipeline} equals $\E[\varphi_{0,\pm}(O)^2]$.
We show that $\hat V_{0,\pm}\overset{p}{\to}\E[\varphi_{0,\pm}(O)^2]$.

Write $\varphi_i:=\varphi_{0,\pm}(O_i)$ and $\hat\varphi_i:=\hat\varphi_{0,\pm,i}$ for brevity.
Using the triangle inequality,
\[
\bigl|\hat V_{0,\pm}-\E[\varphi_{0,\pm}(O)^2]\bigr|
\le
\biggl|\frac{1}{n}\sum_{i=1}^n\bigl(\hat\varphi_i^2-\varphi_i^2\bigr)\biggr|
+
\biggl|\frac{1}{n}\sum_{i=1}^n\varphi_i^2-\E[\varphi_{0,\pm}(O)^2]\biggr|.
\]
The second term converges to zero by the law of large numbers, since $\varphi_{0,\pm}(O)\in L^2(P)$ under Assumption~\ref{ass:pipeline-bounds-dml}.
For the first term, use $a^2-b^2=(a-b)(a+b)$ and Cauchy--Schwarz to obtain
\begin{equation}
\frac{1}{n}\sum_{i=1}^n\bigl|\hat\varphi_i^2-\varphi_i^2\bigr|
\le
\Bigl(\frac{1}{n}\sum_{i=1}^n(\hat\varphi_i-\varphi_i)^2\Bigr)^{1/2}
\Bigl(\frac{1}{n}\sum_{i=1}^n(\hat\varphi_i+\varphi_i)^2\Bigr)^{1/2}.
\label{eq:var-diff-bound}
\end{equation}
The second factor in~\eqref{eq:var-diff-bound} is $O_\Pp(1)$ because
\[
\frac{1}{n}\sum_{i=1}^n\hat\varphi_i^2
\le
2\frac{1}{n}\sum_{i=1}^n\varphi_i^2
+
2\frac{1}{n}\sum_{i=1}^n(\hat\varphi_i-\varphi_i)^2,
\]
and $n^{-1}\sum\varphi_i^2\overset{p}{\to}\E[\varphi_{0,\pm}(O)^2]<\infty$.
Thus it suffices to prove
\begin{equation}
\frac{1}{n}\sum_{i=1}^n(\hat\varphi_i-\varphi_i)^2=o_\Pp(1).
\label{eq:varphi-L2}
\end{equation}

To verify~\eqref{eq:varphi-L2}, observe that $\hat\varphi_i$ is obtained from the three component scores in Lemma~\ref{lem:dml-pipeline-bound-components} and from the scalars $(\hat V_{Y,0},\hat V_{\alpha,0},\hat S_0)$ through smooth algebraic operations.
Because $S_0>0$ and $\hat S_0\overset{p}{\to}S_0$, we have $\hat S_0^{-1}=O_\Pp(1)$ and $\hat S_0$ is bounded away from zero with probability approaching one.

Write $\Psi(O)$ and $\hat\Psi_i$ as in Step~4 of the proof of Lemma~\ref{lem:dml-pipeline-bound-components}, and let $\Psi_{i,1},\Psi_{i,2},\Psi_{i,3}$ denote its components.
Then $\varphi_i$ is a fixed linear combination of $(\Psi_{i,1},\Psi_{i,2},\Psi_{i,3})$ with coefficients depending on $(V_{Y,0},V_{\alpha,0},S_0)$, and $\hat\varphi_i$ is the same combination with $(\Psi_{i,j},V_{Y,0},V_{\alpha,0},S_0)$ replaced by $(\hat\Psi_{i,j},\hat V_{Y,0},\hat V_{\alpha,0},\hat S_0)$.
Using the triangle inequality and the bound $\hat S_0^{-1}=O_\Pp(1)$, we can write
\[
\begin{aligned}
|\hat\varphi_i-\varphi_i|
&\le |\hat\Psi_{i,1}-\Psi_{i,1}|
+C\,\sum_{j=2}^3|\hat\Psi_{i,j}-\Psi_{i,j}|\\
&\quad+C\bigl(|\hat V_{Y,0}-V_{Y,0}|+|\hat V_{\alpha,0}-V_{\alpha,0}|+|\hat S_0-S_0|\bigr)
\bigl(|\Psi_{i,2}|+|\Psi_{i,3}|\bigr)
\end{aligned}
\]
for a constant $C$ depending only on $(c,V_{Y,0},V_{\alpha,0},S_0)$.
Taking squares, averaging, and applying Cauchy--Schwarz shows that the mean square in~\eqref{eq:varphi-L2} is controlled by
(i) $n^{-1}\sum_i\norm{\hat\Psi_i-\Psi_i}^2$, which is $o_\Pp(1)$ by Step~4 of the proof of Lemma~\ref{lem:dml-pipeline-bound-components}, and
(ii) the root-$n$ consistency of $(\hat V_{Y,0},\hat V_{\alpha,0})$ (hence also $\hat S_0$) together with $\E[\Psi_{2}^2+\Psi_{3}^2]<\infty$.
Therefore~\eqref{eq:varphi-L2} holds, which in turn implies $\hat V_{0,\pm}\overset{p}{\to}\E[\varphi_{0,\pm}(O)^2]$.
\end{proof}

\subsection{Proof of Proposition~\ref{prop:onesided}}
\begin{proof}
Write $T^-:=\sqrt n(L-\theta_\phi^-)/\hat s_-$.
Theorem~\ref{thm:bounds-pipeline}, specialized to $\Phi=\{\phi\}$ with $\{g_\phi,\alpha_\phi\}\subseteq\cH$ so that $(\theta_0,S_0)=(\theta_\phi,S_\phi)$, gives $\sqrt n(\hat\theta_\phi^--\theta_\phi^-)/s_-\Rightarrow N(0,1)$ and $\hat s_-/s_-\to1$ in probability.
Hence $T^-\Rightarrow N(0,1)-z_{1-\gamma}$, and $\Pp\{\theta_\phi^-<L\}=\Pp\{T^->0\}\to\gamma$.
The upper endpoint is symmetric, with $T^+:=\sqrt n(\theta_\phi^+-U)/\hat s_+$.
Under a valid scenario, $\theta^\star\in[\theta_\phi^-,\theta_\phi^+]$ by \eqref{eq:bounds}, so $\{\theta^\star<L\}\subseteq\{\theta_\phi^-<L\}$ and $\{\theta^\star>U\}\subseteq\{\theta_\phi^+>U\}$, which gives \eqref{eq:onesided-target}.
For the last claim, $\theta^\star$ lies at distance at least $cS_\phi$ from one endpoint; say $\theta_\phi^+-\theta^\star\ge cS_\phi$, the other case being symmetric.
Then $\Pp\{\theta^\star>U\}\le\Pp\{U-\theta_\phi^+<-cS_\phi\}\to0$, because $U-\theta_\phi^+=O_\Pp(n^{-1/2})$ and $cS_\phi$ is bounded away from zero.
If $L\le U$, the events $\{\theta^\star<L\}$ and $\{\theta^\star>U\}$ are disjoint, so $\Pp\{L\le\theta^\star\le U\}\ge1-\Pp\{\theta^\star<L\}-\Pp\{\theta^\star>U\}-\Pp\{L>U\}$; the last probability tends to zero because $U-L=2c\hat S_\phi+O_\Pp(n^{-1/2})$ and $\hat S_\phi\to S_\phi$ in probability.
Taking limits gives $\liminf_n\Pp\{L\le\theta^\star\le U\}\ge1-\gamma$.
\end{proof}

\section{Approach 1: the pipeline target}
\label{app:pipeline}

This appendix states the mixed-bias identities and the regularity conditions behind Theorem~\ref{thm:pipeline-dml}.
As in Section~\ref{sec:fine-tuned}, $\cH=\cH_\phi$ is the closed span of the functions the pipeline can express, and $(g_0,\alpha_0,\theta_0)$ are given by Definition~\ref{def:pipeline-target}.

\begin{proposition}[{Mixed-bias identities for the pipeline target}]
\label{prop:pipeline-mixed-bias}
Assume the compatibility condition in Section~\ref{sec:setup} and~\eqref{eq:rr-long}.
Let $(g_0,\alpha_0,\theta_0)$ be as in Definition~\ref{def:pipeline-target}.
Then the following hold.

\smallskip

\noindent (i) (\emph{Restricted Riesz representation.})
For every $h\in\mathcal H$,
\[
\E[m(O,h)] = \E\!\left[h(W)\,\alpha_0(W)\right].
\]

\smallskip

\noindent (ii) (\emph{Mixed-bias property.})
For any $g,\alpha\in\mathcal H$,
\begin{equation}
\E\!\left[\psi^\star(O;g,\alpha)\right]-\theta_0
=
-\,\E\!\left[(g-g_0)(\alpha-\alpha_0)\right].
\label{eq:mixed-bias-pipeline}
\end{equation}

\smallskip

\noindent (iii) (\emph{Approximation bias.})
The gap between the full-information and the pipeline targets satisfies
\begin{equation}
\theta^\star-\theta_0
=
\E\!\left[(g^\star-g_0)(\alpha^\star-\alpha_0)\right],
\qquad
|\theta^\star-\theta_0|
\le
\norm{g^\star-g_0}_{P,2}\,\norm{\alpha^\star-\alpha_0}_{P,2}.
\label{eq:ovb-pipeline}
\end{equation}
\end{proposition}

By part (ii), the bias of the orthogonal score at $g,\alpha\in\cH$ is the \emph{mixed bias} $-\E[(g-g_0)(\alpha-\alpha_0)]$.
In particular, if either nuisance is correct (or nearly correct) in $L^2(P)$, the bias is second order.
By part (iii), the distance between the pipeline target $\theta_0$ and the full-information target $\theta^\star$ has the same mixed-bias form.

\begin{assumption}[DML regularity for an end-to-end pipeline]
\label{ass:pipeline-dml}
Let $\mathcal H$ and $(g_0,\alpha_0,\theta_0)$ be as in Definition~\ref{def:pipeline-target}.
Assume that, for some $q>4$, $Y,g_0,\alpha_0\in L^q(P)$, and that $\sigma_0^2:=\Var(\psi^\star(O;g_0,\alpha_0))$ is bounded away from zero and infinity. Write $\mathrm{CI}_0(1-\gamma):=[\hat\theta_0\pm z_{1-\gamma/2}\,\hat\sigma_0/\sqrt n]$.
Assume that $m$ is $L^2$-bounded on $\mathcal H$ in the sense of Lemma~\ref{lem:score-conv}.

For each fold $k$, require that both $\hat g^{(-k)}$ and $\hat\alpha^{(-k)}$ belong to $\mathcal H$ and are constructed using only $I_k^c$. Also assume
\[
\max_{1\le k\le K}\norm{\hat g^{(-k)}}_{P,q}
+\max_{1\le k\le K}\norm{\hat\alpha^{(-k)}}_{P,q}=O_\Pp(1).
\]
These moment bounds prevent rare extreme predictions from dominating the score and allow $L^2(P)$ consistency to imply $L^4(P)$ consistency. The moment requirement on $g_0$ is explicit because projection onto a general linear space need not preserve the moments of $Y$.
For each fold $k$, the cross-fitted nuisance learners satisfy
\[
\begin{aligned}
& \norm{\hat g^{(-k)}-g_0}_{P,2} = o_\Pp(1), \quad
\norm{\hat\alpha^{(-k)}-\alpha_0}_{P,2} = o_\Pp(1),\\
& \sqrt{n}\,\norm{\hat g^{(-k)}-g_0}_{P,2}\,\norm{\hat\alpha^{(-k)}-\alpha_0}_{P,2}
= o_\Pp(1),
\end{aligned}
\]
uniformly over $k\in\{1,\dots,K\}$.
\end{assumption}

\subsection{Proof of Proposition~\ref{prop:pipeline-mixed-bias}}

\begin{proof}
We write $\mathcal H\subseteq L^2(P_W)$ for the closed linear subspace in Definition~\ref{def:pipeline-target}, and $\Pi_{\mathcal H}$ for the $L^2(P)$ projection onto $\mathcal H$.

\noindent\emph{Step 1: The projection interpretation of $\alpha_0$.}
By the long Riesz representation~\eqref{eq:rr-long}, for every $h\in L^2(P_W)$,
\[
\E[m(O,h)] = \E\!\left[h(W)\,\alpha^\star(W)\right].
\]
In particular, for any $\alpha\in\mathcal H$,
\[
\E\!\left[\alpha(W)^2-2m(O,\alpha)\right]
=
\E\!\left[\alpha(W)^2-2\alpha(W)\alpha^\star(W)\right]
=
\norm{\alpha-\alpha^\star}_{P,2}^2-\norm{\alpha^\star}_{P,2}^2.
\]
Since $\mathcal H$ is a closed linear subspace of the Hilbert space $L^2(P_W)$, the minimizer is the $L^2(P)$ projection of $\alpha^\star$ onto $\mathcal H$.
Thus $\alpha_0=\Pi_{\mathcal H}\alpha^\star$, and $\alpha^\star-\alpha_0$ is orthogonal to $\mathcal H$.

\noindent\emph{Step 2: Part (i).}
Fix $h\in\mathcal H$.
Using the long Riesz representation and the orthogonality of $\alpha^\star-\alpha_0$ to $\mathcal H$,
\[
\E[m(O,h)]
=
\E[h\,\alpha^\star]
=
\E[h\,\alpha_0]+\E[h(\alpha^\star-\alpha_0)]
=
\E[h\,\alpha_0],
\]
which proves the restricted Riesz property.

\noindent\emph{Step 3: The projection interpretation of $g_0$.}
Since $g^\star(W)=\E[Y\mid W]$, the residual $Y-g^\star(W)$ is orthogonal in $L^2(P)$ to every square-integrable function of $W$.
In particular, it is orthogonal to $\mathcal H$.
Therefore, minimizing $\E[(Y-g(W))^2]$ over $g\in\mathcal H$ is equivalent to minimizing $\E[(g^\star-g(W))^2]$ over $g\in\mathcal H$, so $g_0=\Pi_{\mathcal H}g^\star$ and $g^\star-g_0$ is orthogonal to $\mathcal H$.

\noindent\emph{Step 4: Part (ii).}
Fix $g,\alpha\in\mathcal H$.
Using $\theta_0=\E[m(O,g_0)]$,
\[
\E[\psi^\star(O;g,\alpha)]-\theta_0
=
\E[m(O,g-g_0)]+\E\!\left[\alpha(W)\{Y-g(W)\}\right].
\]
By part (i), $\E[m(O,g-g_0)]=\E[(g-g_0)\alpha_0]$ because $g-g_0\in\mathcal H$.
Also,
\[
\E\!\left[\alpha(Y-g)\right]
=
\E\!\left[\alpha(Y-g_0)\right]-\E\!\left[\alpha(g-g_0)\right].
\]
Since $Y-g_0$ is orthogonal to $\mathcal H$ and $\alpha\in\mathcal H$, we have $\E[\alpha(Y-g_0)]=0$.
Combining these identities yields
\[
\E[\psi^\star(O;g,\alpha)]-\theta_0
=
\E[(g-g_0)\alpha_0]-\E[\alpha(g-g_0)]
=
-\,\E[(g-g_0)(\alpha-\alpha_0)],
\]
which is~\eqref{eq:mixed-bias-pipeline}.

\noindent\emph{Step 5: Part (iii).}
By the long Riesz representation,
\[
\theta^\star-\theta_0
=
\E[m(O,g^\star-g_0)]
=
\E[(g^\star-g_0)\alpha^\star].
\]
Because $g^\star-g_0$ is orthogonal to $\mathcal H$ and $\alpha_0\in\mathcal H$, we have $\E[(g^\star-g_0)\alpha_0]=0$, so
\[
\E[(g^\star-g_0)\alpha^\star]
=
\E[(g^\star-g_0)(\alpha^\star-\alpha_0)].
\]
This proves the first identity in~\eqref{eq:ovb-pipeline}, and the inequality is immediate from Cauchy--Schwarz.
\end{proof}

\subsection{Proof of Theorem~\ref{thm:pipeline-dml}}

\begin{proof}
The proof follows the same three-step decomposition as the proof of Theorem~\ref{thm:dml-expansion}.
The only substantive change is the population algebra: in place of the short mixed-bias identity used in Step~2 of that proof, we use Proposition~\ref{prop:pipeline-mixed-bias}(ii), which yields
\[
P\Bigl\{\psi^\star(O;\hat g^{(-k)},\hat\alpha^{(-k)})-\psi^\star(O;g_0,\alpha_0)\Bigr\}
=
-\,P\bigl[(\hat g^{(-k)}-g_0)(\hat\alpha^{(-k)}-\alpha_0)\bigr].
\]
Under Assumption~\ref{ass:pipeline-dml}, this term is $o_\Pp(n^{-1/2})$ uniformly over folds, and we handle the remaining empirical-process terms exactly as before because the fitted nuisances are fixed conditional on the training sample.
In particular, interpolation between the assumed $L^2(P)$ convergence and bounded $L^q(P)$ norms gives $L^4(P)$ convergence of both nuisance errors. H\"older's inequality and bounded linearity of $m$ then give convergence of the fitted score in $L^2(P)$.
This yields the linear representation~\eqref{eq:pipeline-expansion}.

The asymptotic normality statement follows from the classical CLT applied to the i.i.d.\ summands $\psi^\star(O_i;g_0,\alpha_0)-\theta_0$.
Consistency of the cross-fitted variance estimator is proved in the same way as in the last part of the proof of Theorem~\ref{thm:dml-expansion}, using that
\[
\max_{1\le k\le K}\norm{\psi^\star(\cdot;\hat g^{(-k)},\hat\alpha^{(-k)})-\psi^\star(\cdot;g_0,\alpha_0)}_{P,2}
= o_\Pp(1)
\]
under Assumption~\ref{ass:pipeline-dml}.
\end{proof}

\section{Approach 1: learning the projections onto $\mathcal H$}
\label{sec:learn-H}

Assumption~\ref{ass:pipeline-dml} concerns only the $L^2(P)$ errors of the fitted nuisances $(\hat g^{(-k)},\hat\alpha^{(-k)})$.
In an end-to-end pipeline, the orthogonal score depends on the embedding rule and the heads only through the resulting pair of predictions.
We verify the product-rate requirement of Assumption~\ref{ass:pipeline-dml} by reducing it to learning-theoretic bounds for two quadratic empirical objectives, one for the regression and one for the Riesz representer.

Because we cross-fit, all learning is foldwise.
For each fold $k$, we train on the complement $I_k^c$ and evaluate on the held-out fold $I_k$.
To keep formulas readable, we adopt the standard cross-fitting convention that, whenever we write an empirical objective $\Pn(\cdot)$ used to \emph{train} $\hat g^{(-k)}$ or $\hat\alpha^{(-k)}$, it is understood as the empirical measure over the part of $I_k^c$ used for the objective (with effective size $n_{\mathrm{tr}}\asymp n$). For a learned candidate menu, this part is independent of the data used to fit the candidates.
Thus, when a bound below displays $1/n$ or $\sqrt{1/n}$, it should be interpreted as $1/n_{\mathrm{tr}}$ or $\sqrt{1/n_{\mathrm{tr}}}$.
For fixed $K$, this distinction affects only constants.

Throughout this appendix, we write $a\lessssim b$ to mean that $a\le Cb$ for a numerical constant $C$ that does not depend on $n$, $M$, or the confidence level $\delta$ (though it may depend on fixed envelope constants that we display explicitly).

\subsection{From quadratic objectives to $L^2(P)$ rates}
\label{sec:learn-H-quadratic}

Both $g_0$ and $\alpha_0$ are projections onto the \emph{same} closed linear space $\mathcal H$ (Definition~\ref{def:pipeline-target}).
With this shared structure, we control both nuisances by the same argument and combine their rates through the product condition.
On $\mathcal H$, both learning problems are quadratic minimization problems, and their excess risks equal squared $L^2(P)$ errors.

To make this explicit, define the population regression risk
\[
\mathcal R(g):=\E\!\left[(Y-g(W))^2\right],
\qquad g\in L^2(P_W),
\]
and the population Riesz loss (as in~\cite{chernozhukov2022riesz})
\[
\mathcal L(\alpha):=\E\!\left[\alpha(W)^2-2m(O,\alpha)\right],
\qquad \alpha\in L^2(P_W).
\]
By Definition~\ref{def:pipeline-target}, $g_0$ minimizes $\mathcal R(g)$ over $\mathcal H$ and $\alpha_0$ minimizes $\mathcal L(\alpha)$ over $\mathcal H$.

\begin{lemma}[Excess-risk identities on $\mathcal H$]
\label{lem:excess-risk-identities}
Let $(g_0,\alpha_0)$ be the pipeline projections from Definition~\ref{def:pipeline-target}.
Then, for every $g\in\mathcal H$,
\begin{equation}
\mathcal R(g)-\mathcal R(g_0)=\norm{g-g_0}_{P,2}^2,
\label{eq:excess-risk-reg}
\end{equation}
and for every $\alpha\in\mathcal H$,
\begin{equation}
\mathcal L(\alpha)-\mathcal L(\alpha_0)=\norm{\alpha-\alpha_0}_{P,2}^2.
\label{eq:excess-risk-riesz}
\end{equation}
\end{lemma}

By Lemma~\ref{lem:excess-risk-identities}, on $\mathcal H$ a bound on excess risk is a bound on the squared $L^2(P)$ error, so learning-theoretic bounds on excess risk give the rates in Assumption~\ref{ass:pipeline-dml}.
A proof is given in Appendix~\ref{app:proofs-sec3}; it is an application of orthogonality for $L^2(P)$ projections and the restricted Riesz representation in Proposition~\ref{prop:pipeline-mixed-bias}(i).

Suppose that the learners satisfy foldwise oracle inequalities of the form
\begin{align}
\mathcal R(\hat g^{(-k)})-\mathcal R(g_{0,n})
&= O_\Pp(r_{g,n}^2),
\label{eq:oracle-reg}\\
\mathcal L(\hat\alpha^{(-k)})-\mathcal L(\alpha_{0,n})
&= O_\Pp(r_{\alpha,n}^2),
\label{eq:oracle-riesz}
\end{align}
where $g_{0,n}$ and $\alpha_{0,n}$ are population minimizers over approximation spaces $\mathcal H_{g,n}\subseteq\mathcal H$ and $\mathcal H_{\alpha,n}\subseteq\mathcal H$, and $r_{g,n},r_{\alpha,n}$ are estimation (complexity) terms.
Writing $a_{g,n}:=\norm{g_{0,n}-g_0}_{P,2}$ and $a_{\alpha,n}:=\norm{\alpha_{0,n}-\alpha_0}_{P,2}$ for approximation errors, Lemma~\ref{lem:excess-risk-identities} yields
\[
\norm{\hat g^{(-k)}-g_0}_{P,2}=O_\Pp(r_{g,n}+a_{g,n}),
\qquad
\norm{\hat\alpha^{(-k)}-\alpha_0}_{P,2}=O_\Pp(r_{\alpha,n}+a_{\alpha,n}),
\]
uniformly over folds.
Given individual nuisance consistency and the remaining regularity conditions in Assumption~\ref{ass:pipeline-dml}, its product-rate requirement holds if
\begin{equation}
\sqrt{n}\,(r_{g,n}+a_{g,n})\,(r_{\alpha,n}+a_{\alpha,n})\to 0.
\label{eq:product-rate-template}
\end{equation}
As in standard DML with fixed covariates, nuisance errors of order $o_\Pp(n^{-1/4})$ suffice, and faster learning of one nuisance can compensate for slower learning of the other. Under the product condition, the remaining bias is smaller than the sampling noise.
In the pipeline setting, the approximation and complexity terms depend on the space $\mathcal H$ induced by the pipeline.

We give two constructions of learners for which~\eqref{eq:product-rate-template} can be verified.
The first is \emph{proper learning}, in which the analyst specifies a controlled star-shaped or convex approximation space inside $\mathcal H$ and minimizes each quadratic loss over that space.
Proper learning suits pipelines that admit an explicit sieve or convex aggregation class whose approximation and critical-radius terms can be evaluated.
The second is \emph{improper learning}, in which the analyst begins with a finite, possibly nonconvex menu of fitted pipelines and applies star aggregation.
Improper learning suits workflows in which prompts or embeddings are discrete candidates, and its complexity grows logarithmically with the menu size.
Both give foldwise $L^2(P)$ rates that can be inserted into~\eqref{eq:product-rate-template}, and they differ in how they search over candidates in the first stage.

\subsection{Proper learning on controlled approximations to $\mathcal H$}
\label{sec:approach-i}

When we can write down controlled approximation spaces inside $\mathcal H$, a natural approach is to learn each nuisance by (regularized) ERM over such a space on each training fold.
Once $\mathcal H_{g,n}$ and $\mathcal H_{\alpha,n}$ are chosen, we can read off $r_{g,n}$ and $r_{\alpha,n}$ from standard complexity calculations, and then check~\eqref{eq:product-rate-template}.

For the regression nuisance, we can take $\hat g^{(-k)}$ to be a (regularized) empirical minimizer of the square loss over $\mathcal H_{g,n}$,
\[
\hat g^{(-k)}\in\arg\min_{g\in\mathcal H_{g,n}} \Pn\!\left[(Y-g(W))^2\right],
\]
and for the representer nuisance we analogously minimize the empirical Riesz loss over $\mathcal H_{\alpha,n}$,
\[
\hat\alpha^{(-k)}\in\arg\min_{\alpha\in\mathcal H_{\alpha,n}} \Pn\!\left[\alpha(W)^2-2m(O,\alpha)\right].
\]
For square loss on a star-shaped class, the sharp way to express $r_{g,n}$ is through a localized fixed point, often called the \emph{critical radius}.
We state a standard form of this bound and cite it without proof.

\begin{proposition}[Square-loss ERM on a star-shaped class]
\label{prop:crit-radius-reg}
Let $\mathcal F\subseteq L^2(P_W)$ be fixed, or constructed independently of the sample used in $\Pn$, and star-shaped around
\[
f_0\in\arg\min_{f\in\mathcal F}\E[(Y-f(W))^2].
\]
Here $f_0$ is the population best predictor in the class, and star-shapedness and the optimality of $f_0$ give the quadratic margin. For an independently constructed random class, the conclusion is conditional on that class.
Assume bounded envelopes: $|Y|\le B_Y$ and $\sup_{f\in\mathcal F}|f(W)|\le B_F$ almost surely.
Let $\hat f\in\arg\min_{f\in\mathcal F}\Pn[(Y-f(W))^2]$.
Define the localized Rademacher functional
\[
\mathcal R_n(r):=\E\Bigl[\sup_{\substack{f\in\mathcal F\\ \norm{f-f_0}_{P,2}\le r}}
\frac{1}{n}\sum_{i=1}^n\varepsilon_i\{f(W_i)-f_0(W_i)\}\Bigr],
\]
and the corresponding critical radius $r_n:=\inf\{r>0:\ \mathcal R_n(r)\le r^2/8\}$.
Then there is a constant $C$ depending only on $(B_Y,B_F)$ such that, for any $\delta\in(0,1)$, with probability at least $1-\delta$,
\begin{equation}
\norm{\hat f-f_0}_{P,2}\ \le\ C\Bigl(r_n+\sqrt{\tfrac{\log(1/\delta)}{n}}\Bigr).
\label{eq:crit-radius-concl}
\end{equation}
\end{proposition}

The bound follows by combining Corollary~5.3 and Lemma~3.4 of~\cite{bartlett2005local}. The factor $1/8$ in the definition of $r_n$ is a normalization. Because $\mathcal F$ is star-shaped around $f_0$, the map $r\mapsto\mathcal R_n(r)/r$ is nonincreasing, so another constant changes $r_n$ by at most a constant factor, which $C$ absorbs.
By~\eqref{eq:crit-radius-concl}, a bound on $r_n$ for the induced approximation space gives an $L^2(P)$ learning rate.

For the representer nuisance, Theorem~3.2 of 
\cite{chernozhukov2021automatic} 
gives high-probability rates for empirical minimizers of the Riesz loss over bounded classes, which match the envelope conditions of candidates produced by a finite pipeline menu.

\begin{proposition}[Riesz regression rates for bounded representer classes]
\label{prop:riesz-theorem32}
Let $\mathcal A\subseteq\mathcal H$ be a representer class fixed or constructed independently of the sample used in $\Pn$; in the latter case the statement is conditional on that class. Let
\[
\hat\alpha\in\arg\min_{\alpha\in\mathcal A}\Pn\!\left[\alpha(W)^2-2m(O,\alpha)\right].
\]
Let $\alpha_{0,n}\in\arg\min_{\alpha\in\mathcal A}\mathcal L(\alpha)$ be the population minimizer over $\mathcal A$.
Assume the conditions of Theorem~3.2 of~\cite{chernozhukov2021automatic}, which in our notation read as follows.
First, $m$ is linear in its second argument and mean-square continuous on $\mathrm{span}(\mathcal A)$: for some $C_m>0$, $\E[m(O,h)^2]\le C_m\,\norm{h}_{P,2}^2$ for all $h$ in that span.
Second, the classes are uniformly bounded, $\sup_{\alpha\in\mathcal A}\norm{\alpha}_\infty\le B_A$ and $\sup_{\alpha\in\mathcal A}\norm{m(O,\alpha)}_\infty\le B_m$, so the deviation sets of the proof, the star-shaped hulls of $\mathcal A-\alpha_{0,n}$ and of $m\circ\mathcal A-m\circ\alpha_{0,n}$, have envelopes bounded by $2B_A$ and $2B_m$.
Third, $\delta_n$ is an upper bound on the critical radii of these two hulls, in the sense of Proposition~\ref{prop:crit-radius-reg} up to the normalizing constant.
Let $\zeta\in(0,1)$.
Then, with probability at least $1-\zeta$,
\begin{equation}
\norm{\hat\alpha-\alpha_0}_{P,2}^2
\ \lessssim\
\norm{\alpha_{0,n}-\alpha_0}_{P,2}^2
\;+\;
C_m\,\delta_n^2
\;+\;
C_m\,\frac{\log(1/\zeta)}{n}.
\label{eq:riesz-theorem32}
\end{equation}
\end{proposition}

Propositions~\ref{prop:crit-radius-reg} and~\ref{prop:riesz-theorem32} show that, to verify~\eqref{eq:oracle-reg}--\eqref{eq:oracle-riesz}, one needs a complexity bound (a critical radius) for the regression class and for the representer class, together with approximation control for the chosen sieves.

As a worked example, consider a finite dictionary, as when the analyst tries several representations and heads.
Suppose an independent part of the training fold yields a finite menu of bounded candidate predictors $\mathcal F=\{f_1,\dots,f_M\}\subseteq\mathcal H$. The aggregation objective uses the remaining part, so its observations are independent of the fitted menu.
For a radius $B>0$, consider the signed $\ell_1$-hull
\[
\mathcal F_B
:=
\left\{f=\sum_{j=1}^M b_j f_j:\ \sum_{j=1}^M |b_j|\le B\right\}.
\]
This is a convex (hence star-shaped) approximation to $\overline{\mathrm{span}}(\mathcal F)$ that is controlled by a single tuning parameter $B$.
A localized complexity calculation (recorded in Appendix~\ref{app:proofs-sec3}) yields the schematic critical-radius bound
\begin{equation}
r_n\ \lessssim\
\min\Biggl\{\sqrt{\frac{M}{n}},\ \sqrt{B\max_{1\le j\le M}\norm{f_j}_\infty}\Bigl(\frac{\log(2M)}{n}\Bigr)^{1/4}\Biggr\}.
\label{eq:l1-radius}
\end{equation}
Thus, under the boundedness assumptions of Proposition~\ref{prop:crit-radius-reg}, ERM over $\mathcal F_B$ yields
\[
\norm{\hat f-f_{0,B}}_{P,2}=O_\Pp(r_n),
\]
where $f_{0,B}$ is the $L^2(P)$ projection onto $\mathcal F_B$.
For representer learning, apply the same $\ell_1$ geometry separately to the bounded dictionaries $a_j$ and $m(O,a_j)$; mean-square continuity enters through Proposition~\ref{prop:riesz-theorem32}, and we give this reduction in Appendix~\ref{app:proofs-sec3}.
By~\eqref{eq:l1-radius}, the product condition~\eqref{eq:product-rate-template} is a condition on the choices of $(M,B)$ for the regression and representer learners.

Proper convex aggregation borrows strength across candidates by fitting a mixed predictor over an $\ell_1$ hull or another controlled convex class.
Its estimation term grows with the effective size of that class; in~\eqref{eq:l1-radius}, the first regime behaves like $\sqrt{M/n}$, so a nuisance rate of $o(n^{-1/4})$ requires $M=o(\sqrt n)$, up to approximation terms and fixed constants.
The next subsection gives the complementary construction for much larger, nonconvex menus.

\subsection{Improper learning by star aggregation}
\label{sec:approach-ii}

Many applied pipelines naturally produce a \emph{menu} of candidates rather than a single convex class.
Varying prompts or checkpoints, searching a grid of hyperparameters, or swapping heads yields a nonconvex collection of fitted predictors in $\mathcal H$ on each training fold.
Star aggregation combines this discrete menu into a predictor in the same space $\mathcal H$ with the optimal remainder $\log(M)/n$ of model selection aggregation; see~\cite{audibert2007progressive}.
It can also approximate the population projections $(g_0,\alpha_0)$ more closely than a rule constrained to output a single candidate.

We start from the empirical risk minimizer over the menu and refit along each segment that connects it to another candidate.
Every segment lies in $\mathcal H$ when its endpoints do, so the refit stays in $\mathcal H$, and its gain comes from the curvature of the quadratic objective.
Even when $M$ is large, the star aggregate satisfies oracle inequalities with a $\log(M)/n$ remainder with high probability, as we show below for the regression and Riesz losses.

For the regression nuisance, take $\mathcal C=\mathcal F$ and use squared loss.
Let $\mathcal F\subseteq\mathcal H$ be a candidate class on a given training fold.
We first compute an empirical risk minimizer
\[
\tilde f\in\arg\min_{f\in\mathcal F}\Pn\!\left[(Y-f(W))^2\right],
\]
and then re-optimize over the star hull around this center,
\[
\hat f\in\arg\min_{f\in\mathrm{star}(\mathcal F,\tilde f)}\Pn\!\left[(Y-f(W))^2\right],
\qquad
\mathrm{star}(\mathcal F,\tilde f)
:=
\{\lambda\tilde f+(1-\lambda)f:\ f\in\mathcal F,\ \lambda\in[0,1]\}.
\]
The same construction applies to the representer objective: for $\mathcal A\subseteq\mathcal H$, compute
\[
\tilde\alpha\in\arg\min_{\alpha\in\mathcal A}\Pn\!\left[\alpha(W)^2-2m(O,\alpha)\right],
\qquad
\hat\alpha\in\arg\min_{\alpha\in\mathrm{star}(\mathcal A,\tilde\alpha)}\Pn\!\left[\alpha(W)^2-2m(O,\alpha)\right].
\]
The second step is a one-dimensional quadratic minimization along each segment, so it has a closed form.
If we fix $\tilde f$ and a candidate $f\in\mathcal F$ and write $f_\lambda:=\lambda\tilde f+(1-\lambda)f$, then the empirical square loss is quadratic in $\lambda$ and the unconstrained minimizer is
\begin{equation}
\lambda^\ast(f)
=
\frac{\Pn[(Y-f(W))(\tilde f(W)-f(W))]}{\Pn[(\tilde f(W)-f(W))^2]},
\label{eq:lambda-reg}
\end{equation}
projected onto $[0,1]$ (with the convention that any $\lambda$ is optimal if the denominator is zero).
Thus, for a finite list $\mathcal F=\{f_1,\dots,f_M\}$, star aggregation can be computed by scanning $j=1,\dots,M$, evaluating~\eqref{eq:lambda-reg}, and picking the best segment fit.
For the empirical Riesz loss, if we fix $\tilde\alpha$ and $\alpha\in\mathcal A$ and set $\alpha_\lambda:=\lambda\tilde\alpha+(1-\lambda)\alpha$, then linearity of $m$ again yields a quadratic in $\lambda$, whose unconstrained minimizer is
\begin{equation}
\lambda^\ast(\alpha)
=
-\frac{\Pn\!\left[\alpha(W)\,(\tilde\alpha-\alpha)(W)-m(O,\tilde\alpha-\alpha)\right]}{\Pn[(\tilde\alpha(W)-\alpha(W))^2]},
\label{eq:lambda-riesz}
\end{equation}
projected onto $[0,1]$.

In a multimodal demand analysis, for example, an analyst may fit outcome regressions and Riesz representers using text embeddings, image embeddings, tabular product attributes, and their combinations, perhaps crossed with a small set of heads or regularization choices.
These pipelines form a discrete, nonconvex menu because changing a foundation-model checkpoint or feature modality changes the fitted function rather than a coefficient in a common parametric class.
Such menus are common in applied work, where researchers compare a finite collection of economically defensible specifications, and the regression and Riesz losses may be minimized by different members of that collection.
Star aggregation combines candidates when the data support a mixture, and its remainder grows only logarithmically with the menu size.

For a finite menu fixed independently of the aggregation sample, the star estimator of Definition~\ref{def:star-estimator} satisfies oracle inequalities with remainder $\log(M)/n$. No additional split between selecting the star center and fitting its weights is required.  

\begin{proposition}[Star aggregation oracle inequality: square loss]
\label{prop:star-finite-reg}
Let $(W_i,Y_i)_{i=1}^n$ be i.i.d.\ with $|Y|\le B_Y$ a.s.,
and let $\mathcal F=\{f_1,\dots,f_M\}\subseteq\mathcal H$
with $\max_{j}\norm{f_j}_\infty\le B_F$, fixed or constructed independently of these observations. In the latter case, the probability statement is conditional on the menu.
Let $\hat f$ be the star estimator in Definition~\ref{def:star-estimator}
for the empirical squared loss
$R_n(f):=\Pn[(Y-f(W))^2]$.
Then there exists a universal constant $C>0$ such that for every
$\delta\in(0,1)$, with probability at least $1-\delta$,
\begin{equation}
\label{eq:star-finite-reg-oracle}
\norm{\hat f-g_0}_{P,2}^2
\ \le\
\min_{1\le j\le M}\norm{f_j-g_0}_{P,2}^2
\;+\;
C\,(B_Y+B_F)^2\,\frac{\log(M/\delta)}{n},
\end{equation}
where $g_0$ is the $L^2(P)$ projection of $Y$ onto $\mathcal H$.
\end{proposition}

\begin{proposition}[Star aggregation oracle inequality: Riesz loss]
\label{prop:star-finite-riesz}
Let $O_1,\dots,O_n$ be i.i.d.\ and let
$\mathcal A=\{a_1,\dots,a_M\}\subseteq\mathcal H$
with $\max_{j}\norm{a_j}_\infty\le B_A$, fixed or constructed independently of these observations. In the latter case, the probability statement is conditional on the menu.
Assume that $m(O,h)$ is linear in $h$ with:
\begin{enumerate}[label=(\roman*),nosep]
\item \emph{mean-square continuity:}
$\E[m(O,h)^2]\le C_m\,\norm{h}_{P,2}^2$
for all $h\in\mathrm{span}(\mathcal A)$;
\item \emph{bounded envelope:}
$\sup_{a\in\mathrm{star}(\mathcal A,a_k)}|m(O,a)|\le B_m$
a.s.\ for every center $k\in\{1,\ldots,M\}$.
\end{enumerate}
Let $\hat a$ be the star estimator in Definition~\ref{def:star-estimator}
for the empirical Riesz loss
$L_n(a):=\Pn[a(W)^2-2\,m(O,a)]$.
Then there exists a universal constant $C>0$ such that for every
$\delta\in(0,1)$, with probability at least $1-\delta$,
\begin{equation}
\label{eq:star-finite-riesz-oracle}
\norm{\hat a-\alpha_0}_{P,2}^2
\ \le\
\min_{1\le j\le M}\norm{a_j-\alpha_0}_{P,2}^2
\;+\;
C\,(B_A^2+C_m+B_m)\,\frac{\log(M/\delta)}{n},
\end{equation}
where $\alpha_0$ minimizes $\mathcal L(\alpha)$ over $\mathcal H$.
\end{proposition}

When each nuisance is learned by star aggregation over a finite
foldwise menu, we obtain high-probability $L^2(P)$
bounds of order $\sqrt{\log(M)/n}$.
These bounds can be inserted into~\eqref{eq:product-rate-template}, provided the menu is independent of the sample used to aggregate it. This independence prevents candidates from fitting the same noise used to assess them. When candidates are trained within a fold, use separate data for aggregation; Appendix~\ref{sec:cross-fold} treats pooled aggregation without this internal split.

\medskip

\subsection{Proofs}
\label{app:proofs-sec3}

We collect the proofs for the results stated in Appendix~\ref{sec:learn-H}.
Throughout, $P$ denotes the population law and $\Pn$ denotes the empirical measure on a training fold (we suppress the fold index, as stated at the start of this appendix).

\subsubsection{Proof of Lemma~\ref{lem:excess-risk-identities}}
\begin{proof}
For~\eqref{eq:excess-risk-reg}, write $Y=g_0(W)+\varepsilon$ where $\varepsilon:=Y-g_0(W)$.
Since $g_0$ minimizes $\mathcal R$ over the linear subspace $\mathcal H$, the normal equation gives
$\E[\varepsilon\,h(W)]=0$ for every $h\in\mathcal H$.
Then for any $g\in\mathcal H$,
\[
\mathcal R(g)-\mathcal R(g_0)
=
\E[\{(g-g_0)(W)\}^2]-2\,\E[(g-g_0)(W)\varepsilon]
=
\norm{g-g_0}_{P,2}^2,
\]
because $g-g_0\in\mathcal H$ and hence $\E[(g-g_0)(W)\,\varepsilon]=0$ by the normal equation.

For~\eqref{eq:excess-risk-riesz}, expand
\[
\mathcal L(\alpha)-\mathcal L(\alpha_0)
=
\E[\{(\alpha-\alpha_0)(W)\}^2]
+
2\,\E[\alpha_0(W)\{\alpha-\alpha_0\}(W)]
-2\,\E[m(O,\alpha-\alpha_0)].
\]
Since $\alpha-\alpha_0\in\mathcal H$, Proposition~\ref{prop:pipeline-mixed-bias}(i) gives the restricted Riesz identity
$\E[m(O,h)]=\E[h(W)\alpha_0(W)]$ for all $h\in\mathcal H$.
Taking $h=\alpha-\alpha_0$ cancels the last two terms and yields
$\mathcal L(\alpha)-\mathcal L(\alpha_0)=\norm{\alpha-\alpha_0}_{P,2}^2$.
\end{proof}

\medskip

We next record the localized complexity calculation behind~\eqref{eq:l1-radius}, and the corresponding extension to Riesz regression via mean-square continuity.

\subsubsection{Derivation of~\eqref{eq:l1-radius}}
\begin{proof}
Let $\mathcal F_B$ be the signed $\ell_1$ hull and let $f_{0,B}$ be the minimizer of $\mathcal R(f)$ over $\mathcal F_B$.
For $r>0$ consider the localized difference set
\[
\Delta(r):=\{f-f_{0,B}:\ f\in \mathcal F_B,\ \norm{f-f_{0,B}}_{P,2}\le r\}.
\]
Since $\mathcal F_B$ lies in a span of dimension at most $M$, the localized difference set lies in its $L^2(P)$ ball of radius $r$. Expanding in an orthonormal basis and applying Cauchy--Schwarz gives
\[
\mathcal R_n(r)\le r\sqrt{\frac{M}{n}}.
\]
On the other hand, $\ell_1$--$\ell_\infty$ duality and the bounded-dictionary Rademacher maximal inequality give
\[
\E\sup_{f\in\mathcal F_B}\frac{1}{n}\sum_{i=1}^n\varepsilon_i f(W_i)
=B\,\E\max_{1\le j\le M}\Bigl|\frac{1}{n}\sum_{i=1}^n\varepsilon_i f_j(W_i)\Bigr|
\lessssim B\max_j\norm{f_j}_\infty\sqrt{\frac{\log(2M)}{n}}.
\]
Subtracting the fixed comparator $f_{0,B}$ does not change this expectation. Hence
\[
\mathcal R_n(r)\lessssim
\min\Biggl\{r\sqrt{\frac{M}{n}},\ B\max_j\norm{f_j}_\infty\sqrt{\frac{\log(2M)}{n}}\Biggr\}.
\]
Solving $\mathcal R_n(r)\le r^2/8$ yields~\eqref{eq:l1-radius}. The dimension bound is independent of dictionary scaling, while the global bound uses the dictionary's uniform envelope.
\end{proof}

\medskip

\begin{proof}[{{Reduction to $\ell_1$ geometry for the Riesz loss}}]
Let $\mathcal A=\{a_1,\ldots,a_M\}\subseteq\mathcal H$ and define
\[
\mathcal A_B:=\Bigl\{\sum_{j=1}^M b_j a_j:\ \sum_j|b_j|\le B\Bigr\}.
\]
Assume the dictionary is fixed independently of the objective sample and both $a_j$ and $m(O,a_j)$ have bounded envelopes.
Apply the calculation for~\eqref{eq:l1-radius} to $\mathcal A_B-\alpha_{0,B}$ using $\max_j\norm{a_j}_\infty$.
By linearity, $m\circ\mathcal A_B$ is the signed $\ell_1$ hull of $\{m(O,a_j):1\le j\le M\}$, so apply the same calculation separately using $\max_j\norm{m(O,a_j)}_\infty$.
Both critical radii are therefore bounded by
\[
\delta_n\lessssim
\min\Biggl\{\sqrt{\frac{M}{n}},\
\sqrt{B\max_j\bigl(\norm{a_j}_\infty\vee\norm{m(O,a_j)}_\infty\bigr)}
\Bigl(\frac{\log(2M)}{n}\Bigr)^{1/4}\Biggr\}.
\]
The deviation sets are already star-shaped because the signed hulls are convex. Insert this common radius into~\eqref{eq:riesz-theorem32}; mean-square continuity supplies the factor $C_m$ in that theorem, rather than a direct rescaling of the critical radius.
\end{proof}

\medskip

We now prove Propositions~\ref{prop:star-finite-reg} and~\ref{prop:star-finite-riesz}.
Both results share a common structure, which we isolate in abstract terms before specializing to each loss.
In both, star aggregation is applied over a finite class of functions with strongly convex and sufficiently regular losses.
We use the empirical-process notation of Appendix~\ref{app:proofs-sec2}, with $\Pn$ the empirical measure of the aggregation sample, and in addition the symmetrized process
$${\mathbb{P}_{n,\varepsilon}} f := n^{-1} \sum_{i=1}^n \varepsilon_i f(O_i),$$ where $\varepsilon_1,\dots,\varepsilon_n$ are i.i.d.\ Rademacher signs independent of the sample.

Let $\cF=\{f_1,\dots,f_M\}$ be a finite candidate set of measurable functions $f\colon \mathcal{W}\to\R$, and define the $2$-convex hull
\[
\operatorname{conv}_2(\cF)
:=\bigl\{\lambda f_j+(1-\lambda)f_k : 1\le j,k\le M,\;
\lambda\in[0,1]\bigr\}.
\]
Throughout, we fix a comparator $f^*\in\cF$ and write
$$\cE_f:=\ell_f-\ell_{f^*}$$for the excess loss,
where $\ell_f(O):=\ell(O,f)$.

\subsubsection{Proofs of Propositions~\ref{prop:star-finite-reg} and~\ref{prop:star-finite-riesz}}
Propositions~\ref{prop:star-finite-reg} and~\ref{prop:star-finite-riesz} bound $P\cE_{\hat f}$, where $\hat f$ is the star estimator, for the following two losses over the common function space~$\mathcal{H}$.

\medskip
\noindent\emph{(i) Regression (squared loss).}
We observe $O=(Y,W)$ and take
$\ell^{\mathrm{ls}}(O,f) := (Y-f(W))^2.$

\medskip
\noindent\emph{(ii) Riesz regression (Riesz loss).}
We observe a generic $O$ that contains $W$ as a measurable component
and take
$\ell^{\mathrm{R}}(O,f) := f(W)^2 - 2\,m(O,f),$
where $m(O,\cdot)$ is linear in its second argument.
In the ATE case, $m\{(y,d,z),f\}=f(1,z)-f(0,z)$.

\medskip

In both cases, population excess risk is a difference of squared $L^2(P)$ distances to a particular function $f_0$, since by Lemma~\ref{lem:excess-risk-identities}, for $f\in\cH$ and a comparator $f^*\in\cH$,
\[
P\cE_f^{\mathrm{ls}} = \norm{f-f_0^{\mathrm{ls}}}_{P,2}^2-\norm{f^*-f_0^{\mathrm{ls}}}_{P,2}^2,
\qquad
P\cE_f^{\mathrm{R}} = \norm{f-f_0^{\mathrm{R}}}_{P,2}^2-\norm{f^*-f_0^{\mathrm{R}}}_{P,2}^2.
\]
For squared loss, $f_0^{\mathrm{ls}}$ is the $L^2(P)$ projection of $Y$ onto $\cH$; for the Riesz loss,
$f_0^{\mathrm{R}}$ is the Riesz representer of the restricted functional
$h\mapsto Pm(O,h)$ on~$\cH$.
Note that, in general, $f_0$ does not belong to the candidate class $\cF$, so we state oracle inequalities relative to an arbitrary comparator $f^* \in \cF$.

With this shared structure, Propositions~\ref{prop:star-finite-reg} and~\ref{prop:star-finite-riesz} can be proven in tandem, and we do so here.

\medskip
\noindent\textbf{Standing assumptions.}
The proof relies on an offset quantity
\[
\cQ_f \;:=\; \cE_f + c_0\,(f-f^*)^2,
\qquad c_0:=\tfrac{1}{9},
\]
and on the following regularity conditions, which hold for both loss functions.

\begin{enumerate}[label=\textup{(A\arabic*)}]
\item\label{A:bounded}
(\emph{Boundedness.})
$\|f\|_\infty\le b_F$ for all $f\in\cF$, and
$|\cQ_f|\le b_Q$ a.s.\ for all $f\in\conv_2(\cF)$.

\item\label{A:var}
(\emph{Population second-moment bound.})
There exists $B_0<\infty$ such that
\[
P\,\cQ_f^2 \;\le\; B_0\,\norm{f-f^*}_{P,2}^2
\qquad \forall\, f\in\conv_2(\cF).
\]

\item\label{A:inc}
(\emph{Population increment bound.})
There exists $B_1<\infty$ such that
\[
P\,(\cQ_f-\cQ_g)^2 \;\le\; B_1\,\norm{f-g}_{P,2}^2
\qquad \forall\, f,g\in\conv_2(\cF).
\]

\item\label{A:lip}
(\emph{Restricted Lipschitz along line segments.})
There exists $L<\infty$ such that for all
$f_1,f_2\in\cF$ and $\lambda,\lambda'\in[0,1]$,
\[
\bigl|\cQ_{f_1+\lambda(f_2-f_1)}-\cQ_{f_1+\lambda'(f_2-f_1)}\bigr|
\;\le\; L\,|\lambda-\lambda'|.
\]
\end{enumerate}

\medskip
The proof argument proceeds in two steps.

Step 1. Strong convexity of $\ell^{\mathrm{ls}}$ and
$\ell^{\mathrm{R}}$ implies that the star estimator's excess risk
admits a \emph{localized} bound over $\conv_2(\cF)$
(Lemma~\ref{lemma:star-margin}).

Step 2. Regularity of the respective losses allows us to control the
complexity of $\conv_2(\cF)$, yielding the desired
$\log(M)/n$ rate.
Mixed-tail chaining first bounds the expected localized empirical process. Concentration and a peeling argument then retain the quadratic offset and give the oracle remainder.

We first give the shared localization step, then the complexity bound, and finally verify the standing assumptions for each loss, thus allowing us to prove both Propositions~\ref{prop:star-finite-reg} and~\ref{prop:star-finite-riesz}.

\medskip
\noindent\textbf{Step~1: Localized bound via strong convexity.}
Because the loss is strongly convex, the star algorithm's excess risk is bounded by a localized (or offset) empirical process
\cite{vijaykumar2021localization}. The following result shows that the star estimator behaves like a convex minimizer, even when applied to a nonconvex candidate set.

\begin{lemma}[Star margin]\label{lemma:star-margin}
Let $\ell \colon \mathcal{H} \times \mathcal{O} \to \R$ be given.
Suppose that $\ell$ almost surely satisfies the strong convexity
condition: for all $f,g\in\cH$ and $\lambda\in[0,1]$,
\[
\ell\!\bigl(O,\,f + \lambda(g-f)\bigr)
\;\le\;
(1-\lambda)\,\ell(O,f) + \lambda\,\ell(O,g)
- \lambda(1-\lambda)\bigl(f(W)-g(W)\bigr)^2.
\]
Let $\hat f$ denote the star estimator:
for $\mathcal{F} \subset \mathcal{H}$,
let $\tilde f$ minimize $f \mapsto \Pn\cE_f$ over $\cF$
and let $\hat f$ minimize $f \mapsto \Pn\cE_f$ over
$\mathrm{star}(\cF,\tilde f)$.
Then, for any comparator $f^* \in \cF$:
\begin{enumerate}[label=(\roman*)]
\item
\emph{(Empirical margin.)}
$-\Pn \cE_{\hat f} - \tfrac{1}{9}\norm{f^* - \hat f}_{\Pn,2}^2 \ge 0.$
\item
\emph{(Population-risk bound.)}
\begin{equation}
P\cE_{\hat f}
\;\le\;
\sup_{f \in \conv_2 \cF}\left\{
(P - \Pn)\!\left[\cE_{ f} + \tfrac{1}{9}( f - f^*)^2\right]
- \tfrac{1}{9}\norm{f^* - f}_{P,2}^2\right\}.
\label{eq:star-empirical-process}
\end{equation}
\end{enumerate}
\end{lemma}

\begin{proof}
Claim (i) is Proposition~4 of~\cite{vijaykumar2021localization}, applied to the empirical objective $f\mapsto\Pn\ell(\cdot,f)$, which is $(\mu,d)$-convex in the sense of Definition~1 of that paper with $\mu(z)=z^2$ and $d$ the empirical $L^2$ pseudometric.
For the square loss this is Lemma~5 of~\cite{vijaykumar2021localization}, and Example~6 of that paper carries out exactly this computation, with $d(x,y)=|x-y|$ and $\mu(z)=z^2$; the star algorithm costs the factor $\mu(d/3)=d^2/9$, which is the $1/9$ in~\eqref{eq:star-empirical-process}.
The Riesz loss depends on $f$ through $m(O,f)$ as well as through $f(W)$, so that lemma does not apply verbatim; but $m(O,\cdot)$ is linear and has zero Bregman divergence, so each summand again satisfies the same condition, and the averaging argument is unchanged.
For~(ii), start from~(i):
\begin{align*}
P \cE_{\hat f}
&\le (P - \Pn) \cE_{\hat f}
- \tfrac{1}{9}\norm{f^* - \hat f}_{\Pn,2}^2 \\
&= (P - \Pn)\!\left[\cE_{\hat f}
+ \tfrac{1}{9}(\hat f - f^*)^2\right]
- \tfrac{1}{9}\norm{f^* - \hat f}_{P,2}^2 \\
&\le \sup_{f \in \conv_2 \cF}\left\{
(P - \Pn)\!\left[\cE_{ f}
+ \tfrac{1}{9}( f - f^*)^2\right]
- \tfrac{1}{9}\norm{f^* - f}_{P,2}^2\right\},
\end{align*}
where the last step uses $\hat f \in \conv_2 \cF$.
\end{proof}

\medskip
\noindent
We therefore seek to bound the offset empirical process
\begin{equation}
\sup_{f \in \conv_2(\cF)}\left\{
 (P - \Pn)\cQ_f - c_0\,\norm{f - f^*}_{P,2}^2\right\},
\qquad
\cQ_f := \cE_f + c_0 \,(f-f^*)^2,
\quad c_0=\tfrac{1}{9}.
\end{equation}

\bigskip
\noindent\textbf{Step~2: Mixed-tail chaining and the offset uniform law of large numbers (ULLN) via peeling
(Lemmas~\ref{lem:mixed-tail-chaining-bound}--\ref{lem:second_part_peeling}).}
For any $t > 0$, define
\[
\cF(t) := \bigl\{f\in\conv_2(\cF):
\norm{f-f^*}_{P,2}\le\sqrt{t}\bigr\},
\qquad
Z(t):=\sup_{f\in\cF(t)} (\mathbb{P}_{n}-P)\cQ_f.
\]
We will first control $PZ(t)$; under our standing assumptions, this follows from a standard chaining bound for processes with sub-Gamma increments (Lemma~\ref{lem:mixed-tail-chaining-bound}). From this we deduce a localized bound using a peeling argument for the offset process (Lemma~\ref{lem:second_part_peeling}).

\begin{lemma}[Mixed-tail chaining]\label{lem:mixed-tail-chaining-bound} Under Assumptions~\ref{A:inc} and~\ref{A:lip}, it holds for a universal constant $C_3$ that
\begin{equation}\label{eq:mixed-tail-chaining-bound}
     P Z(t) \le C_3\left\{ \sqrt{\frac{B_1 t \log(M+1)}{n}} + \frac{L\log(M+1)}{n}\right\}.
\end{equation}
\end{lemma}
\begin{proof}
For a finite collection of centered empirical averages with sub-Gamma variance and scale parameters bounded by $\sigma_U^2$ and $b_U$, Bernstein's maximal inequality~\cite[Lemma~2.2.10]{vanderVaartWellner1996} gives
\begin{equation}
\E\max_{j\in\mathcal I}U_j
\lesssim\sqrt{\sigma_U^2\log(1+|\mathcal I|)}+b_U\log(1+|\mathcal I|).
\label{eq:sub-gamma-maximal}
\end{equation}
For $(\Pn-P)h$, these parameters are bounded by constant multiples of $Ph^2/n$ and $\norm{h}_\infty/n$, respectively.

The set $\cF(t)$ is a union of at most $M^2$ line segments, each of $L^2(P)$ diameter at most $2\sqrt t$. A dyadic grid of $2^m+1$ points on each intersected segment gives a net $A_m$ with $|A_m|\le M^2(2^m+1)$ and projections $\pi_m(f)$ satisfying
\begin{enumerate}[label=(\alph*)]
\item $\norm{\cQ_f-\cQ_{\pi_m(f)}}_{P,2}\le 2\sqrt{B_1t}\,2^{-m}$ by~\ref{A:inc};\label{l2-cov}
\item $\norm{\cQ_f-\cQ_{\pi_m(f)}}_\infty\le L2^{-m}$ by~\ref{A:lip}.\label{unif-cov}
\end{enumerate}
For the initial net, $\cQ_{f^*}=0$ and~\ref{A:inc} give $\norm{\cQ_{\pi_0(f)}}_{P,2}\le\sqrt{B_1t}$. Also $\norm{\cQ_{\pi_0(f)}}_\infty\le2L$: connect $f^*$ to an endpoint of the segment containing $\pi_0(f)$, and then to $\pi_0(f)$, using~\ref{A:lip} twice.
The telescoping expansion converges uniformly by~\ref{unif-cov}, so
\begin{align*}
PZ(t)&\le \E\sup_{f\in\cF(t)}(\Pn-P)\cQ_{\pi_0(f)}\\
&\quad+\sum_{m\ge1}\E\sup_{f\in\cF(t)}(\Pn-P)
\{\cQ_{\pi_m(f)}-\cQ_{\pi_{m-1}(f)}\}.
\end{align*}
Each increment maximum has at most $|A_m||A_{m-1}|$ terms. Applying~\eqref{eq:sub-gamma-maximal} at each level therefore gives
\begin{align*}
PZ(t)&\lesssim\sqrt{\frac{B_1t}{n}}\sum_{m\ge0}2^{-m}\sqrt{\log(M+1)+m}
+\frac Ln\sum_{m\ge0}2^{-m}\{\log(M+1)+m\}\\
&\lesssim\sqrt{\frac{B_1t\log(M+1)}n}+\frac{L\log(M+1)}n.
\end{align*}
\end{proof}

\begin{lemma}[Offset ULLN via peeling]
\label{lem:second_part_peeling}
Under Assumptions~\ref{A:bounded}--\ref{A:lip}, for every $\delta\in(0,1)$, with probability at least $1-\delta$,
\begin{equation}
\label{eq:strat-B-conclusion}
\sup_{f\in\conv_2(\cF)}\!
\bigl\{(\mathbb P_n-P)\cQ_f-c_0\norm{f-f^*}_{P,2}^2\bigr\}
\le C_4\Bigl(\frac{B_0+B_1}{c_0}\vee b_Q\vee L\Bigr)
\frac{\log(2M/\delta)}n,
\end{equation}
where $C_4>0$ is a universal constant and $c_0=1/9$ as above.
\end{lemma}

\begin{proof}
For $f\in\cF(t)$, $\Var(\Pn\cQ_f)\le n^{-1}P\cQ_f^2\le B_0t/n$.
The centered summands are bounded by $2b_Q$. Bousquet's inequality~\cite[Theorem~2.3]{bousquet2002bennett} gives, with probability at least $1-\eta$,
\begin{equation}
Z(t)\le PZ(t)+\sqrt{\frac{\{2B_0t+8b_QPZ(t)\}\log(1/\eta)}n}
+\frac{2b_Q\log(1/\eta)}{3n}.
\end{equation}
Combining this with~\eqref{eq:mixed-tail-chaining-bound} and applying AM--GM yields
\begin{equation}
\label{eq:Zt-offset-B}
Z(t)\le\frac{c_0t}{4}
+C_4\Bigl(\frac{B_0+B_1}{c_0}\vee b_Q\vee L\Bigr)
\frac{\log(2M/\eta)}n.
\end{equation}
The $L$ term controls increments along the segments and must be retained through this step.

Set
\[
\tau:=\frac1{c_0}\Bigl(\frac{B_0+B_1}{c_0}\vee b_Q\vee L\Bigr)
\frac{\log(2M/\delta)}n,
\qquad t_j:=2^j\tau\quad(j\ge0).
\]
If the coefficient is zero the conclusion is immediate. Otherwise apply~\eqref{eq:Zt-offset-B} with $\eta_j:=\delta/\{\pi^2(j+1)^2\}$. A union bound gives an event of probability at least $1-\delta$ on which all these bounds hold.
For $\norm{f-f^*}_{P,2}^2\le\tau$, the offset process is at most $Z(\tau)$, of the required order.
For $t_{j-1}<\norm{f-f^*}_{P,2}^2\le t_j$ with $j\ge1$, it is at most
\[
-\frac{c_0t_j}{4}
+C_4\Bigl(\frac{B_0+B_1}{c_0}\vee b_Q\vee L\Bigr)
\frac{\log(2M/\delta)+2\log(j+1)+2\log\pi}{n}.
\]
Substitution of $t_j=2^j\tau$ factors out the same coefficient divided by $n$. The remaining expression is bounded by a universal multiple of $\log(2M/\delta)$, because $\log(2M/\delta)\ge\log2$ and $2^j$ dominates $\log(j+1)$. This proves~\eqref{eq:strat-B-conclusion} with a larger universal constant.
\end{proof}

We now combine the star margin (Lemma~\ref{lemma:star-margin}) with the rest of the argument. Combining~\eqref{eq:star-empirical-process}
with~\eqref{eq:strat-B-conclusion} applied to $-\cQ_f$, which satisfies the same four assumptions, gives,
for any $f^*\in\cF$ and $\delta\in(0,1)$,
with probability at least $1-\delta$,
\begin{equation}
\label{eq:combined-B}
P\cE_{\hat f}
\;\le\;
C_8\,\Bigl(\frac{B_0+B_1}{c_0}\vee b_Q\vee L\Bigr)\,
\frac{\log(2M/\delta)}{n}.
\end{equation}
For $M\ge2$, $\log(2M/\delta)\lesssim\log(M/\delta)$, giving the stated rate. For $M=1$, the estimator equals the sole comparator and its excess risk is zero.

\noindent
Choosing $f^*\in\cF$ to minimize the population risk before applying the probability bound yields
\[
\norm{\hat f - f_0}_{P,2}^2
\;\le\;
\min_{f^*\in\cF}\norm{f^* - f_0}_{P,2}^2
\;+\; \text{(estimation remainder)},
\]
where the estimation remainder is the right-hand side of~\eqref{eq:combined-B}.

It remains to verify~\ref{A:bounded}--\ref{A:lip} for each loss.

\paragraph{Proof of Proposition~\ref{prop:star-finite-reg} (square loss)}
\begin{proof}
We verify~\ref{A:bounded}--\ref{A:lip} for the squared loss
$\ell(O,f)=(Y-f(W))^2$ with $|Y|\le B_Y$ and
$\|f\|_\infty\le B_F$.
Write $\Delta:=f-f^*$ and $\bar B:=B_Y+B_F$.
The excess loss and offset quantity are
\[
\cE_f = -2\,\Delta\,(Y-f^*) + \Delta^2,
\qquad
\cQ_f = -2\,\Delta\,(Y-f^*) + (1+c_0)\,\Delta^2.
\]

\smallskip
\noindent\emph{Checking~\ref{A:bounded}.}
For $f\in\conv_2(\cF)$,
$\|\Delta\|_\infty\le 2B_F$ and $|Y-f^*|\le\bar B$, so
\[
b_F = B_F,\qquad
b_Q
\;\le\; 4B_F\bar B + 4(1+c_0)B_F^2
\;\le\; 9\,\bar B^2.
\]

\smallskip
\noindent\emph{Checking~\ref{A:var}.}
Since $\cQ_f = \Delta\{-2(Y-f^*)+(1+c_0)\Delta\}$
and the factor in braces is bounded by $2\bar B + 2(1+c_0)B_F \le 5\bar B$,
\[
P\cQ_f^2
\;\le\; P[\Delta^2\cdot 25\bar B^2]
\;=\; 25\,\bar B^2\,\norm{\Delta}_{P,2}^2.
\]
Thus $B_0 = 25\,\bar B^2$.

\smallskip
\noindent\emph{Checking~\ref{A:inc}.}
For $f,g\in\conv_2(\cF)$ write $\delta_{fg}:=f-g$.  Then
\[
\cQ_f - \cQ_g
= \delta_{fg}\bigl\{-(2Y-f-g) + c_0(f+g-2f^*)\bigr\}.
\]
The factor in braces is bounded by
$2\bar B+4c_0B_F\le 3\bar B\le 9\bar B$.
Hence $P(\cQ_f-\cQ_g)^2 \le 81\,\bar B^2\,\norm{f-g}_{P,2}^2$,
giving $B_1 = 81\,\bar B^2$.

\smallskip
\noindent\emph{Checking~\ref{A:lip}.}
Along $f_\lambda = f_1+\lambda(f_2-f_1)$,
$f_\lambda - f_{\lambda'} = (\lambda-\lambda')(f_2-f_1)$,
so $|\cQ_{f_\lambda}-\cQ_{f_{\lambda'}}|
\le |\lambda-\lambda'|\cdot 2B_F\cdot 8\bar B
= 16\,B_F\bar B\,|\lambda-\lambda'|$.
Thus $L = 16\,B_F\bar B \le 16\,\bar B^2$.

\smallskip
\noindent\emph{Substituting into~\eqref{eq:combined-B}.}
For $M\ge2$, $B_0+B_1+b_Q+L=O(\bar B^2)$ and $c_0=1/9$, so~\eqref{eq:combined-B} gives the final rate (the case $M=1$ is exact):
\[
P\cE_{\hat f}
\;\le\;
C\,(B_Y+B_F)^2\,\frac{\log(M/\delta)}{n}.
\]

\smallskip
The $L^2(P)$ oracle inequality follows from the identity
$P\cE_{\hat f} = \norm{\hat f - g_0}_{P,2}^2
- \norm{f^*-g_0}_{P,2}^2$ (by the excess-risk identity of
Lemma~\ref{lem:excess-risk-identities}),
after minimizing over $f^*\in\cF$.
\end{proof}

\paragraph{Proof of Proposition~\ref{prop:star-finite-riesz} (Riesz loss)}
\begin{proof}
We verify~\ref{A:bounded}--\ref{A:lip} for the Riesz loss
$\ell(O,f)=f(W)^2-2\,m(O,f)$ with $\|f\|_\infty\le B_A$,
$\E[m(O,h)^2]\le C_m\,\norm{h}_{P,2}^2$ for
$h\in\mathrm{span}(\mathcal{A})$,
and $\sup_{a\in\mathrm{star}(\mathcal{A},a_k)}|m(O,a)|\le B_m$ a.s.
Write $\Delta := f - f^*$.
The excess loss and offset quantity are
\[
\cE_f = 2\,f^*\,\Delta + \Delta^2 - 2\,m(O,\Delta),
\qquad
\cQ_f = 2\,f^*\,\Delta + (1+c_0)\,\Delta^2 - 2\,m(O,\Delta).
\]

\smallskip
\noindent\emph{Checking~\ref{A:bounded}.}
For $f\in\conv_2(\mathcal{A})$, $\|\Delta\|_\infty\le 2B_A$
and $|m(O,\Delta)|\le 2B_m$ (using boundedness of $m$ on star
segments and linearity), so
\[
b_F = B_A,
\qquad
b_Q \;\le\; 4B_A^2 + 4(1+c_0)B_A^2 + 4B_m
\;\le\; 9B_A^2 + 4B_m.
\]

\smallskip
\noindent\emph{Checking~\ref{A:var}.}
We bound $P\cQ_f^2$ by expanding the square,
\[
P\cQ_f^2
\;\le\; 3\bigl\{
4B_A^2\,P\Delta^2
+ 4(1+c_0)^2 B_A^2\,P\Delta^2
+ 4\,P[m(O,\Delta)^2]
\bigr\}
\;\le\; 3(12B_A^2 + 4C_m)\,\norm{\Delta}_{P,2}^2,
\]
using $P[m(O,\Delta)^2]\le C_m\norm{\Delta}_{P,2}^2$
(since $\Delta\in\mathrm{span}(\mathcal{A})$ for
$f,f^*\in\mathcal{A}$).
Thus $B_0 = 3(12B_A^2 + 4C_m) = 36B_A^2 + 12C_m$.

\smallskip
\noindent\emph{Checking~\ref{A:inc}.}
For $f,g\in\conv_2(\mathcal{A})$ with $\delta_{fg}:=f-g$,
\[
\cQ_f - \cQ_g
= \delta_{fg}\{2f^* + (1+c_0)(f+g-2f^*)\} - 2\,m(O,\delta_{fg}).
\]
The first term satisfies
$|\delta_{fg}\{\cdots\}| \le 2B_A\cdot|\delta_{fg}|
+ 4(1+c_0)B_A\cdot|\delta_{fg}|
\le 7B_A|\delta_{fg}|$, which gives
$P[\delta_{fg}^2\cdot 49B_A^2] = 49B_A^2\norm{f-g}_{P,2}^2$.
The second term gives $4P[m(O,\delta_{fg})^2]\le 4C_m\norm{f-g}_{P,2}^2$.
By Cauchy--Schwarz on the cross term,
\[
B_1 \;\le\; 2(49B_A^2 + 4C_m) = 98B_A^2 + 8C_m.
\]

\smallskip
\noindent\emph{Checking~\ref{A:lip}.}
Along $f_\lambda = f_1 + \lambda(f_2-f_1)$,
linearity of $m$ gives
$m(O,f_\lambda)-m(O,f_{\lambda'})
= (\lambda-\lambda')\,m(O,f_2-f_1)$,
so
\[
|\cQ_{f_\lambda}-\cQ_{f_{\lambda'}}|
\;\le\; |\lambda-\lambda'|\cdot\bigl(
7B_A\cdot 2B_A + 2\cdot 2B_m\bigr)
\;=\; |\lambda-\lambda'|\cdot(14B_A^2 + 4B_m).
\]
Thus $L = 14B_A^2 + 4B_m$.

\smallskip
\noindent\emph{Substituting into~\eqref{eq:combined-B}.}
Since $B_0+B_1+b_Q+L=O(B_A^2+C_m+B_m)$ and $c_0=1/9$, substituting these constants yields (for $M\ge2$, with $M=1$ exact)
\[
P\cE_{\hat a}
\;\le\;
C\,(B_A^2+C_m+B_m)\,\frac{\log(M/\delta)}{n}.
\]
The $L^2(P)$ oracle inequality follows from
Lemma~\ref{lem:excess-risk-identities}
exactly as in the squared-loss case.
\end{proof}

\section{Approach 1 with menus: aggregation across folds}
\label{app:cross-fold}
\label{sec:cross-fold}

Propositions~\ref{prop:star-finite-reg} and~\ref{prop:star-finite-riesz} treat the menu of candidates as fixed.
Inside a training fold, the candidates and the star weights must therefore be fitted on disjoint data.
That calls for a second split within each training fold, or for nested cross-fitting, and either way it multiplies the number of candidate fits.
When a candidate is a fine-tuned foundation model, each fit is expensive.

Practice often follows a cheaper design. Train every candidate once on each training fold, as cross-fitting requires anyway, and fit a single set of star weights to the pooled out-of-fold predictions.
Each prediction in the aggregation criterion then comes from a model that never saw that observation, as in stacking~\cite{wolpert1992stacked,breiman1996stacked} and the super learner~\cite{vanderlaan2007super}.
The weights, however, depend on every fold, including the fold on which the score is later evaluated.
Cross-fitting then holds for the candidates but not for the aggregate.
We show that this failure of cross-fitting adds one remainder term of order $(r_g+r_\alpha)\sqrt{\log M}$ for star weights and $(r_g+r_\alpha)\sqrt{M\log n}$ for convex weights, where $r_g$ and $r_\alpha$ are the nuisance rates.

\medskip
\noindent\emph{Construction.}
For each fold $k$, train on $I_k^c$ a regression menu
$\cF^{(-k)}=\{\hat f_1^{(-k)},\dots,\hat f_M^{(-k)}\}\subseteq\cH$
and a representer menu
$\mathcal A^{(-k)}=\{\hat a_1^{(-k)},\dots,\hat a_M^{(-k)}\}\subseteq\cH$.
Candidate~$j$ is one recipe, such as an encoder with a head, refitted on each training sample.
Index the segments of the star hull by
\[
\mathcal T:=\{1,\dots,M\}^2\times[0,1],
\qquad
f_\tau^{(-k)}:=\lambda\hat f_a^{(-k)}+(1-\lambda)\hat f_b^{(-k)}
\quad\text{for }\tau=(a,b,\lambda)\in\mathcal T,
\]
and define $a_\tau^{(-k)}$ from $\mathcal A^{(-k)}$ in the same way.
The index $\tau$ is common to all folds; the functions it selects are not.
The pooled out-of-fold objectives are
\begin{align}
\bar R_n(\tau)
&:=\frac1n\sum_{k=1}^K\sum_{i\in I_k}\bigl\{Y_i-f_\tau^{(-k)}(W_i)\bigr\}^2,
\label{eq:pooled-reg}\\
\bar L_n(\tau)
&:=\frac1n\sum_{k=1}^K\sum_{i\in I_k}\Bigl\{a_\tau^{(-k)}(W_i)^2-2\,m\bigl(O_i,a_\tau^{(-k)}\bigr)\Bigr\}.
\label{eq:pooled-riesz}
\end{align}

\begin{definition}[Pooled star estimator]
\label{def:pooled-star}
Choose a center $\tilde a\in\arg\min_{1\le j\le M}\bar R_n(j,j,1)$ and a segment
\[
\hat\tau_g\in\arg\min\bigl\{\bar R_n(\tau):\ \tau=(\tilde a,b,\lambda)\in\mathcal T\bigr\}.
\]
Set $\hat g^{(-k)}:=f^{(-k)}_{\hat\tau_g}$ for every $k$.
Define $\hat\tau_\alpha$ from $\bar L_n$ in the same way, and set $\hat\alpha^{(-k)}:=a^{(-k)}_{\hat\tau_\alpha}$.
The estimators $\hat\theta_0$ and $\hat\sigma_0^2$ are then computed from these nuisances exactly as in Section~\ref{sec:pipeline-inference}.
\end{definition}

The segment search still has the closed forms~\eqref{eq:lambda-reg} and~\eqref{eq:lambda-riesz}, with $\Pn$ replaced by the pooled out-of-fold average.
Equivalently, the pooled star estimator is the star estimator of Definition~\ref{def:star-estimator}, computed from the augmented observations $(O_i,k(i))$ and the stacked candidates $(o,k)\mapsto\hat f_j^{(-k)}(o)$, $j=1,\dots,M$.

Write $p_k:=n_k/n$.
The natural benchmark is the best single recipe, measured by its fold-averaged error:
\begin{equation}
e_g^2:=\min_{1\le j\le M}\sum_{k=1}^K p_k\norm{\hat f_j^{(-k)}-g_0}_{P,2}^2,
\qquad
e_\alpha^2:=\min_{1\le j\le M}\sum_{k=1}^K p_k\norm{\hat a_j^{(-k)}-\alpha_0}_{P,2}^2.
\label{eq:oracle-errors}
\end{equation}

\begin{proposition}[Pooled star aggregation: oracle inequalities]
\label{prop:pooled-oracle}
Let $K$ be fixed, and let $n_k\ge n/(2K)$ for every $k$.
Suppose that, for every $k$ and conditional on $\{O_i:i\in I_k^c\}$, the menu $\cF^{(-k)}$ satisfies the conditions of Proposition~\ref{prop:star-finite-reg} and the menu $\mathcal A^{(-k)}$ satisfies those of Proposition~\ref{prop:star-finite-riesz}, with constants $(B_Y,B_F,B_A,C_m,B_m)$ that do not depend on $k$.
Then there is a universal constant $C$ such that, for every $\delta\in(0,1)$, each of the following bounds holds with probability at least $1-\delta$:
\begin{align}
\sum_{k=1}^K p_k\norm{\hat g^{(-k)}-g_0}_{P,2}^2
&\le
e_g^2+C\,K\,(B_Y+B_F)^2\,\frac{\log(KM/\delta)}{n},
\label{eq:pooled-oracle-reg}\\
\sum_{k=1}^K p_k\norm{\hat\alpha^{(-k)}-\alpha_0}_{P,2}^2
&\le
e_\alpha^2+C\,K\,(B_A^2+C_m+B_m)\,\frac{\log(KM/\delta)}{n}.
\label{eq:pooled-oracle-riesz}
\end{align}
\end{proposition}

The remainder matches that of Propositions~\ref{prop:star-finite-reg} and~\ref{prop:star-finite-riesz} up to the factor $K$, which comes from treating each fold separately.
Unlike those propositions, Proposition~\ref{prop:pooled-oracle} needs no independence between the menu and the data used to aggregate it.
The proof uses only that each fold's candidates are independent of that fold's observations, the same structure that underlies oracle inequalities for cross-validated selectors~\cite{vandervaart2006oracle,vanderlaan2007super}.

Inference needs one more ingredient.
With strict cross-fitting, the empirical-process remainder in the DML expansion vanishes by conditioning alone.
Here the weights depend on the evaluation fold, so we control the remainder uniformly over the star segments that could be selected.

\begin{assumption}[Regularity for pooled star aggregation]
\label{ass:pooled-star}
Let $K$ be fixed, let $n_k\ge n/(2K)$ for every $k$, and let $M=M_n\ge2$.
The constants $B_Y$, $B$, $B_m$, and $C_m$ below do not depend on $n$.
\begin{enumerate}[label=(\roman*)]
\item \emph{Boundedness.}
$|Y|\le B_Y$ almost surely; every element of $\cF^{(-k)}\cup\mathcal A^{(-k)}$, as well as $g_0$ and $\alpha_0$, has sup norm at most $B$; and $|m(O,f)|\le B_m$ almost surely for every $f\in\cF^{(-k)}\cup\mathcal A^{(-k)}\cup\{g_0\}$ and every $k$.
\item \emph{Mean-square continuity.}
$\E[m(O,h)^2]\le C_m\norm{h}_{P,2}^2$ for all $h\in\cH$.
\item \emph{Rates.}
$e_g=O_\Pp(a_{g,n})$ and $e_\alpha=O_\Pp(a_{\alpha,n})$ for deterministic sequences such that, with
$r_{g,n}:=a_{g,n}+\sqrt{\log M/n}$ and $r_{\alpha,n}:=a_{\alpha,n}+\sqrt{\log M/n}$,
\begin{equation}
\sqrt n\,r_{g,n}\,r_{\alpha,n}\to0
\qquad\text{and}\qquad
\sqrt{\log M}\,\bigl(r_{g,n}+r_{\alpha,n}\bigr)\to0.
\label{eq:pooled-rates}
\end{equation}
\item \emph{Nondegeneracy.}
$\sigma_0^2$ of Assumption~\ref{ass:pipeline-dml} is bounded away from zero.
\end{enumerate}
\end{assumption}

\begin{theorem}[DML inference with pooled star aggregation]
\label{thm:pooled-star}
Suppose Assumption~\ref{ass:pooled-star} holds, and let $\hat g^{(-k)}$ and $\hat\alpha^{(-k)}$ be the pooled star estimators of Definition~\ref{def:pooled-star}.
Then $\hat\theta_0$ satisfies the expansion~\eqref{eq:pipeline-expansion},
\[
\sigma_0^{-1}\sqrt n\bigl(\hat\theta_0-\theta_0\bigr)\Rightarrow N(0,1),
\]
and $\hat\sigma_0^2$ is consistent for $\sigma_0^2$.
Consequently, $\Pp\{\theta_0\in\mathrm{CI}_0(1-\gamma)\}=1-\gamma+o(1)$.
\end{theorem}

The first condition in~\eqref{eq:pooled-rates} is the product-rate template~\eqref{eq:product-rate-template} with the star remainder of Proposition~\ref{prop:pooled-oracle}.
Here $a_{g,n}$ and $a_{\alpha,n}$ bound the oracle errors~\eqref{eq:oracle-errors}, and $r_{g,n}$ and $r_{\alpha,n}$ add the star remainder; in Appendix~\ref{sec:learn-H} the same symbols denote approximation and complexity terms.
The second condition is the price of pooling. It controls the empirical-process remainder uniformly over the star segments whose errors are of order $r_{g,n}$ and $r_{\alpha,n}$.
That set consists of at most $M^4$ pieces of dimension two, one for each pair of segments, so its complexity grows like $\sqrt{\log M}$.
If $a_{g,n}a_{\alpha,n}=o(n^{-1/2})$ and $a_{g,n}+a_{\alpha,n}=O(n^{-1/4})$, both conditions reduce to $\log M=o(\sqrt n)$.
The menu may therefore grow faster than any power of $n$, as it may for star aggregation inside the folds.

\begin{remark}[{What must stay inside the folds}]
\label{rem:inside-folds}
Only the aggregation step crosses folds.
Each candidate must still be trained on $I_k^c$ alone; otherwise the candidates themselves would depend on the evaluation fold, in addition to the low-dimensional choice among them.
Two consequences follow.
First, the empirical-process part of the proof bounds a supremum over the localized star hulls, so it applies to any rule that selects a star segment, provided the selected aggregates attain the rates in Assumption~\ref{ass:pooled-star}(iii).
Second, a discrete choice made on the full sample, such as a prompt or a checkpoint, is covered by adding the rejected alternatives to the menu, at the cost of a larger $M$.
\end{remark}

\medskip
\noindent\emph{Convex weights.}
The same design applies to convex aggregation.
Let $\Delta_M:=\{w\in[0,1]^M:\sum_jw_j=1\}$ and, for $w\in\Delta_M$, write $f_w^{(-k)}:=\sum_jw_j\hat f_j^{(-k)}$ and $a_w^{(-k)}:=\sum_jw_j\hat a_j^{(-k)}$.
Let $\bar R_n(w)$ and $\bar L_n(w)$ be the pooled objectives~\eqref{eq:pooled-reg}--\eqref{eq:pooled-riesz} with $f_w^{(-k)}$ and $a_w^{(-k)}$ in place of $f_\tau^{(-k)}$ and $a_\tau^{(-k)}$.

\begin{definition}[Pooled convex estimator]
\label{def:pooled-convex}
Choose $\hat w_g\in\arg\min_{w\in\Delta_M}\bar R_n(w)$ and $\hat w_\alpha\in\arg\min_{w\in\Delta_M}\bar L_n(w)$, and set $\hat g^{(-k)}:=f^{(-k)}_{\hat w_g}$ and $\hat\alpha^{(-k)}:=a^{(-k)}_{\hat w_\alpha}$ for every $k$.
\end{definition}

Both problems are convex quadratic programs over the simplex.
The benchmark is now the best convex combination, measured by its fold-averaged error:
\begin{equation}
\bar e_g^2:=\min_{w\in\Delta_M}\sum_{k=1}^K p_k\norm{f_w^{(-k)}-g_0}_{P,2}^2,
\qquad
\bar e_\alpha^2:=\min_{w\in\Delta_M}\sum_{k=1}^K p_k\norm{a_w^{(-k)}-\alpha_0}_{P,2}^2.
\label{eq:oracle-errors-convex}
\end{equation}

\begin{proposition}[Pooled convex aggregation: oracle inequalities]
\label{prop:pooled-convex-oracle}
Let $K$ be fixed, let $n_k\ge n/(2K)$ for every $k$, and suppose Assumption~\ref{ass:pooled-star}(i)--(ii) hold.
Then there is a universal constant $C$ such that, for every $\delta\in(0,1)$, each of the following bounds holds with probability at least $1-\delta$:
\begin{align}
\sum_{k=1}^K p_k\norm{\hat g^{(-k)}-g_0}_{P,2}^2
&\le
\bar e_g^2+C\,K\,(B_Y+B)^2\,\frac{M\log(3n)+\log(K/\delta)}{n},
\label{eq:pooled-convex-reg}\\
\sum_{k=1}^K p_k\norm{\hat\alpha^{(-k)}-\alpha_0}_{P,2}^2
&\le
\bar e_\alpha^2+C\,K\,(B^2+C_m+B_m)\,\frac{M\log(3n)+\log(K/\delta)}{n}.
\end{align}
\end{proposition}

\begin{theorem}[DML inference with pooled convex aggregation]
\label{thm:pooled-convex}
Suppose Assumption~\ref{ass:pooled-star}(i), (ii), and~(iv) hold, and let $\hat g^{(-k)}$ and $\hat\alpha^{(-k)}$ be the pooled convex estimators of Definition~\ref{def:pooled-convex}.
Suppose $\bar e_g=O_\Pp(a_{g,n})$ and $\bar e_\alpha=O_\Pp(a_{\alpha,n})$ for deterministic sequences such that, with
$\bar r_{g,n}:=a_{g,n}+\sqrt{M\log n/n}$ and $\bar r_{\alpha,n}:=a_{\alpha,n}+\sqrt{M\log n/n}$,
\begin{equation}
\sqrt n\,\bar r_{g,n}\,\bar r_{\alpha,n}\to0
\qquad\text{and}\qquad
\sqrt{M\log n}\,\bigl(\bar r_{g,n}+\bar r_{\alpha,n}\bigr)\to0.
\label{eq:pooled-convex-rates}
\end{equation}
Then the conclusions of Theorem~\ref{thm:pooled-star} hold.
\end{theorem}

Convex aggregation pays for its richer class twice. Its estimation term grows like $M\log n/n$ rather than $\log M/n$, and pooling costs a factor $\sqrt{M\log n}$ rather than $\sqrt{\log M}$.
If $a_{g,n}a_{\alpha,n}=o(n^{-1/2})$ and $a_{g,n}+a_{\alpha,n}=O(n^{-1/4})$, both conditions in~\eqref{eq:pooled-convex-rates} reduce to $M\log n=o(\sqrt n)$.
Up to the logarithm, this is the condition $M=o(\sqrt n)$ that proper convex aggregation already requires after~\eqref{eq:l1-radius}; pooling the weights across folds adds no further condition.
In return, $\bar e_g\le e_g$ and $\bar e_\alpha\le e_\alpha$: a convex combination can only improve approximation power.

\begin{remark}[Computation]
\label{rem:computation}
Star aggregation inside the folds with an internal split trains each candidate on part of each training fold and aggregates on the rest; nested cross-fitting with $K'$ inner folds requires $KK'$ fits of every candidate.
Pooled aggregation requires $K$ fits, the fewest that cross-fitting allows, and fits the weights to all $n$ out-of-fold predictions.
% With fine-tuned foundation models, the saving can decide whether the analysis is feasible.
\end{remark}

\subsection{Proofs}
\label{app:proofs-cross-fold}

We use the fold notation of Appendix~\ref{app:proofs-sec2}, $n_k=|I_k|$, $\Pn^{(k)}$ and $\Gn^{(k)}=\sqrt{n_k}\,(\Pn^{(k)}-P)$, together with $p_k=n_k/n$ from~\eqref{eq:oracle-errors}.
Write $D_k:=\{O_i:i\in I_k\}$ and $D_k^c:=\{O_i:i\in I_k^c\}$.
Conditioning on $D_k^c$ fixes the fold-$k$ menus, and the observations in $D_k$ are then i.i.d.\ draws from $P$.
For a function $h_u$ indexed by a random $\hat u$, we write $Ph_{\hat u}$ for $(Ph_u)|_{u=\hat u}$: the integral over a fresh observation, with the index held fixed.
Every family indexed by a continuum below is continuous in its index in the sup norm, so its suprema equal suprema over countable dense subsets and are measurable.

\subsubsection{Selection outside the folds}

The first lemma isolates what is lost when a tuning index is chosen on the full sample. The bias term needs nothing new. The empirical-process term, which strict cross-fitting disposes of by conditioning, must instead be bounded uniformly over the indices that could be selected.

\begin{lemma}[Selection outside the folds]
\label{lem:selection-outside}
Let $\mathcal U$ be an index set.
For each fold $k$, let the family $\{(g_u^{(-k)},\alpha_u^{(-k)}):u\in\mathcal U\}$ of pairs in $\cH\times\cH$ be determined by $D_k^c$, and let $\hat u\in\mathcal U$ be a random index that may depend on the entire sample.
Set $\hat g^{(-k)}:=g_{\hat u}^{(-k)}$ and $\hat\alpha^{(-k)}:=\alpha_{\hat u}^{(-k)}$, define $\hat\theta_0$ and $\hat\sigma_0^2$ as in Section~\ref{sec:pipeline-inference}, assume $0<\sigma_0^2<\infty$, and write
\[
h_u^{(k)}:=\psi^\star\bigl(\cdot\,;g_u^{(-k)},\alpha_u^{(-k)}\bigr)-\psi^\star(\cdot\,;g_0,\alpha_0).
\]
Suppose there are sets $\mathcal U_k=\mathcal U_{k,n}\subseteq\mathcal U$, determined by $D_k^c$, such that $\Pp(\hat u\in\mathcal U_k\text{ for every }k)\to1$ and, for every $k$,
\begin{enumerate}[label=(\alph*)]
\item $\sqrt n\,\sup_{u\in\mathcal U_k}\norm{g_u^{(-k)}-g_0}_{P,2}\norm{\alpha_u^{(-k)}-\alpha_0}_{P,2}=o_\Pp(1)$;
\item $\E\bigl[\sup_{u\in\mathcal U_k}\bigl|\Gn^{(k)}h_u^{(k)}\bigr|\bigm| D_k^c\bigr]=o_\Pp(1)$.
\end{enumerate}
Then the expansion~\eqref{eq:pipeline-expansion} holds, and $\sigma_0^{-1}\sqrt n(\hat\theta_0-\theta_0)\Rightarrow N(0,1)$.
If, in addition,
\begin{enumerate}[label=(\alph*),resume]
\item $\sup_{u\in\mathcal U_k}\Pn^{(k)}\bigl(h_u^{(k)}\bigr)^2=o_\Pp(1)$ for every $k$,
\end{enumerate}
then $\hat\sigma_0^2$ is consistent for $\sigma_0^2$.
\end{lemma}

\begin{proof}
Let $E_n:=\{\hat u\in\mathcal U_k\text{ for every }k\}$; by assumption $\Pp(E_n)\to1$, so it suffices to argue on $E_n$.
Write $\psi_0:=\psi^\star(\cdot\,;g_0,\alpha_0)$.
Since $\hat\theta_0=\sum_k p_k\Pn^{(k)}\bigl(\psi_0+h^{(k)}_{\hat u}\bigr)$ and $P\psi_0=\theta_0$,
\[
\sqrt n\bigl(\hat\theta_0-\theta_0\bigr)
=\Gn\psi_0+\sqrt n\,R_{1n}+\sqrt n\,R_{2n},
\]
where
\[
\sqrt n\,R_{1n}:=\sum_{k=1}^K\sqrt{p_k}\;\Gn^{(k)}h^{(k)}_{\hat u},
\qquad
R_{2n}:=\sum_{k=1}^K p_k\,Ph^{(k)}_{\hat u}.
\]

\emph{Bias term.}
For every $u$, both nuisances lie in $\cH$, so Proposition~\ref{prop:pipeline-mixed-bias}(ii) gives
$Ph_u^{(k)}=-P\bigl[(g_u^{(-k)}-g_0)(\alpha_u^{(-k)}-\alpha_0)\bigr]$.
By Cauchy--Schwarz, on $E_n$,
\[
\sqrt n\,|R_{2n}|\le\max_{1\le k\le K}\sqrt n\,\sup_{u\in\mathcal U_k}\norm{g_u^{(-k)}-g_0}_{P,2}\norm{\alpha_u^{(-k)}-\alpha_0}_{P,2}=o_\Pp(1)
\]
by (a).
This step uses no independence between $\hat u$ and the data.

\emph{Empirical-process term.}
On $E_n$, $\sqrt n\,|R_{1n}|\le\sum_k S_k$ with $S_k:=\sup_{u\in\mathcal U_k}|\Gn^{(k)}h_u^{(k)}|$.
Fix $\eta>0$.
By the conditional Markov inequality, $\Pp(S_k>\eta\mid D_k^c)\le\min\{1,\eta^{-1}\E[S_k\mid D_k^c]\}$, which is $o_\Pp(1)$ by (b); since it is bounded by one, dominated convergence gives $\Pp(S_k>\eta)\to0$.
Hence $\sqrt n\,R_{1n}=o_\Pp(1)$, which proves~\eqref{eq:pipeline-expansion}.
The summands of $\Gn\psi_0$ are i.i.d.\ with variance $\sigma_0^2\in(0,\infty)$, so the central limit theorem gives the normal limit.

\emph{Variance.}
Write $\hat\psi_i=\psi_0(O_i)+H_i$, where $H_i:=h^{(k(i))}_{\hat u}(O_i)$.
On $E_n$,
\[
\frac1n\sum_{i=1}^nH_i^2=\sum_{k=1}^Kp_k\Pn^{(k)}\bigl(h^{(k)}_{\hat u}\bigr)^2\le\max_{k}\sup_{u\in\mathcal U_k}\Pn^{(k)}\bigl(h_u^{(k)}\bigr)^2=o_\Pp(1)
\]
by (c).
Cauchy--Schwarz gives
\[
\Bigl|\Pn\hat\psi^2-\Pn\psi_0^2\Bigr|
\le\Bigl(\frac1n\sum_{i=1}^nH_i^2\Bigr)^{1/2}
\Bigl\{\Bigl(\frac1n\sum_{i=1}^nH_i^2\Bigr)^{1/2}+2\bigl(\Pn\psi_0^2\bigr)^{1/2}\Bigr\}
=o_\Pp(1),
\]
and $\Pn\psi_0^2\to P\psi_0^2$ by the law of large numbers.
Since $\hat\sigma_0^2=\Pn\hat\psi^2-\hat\theta_0^2$ and $\hat\theta_0\to\theta_0$ in probability, $\hat\sigma_0^2\to P\psi_0^2-\theta_0^2=\sigma_0^2$ in probability.
\end{proof}

\subsubsection{Chaining over localized pairs of segments}

We use a Bernstein-type maximal inequality \cite[Lemmas~2.2.9 and~2.2.10]{vanderVaartWellner1996}.
Conditional on $D_k^c$, let $f_1,\dots,f_N$ be fixed functions with $\norm{f_j}_{P,2}\le\sigma$ and $\norm{f_j}_\infty\le\beta$.
Then, for a universal constant $C_0$,
\begin{equation}
\E\Bigl[\max_{j\le N}\bigl|\Gn^{(k)}f_j\bigr|\Bigm| D_k^c\Bigr]
\le C_0\Bigl\{\sigma\sqrt{\log(1+N)}+\beta\,\frac{\log(1+N)}{\sqrt{n_k}}\Bigr\}.
\label{eq:bernstein-max}
\end{equation}

\begin{lemma}[Chaining over localized pairs of segments]
\label{lem:pair-chaining}
Fix $k$ and condition on $D_k^c$.
Let $\mathcal U=\bigcup_{s=1}^S\mathcal U_s$, where each piece $\mathcal U_s=\{(s,\lambda,\mu):\lambda\in J_s,\ \mu\in J_s'\}$ is indexed by a product of closed intervals $J_s,J_s'\subseteq[0,1]$ of lengths $\ell_s,\ell_s'$.
Let $\{h_u:u\in\mathcal U\}$ be fixed functions and $\rho,b,L$ constants such that, for all $u$ and for all $u=(s,\lambda,\mu)$ and $u'=(s,\lambda',\mu')$ in a common piece,
\begin{align}
&\norm{h_u}_{P,2}\le\rho,
\qquad
\norm{h_u}_\infty\le b,
\label{eq:pc-size}\\
&\norm{h_u-h_{u'}}_{P,2}\le\rho\Bigl(\frac{|\lambda-\lambda'|}{\ell_s}+\frac{|\mu-\mu'|}{\ell_s'}\Bigr),
\label{eq:pc-l2}\\
&\norm{h_u-h_{u'}}_\infty\le L\bigl(|\lambda-\lambda'|+|\mu-\mu'|\bigr),
\label{eq:pc-lip}
\end{align}
with the convention $0/0:=0$.
Then, for a universal constant $C$,
\[
\E\Bigl[\sup_{u\in\mathcal U}\bigl|\Gn^{(k)}h_u\bigr|\Bigm| D_k^c\Bigr]
\le C\Bigl\{\rho\sqrt{\log(S+1)}+(b+L)\,\frac{\log(S+1)}{\sqrt{n_k}}\Bigr\}.
\]
\end{lemma}

\begin{proof}
Write $J_s=[\lambda_s,\lambda_s+\ell_s]$ and $J_s'=[\mu_s,\mu_s+\ell_s']$.
For $m\ge0$, let $A_m$ be the dyadic grid
$\{(s,\lambda_s+\ell_si2^{-m},\mu_s+\ell_s'i'2^{-m}):s\le S,\ 0\le i,i'\le2^m\}$,
so that $|A_m|\le S\,4^{m+1}$, and let $\pi_m(u)$ round each coordinate of $u=(s,\lambda,\mu)$ down to the grid.
Rounding down to dyadic grids nests, so $\pi_{m-1}(\pi_m(u))=\pi_{m-1}(u)$.
Since $|\lambda-\lambda'|\le\ell_s2^{-m}\le2^{-m}$ for $\lambda'$ the rounded coordinate, and likewise for $\mu$, conditions~\eqref{eq:pc-l2} and~\eqref{eq:pc-lip} give
\[
\norm{h_u-h_{\pi_m(u)}}_{P,2}\le2\rho\,2^{-m},
\qquad
\norm{h_u-h_{\pi_m(u)}}_\infty\le2L\,2^{-m}.
\]
Thus $h_u=h_{\pi_0(u)}+\sum_{m\ge1}\bigl(h_{\pi_m(u)}-h_{\pi_{m-1}(u)}\bigr)$, with the series converging in the sup norm, and
\[
\sup_{u\in\mathcal U}\bigl|\Gn^{(k)}h_u\bigr|
\le\max_{v\in A_0}\bigl|\Gn^{(k)}h_v\bigr|
+\sum_{m\ge1}\max_{v\in A_m}\bigl|\Gn^{(k)}\bigl(h_v-h_{\pi_{m-1}(v)}\bigr)\bigr|.
\]
The increments in the $m$th maximum have $L^2(P)$ norm at most $4\rho2^{-m}$ and sup norm at most $4L2^{-m}$, and there are at most $S4^{m+1}$ of them.
Apply~\eqref{eq:bernstein-max} to each maximum and use $\log(1+S4^{m+1})\le4(m+1)\log(S+1)$:
\[
\E\Bigl[\sup_{u}\bigl|\Gn^{(k)}h_u\bigr|\Bigm| D_k^c\Bigr]
\le C_0\sum_{m\ge0}2^{4-m}(m+1)
\Bigl\{\rho\sqrt{\log(S+1)}+(b+L)\frac{\log(S+1)}{\sqrt{n_k}}\Bigr\}.
\]
The series converges, which proves the lemma.
\end{proof}

\subsubsection{Proof of Proposition~\ref{prop:pooled-oracle}}

\begin{proof}
We prove~\eqref{eq:pooled-oracle-reg}.
The proof of~\eqref{eq:pooled-oracle-riesz} is identical, with the constants of the proof of Proposition~\ref{prop:star-finite-riesz} in place of those of Proposition~\ref{prop:star-finite-reg}.
Fix a comparator index $j\in\{1,\dots,M\}$, and write $f_*^{(-k)}:=\hat f_j^{(-k)}$, $\ell_f:=\ell(\cdot,f)$ for the square loss, and $c_0:=1/9$.
For $f\in\cH$ set
\[
\cE^{(k)}_f:=\ell_f-\ell_{f_*^{(-k)}},
\qquad
\cQ^{(k)}_f:=\cE^{(k)}_f+c_0\bigl(f-f_*^{(-k)}\bigr)^2.
\]

\emph{Step 1: a pooled star margin.}
Consider the augmented observations $\tilde O_i:=(O_i,k(i))$, and let $\mathcal B$ be the vector space of real functions of $(o,k)$.
The stacked candidates $\bar f_j(o,k):=\hat f_j^{(-k)}(o)$ belong to $\mathcal B$, and $\bar R_n$ in~\eqref{eq:pooled-reg} is the average over $\tilde O_1,\dots,\tilde O_n$ of the summands $\bar f\mapsto\{Y_i-\bar f(\tilde O_i)\}^2$.
The pooled star estimator is therefore the output of the star algorithm of \cite[Section~2.1]{vijaykumar2021localization}, applied to the set $\{\bar f_1,\dots,\bar f_M\}$ and the function $\bar R_n$ on $\mathcal B$.
We check the two facts that the star margin requires.

First, $\bar R_n$ is $(\mu,d_n)$-convex on $\mathcal B$ in the sense of \cite[Definition~1]{vijaykumar2021localization}, with $\mu(z)=z^2$ and the pooled empirical pseudometric
\[
d_n(\bar f,\bar g)^2:=\frac1n\sum_{i=1}^n\bigl\{\bar f(\tilde O_i)-\bar g(\tilde O_i)\bigr\}^2
=\sum_{k=1}^Kp_k\,\Pn^{(k)}\bigl\{\bar f(\cdot,k)-\bar g(\cdot,k)\bigr\}^2.
\]
Each summand is a convex quadratic in $\bar f$, and its Bregman divergence at $\bar g$ is exactly $\{\bar f(\tilde O_i)-\bar g(\tilde O_i)\}^2$.
Averaging these identities gives the claim; this is the argument of \cite[Lemma~5]{vijaykumar2021localization}, carried out on the augmented sample.
For the Riesz loss the summand is $\bar a\mapsto\bar a(\tilde O_i)^2-2m\bigl(O_i,\bar a(\cdot,k(i))\bigr)$.
Lemma~5 of~\cite{vijaykumar2021localization} is stated for losses that depend on $\bar a$ only through the prediction $\bar a(\tilde O_i)$, which this summand does not; but the linear term has zero Bregman divergence, so the summand again satisfies Definition~1 with the same $\mu$ and $d_n$, and the averaging step is unchanged.
The same remark applies to Lemma~\ref{lemma:star-margin} itself when it is used for the Riesz loss.

Second, Proposition~4 of~\cite{vijaykumar2021localization} is deterministic. For any $(\mu,d)$-convex function on a vector space and any set $S$, the output $\tilde x$ of the star algorithm satisfies $f(x)-f(\tilde x)\ge\mu\bigl(d(x,\tilde x)/3\bigr)$ for every $x\in S$.
Its proof uses the geometry of the lower sets of $f$ and the triangle inequality for $d$, and nothing about the sample.
Applied to $f=\bar R_n$, $S=\{\bar f_1,\dots,\bar f_M\}$, and $x=\bar f_j$, and written fold by fold, it gives
\begin{equation}
\sum_{k=1}^Kp_k\,\Pn^{(k)}\cQ^{(k)}_{f^{(-k)}_{\hat\tau_g}}\le0.
\label{eq:pooled-margin}
\end{equation}

\emph{Step 2: reduction to fold-wise offset processes.}
Adding $P$ to both sides of~\eqref{eq:pooled-margin} and using $P\cQ^{(k)}_f=P\cE^{(k)}_f+c_0\norm{f-f_*^{(-k)}}_{P,2}^2$,
\[
\sum_{k=1}^Kp_k\,P\cE^{(k)}_{f^{(-k)}_{\hat\tau_g}}
\le\sum_{k=1}^Kp_k\Bigl\{(P-\Pn^{(k)})\cQ^{(k)}_{f^{(-k)}_{\hat\tau_g}}-c_0\norm{f^{(-k)}_{\hat\tau_g}-f_*^{(-k)}}_{P,2}^2\Bigr\}
\le\sum_{k=1}^Kp_k\,\Xi_k,
\]
where
\[
\Xi_k:=\sup_{f\in\conv_2(\cF^{(-k)})}\Bigl\{(P-\Pn^{(k)})\cQ^{(k)}_f-c_0\norm{f-f_*^{(-k)}}_{P,2}^2\Bigr\},
\]
because $f_\tau^{(-k)}\in\conv_2(\cF^{(-k)})$ for every $\tau\in\mathcal T$.
The index $\hat\tau_g$ depends on all folds, but $\Xi_k$ does not involve it.

\emph{Step 3: fold-wise control.}
Conditional on $D_k^c$, the menu $\cF^{(-k)}$ and the comparator $f_*^{(-k)}$ are fixed, the observations in $D_k$ are i.i.d., and Assumptions~\ref{A:bounded}--\ref{A:lip} hold with the constants verified in the proof of Proposition~\ref{prop:star-finite-reg}.
These assumptions are unchanged when $\cQ^{(k)}_f$ is replaced by $-\cQ^{(k)}_f$, so Lemma~\ref{lem:second_part_peeling} applies to $(P-\Pn^{(k)})$ as well as to $(\Pn^{(k)}-P)$.
At confidence level $\delta/(KM)$ it gives
$\Xi_k\le C(B_Y+B_F)^2\log(KM^2/\delta)/n_k\le2C(B_Y+B_F)^2\log(KM/\delta)/n_k$
with conditional probability at least $1-\delta/(KM)$; integrating over $D_k^c$ removes the conditioning.
A union bound over the $K$ folds and the $M$ comparator indices yields an event of probability at least $1-\delta$ on which, for every $j$ simultaneously and with $n_k\ge n/(2K)$,
\[
\sum_{k=1}^Kp_k\,P\cE^{(k)}_{f^{(-k)}_{\hat\tau_g}}\le4C\,K\,(B_Y+B_F)^2\,\frac{\log(KM/\delta)}{n}.
\]
The union over comparators matters: the best index for the fold average depends on every training sample, and so on $D_k$, and cannot be fixed in advance of fold $k$.

\emph{Step 4: from excess risk to $L^2(P)$ error.}
All functions involved lie in $\cH$, so Lemma~\ref{lem:excess-risk-identities} gives
$P\cE^{(k)}_f=\norm{f-g_0}_{P,2}^2-\norm{f_*^{(-k)}-g_0}_{P,2}^2$.
Hence, on the same event,
\[
\sum_{k=1}^Kp_k\norm{\hat g^{(-k)}-g_0}_{P,2}^2
\le\sum_{k=1}^Kp_k\norm{\hat f_j^{(-k)}-g_0}_{P,2}^2
+4C\,K\,(B_Y+B_F)^2\,\frac{\log(KM/\delta)}{n}
\]
for every $j$.
Minimizing over $j$ gives~\eqref{eq:pooled-oracle-reg}.
\end{proof}

\subsubsection{Proof of Theorem~\ref{thm:pooled-star}}

\begin{proof}
We verify the conditions of Lemma~\ref{lem:selection-outside} with $u=(\tau,\sigma)\in\mathcal T\times\mathcal T$, $g_u^{(-k)}:=f_\tau^{(-k)}$, and $\alpha_u^{(-k)}:=a_\sigma^{(-k)}$.
These functions lie in $\cH$, since $\cH$ is linear.
Assumption~\ref{ass:pooled-star}(i)--(ii) imply the conditions of Proposition~\ref{prop:pooled-oracle} with $B_F=B_A=B$, because $m(O,\cdot)$ is linear and so $|m(O,a)|\le B_m$ on every segment.

\emph{Step 1: localization.}
The limits in~\eqref{eq:pooled-rates} are strict, so there is a sequence $C_n\to\infty$ such that
\[
\sqrt n\,C_n^2\,r_{g,n}r_{\alpha,n}\to0
\qquad\text{and}\qquad
C_n\sqrt{\log M}\,(r_{g,n}+r_{\alpha,n})\to0.
\]
Set $R_g:=C_nr_{g,n}$ and $R_\alpha:=C_nr_{\alpha,n}$.
Proposition~\ref{prop:pooled-oracle} with $\delta=1/C_n$ and $\max_kp_k^{-1}\le2K$ gives, with probability at least $1-2/C_n$,
\[
\max_{1\le k\le K}\norm{\hat g^{(-k)}-g_0}_{P,2}^2
\le2K\,e_g^2+2C\,K^2(B_Y+B)^2\,\frac{\log(KMC_n)}{n}.
\]
Since $e_g=O_\Pp(a_{g,n})$ and $C_n\to\infty$, the first term is at most $R_g^2/2$ with probability tending to one.
Since $\log(KMC_n)/\log M\le1+(\log K+\log C_n)/\log2=o(C_n^2)$, the second term is at most $R_g^2/2$ for all large $n$.
The same argument applies to $\hat\alpha^{(-k)}$.
Hence the event
\[
E_n:=\Bigl\{\max_{k}\norm{\hat g^{(-k)}-g_0}_{P,2}\le R_g,\ \ \max_{k}\norm{\hat\alpha^{(-k)}-\alpha_0}_{P,2}\le R_\alpha\Bigr\}
\]
has probability tending to one.
Define
\[
\mathcal U_k:=\Bigl\{(\tau,\sigma)\in\mathcal T\times\mathcal T:\ \norm{f_\tau^{(-k)}-g_0}_{P,2}\le R_g,\ \norm{a_\sigma^{(-k)}-\alpha_0}_{P,2}\le R_\alpha\Bigr\}.
\]
The set $\mathcal U_k$ is determined by $D_k^c$, and $\hat u:=(\hat\tau_g,\hat\tau_\alpha)\in\mathcal U_k$ for every $k$ on $E_n$.

\emph{Step 2: condition (a).}
On $\mathcal U_k$ the product of errors is at most $R_gR_\alpha$, and $\sqrt n\,R_gR_\alpha=\sqrt n\,C_n^2r_{g,n}r_{\alpha,n}\to0$.

\emph{Step 3: condition (b).}
Fix $k$ and condition on $D_k^c$.
For fixed $(a,b)$, the map $\lambda\mapsto\norm{f^{(-k)}_{(a,b,\lambda)}-g_0}_{P,2}$ is convex, so the set of $\lambda$ with $\norm{f^{(-k)}_{(a,b,\lambda)}-g_0}_{P,2}\le R_g$ is a closed interval $J_{ab}$, possibly empty; define $J'_{a'b'}$ from the representer menu in the same way.
Hence $\mathcal U_k$ is a union of at most $S=M^4$ pieces of the form required by Lemma~\ref{lem:pair-chaining}, one for each $(a,b,a',b')$, with $(\lambda,\mu)\in J_{ab}\times J'_{a'b'}$.

We verify~\eqref{eq:pc-size}--\eqref{eq:pc-lip}.
Write $\delta g:=g_u^{(-k)}-g_0$, $\delta\alpha:=\alpha_u^{(-k)}-\alpha_0$, and $U:=Y-g_0(W)$.
Linearity of $m$ gives the exact expansion
\begin{equation}
h_u^{(k)}=m(O,\delta g)-\alpha_0\,\delta g+\delta\alpha\,U-\delta\alpha\,\delta g.
\label{eq:score-diff-expansion}
\end{equation}
Let $C_1:=\sqrt{C_m}+B_Y+5B$ and $C_2:=4(B_m+BB_Y+3B^2)$.
Mean-square continuity, $|\alpha_0|\le B$, $|U|\le B_Y+B$, and $\norm{\delta\alpha}_\infty\le2B$ give
$\norm{h_u^{(k)}}_{P,2}\le C_1\bigl(\norm{\delta g}_{P,2}+\norm{\delta\alpha}_{P,2}\bigr)\le C_1(R_g+R_\alpha)$ on $\mathcal U_k$.
For $u,u'$ in a common piece, write $\Delta g:=g_u^{(-k)}-g_{u'}^{(-k)}$ and $\Delta\alpha:=\alpha_u^{(-k)}-\alpha_{u'}^{(-k)}$.
Then
\[
h_u^{(k)}-h_{u'}^{(k)}
=m(O,\Delta g)-\alpha_0\,\Delta g+\Delta\alpha\,U-\Delta\alpha\,\delta g-\delta\alpha'\,\Delta g,
\]
where $\delta\alpha'$ is the representer error at $u'$, so $\norm{h_u^{(k)}-h_{u'}^{(k)}}_{P,2}\le C_1\bigl(\norm{\Delta g}_{P,2}+\norm{\Delta\alpha}_{P,2}\bigr)$.
On the segment $(a,b)$, $\Delta g=(\lambda-\lambda')(\hat f_a^{(-k)}-\hat f_b^{(-k)})$.
The endpoints of $J_{ab}$ both lie within $R_g$ of $g_0$, so $\ell_{ab}\norm{\hat f_a^{(-k)}-\hat f_b^{(-k)}}_{P,2}\le2R_g$, where $\ell_{ab}$ is the length of $J_{ab}$.
Hence $\norm{\Delta g}_{P,2}\le2R_g|\lambda-\lambda'|/\ell_{ab}$, and likewise $\norm{\Delta\alpha}_{P,2}\le2R_\alpha|\mu-\mu'|/\ell'_{a'b'}$.
Therefore~\eqref{eq:pc-size} and~\eqref{eq:pc-l2} hold with $\rho:=2C_1(R_g+R_\alpha)$.
For the sup norms, $|m(O,\delta g)|\le2B_m$, $\norm{\delta g}_\infty\le2B$, $\norm{\Delta g}_\infty\le2B|\lambda-\lambda'|$, and $|m(O,\Delta g)|\le2B_m|\lambda-\lambda'|$; inserting these bounds in~\eqref{eq:score-diff-expansion} and in the increment identity gives~\eqref{eq:pc-size} and~\eqref{eq:pc-lip} with $b=L=C_2$.
Lemma~\ref{lem:pair-chaining} with $\log(M^4+1)\le5\log M$ now gives
\[
\E\Bigl[\sup_{u\in\mathcal U_k}\bigl|\Gn^{(k)}h_u^{(k)}\bigr|\Bigm| D_k^c\Bigr]
\le C\Bigl\{C_1C_n\,(r_{g,n}+r_{\alpha,n})\sqrt{\log M}+C_2\,\frac{\log M}{\sqrt{n_k}}\Bigr\}.
\]
The first term vanishes by the choice of $C_n$.
The second vanishes because $\log M/\sqrt n\le\sqrt n\,r_{g,n}r_{\alpha,n}\to0$.

\emph{Step 4: condition (c).}
On $\mathcal U_k$, $P(h_u^{(k)})^2\le\rho^2\to0$.
The squared family $\{(h_u^{(k)})^2\}$ satisfies~\eqref{eq:pc-size}--\eqref{eq:pc-lip} with constants $(2C_2\rho,\,C_2^2,\,2C_2^2)$ in place of $(\rho,b,L)$, because $|x^2-y^2|\le2C_2|x-y|$ when $|x|,|y|\le C_2$.
Lemma~\ref{lem:pair-chaining} therefore gives $\sup_{u\in\mathcal U_k}|\Gn^{(k)}(h_u^{(k)})^2|=O_\Pp(1)$, and
\[
\sup_{u\in\mathcal U_k}\Pn^{(k)}(h_u^{(k)})^2
\le\rho^2+n_k^{-1/2}\sup_{u\in\mathcal U_k}\bigl|\Gn^{(k)}(h_u^{(k)})^2\bigr|
=o_\Pp(1).
\]

Assumption~\ref{ass:pooled-star}(i) and~(iv) give $0<\sigma_0^2<\infty$.
Lemma~\ref{lem:selection-outside} now yields the expansion, the normal limit, and consistency of $\hat\sigma_0^2$.
The statement about $\mathrm{CI}_0(1-\gamma)$ follows from the normal limit and $\hat\sigma_0^2/\sigma_0^2\to1$ in probability.
\end{proof}

\subsubsection{Convex weights}

The convex case replaces the star margin by the first-order condition of a convex program, and the chaining over segments by a covering of the simplex.
We use the $\ell_1$ norm $|w|_1:=\sum_j|w_j|$ on weights.
Because every candidate is bounded by $B$ and $m(O,\cdot)$ is linear with $|m(O,\hat f_j^{(-k)})|\le B_m$,
\begin{equation}
\bigl|f_w^{(-k)}-f_v^{(-k)}\bigr|\le B\,|w-v|_1,
\qquad
\bigl|m\bigl(O,f_w^{(-k)}-f_v^{(-k)}\bigr)\bigr|\le B_m\,|w-v|_1,
\label{eq:simplex-lipschitz}
\end{equation}
and the same bounds hold for the representer menu.
A maximal $1/n_k$-separated subset of $\Delta_M$ in $|\cdot|_1$ is a $1/n_k$-net of $\Delta_M$ lying in $\Delta_M$; the $\ell_1$ balls of radius $1/(2n_k)$ about its points are disjoint and lie in the $\ell_1$ ball of radius $1+1/(2n_k)$, so a volume comparison bounds its cardinality by $(2n_k+1)^M\le(3n_k)^M$.

\begin{lemma}[Offset bound over pairs of convex weights]
\label{lem:simplex-offset}
Fix $k$ and condition on $D_k^c$.
For $u=(w,w')\in\Delta_M\times\Delta_M$, let $Q_u$ and $\Delta_u$ be fixed functions, and let $v_0,b,L,B$ be constants such that, for all $u$ and all $u'=(v,v')$,
\begin{enumerate}[label=(\alph*)]
\item $|Q_u|\le b$, \ $PQ_u^2\le v_0\norm{\Delta_u}_{P,2}^2$, \ and $\norm{\Delta_u}_\infty\le2B$;
\item $|Q_u-Q_{u'}|\le L\,(|w-v|_1+|w'-v'|_1)$ and $|\Delta_u-\Delta_{u'}|\le B\,(|w-v|_1+|w'-v'|_1)$.
\end{enumerate}
Then, for every $\eta\in(0,1)$, with conditional probability at least $1-\eta$,
\[
\sup_{u\in\Delta_M\times\Delta_M}\Bigl\{2(P-\Pn^{(k)})Q_u-\norm{\Delta_u}_{P,2}^2\Bigr\}
\le C\,(v_0+b+L+B^2)\,\frac{M\log(3n_k)+\log(1/\eta)}{n_k}
\]
for a universal constant $C$.
\end{lemma}

\begin{proof}
Let $N$ be the $1/n_k$-net of $\Delta_M$ described above, and for $u=(w,w')$ choose $u'=(v,v')\in N\times N$ with $|w-v|_1,|w'-v'|_1\le1/n_k$.
By (b), $|(P-\Pn^{(k)})(Q_u-Q_{u'})|\le4L/n_k$, and by (a) and (b),
\[
\bigl|\norm{\Delta_u}_{P,2}^2-\norm{\Delta_{u'}}_{P,2}^2\bigr|
\le\norm{\Delta_u-\Delta_{u'}}_{P,2}\bigl(\norm{\Delta_u}_{P,2}+\norm{\Delta_{u'}}_{P,2}\bigr)
\le\frac{8B^2}{n_k}.
\]
Hence the supremum is at most $\max_{u'\in N\times N}\{2(P-\Pn^{(k)})Q_{u'}-\norm{\Delta_{u'}}_{P,2}^2\}+(8L+8B^2)/n_k$.
Fix $u'$.
The centered variables $PQ_{u'}-Q_{u'}(O_i)$, $i\in I_k$, are bounded by $2b$ and have variance at most $v_0\norm{\Delta_{u'}}_{P,2}^2$, so Bernstein's inequality gives, with probability at least $1-e^{-x}$,
\[
2(P-\Pn^{(k)})Q_{u'}-\norm{\Delta_{u'}}_{P,2}^2
\le2\sqrt{\frac{2v_0\norm{\Delta_{u'}}_{P,2}^2\,x}{n_k}}+\frac{4bx}{3n_k}-\norm{\Delta_{u'}}_{P,2}^2
\le\frac{2v_0x}{n_k}+\frac{4bx}{3n_k},
\]
where the last step uses $2\sqrt{2as}\le a+2s$.
Take $x:=2M\log(3n_k)+\log(1/\eta)$ and apply a union bound over the at most $(3n_k)^{2M}$ points of $N\times N$.
Since $x\ge1$, the discretization error is absorbed into the constant.
\end{proof}

\paragraph*{Proof of Proposition~\ref{prop:pooled-convex-oracle}}

\begin{proof}
\emph{Regression.}
Write $\hat w:=\hat w_g$, and let $w^\ast\in\Delta_M$ attain the minimum in~\eqref{eq:oracle-errors-convex}; it exists by continuity and compactness, and it depends on every training sample.
Set $\Delta^{(k)}:=f_{\hat w}^{(-k)}-f_{w^\ast}^{(-k)}\in\cH$.
Because $\hat w$ minimizes the convex quadratic $\bar R_n$ over the convex set $\Delta_M$, the derivative of $t\mapsto\bar R_n(\hat w+t(w^\ast-\hat w))$ at $t=0$ is nonnegative, so that
\begin{equation}
\sum_{k=1}^Kp_k\,\Pn^{(k)}\Bigl[\bigl(f_{\hat w}^{(-k)}-Y\bigr)\Delta^{(k)}\Bigr]\le0.
\label{eq:convex-foc}
\end{equation}
For each $k$, expanding the square gives the identity
\[
\norm{f_{\hat w}^{(-k)}-g_0}_{P,2}^2
=\norm{f_{w^\ast}^{(-k)}-g_0}_{P,2}^2-\norm{\Delta^{(k)}}_{P,2}^2+2P\bigl[(f_{\hat w}^{(-k)}-g_0)\Delta^{(k)}\bigr].
\]
Since $g_0$ is the projection of $Y$ onto $\cH$ and $\Delta^{(k)}\in\cH$, $P[(f_{\hat w}^{(-k)}-g_0)\Delta^{(k)}]=P[(f_{\hat w}^{(-k)}-Y)\Delta^{(k)}]$.
Average the identities with weights $p_k$, and subtract twice the left side of~\eqref{eq:convex-foc}, which is nonpositive, to obtain
\[
\sum_{k=1}^Kp_k\norm{\hat g^{(-k)}-g_0}_{P,2}^2
\le\bar e_g^2+\sum_{k=1}^Kp_k\Bigl\{2(P-\Pn^{(k)})Q^{(k)}_{(\hat w,w^\ast)}-\norm{\Delta^{(k)}}_{P,2}^2\Bigr\},
\]
where $Q^{(k)}_{(w,w')}:=(f_w^{(-k)}-Y)(f_w^{(-k)}-f_{w'}^{(-k)})$.
The pair $(\hat w,w^\ast)$ depends on $D_k$, so we bound the $k$th term by its supremum over all pairs, which is the quantity in Lemma~\ref{lem:simplex-offset} with $\Delta_{(w,w')}:=f_w^{(-k)}-f_{w'}^{(-k)}$.
Its conditions hold with $v_0=(B_Y+B)^2$, $b=2B(B_Y+B)$, and $L=B(B_Y+3B)$, because $|f_w^{(-k)}-Y|\le B_Y+B$, by~\eqref{eq:simplex-lipschitz}, and by the decomposition $Q_u-Q_{u'}=(f_w-f_v)\Delta_u+(f_v-Y)(\Delta_u-\Delta_{u'})$.
Lemma~\ref{lem:simplex-offset} with $\eta=\delta/K$, a union bound over folds, and $n_k\ge n/(2K)$ give~\eqref{eq:pooled-convex-reg}.

\emph{Riesz representer.}
The argument is the same with three substitutions.
The first-order condition becomes $\sum_kp_k\Pn^{(k)}[a_{\hat w}^{(-k)}\Delta^{(k)}-m(O,\Delta^{(k)})]\le0$, since the derivative of $a^2-2m(O,a)$ in the direction $-\Delta$ is $-2\{a\Delta-m(O,\Delta)\}$.
The restricted Riesz representation of Proposition~\ref{prop:pipeline-mixed-bias}(i) gives $P[(a_{\hat w}^{(-k)}-\alpha_0)\Delta^{(k)}]=P[a_{\hat w}^{(-k)}\Delta^{(k)}-m(O,\Delta^{(k)})]$.
And $Q^{(k)}_{(w,w')}:=a_w^{(-k)}\Delta_{(w,w')}-m(O,\Delta_{(w,w')})$ satisfies the conditions of Lemma~\ref{lem:simplex-offset} with $v_0=2(B^2+C_m)$, $b=2B^2+2B_m$, and $L=3B^2+B_m$.
\end{proof}

\paragraph*{Proof of Theorem~\ref{thm:pooled-convex}}

\begin{proof}
The proof follows that of Theorem~\ref{thm:pooled-star}, with Proposition~\ref{prop:pooled-convex-oracle} in place of Proposition~\ref{prop:pooled-oracle} and a covering of the simplex in place of Lemma~\ref{lem:pair-chaining}.

\emph{Localization.}
Choose $C_n\to\infty$ such that~\eqref{eq:pooled-convex-rates} still holds when $\bar r_{g,n}$ and $\bar r_{\alpha,n}$ are multiplied by $C_n$, and set $R_g:=C_n\bar r_{g,n}$ and $R_\alpha:=C_n\bar r_{\alpha,n}$.
Proposition~\ref{prop:pooled-convex-oracle} with $\delta=1/C_n$ and $\max_kp_k^{-1}\le2K$ gives, with probability at least $1-2/C_n$,
\[
\max_{1\le k\le K}\norm{\hat g^{(-k)}-g_0}_{P,2}^2
\le2K\,\bar e_g^2+2C\,K^2(B_Y+B)^2\,\frac{M\log(3n)+\log(KC_n)}{n},
\]
and the right side is at most $R_g^2$ with probability tending to one, as in Step~1 of the proof of Theorem~\ref{thm:pooled-star}.
The same holds for $\hat\alpha^{(-k)}$.
Define
\[
\mathcal U_k:=\Bigl\{(w,v)\in\Delta_M\times\Delta_M:\ \norm{f_w^{(-k)}-g_0}_{P,2}\le R_g,\ \norm{a_v^{(-k)}-\alpha_0}_{P,2}\le R_\alpha\Bigr\},
\]
so that $\hat u:=(\hat w_g,\hat w_\alpha)\in\mathcal U_k$ for every $k$ with probability tending to one.
We verify the conditions of Lemma~\ref{lem:selection-outside} with $g_u^{(-k)}:=f_w^{(-k)}$ and $\alpha_u^{(-k)}:=a_v^{(-k)}$.
Condition~(a) holds because $\sqrt n\,R_gR_\alpha\to0$.

\emph{Condition (b).}
Fix $k$ and condition on $D_k^c$.
By~\eqref{eq:score-diff-expansion} and~\eqref{eq:simplex-lipschitz}, and with the constants $C_1$ and $C_2$ of the proof of Theorem~\ref{thm:pooled-star},
\[
\norm{h_u^{(k)}}_{P,2}\le C_1\bigl(\norm{f_w^{(-k)}-g_0}_{P,2}+\norm{a_v^{(-k)}-\alpha_0}_{P,2}\bigr),
\qquad
\norm{h_u^{(k)}}_\infty\le C_2,
\]
\[
\bigl|h_u^{(k)}-h_{u'}^{(k)}\bigr|\le C_2\bigl(|w-w'|_1+|v-v'|_1\bigr)
\qquad\text{for }u=(w,v),\ u'=(w',v').
\]
Let $N$ be the $1/n_k$-net of $\Delta_M$ above and $N_k:=\{u'\in N\times N:\ |u'-u|_1\le2/n_k\text{ for some }u\in\mathcal U_k\}$.
Every $u'\in N_k$ has errors at most $R_g+2B/n_k$ and $R_\alpha+2B/n_k$, so $\norm{h_{u'}^{(k)}}_{P,2}\le\sigma_n:=C_1(R_g+R_\alpha+4B/n_k)$.
Hence
\[
\sup_{u\in\mathcal U_k}\bigl|\Gn^{(k)}h_u^{(k)}\bigr|
\le\max_{u'\in N_k}\bigl|\Gn^{(k)}h_{u'}^{(k)}\bigr|+\frac{8C_2}{\sqrt{n_k}},
\]
and~\eqref{eq:bernstein-max} with $\log(1+|N_k|)\le2M\log(3n_k)+1$ gives
\[
\E\Bigl[\sup_{u\in\mathcal U_k}\bigl|\Gn^{(k)}h_u^{(k)}\bigr|\Bigm| D_k^c\Bigr]
\le C\Bigl\{\sigma_n\sqrt{M\log n}+C_2\,\frac{M\log n}{\sqrt n}\Bigr\}+\frac{8C_2}{\sqrt{n_k}}.
\]
The first term vanishes by the choice of $C_n$, and the second because $M\log n/\sqrt n\le\sqrt n\,\bar r_{g,n}\bar r_{\alpha,n}\to0$.

\emph{Condition (c).}
On $\mathcal U_k$, $P(h_u^{(k)})^2\le\sigma_n^2\to0$.
The squared family has $L^2(P)$ norms at most $C_2\sigma_n$, sup norms at most $C_2^2$, and Lipschitz constant $2C_2^2$ in $|\cdot|_1$, so the same covering argument gives $\sup_{u\in\mathcal U_k}|\Gn^{(k)}(h_u^{(k)})^2|=O_\Pp(\sqrt{M\log n}+M\log n/\sqrt n)$.
Dividing by $\sqrt{n_k}$ and using $M\log n/n\to0$ gives $\sup_{u\in\mathcal U_k}\Pn^{(k)}(h_u^{(k)})^2=o_\Pp(1)$.

Lemma~\ref{lem:selection-outside} completes the proof.
\end{proof}

\section{Approach 2: joint-loss tuning}
\label{app:joint}

This appendix proves the Approach~2 results of Section~\ref{sec:joint}: Lemma~\ref{lem:profile} behind~\eqref{eq:profile-bound}, Proposition~\ref{prop:joint}, Corollary~\ref{cor:joint}, and Theorem~\ref{thm:pooled-joint}.

Throughout this appendix the folds have equal size $n/K$; $\E_{n,k}$ and $\E_{n,-k}$ average over $I_k$ and over its complement.
For a head $\chi$ and a representation $\phi$, write $\chi\circ W_\phi$ for the function $w\mapsto\chi(W_\phi(w))$, and recall $R(h)=\E[\{Y-h(W)\}^2]$ and $Q(a)=\E[a(W)^2-2m(O,a)]$.

\subsection{The profiled joint loss}

\begin{lemma}[Profiled joint loss]
\label{lem:profile}
For every representation $\phi$, \eqref{eq:profile-reg}--\eqref{eq:profile-riesz} hold with the minima attained at $g_\phi$ and $\alpha_\phi$.
Consequently the population version of~\eqref{eq:joint-loss}, profiled over the heads, equals a constant plus
$\mathcal J(\phi)=\norm{\Delta_g(\phi)}_{P,2}^2/V_Y+\norm{\Delta_\alpha(\phi)}_{P,2}^2/V_\alpha$, and~\eqref{eq:profile-bound} holds.
If $V_Y=\E[(Y-g_\phi)^2]$ and $V_\alpha=\E[\alpha_\phi^2]$, then $\mathcal J(\phi)=C_Y^2(\phi)+C_D^2(\phi)$ and $\sqrt{V_YV_\alpha}=S_\phi$.
\end{lemma}

\begin{proof}
The first identity follows from the projection theorem, because $g_\phi=\E[g^\star\mid W_\phi]$ and $Y-g^\star$ is orthogonal to every function of $W$.
For the second, the Riesz property~\eqref{eq:rr-short} gives $\E[m(O,a)]=\E[a\,\alpha_\phi]$ for every square-integrable function $a$ of $W_\phi$, so $Q(a)=\E[(a-\alpha_\phi)^2]-\E[\alpha_\phi^2]$, which is minimized at $a=\alpha_\phi$.
Since $\alpha_\phi=\E[\alpha^\star\mid W_\phi]$ (see the proof of Proposition~\ref{prop:ovb}), $\E[\alpha_\phi^2]=\E[\alpha^{\star2}]-\norm{\Delta_\alpha(\phi)}_{P,2}^2$.
For the bound, \eqref{eq:cs-bound} gives $|\theta^\star-\theta_\phi|\le\norm{\Delta_g}_{P,2}\norm{\Delta_\alpha}_{P,2}$, and $2xy\le\sqrt{V_YV_\alpha}\,(x^2/V_Y+y^2/V_\alpha)$ completes it.
For the last statement, the projection identities give $C_Y^2(\phi)=\norm{\Delta_g(\phi)}^2_{P,2}/\E[(Y-g_\phi)^2]$ and $C_D^2(\phi)=\norm{\Delta_\alpha(\phi)}^2_{P,2}/\E[\alpha_\phi^2]$.
\end{proof}

\subsection{The fold-averaged target}
Because $\hat g_k$ and $\hat\alpha_k$ are both functions of $W_{\hat\phi_k}$, the mixed-bias identity holds exactly at $\theta_{\hat\phi_k}$:
\begin{equation}
\E\bigl[m(O,\hat g_k)+\hat\alpha_k(Y-\hat g_k)\bigr]-\theta_{\hat\phi_k}
=-\E\bigl[(\hat g_k-g_{\hat\phi_k})(\hat\alpha_k-\alpha_{\hat\phi_k})\bigr],
\label{eq:mixed-bias-k}
\end{equation}
where the expectation is over a new observation, with the training folds held fixed.
Indeed, $\E[\hat\alpha_k(Y-g_{\hat\phi_k})]=0$, and the Riesz property gives $\E[m(O,\hat g_k-g_{\hat\phi_k})]=\E[\alpha_{\hat\phi_k}(\hat g_k-g_{\hat\phi_k})]$.
If the two heads were functions of different representations, \eqref{eq:mixed-bias-k} would fail, and there would be no target around which the bias is a product of errors.

\begin{assumption}[Joint-loss tuning]
\label{ass:joint}
For each $k$:
\begin{enumerate}[label=(\alph*),leftmargin=*]
\item \emph{Heads.}
The compatibility condition of Section~\ref{sec:setup} holds for $W_{\hat\phi_k}$, and
\[
\sqrt n\,\norm{\hat g_k-g_{\hat\phi_k}}_{P,2}\,\norm{\hat\alpha_k-\alpha_{\hat\phi_k}}_{P,2}=o_\Pp(1).
\]
\item \emph{Consistency of the score.}
There are nonrandom functions $\psi_n^\ast$, with $\Var(\psi_n^\ast)$ bounded away from zero and $\E|\psi_n^\ast-\E\psi_n^\ast|^{2+\delta}$ bounded for some $\delta>0$, such that $\norm{\hat\psi_k-\psi_n^\ast}_{P,2}=o_\Pp(1)$.
\end{enumerate}
\end{assumption}

Part~(a) requires each head to be close to the best predictor given its own fold's representation; it is the usual DML condition with the representation held fixed, and it places no condition on $\hat\phi_k$.
Part~(b) requires only that the fitted scores converge, without a rate.
Because the score depends on $\hat\phi_k$ only through $(\hat g_k,\hat\alpha_k)$, under the boundedness and mean-square continuity conditions of Assumption~\ref{ass:dml} it suffices that $\hat g_k$ and $\hat\alpha_k$ converge in $L^2(P)$ to nonrandom $g_n^\ast$ and $\alpha_n^\ast$.
The parameter $\hat\phi_k$ itself need not converge, and many parameterizations with the same information set are harmless.

\begin{proof}[Proof of Proposition~\ref{prop:joint}]
Write $P\hat\psi_k$ for the expectation of $\hat\psi_k$ over a new observation, with the training folds held fixed.
Then
\[
\hat\theta-\bar\theta_n
=\frac1K\sum_k(\E_{n,k}-P)\hat\psi_k
+\frac1K\sum_k\bigl(P\hat\psi_k-\theta_{\hat\phi_k}\bigr).
\]
By~\eqref{eq:mixed-bias-k}, Cauchy--Schwarz, and part~(a), each term of the second sum is $o_\Pp(n^{-1/2})$.
For the first sum, write $(\E_{n,k}-P)\hat\psi_k=(\E_{n,k}-P)\psi_n^\ast+(\E_{n,k}-P)(\hat\psi_k-\psi_n^\ast)$.
Conditionally on the complement of $I_k$, the last term is an average of $n/K$ i.i.d.\ centered variables, each with variance at most $\norm{\hat\psi_k-\psi_n^\ast}_{P,2}^2$.
By Chebyshev's inequality and part~(b), it is $o_\Pp(n^{-1/2})$; the conditional bound passes to the unconditional one by dominated convergence.
Therefore
\begin{equation}
\sqrt n(\hat\theta-\bar\theta_n)=\frac1{\sqrt n}\sum_{i=1}^n\bigl\{\psi_n^\ast(O_i)-\E\psi_n^\ast\bigr\}+o_\Pp(1),
\label{eq:joint-expansion}
\end{equation}
and Lyapunov's central limit theorem gives the first claim.

For the variance, the triangle inequality in $L^2(\Pn)$ reduces $\hat\sigma^2$ to $n^{-1}\sum_i\{\psi_n^\ast(O_i)-\E\psi_n^\ast\}^2+o_\Pp(1)$, using part~(b), conditional Markov inequalities, $|P\hat\psi_k-\E\psi_n^\ast|\le\norm{\hat\psi_k-\psi_n^\ast}_{P,2}$, and $\hat\theta-\E\psi_n^\ast=o_\Pp(1)$.
The law of large numbers, valid under the moment bound in part~(b), finishes the proof.
\end{proof}

For one-sided bounds we also need the plug-in scale of each fold.
With the fold-$k$ nuisances, fitted on $I_k^c$, and the held-out observations $I_k$, set
\begin{equation*}
\hat V_{Y,k}:=\E_{n,k}(Y-\hat g_k)^2,
\qquad
\hat V_{\alpha,k}:=\E_{n,k}\bigl\{2m(O,\hat\alpha_k)-\hat\alpha_k^2\bigr\},
\qquad
\hat S_k:=(\hat V_{Y,k}\hat V_{\alpha,k})^{1/2},
\end{equation*}
the estimators of $V_{Y,\hat\phi_k}:=\E(Y-g_{\hat\phi_k})^2$, $V_{\alpha,\hat\phi_k}:=\E\alpha_{\hat\phi_k}^2$, and $S_{\hat\phi_k}$.
Let $\hat S:=K^{-1}\sum_k\hat S_k$ and $\hat\theta^\pm:=\hat\theta\pm c\hat S$.
The estimated endpoint scores and their variance estimates are
\begin{equation*}
\hat\psi_{k,\pm}:=\hat\psi_k\pm\frac{c}{2\hat S_k}\Bigl\{\hat V_{\alpha,k}(Y-\hat g_k)^2+\hat V_{Y,k}\bigl[2m(O,\hat\alpha_k)-\hat\alpha_k^2\bigr]\Bigr\},
\qquad
\hat s_\pm^2:=\frac1K\sum_{k=1}^K\E_{n,k}\bigl(\hat\psi_{k,\pm}-\hat\theta^\pm\bigr)^2,
\end{equation*}
and, as in~\eqref{eq:one-sided}, $L:=\hat\theta^--z_{1-\gamma}\hat s_-/\sqrt n$ and $U:=\hat\theta^++z_{1-\gamma}\hat s_+/\sqrt n$.
Write $\bar\theta_n^\pm:=\bar\theta_n\pm c\bar S_n$.

\begin{corollary}[Sensitivity bounds for the fold-averaged target]
\label{cor:joint}
Fix a scenario $c=\bar\rho\,\bar C_Y\bar C_D$ that is valid for each $W_{\hat\phi_k}$.
Then
\begin{equation*}
\theta^\star\in\bigl[\bar\theta_n-c\bar S_n,\;\bar\theta_n+c\bar S_n\bigr],
\qquad
\bar S_n:=\frac1K\sum_{k=1}^KS_{\hat\phi_k}.
\end{equation*}
Suppose in addition that:
\begin{enumerate}[label=(\alph*),leftmargin=*]
\item for each $k$, Assumption~\ref{ass:joint}(a) holds with the product rate strengthened to $\norm{\hat g_k-g_{\hat\phi_k}}_{P,2}+\norm{\hat\alpha_k-\alpha_{\hat\phi_k}}_{P,2}=o_\Pp(n^{-1/4})$; for some $q>4$, $Y\in L^q(P)$ and $\norm{\hat g_k}_{P,q}+\norm{\hat\alpha_k}_{P,q}=O_\Pp(1)$; and $\norm{m(O,h)}_{P,2}\le C_m\norm{h}_{P,2}$ for all $h\in L^2(P_W)$;
\item there is a constant $c_S>0$ with $\Pp\{\min_kS_{\hat\phi_k}\ge c_S\}\to1$;
\item there are nonrandom $\psi_{n,\pm}^\ast$, with $\Var(\psi_{n,\pm}^\ast)=:s_{n,\pm}^2$ bounded away from zero and $\E|\psi_{n,\pm}^\ast-\E\psi_{n,\pm}^\ast|^{2+\delta}$ bounded for some $\delta>0$, such that $\norm{\psi^\circ_{k,\pm}-\psi_{n,\pm}^\ast}_{P,2}=o_\Pp(1)$ for each $k$, where $\psi^\circ_{k,\pm}$ is $\hat\psi_{k,\pm}$ with $(\hat S_k,\hat V_{Y,k},\hat V_{\alpha,k})$ replaced by $(S_{\hat\phi_k},V_{Y,\hat\phi_k},V_{\alpha,\hat\phi_k})$.
\end{enumerate}
Then $\sqrt n(\hat\theta^\pm-\bar\theta_n^\pm)/s_{n,\pm}\Rightarrow N(0,1)$ and $\hat s_\pm^2/s_{n,\pm}^2\to1$ in probability, and the one-sided bounds $L$ and $U$ satisfy~\eqref{eq:onesided-target}.
\end{corollary}

\begin{proof}
For each $k$, \eqref{eq:bounds} applied to the representation $W_{\hat\phi_k}$ gives $|\theta^\star-\theta_{\hat\phi_k}|\le cS_{\hat\phi_k}$; average over $k$ and apply the triangle inequality.

For the second statement, fix $k$, condition on $I_k^c$, and write $V_{Y,k}$, $V_{\alpha,k}$, $S_k$ for $V_{Y,\hat\phi_k}$, $V_{\alpha,\hat\phi_k}$, $S_{\hat\phi_k}$.
By the Riesz property~\eqref{eq:rr-short}, $P(Y-\hat g_k)^2=V_{Y,k}+\norm{\hat g_k-g_{\hat\phi_k}}_{P,2}^2$ and $P\{2m(O,\hat\alpha_k)-\hat\alpha_k^2\}=V_{\alpha,k}-\norm{\hat\alpha_k-\alpha_{\hat\phi_k}}_{P,2}^2$, and by (a) the two squared norms are $o_\Pp(n^{-1/2})$.
By (a), $q>4$ and the bound on $m$ give the summands bounded second moments, so conditional Chebyshev gives $\hat V_{Y,k}-V_{Y,k}=(\E_{n,k}-P)(Y-\hat g_k)^2+o_\Pp(n^{-1/2})=O_\Pp(n^{-1/2})$, and likewise for $\hat V_{\alpha,k}$.
On the event in~(b), $V_{Y,k}\le\E Y^2$ and $V_{\alpha,k}\le\E\alpha^{\star2}$, so $S_k\ge c_S$ bounds both away from zero, and a second-order Taylor expansion of $(v,w)\mapsto(vw)^{1/2}$ gives
\begin{equation*}
\hat S_k=S_k+\frac1{2S_k}\bigl\{V_{\alpha,k}(\hat V_{Y,k}-V_{Y,k})+V_{Y,k}(\hat V_{\alpha,k}-V_{\alpha,k})\bigr\}+O_\Pp(n^{-1}).
\end{equation*}
Let $\psi^\circ_{k,\pm}$ be as in~(c).
Since $\E_{n,k}\psi^\circ_{k,\pm}=\E_{n,k}\hat\psi_k\pm c\{S_k+(2S_k)^{-1}[V_{\alpha,k}(\hat V_{Y,k}-V_{Y,k})+V_{Y,k}(\hat V_{\alpha,k}-V_{\alpha,k})]\}$, the expansion gives $\E_{n,k}\hat\psi_k\pm c\hat S_k=\E_{n,k}\psi^\circ_{k,\pm}+O_\Pp(n^{-1})$.
Also $P\psi^\circ_{k,\pm}=P\hat\psi_k\pm cS_k+o_\Pp(n^{-1/2})=\theta_{\hat\phi_k}\pm cS_k+o_\Pp(n^{-1/2})$, by the two identities above and~\eqref{eq:mixed-bias-k}.
Averaging over $k$,
\begin{equation*}
\sqrt n\bigl(\hat\theta^\pm-\bar\theta_n^\pm\bigr)=\frac1K\sum_{k=1}^K\sqrt n(\E_{n,k}-P)\psi^\circ_{k,\pm}+o_\Pp(1).
\end{equation*}
By~(c) and the argument of Proposition~\ref{prop:joint} (conditional Chebyshev, then dominated convergence), the right side is $\Gn\psi_{n,\pm}^\ast+o_\Pp(1)$, and Lyapunov's theorem gives the limit law.
For the variance, note that $\hat\psi_{k,\pm}$ depends on the observations in $I_k$ through $(\hat S_k,\hat V_{Y,k},\hat V_{\alpha,k})$, so we compare it with $\psi^\circ_{k,\pm}$ in the empirical norm of $I_k$.
Their difference is $\pm\frac c2\{a_k(Y-\hat g_k)^2+b_k[2m(O,\hat\alpha_k)-\hat\alpha_k^2]\}$, with $a_k:=\hat V_{\alpha,k}/\hat S_k-V_{\alpha,k}/S_k$ and $b_k:=\hat V_{Y,k}/\hat S_k-V_{Y,k}/S_k$; both are $o_\Pp(1)$, because $\hat V_{Y,k}$, $\hat V_{\alpha,k}$, and $\hat S_k$ converge in probability to their targets and $S_k\ge c_S$.
By~(a), $P(Y-\hat g_k)^4\le8(\norm{Y}_{P,4}^4+\norm{\hat g_k}_{P,4}^4)=O_\Pp(1)$ and $P\{2m(O,\hat\alpha_k)-\hat\alpha_k^2\}^2\le8C_m^2\norm{\hat\alpha_k}_{P,2}^2+2\norm{\hat\alpha_k}_{P,4}^4=O_\Pp(1)$, so the conditional Markov inequality gives
\begin{equation*}
\E_{n,k}(Y-\hat g_k)^4+\E_{n,k}\{2m(O,\hat\alpha_k)-\hat\alpha_k^2\}^2=O_\Pp(1),
\end{equation*}
and hence $\E_{n,k}(\hat\psi_{k,\pm}-\psi^\circ_{k,\pm})^2\le\frac{c^2}2\bigl[a_k^2\,\E_{n,k}(Y-\hat g_k)^4+b_k^2\,\E_{n,k}\{2m(O,\hat\alpha_k)-\hat\alpha_k^2\}^2\bigr]=o_\Pp(1)$.
Since $\psi^\circ_{k,\pm}$ is fixed given $I_k^c$, (c) and the conditional Markov inequality give $\E_{n,k}(\psi^\circ_{k,\pm}-\psi^\ast_{n,\pm})^2=o_\Pp(1)$, so $\E_{n,k}(\hat\psi_{k,\pm}-\psi^\ast_{n,\pm})^2=o_\Pp(1)$.
With this bound in place of the one drawn from Assumption~\ref{ass:joint}(b), the variance claim follows as in the proof of Proposition~\ref{prop:joint}.
Finally, $\{\bar\theta_n^-<L\}=\{\sqrt n(\hat\theta^--\bar\theta_n^-)/\hat s_->z_{1-\gamma}\}$ has probability tending to $\gamma$, and likewise $\{\bar\theta_n^+>U\}$.
By the first statement, $\theta^\star\in[\bar\theta_n^-,\bar\theta_n^+]$, so $\{\theta^\star<L\}\subseteq\{\bar\theta_n^-<L\}$ and $\{\theta^\star>U\}\subseteq\{\bar\theta_n^+>U\}$, which gives~\eqref{eq:onesided-target}.
\end{proof}

Two cautions complete the picture.
First, the empirical Riesz loss is unbounded below over flexible heads; with propensities near zero or one it rewards extreme weights.
The Riesz head needs regularization, clipping, or early stopping.
Second, a common shortcut is to tune $\phi$ once on all the data and to cross-fit only the heads.
Then $\hat\psi_k$ depends on the observations in $I_k$, the term $(\E_{n,k}-P)\hat\psi_k$ need not have mean zero, and the proof of Proposition~\ref{prop:joint} fails at its first step.
The shortcut is safe when the tuning carries little information about the sample, as when one chooses among a handful of recipes; Theorem~\ref{thm:pooled-joint} makes this precise.
It is not safe for fine-tuning millions of weights, which can memorize the sample.

\subsection{Pooled selection across folds}
\label{app:pooled-joint}

For each recipe $j$ and fold $k$, let $(\hat\phi_{jk},\hat g_{jk},\hat\alpha_{jk})$ be computed on $I_k^c$ by~\eqref{eq:joint-loss} restricted to recipe $j$, and let $\hat\psi_{jk}$ be the corresponding score.
Write $\hat\theta_j:=K^{-1}\sum_k\E_{n,k}\hat\psi_{jk}$,
\[
\hat\sigma_{jl}:=\frac1n\sum_{k=1}^K\sum_{i\in I_k}\bigl\{\hat\psi_{jk}(O_i)-\hat\theta_j\bigr\}\bigl\{\hat\psi_{lk}(O_i)-\hat\theta_l\bigr\},
\qquad
\hat\sigma_j^2:=\hat\sigma_{jj},
\qquad
\hat\Omega_{jl}:=\frac{\hat\sigma_{jl}}{\hat\sigma_j\hat\sigma_l}.
\]
Let $\Sigma_n$ be the covariance matrix of $(\psi_{n,1}^\ast,\dots,\psi_{n,M}^\ast)$, with diagonal entries $\sigma_{n,j}^2$, and $\Omega_n$ its correlation matrix.
For a correlation matrix $\Omega$, let $F_\Omega(t):=\Pp\{\max_j|\xi_j|\le t\}$ with $\xi\sim N(0,\Omega)$, and let $c_{1-\gamma}(\Omega)$ be its $(1-\gamma)$-quantile, so that $\hat c_{1-\gamma}=c_{1-\gamma}(\hat\Omega)$.

\begin{lemma}[Properties of the max-$t$ quantile]
\label{lem:maxt-quantile}
For every correlation matrix $\Omega$, $F_\Omega$ is continuous and strictly increasing on $[0,\infty)$, and $\Omega\mapsto c_{1-\gamma}(\Omega)$ is continuous.
\end{lemma}

\begin{proof}
Each $|\xi_j|$ has a continuous distribution, so $\Pp\{\max_j|\xi_j|=t\}\le\sum_j\Pp\{|\xi_j|=t\}=0$, and $F_\Omega$ is continuous.
For $0\le t<t'$, the set $\{x\in\mathrm{range}(\Omega):\ t<\max_j|x_j|<t'\}$ is relatively open in the support $\mathrm{range}(\Omega)$ of $N(0,\Omega)$, and it is nonempty: it contains $s\,\Omega e_1$ for $s\in(t,t')$, because $\Omega_{11}=1$ and $|\Omega_{j1}|\le1$.
A nondegenerate Gaussian law on its support charges every nonempty relatively open set, so $F_\Omega(t')>F_\Omega(t)$.
If $\Omega_m\to\Omega$, then $N(0,\Omega_m)\Rightarrow N(0,\Omega)$, and, since $F_\Omega$ is continuous, $F_{\Omega_m}(t)\to F_\Omega(t)$ for every $t$.
A sequence of distribution functions converging pointwise to a continuous, strictly increasing limit has quantiles converging to the limit's quantiles.
\end{proof}

\begin{lemma}[One-sided max-$t$ quantile]
\label{lem:maxt-quantile-onesided}
For a correlation matrix $\Omega$ of order $M$ and $\xi\sim N(0,\Omega)$, let $F^+_\Omega(t):=\Pp\{\max_j\xi_j\le t\}$.
Then $F^+_\Omega$ is continuous, and $F^+_\Omega(t')>F^+_\Omega(t)$ whenever $t<t'$ and $t'>0$.
Consequently, for $\gamma\in(0,1/2)$, the $(1-\gamma)$-quantile $c^+_{1-\gamma}(\Omega)$ of $\max_j\xi_j$ is the unique solution of $F^+_\Omega(c)=1-\gamma$, it is positive, $\Omega\mapsto c^+_{1-\gamma}(\Omega)$ is continuous, and
\begin{equation*}
z_{1-\gamma}\le c^+_{1-\gamma}(\Omega)\le z_{1-\gamma/M}.
\end{equation*}
\end{lemma}

\begin{proof}
Each $\xi_j$ has a continuous distribution, so $\Pp\{\max_j\xi_j=t\}\le\sum_j\Pp\{\xi_j=t\}=0$, and $F^+_\Omega$ is continuous.
For $t<t'$ with $t'>0$, the set $\{x\in\mathrm{range}(\Omega):\ t<\max_jx_j<t'\}$ is relatively open in the support $\mathrm{range}(\Omega)$ of $N(0,\Omega)$, and it is nonempty: it contains $s\,\Omega e_1$ for $s\in(\max\{t,0\},t')$, because $\Omega_{11}=1$ and $\Omega_{j1}\le1$ give $\max_j s\Omega_{j1}=s$.
A nondegenerate Gaussian law on its support charges every nonempty relatively open set, so $F^+_\Omega(t')>F^+_\Omega(t)$.
Now $F^+_\Omega(0)\le\Pp\{\xi_1\le0\}=1/2<1-\gamma$, so every solution of $F^+_\Omega(c)=1-\gamma$ is positive; a solution exists by continuity, and it is unique by the strict increase.
If $\Omega_m\to\Omega$, then $N(0,\Omega_m)\Rightarrow N(0,\Omega)$, so $F^+_{\Omega_m}(t)\to F^+_\Omega(t)$ for every $t$ because $F^+_\Omega$ is continuous; a sequence of distribution functions converging pointwise to a limit that is continuous and strictly increasing at the level $1-\gamma$ has quantiles converging to the limit's quantile.
For the bounds, $F^+_\Omega(t)\le\Pp\{\xi_1\le t\}=\Phi(t)$ gives $\Phi(c^+_{1-\gamma}(\Omega))\ge1-\gamma$, and Bonferroni's inequality gives $\Pp\{\max_j\xi_j>z_{1-\gamma/M}\}\le\sum_j\Pp\{\xi_j>z_{1-\gamma/M}\}=\gamma$, so $F^+_\Omega(z_{1-\gamma/M})\ge1-\gamma=F^+_\Omega(c^+_{1-\gamma}(\Omega))$; since $z_{1-\gamma/M}>0$, strict increase gives $c^+_{1-\gamma}(\Omega)\le z_{1-\gamma/M}$.
\end{proof}

\begin{lemma}[{Stability of pooled selection under a margin}]
\label{lem:pooled-margin}
Let $M$ be fixed.
Suppose that for each $j$ there are nonrandom $g_{n,j}^\ast,\alpha_{n,j}^\ast\in L^2(P_W)$ such that
\[
\max_{k}\bigl\{\norm{\hat g_{jk}-g_{n,j}^\ast}_{P,2}+\norm{\hat\alpha_{jk}-\alpha_{n,j}^\ast}_{P,2}\bigr\}=o_\Pp(1),
\qquad
\max_{k}\bigl\{\norm{\hat g_{jk}}_{P,4}+\norm{\hat\alpha_{jk}}_{P,4}\bigr\}=O_\Pp(1);
\]
that $\norm{Y}_{P,4}<\infty$ and $\sup_n\{\norm{g_{n,j}^\ast}_{P,2}+\norm{\alpha_{n,j}^\ast}_{P,2}\}<\infty$; and that $\norm{m(O,h)}_{P,2}\le C_m\norm{h}_{P,2}$ for all $h\in L^2(P_W)$.
Let the normalizers be constants $V_Y,V_\alpha>0$, or estimators converging in probability to such constants.
Define $J_n^\ast(j):=R(g_{n,j}^\ast)/V_Y+Q(\alpha_{n,j}^\ast)/V_\alpha$.
If there are $j_n^\ast$ and $\eta>0$ such that $J_n^\ast(j)\ge J_n^\ast(j_n^\ast)+\eta$ for all $j\ne j_n^\ast$ and all $n$, then the rule~\eqref{eq:pooled-joint} satisfies $\Pp\{\hat\jmath=j_n^\ast\}\to1$.
\end{lemma}

\begin{proof}
It suffices to show $\max_j|\bar J_n(j)-J_n^\ast(j)|=o_\Pp(1)$, for then, with probability tending to one, $\bar J_n(j)-\bar J_n(j_n^\ast)>\eta/2>0$ for every $j\ne j_n^\ast$.
Since $M$ is fixed, it suffices to treat one $j$, and, since the normalizers converge to positive constants and the unnormalized pooled losses are $O_\Pp(1)$ by the bounds below, we may take $V_Y$ and $V_\alpha$ fixed.

\emph{Regression part.}
Write $\bar R_n(j):=K^{-1}\sum_k\E_{n,k}(Y-\hat g_{jk})^2$ and $P(Y-\hat g_{jk})^2=R(\hat g_{jk})$ for the expectation over a new observation with the training folds held fixed.
Then
\[
\bar R_n(j)-R(g_{n,j}^\ast)=\frac1K\sum_k(\E_{n,k}-P)(Y-\hat g_{jk})^2+\frac1K\sum_k\bigl\{R(\hat g_{jk})-R(g_{n,j}^\ast)\bigr\}.
\]
Conditionally on $I_k^c$, the $k$th term of the first sum is an average of $n/K$ i.i.d.\ centered variables with variance at most $(K/n)\norm{Y-\hat g_{jk}}_{P,4}^4=O_\Pp(n^{-1})$; by conditional Chebyshev it is $O_\Pp(n^{-1/2})$.
For the second sum, $(Y-\hat g)^2-(Y-g^\ast)^2=(g^\ast-\hat g)(2Y-\hat g-g^\ast)$, so
$|R(\hat g_{jk})-R(g_{n,j}^\ast)|\le\norm{\hat g_{jk}-g_{n,j}^\ast}_{P,2}\bigl(2\norm{Y}_{P,2}+\norm{\hat g_{jk}}_{P,2}+\norm{g_{n,j}^\ast}_{P,2}\bigr)=o_\Pp(1)$.

\emph{Riesz part.}
The same decomposition applies to $\bar Q_n(j):=K^{-1}\sum_k\E_{n,k}\{\hat\alpha_{jk}^2-2m(O,\hat\alpha_{jk})\}$.
The conditional variance of the $k$th centered term is at most $(K/n)\E\{\hat\alpha_{jk}^2-2m(O,\hat\alpha_{jk})\}^2\le(2K/n)\bigl(\norm{\hat\alpha_{jk}}_{P,4}^4+4C_m^2\norm{\hat\alpha_{jk}}_{P,2}^2\bigr)=O_\Pp(n^{-1})$.
By linearity of $m$,
\[
|Q(\hat\alpha_{jk})-Q(\alpha_{n,j}^\ast)|
\le\norm{\hat\alpha_{jk}-\alpha_{n,j}^\ast}_{P,2}\bigl(\norm{\hat\alpha_{jk}}_{P,2}+\norm{\alpha_{n,j}^\ast}_{P,2}+2C_m\bigr)=o_\Pp(1).
\]
Combining the two parts proves the claim.
\end{proof}

\begin{proof}[Proof of Theorem~\ref{thm:pooled-joint}]
\emph{Step 1: expansions for each recipe.}
For a fixed recipe $j$ no selection is involved, so~\eqref{eq:joint-expansion} gives
$\sqrt n\{\hat\theta_j-\bar\theta_n(j)\}=\Gn\psi_{n,j}^\ast+o_\Pp(1)$.
The proof of Proposition~\ref{prop:joint} also shows $\hat\sigma_j^2-\sigma_{n,j}^2=o_\Pp(1)$.
Because $\psi_{n,j}^\ast+\psi_{n,l}^\ast$ satisfies the moment bound of Assumption~\ref{ass:joint}(b), and $\hat\psi_{jk}+\hat\psi_{lk}$ converges to it in $L^2(P)$, the same argument gives consistency of the empirical variance of the summed scores; polarization then yields $\hat\sigma_{jl}-(\Sigma_n)_{jl}=o_\Pp(1)$.
Since the variances are bounded away from zero, $\hat\Omega-\Omega_n=o_\Pp(1)$.
As $M$ is fixed, all these statements hold jointly over $j$ and $l$.

\emph{Step 2: a joint limit along subsequences.}
Let $T_n:=\max_j\sqrt n\,|\hat\theta_j-\bar\theta_n(j)|/\hat\sigma_j$ and $Z_{n,j}:=\Gn\psi_{n,j}^\ast/\sigma_{n,j}$.
By Step~1, $T_n=\max_j|Z_{n,j}|+o_\Pp(1)$.
Correlation matrices form a compact set, so every subsequence has a further subsequence along which $\Omega_n\to\Omega$.
Along it, for every $t\in\R^M$, $t^\top Z_n$ is a normalized sum of i.i.d.\ variables with variance $t^\top\Omega_nt\to t^\top\Omega t$; if $t^\top\Omega t=0$ it tends to zero in probability by Chebyshev, and otherwise Lyapunov's theorem applies, because the standardized scores have bounded moments of order $2+\delta$.
By the Cram\'er--Wold device, $Z_n\Rightarrow N(0,\Omega)$, and so $T_n\Rightarrow\max_j|\xi_j|$ with $\xi\sim N(0,\Omega)$.
By Lemma~\ref{lem:maxt-quantile} and Step~1, $\hat c_{1-\gamma}\to c_{1-\gamma}(\Omega)$ in probability along the same subsequence, and, since $F_\Omega$ is continuous,
$\Pp\{T_n\le\hat c_{1-\gamma}\}\to F_\Omega\bigl(c_{1-\gamma}(\Omega)\bigr)=1-\gamma$.
As every subsequence has a further subsequence with this limit, $\Pp\{T_n\le\hat c_{1-\gamma}\}\to1-\gamma$.

\emph{Step 3: part (i).}
Whatever the value of $\hat\jmath$, $\sqrt n|\hat\theta_{\hat\jmath}-\bar\theta_n(\hat\jmath)|/\hat\sigma_{\hat\jmath}\le T_n$.
Hence the event $\{T_n\le\hat c_{1-\gamma}\}$ is contained in the coverage event of part~(i), and Step~2 completes the proof.

\emph{Step 4: part (ii).}
For each fixed $j$, Proposition~\ref{prop:joint} shows that the coverage probability $\pi_{n,j}$ of $[\hat\theta_j\pm z_{1-\gamma/2}\hat\sigma_j/\sqrt n]$ for $\bar\theta_n(j)$ tends to $1-\gamma$; as $M$ is fixed, $\max_j|\pi_{n,j}-(1-\gamma)|\to0$.
The coverage events for $\hat\jmath$ and for $j_n^\ast$ coincide on $\{\hat\jmath=j_n^\ast\}$, so the coverage probability for $\hat\jmath$ differs from $\pi_{n,j_n^\ast}$ by at most $\Pp\{\hat\jmath\ne j_n^\ast\}\to0$.

\emph{Step 5: part (iii).}
This is Lemma~\ref{lem:pooled-margin}.
\end{proof}

\begin{remark}[Sensitivity analysis after pooled selection]
\label{rem:pooled-sensitivity}
The first statement of Corollary~\ref{cor:joint} holds for every recipe, hence for $\hat\jmath$, whenever the scenario is valid for each fold-wise representation of the selected recipe.
Under the margin condition, the one-sided bounds for recipe $j_n^\ast$ apply unchanged.
Without it, suppose the scenario is valid for the fold-wise representations of every recipe, the hypotheses of Corollary~\ref{cor:joint} hold for every recipe $j$, and $\gamma<1/2$.
Let $\hat\theta_j^\pm$ and $\hat s_{j,\pm}$ be the endpoint estimates and standard errors of recipe $j$, let $\hat\Omega_\pm$ be the empirical correlation matrix of the $M$ cross-fitted endpoint scores $(\hat\psi_{jk,\pm})_{j\le M}$, and let $c^+_{1-\gamma}(\Omega)$ be the $(1-\gamma)$-quantile of $\max_j\xi_j$ for $\xi\sim N(0,\Omega)$, which Lemma~\ref{lem:maxt-quantile-onesided} shows to exist, to depend continuously on $\Omega$, and to satisfy $c^+_{1-\gamma}(\Omega)\le z_{1-\gamma/M}$.
In the one-sided bounds of recipe $\hat\jmath$, replacing $z_{1-\gamma}$ by $c^+_{1-\gamma}(\hat\Omega_-)$ in the lower bound and by $c^+_{1-\gamma}(\hat\Omega_+)$ in the upper bound restores~\eqref{eq:onesided-target}.
Indeed, for each $j$ the proof of Corollary~\ref{cor:joint} gives $\sqrt n(\hat\theta_j^--\bar\theta_n^-(j))=\Gn\psi_{n,j,-}^\ast+o_\Pp(1)$ and $\hat s_{j,-}^2-s_{n,j,-}^2=o_\Pp(1)$, with $\bar\theta_n^\pm(j):=\bar\theta_n(j)\pm c\bar S_n(j)$; the argument of Step~1 of the proof of Theorem~\ref{thm:pooled-joint} makes $\hat\Omega_--\Omega_{n,-}=o_\Pp(1)$ for the correlation matrix $\Omega_{n,-}$ of the limits, and Step~2, with Lemma~\ref{lem:maxt-quantile-onesided} in place of Lemma~\ref{lem:maxt-quantile}, gives $\Pp\{T_n^-\le c^+_{1-\gamma}(\hat\Omega_-)\}\to1-\gamma$ for $T_n^-:=\max_j\sqrt n(\hat\theta_j^--\bar\theta_n^-(j))/\hat s_{j,-}$.
For any selection rule, $\{\bar\theta_n^-(\hat\jmath)<L\}\subseteq\{T_n^->c^+_{1-\gamma}(\hat\Omega_-)\}$, whose probability tends to $\gamma$.
The upper bound is symmetric, with $T_n^+:=\max_j\sqrt n(\bar\theta_n^+(j)-\hat\theta_j^+)/\hat s_{j,+}$, because $N(0,\Omega)$ is symmetric.
Since $\theta^\star\in[\bar\theta_n^-(\hat\jmath),\bar\theta_n^+(\hat\jmath)]$ by the first part of Corollary~\ref{cor:joint}, the claim~\eqref{eq:onesided-target} follows.
\end{remark}

\section{Approach 3: double representation selection}
\label{app:sfu}

We prove here Lemma~\ref{lem:union}, Proposition~\ref{prop:sfu} and Corollary~\ref{cor:adaptive} for Approach~3 (Definition~\ref{def:sfu}). The only new element of the proofs is uniformity over realizations of the tuning sample.
The tuning sample $\mathcal S_0$ is unrelated to the scale $S_0$ of Appendix~\ref{app:sens}.

Throughout, $\hat\phi=(\hat\phi_g,\hat\phi_\alpha)$ denotes the pair tuned on $\mathcal S_0$, not the fold-wise representations of Approach~1.
The estimand $\theta_{\hat\phi}:=\theta_\phi|_{\phi=\hat\phi}$ is the representation-based target of the frozen union; it depends on $\mathcal S_0$ and is fixed once $\mathcal S_0$ is.
We assume the compatibility condition of Section~\ref{sec:setup} for $W_{\hat\phi}$, which holds in the leading examples.
If the two embeddings come from different models or layers, the choice of model is part of the index; if $\phi_g=\phi_\alpha$, the concatenation repeats coordinates, which is harmless for conditional expectations, and in practice the heads should use each coordinate once.

\begin{proof}[Proof of Lemma~\ref{lem:union}]
Since $\sigma(W_{\phi_g})\subseteq\sigma(W_\phi)$, and $g_\phi$ and $g_{\phi_g}$ are the $L^2(P)$ projections of $g^\star$ onto the square-integrable functions of $W_\phi$ and of $W_{\phi_g}$, the first projection is the closer one.
The proof of Proposition~\ref{prop:ovb} shows that $\alpha_\phi=\E[\alpha^\star\mid W_\phi]$ and $\alpha_{\phi_\alpha}=\E[\alpha^\star\mid W_{\phi_\alpha}]$, and $\sigma(W_{\phi_\alpha})\subseteq\sigma(W_\phi)$ gives the second inequality.
Proposition~\ref{prop:ovb}, applied to $W_\phi$, and Cauchy--Schwarz give the bound on $|\theta^\star-\theta_\phi|$.

For the sensitivity parameters, fix any representation $\varphi$ and write $a:=\norm{g^\star-g_\varphi}_{P,2}^2$ and $V_{Y,\star}:=\E[(Y-g^\star)^2]$.
Because $Y-g^\star$ is orthogonal to every function of $W$, $\E[(Y-g_\varphi)^2]=V_{Y,\star}+a$, and so $C_Y^2(\varphi)=a/(V_{Y,\star}+a)$, which increases with $a$.
Likewise, with $b:=\norm{\alpha^\star-\alpha_\varphi}_{P,2}^2$, the projection identity gives $\E[\alpha_\varphi^2]=\E[\alpha^{\star2}]-b$, and so $C_D^2(\varphi)=b/(\E[\alpha^{\star2}]-b)$, which also increases with $b$.
The two inequalities for the errors complete the proof.
\end{proof}

Write $\Pp_s$ for probability conditional on $\mathcal S_0=s$.
Given $\mathcal S_0=s$, the pair $\hat\phi=\hat\phi(s)$ is fixed, and $\mathcal S_1$ is an i.i.d.\ sample of size $n_1$ from $P$, independent of $\mathcal S_0$.
The theory of Sections~\ref{sec:fixed} and~\ref{sec:bounds} for a fixed, possibly drifting, representation therefore applies to $\mathcal S_1$.

\begin{assumption}[Regularity on most tuning samples]
\label{ass:sfu}
There are events $A_n$ with $\Pp(\mathcal S_0\in A_n)\to1$ such that, for every sequence of realizations $s_n\in A_n$, the frozen representation $W_{\hat\phi(s_n)}$ and the cross-fitted learners on $\mathcal S_1$ satisfy Assumption~\ref{ass:dml} with $n_1$ in place of $n$.
For the sensitivity statements, they also satisfy Assumption~\ref{ass:pipeline-bounds-dml} with $\Phi=\{\hat\phi(s_n)\}$, the endpoint variances are bounded away from zero, $\alpha^\star\in L^q(P)$ for the exponent $q>4$ of that assumption, and $\sup_n\norm{m(O,\alpha_{\hat\phi(s_n)})}_{P,2+\delta}<\infty$ for some $\delta>0$.
\end{assumption}

The assumption concerns only the heads fitted on $\mathcal S_1$ to a frozen covariate vector, the textbook DML setting.
It places no condition on how $\hat\phi_g$ and $\hat\phi_\alpha$ were obtained.

\begin{proof}[Proof of Proposition~\ref{prop:sfu}]
Let $\mathrm{CI}:=[\hat\theta\pm z_{1-\gamma/2}\hat\sigma/\sqrt{n_1}]$.
We first show
\begin{equation}
\sup_{s\in A_n}\Bigl|\Pp_s\bigl\{\theta_{\hat\phi}\in\mathrm{CI}\bigr\}-(1-\gamma)\Bigr|\to0.
\end{equation}
Suppose it fails.
Then there are $\epsilon>0$, a subsequence $n'$, and realizations $s_{n'}\in A_{n'}$ with $|\Pp_{s_{n'}}\{\theta_{\hat\phi}\in\mathrm{CI}\}-(1-\gamma)|\ge\epsilon$.
Along this subsequence, $W_{\hat\phi(s_{n'})}$ is a nonrandom drifting representation, and $\mathcal S_1$ is an i.i.d.\ sample from $P$ of size $n_1'$.
Under Assumption~\ref{ass:sfu}, Assumption~\ref{ass:dml} holds for this sequence, so Theorem~\ref{thm:dml-expansion}, together with consistency of $\hat\sigma^2$, gives $\Pp_{s_{n'}}\{\theta_{\hat\phi}\in\mathrm{CI}\}\to1-\gamma$, a contradiction.
For the unconditional statement,
\[
\Bigl|\Pp\{\theta_{\hat\phi}\in\mathrm{CI}\}-(1-\gamma)\Bigr|
\le\sup_{s\in A_n}\Bigl|\Pp_s\{\theta_{\hat\phi}\in\mathrm{CI}\}-(1-\gamma)\Bigr|+\Pp(\mathcal S_0\notin A_n)\to0.
\]
For the one-sided bounds, conditionally on $\mathcal S_0=s_n\in A_n$ the representation is fixed, and Proposition~\ref{prop:onesided}, applied on $\mathcal S_1$, gives $\Pp_{s_n}\{\theta_{\hat\phi}^-<L\}\to\gamma$ and $\Pp_{s_n}\{\theta_{\hat\phi}^+>U\}\to\gamma$.
The proof of that proposition uses Theorem~\ref{thm:bounds-pipeline}, whose central limit theorem and variance consistency for the endpoints are proved for a fixed representation; along $s_n\in A_n$ we need their triangular-array versions.
By conditional Jensen, $\norm{g_\phi}_{P,q}\le\norm{Y}_{P,q}$ and $\norm{\alpha_\phi}_{P,q}\le\norm{\alpha^\star}_{P,q}$ for every $\phi$, so $(Y-g_\phi)^2$, $\alpha_\phi^2$, and $\alpha_\phi(Y-g_\phi)$ are bounded in $L^{q/2}$, with $q/2>2$, uniformly along $s_n$; the terms $m(O,g_\phi)$ and $m(O,\alpha_\phi)$ are bounded in $L^{2+\delta}$ by Assumption~\ref{ass:dml}(v) and Assumption~\ref{ass:sfu}.
Hence the endpoint influence functions are bounded in $L^{2+\delta'}$ for some $\delta'>0$, uniformly along $s_n$, and, with their variances bounded away from zero, the Lyapunov ratio tends to zero.
Lyapunov's central limit theorem for triangular arrays, as in the proof of Theorem~\ref{thm:dml-expansion}, gives the endpoint central limit theorems along $s_n$, and the same moment bound gives the triangular-array law of large numbers for the variance estimators.
The other steps of the proof of Theorem~\ref{thm:bounds-pipeline} use only the rates of Assumption~\ref{ass:pipeline-bounds-dml}.
By the subsequence argument above, these limits are uniform over $A_n$, and since $\Pp(\mathcal S_0\notin A_n)\to0$ they also hold unconditionally.
Under a scenario valid for the union, $\theta^\star\in[\theta_{\hat\phi}^-,\theta_{\hat\phi}^+]$ by~\eqref{eq:bounds}, applied conditionally on $\mathcal S_0$, and~\eqref{eq:onesided-target} follows as in the proof of Proposition~\ref{prop:onesided}.
\end{proof}

\begin{corollary}[{Coverage of the full-information target}]
\label{cor:adaptive}
If
\[
\sqrt{n_1}\,\norm{g^\star-g_{\hat\phi_g}}_{P,2}\,\norm{\alpha^\star-\alpha_{\hat\phi_\alpha}}_{P,2}=o_\Pp(1),
\]
then $\mathrm{CI}$ covers $\theta^\star$ with asymptotic probability $1-\gamma$.
\end{corollary}

\begin{proof}
By Lemma~\ref{lem:union}, $\sqrt{n_1}|\theta^\star-\theta_{\hat\phi}|=o_\Pp(1)$; combine with Proposition~\ref{prop:sfu}.
\end{proof}

Corollary~\ref{cor:adaptive} is target adaptivity for the union, and its condition involves each embedding only through its own nuisance.

\begin{remark}[{Cost of the held-out sample}]
The coverage statement needs no rate, stability, or identifiability condition on the tuning.
In exchange, the target $\theta_{\hat\phi}$ depends on $\mathcal S_0$, and the interval uses only the $n_1$ observations in $\mathcal S_1$.
Sample splitting for assumption-lean inference has the same costs~\cite{cox1975note,rinaldo2019bootstrapping}.
\end{remark}

\begin{remark}[Dimension]
The union is larger than either part, so the heads fitted on $\mathcal S_1$ face a harder learning problem.
Whether the gain in approximation outweighs this cost depends on the application.
Rate conditions stated for the heads are unaffected, though a sieve condition that grows with the dimension changes its constants.
A menu of heads on the frozen union can still be aggregated by star or convex aggregation, inside the folds or across them (Appendix~\ref{app:cross-fold}).
\end{remark}

\clearpage
\section{Sensitivity grid for the illustration}
\label{app:sens-grid}

Table~\ref{tab:empirical-sensitivity} gives the nine scenarios of Section~\ref{sec:empirical-sensitivity}.
The endpoints depend on the scenario only through $c=\rho C_YC_D$, which equals $\rho r^2/(1-r^2)^{1/2}$ when both shares equal $r^2$.
The columns $\hat\theta^-$ and $\hat\theta^+$ are the estimated endpoints of~\eqref{eq:bounds} with $\hat S_\phi=3.51$; the outer columns are the one-sided bounds $L$ and $U$ of~\eqref{eq:one-sided} with $\gamma=0.10$, with the standard error of each endpoint evaluated at the value of $c$ of its row.

\begin{table}[htbp]
\centering
\caption[Sensitivity of the Approach~1 full-model estimate]{Sensitivity of the Approach~1 full-model estimate}
\label{tab:empirical-sensitivity}
\footnotesize
\setlength{\tabcolsep}{6pt}
\begin{tabular}{@{}clccccc@{}}
\toprule
$\rho$ & Scenario & $L$ & $\hat\theta^-$ & $\hat\theta_{\mathrm{full}}$ & $\hat\theta^+$ & $U$ \\
\midrule
1.0 & Mild: $r^2=0.010$   & $-1.187$ & $-1.164$ & $-1.129$ & $-1.093$ & $-1.070$ \\
1.0 & Medium: $r^2=0.025$ & $-1.241$ & $-1.217$ & $-1.129$ & $-1.040$ & $-1.016$ \\
1.0 & Strong: $r^2=0.050$ & $-1.332$ & $-1.309$ & $-1.129$ & $-0.949$ & $-0.924$ \\
\addlinespace
0.5 & Mild: $r^2=0.010$   & $-1.170$ & $-1.146$ & $-1.129$ & $-1.111$ & $-1.087$ \\
0.5 & Medium: $r^2=0.025$ & $-1.197$ & $-1.173$ & $-1.129$ & $-1.084$ & $-1.061$ \\
0.5 & Strong: $r^2=0.050$ & $-1.242$ & $-1.219$ & $-1.129$ & $-1.039$ & $-1.015$ \\
\addlinespace
0.2 & Mild: $r^2=0.010$   & $-1.159$ & $-1.136$ & $-1.129$ & $-1.122$ & $-1.098$ \\
0.2 & Medium: $r^2=0.025$ & $-1.170$ & $-1.146$ & $-1.129$ & $-1.111$ & $-1.087$ \\
0.2 & Strong: $r^2=0.050$ & $-1.188$ & $-1.165$ & $-1.129$ & $-1.093$ & $-1.069$ \\
\bottomrule
\end{tabular}
\begin{minipage}{0.94\textwidth}
\footnotesize
\emph{Note:} In each row, $C_Y^2=1-R^2_{\alpha^\star\sim\alpha_\phi}=r^2$. The region $[\hat\theta^-,\hat\theta^+]$ reflects the approximation bias; the bounds $L$ and $U$ add sampling uncertainty, each failing with probability at most $10\%$ under a valid scenario. In all nine rows $U<-0.605$, the elasticity threshold of Section~\ref{sec:empirical-specification}.
\end{minipage}
\end{table}

\clearpage
% Reviewed average-derivative appendix; also supplied as a separate source fragment.
\section{Average derivatives with varying-coefficient heads}
\label{app:avg-deriv}

The average derivative $m(O,h)=\partial_dh(D,X)$ differs from the average treatment effect in one respect that matters for the conditions of this paper.
Several results assume mean-square continuity, $\norm{m(O,h)}_{P,2}\le C_m\norm{h}_{P,2}$ for every $h$ in $L^2(P_W)$ or $L^2(P_{W_\phi})$.
For a derivative, mean-square continuity requires the differentiation operator to be bounded on the chosen $L^2$ space.
This bound is sometimes called a \emph{reverse Poincar\'e inequality}: it bounds the mean square of the derivative by the mean square of the function.
It fails on spaces containing arbitrarily high frequencies (Lemma~\ref{lem:ad-no-rp}).
When the treatment has an atom, even the weaker mean continuity fails on $L^2$, and the long representer $\alpha^\star$ of Section~\ref{sec:identity} does not exist.
The remedy, as in~\cite{chernozhukov2024long}, is to define the regression, the representer, and the target on a closed linear subspace of $L^2$ on which the functional is continuous (Lemma~\ref{lem:ad-closure}), and to apply the results of the paper there (Proposition~\ref{prop:ad-where}).

A useful example is the varying-coefficient class $h(d,x)=d\,a(x)+c(x)$.
On it the mean-square continuity bound~\eqref{eq:ad-rp} holds as stated, with the sharp constant $C_m=(\operatorname*{ess\,inf}\Var(D\mid X))^{-1/2}$, for any law of $D$ with $\E D^2<\infty$ and $\operatorname*{ess\,inf}\Var(D\mid X)>0$; in particular $D$ may have atoms.
The price is the form of the target: the derivative of a head is the coefficient $a(x)$, and the target is the average within-$X$ least-squares slope (Lemma~\ref{lem:ad-vc} and Corollary~\ref{cor:ad-vc}).
For heads nonlinear in the treatment, continuity requires separate conditions on the function class and the treatment law; those extensions are beyond the scope of this appendix.

\subsection{Setting}
\label{app:ad-setting}

Let $X$ denote $Z$ or $X_\phi=x(Z,\phi)$, and write $V=(D,X)$, so that $V=W$ or $V=W_\phi$.
Here $m(O,h)=\partial_dh(D,X)$, with weight $w\equiv1$; for a weighted derivative $w(X)\,\partial_dh(D,X)$ with $|w|\le B_w$, the upper bounds on $\norm{m(O,h)}_{P,2}$ below are multiplied by $B_w$.
Begin with a linear class $\Gamma$ of specified functions $h(d,x)$ that are differentiable in $d$, with $h(D,X)$ and $\partial_dh(D,X)$ in $L^2(P)$. We identify functions equal $P_V$-almost surely only after checking that the functional respects this identification. Mean continuity below makes the mean functional well defined on these equivalence classes; the stronger bound does the same for the derivative operator.
Consider
\begin{align}
\bigl|\E[m(O,h)]\bigr|&\le C\,\norm{h}_{P,2}
&&\text{for all }h\in\Gamma;
\label{eq:ad-mean}
\\
\norm{m(O,h)}_{P,2}&\le C_m\,\norm{h}_{P,2}
&&\text{for all }h\in\Gamma.
\label{eq:ad-rp}
\end{align}
We call~\eqref{eq:ad-mean} \emph{mean continuity} and~\eqref{eq:ad-rp} \emph{mean-square continuity} on $\Gamma$, and write $C_m(\Gamma)\in[0,\infty]$ for the smallest constant in~\eqref{eq:ad-rp}.
Mean continuity is what the Riesz representation~\eqref{eq:rr-long}--\eqref{eq:rr-short} and the identity of Proposition~\ref{prop:ovb} need; the mean-square continuity bound is what the remainder and variance bounds of the inference results need.
By Jensen's inequality, \eqref{eq:ad-rp} implies~\eqref{eq:ad-mean} with $C=C_m(\Gamma)$.

\subsection{Failure of mean-square continuity on unrestricted spaces}
\label{app:ad-norp}

\begin{lemma}[Failure of mean-square continuity on unrestricted spaces]
\label{lem:ad-no-rp}
\begin{enumerate}[label=(\roman*),leftmargin=*]
\item For every law of $D$ and every $\omega>0$, one of the two functions $h(d)=\sin(\omega d)$ and $h(d)=\cos(\omega d)$ satisfies $\norm{h'(D)}_{P,2}\ge\omega\,\norm{h(D)}_{P,2}$ and $h'(D)\ne0$ in $L^2(P)$.
Hence~\eqref{eq:ad-rp} fails on every $\Gamma$ that contains these functions for arbitrarily large $\omega$; in particular it fails on the differentiable elements of $L^2(P_V)$.
\item Suppose $\Pp\{D=d_0,\ X\in B\}>0$ for some point $d_0$ and some set $B$.
Then~\eqref{eq:ad-mean} fails on every $\Gamma$ that contains the functions $h_\varepsilon(d,x):=\varepsilon\,\psi\bigl((d-d_0)/\varepsilon\bigr)\1_B(x)$, $\varepsilon>0$, where $\psi$ is a fixed continuously differentiable function supported in $[-1,1]$ with $\psi'(0)=1$.
\end{enumerate}
\end{lemma}

\begin{proof}
(i) Let $s:=\E\sin^2(\omega D)$, so that $\E\cos^2(\omega D)=1-s$, $\E\{\omega\cos(\omega D)\}^2=\omega^2(1-s)$, and $\E\{\omega\sin(\omega D)\}^2=\omega^2s$.
If $0<s<1$, the two ratios of squared norms are $\omega^2(1-s)/s$ and $\omega^2s/(1-s)$; their product is $\omega^4$, so the larger is at least $\omega^2$.
If $s=0$, then $\sin(\omega D)=0$ almost surely while its derivative has norm $\omega$; if $s=1$, the same holds for $\cos(\omega D)$.
In these two boundary cases, no finite constant works even at that frequency, and $m$ is not a function of the $L^2$ class of $h$. In the other cases, arbitrarily large frequencies rule out a common finite constant.

(ii) Since $|\psi|\le\norm{\psi}_\infty$, $\norm{h_\varepsilon}_{P,2}\le\varepsilon\norm{\psi}_\infty\to0$.
On the other hand, $\partial_dh_\varepsilon(D,X)=\psi'\bigl((D-d_0)/\varepsilon\bigr)\1_B(X)$, which equals $\psi'(0)=1$ on $\{D=d_0,X\in B\}$ and vanishes unless $|D-d_0|\le\varepsilon$, so
\[
\Bigl|\E[\partial_dh_\varepsilon(D,X)]-\Pp\{D=d_0,\ X\in B\}\Bigr|
\le\norm{\psi'}_\infty\,\Pp\{0<|D-d_0|\le\varepsilon\}\to0
\]
by countable additivity.
Thus $\E[m(O,h_\varepsilon)]$ stays away from zero while $\norm{h_\varepsilon}_{P,2}\to0$.
\end{proof}

Part~(ii) says that when the treatment has an atom, the average derivative of an unrestricted regression has no Riesz representer and is not a continuous functional of the regression.
It also rules out the integration-by-parts representer $-\partial_d\log f(D\mid X)$, which needs a conditional density with finite Fisher information~\cite[Remark~1]{chernozhukov2024long}.
A restricted subspace, as in Lemma~\ref{lem:ad-vc}, resolves this problem without discarding the atom. A weighted derivative that excludes the atom is a different target and still needs its own continuity and boundary conditions.

\subsection{Restricted representers and the identity}
\label{app:ad-closure}

\begin{lemma}[Closure, restricted representers, and the identity]
\label{lem:ad-closure}
Let $\Gamma\subseteq L^2(P_V)$ be a linear subspace on which~\eqref{eq:ad-mean} holds, and let $\bar\Gamma$ be its closure.
\begin{enumerate}[label=(\roman*),leftmargin=*]
\item The functional $h\mapsto\E[m(O,h)]$ extends uniquely to a continuous linear functional on $\bar\Gamma$, and there is a unique $\alpha_\Gamma\in\bar\Gamma$ with $\E[m(O,h)]=\E[h\,\alpha_\Gamma]$ for all $h\in\bar\Gamma$.
If~\eqref{eq:ad-rp} holds on $\Gamma$, then $h\mapsto m(O,h)$ extends uniquely to a bounded linear map from $\bar\Gamma$ into $L^2(P)$, and~\eqref{eq:ad-rp} holds on $\bar\Gamma$ with the same constant.
\item If $\Gamma_1\subseteq\Gamma_2$ both satisfy~\eqref{eq:ad-mean}, then $\alpha_{\Gamma_1}$ is the $L^2(P)$ projection of $\alpha_{\Gamma_2}$ onto $\bar\Gamma_1$.
\item Let $Y\in L^2(P)$, and for $\Gamma_1\subseteq\Gamma_2$ as in~(ii) let $g_{\Gamma_i}$ be the projection of $Y$ onto $\bar\Gamma_i$ and $\theta_{\Gamma_i}:=\E[m(O,g_{\Gamma_i})]$.
Then, with $\E[m(O,\cdot)]$ the extended functional of~(i),
\[
\theta_{\Gamma_2}-\theta_{\Gamma_1}=\E\bigl[(g_{\Gamma_2}-g_{\Gamma_1})(\alpha_{\Gamma_2}-\alpha_{\Gamma_1})\bigr].
\]
\end{enumerate}
\end{lemma}

\begin{proof}
(i) A bounded linear map on a dense subspace of $\bar\Gamma$ has a unique bounded extension with the same norm; the representer then comes from the Riesz representation theorem on the Hilbert space $\bar\Gamma$.
(ii) For $h\in\bar\Gamma_1$, $\E[h\,\alpha_{\Gamma_2}]=\E[m(O,h)]=\E[h\,\alpha_{\Gamma_1}]$, so $\alpha_{\Gamma_2}-\alpha_{\Gamma_1}$ is orthogonal to $\bar\Gamma_1$, and $\alpha_{\Gamma_1}\in\bar\Gamma_1$.
(iii) Since $g_{\Gamma_2}-g_{\Gamma_1}\in\bar\Gamma_2$, $\theta_{\Gamma_2}-\theta_{\Gamma_1}=\E[(g_{\Gamma_2}-g_{\Gamma_1})\alpha_{\Gamma_2}]$.
The projections are nested, so $g_{\Gamma_1}$ is also the projection of $g_{\Gamma_2}$ onto $\bar\Gamma_1$; as $\alpha_{\Gamma_1}\in\bar\Gamma_1$, $\E[(g_{\Gamma_2}-g_{\Gamma_1})\alpha_{\Gamma_1}]=0$.
Subtracting gives the claim.
\end{proof}

Part~(iii) is the identity~\eqref{eq:ovb-identity} with restricted representers, as in~\cite[Lemma~1 and Theorem~2]{chernozhukov2024long}. The Cauchy--Schwarz sensitivity bound and its projection-based scale decomposition therefore remain valid. Inference for their estimates also requires the moment, rate, and continuity conditions discussed below.
For a representation $\phi$, the space $\cH_\phi$ of Section~\ref{sec:subordinated-geometry} plays the role of $\bar\Gamma_1$.

\subsection{Varying-coefficient heads}
\label{app:ad-vc}

Consider heads that are linear in the treatment, with $m(O,h)=\partial_dh(D,X)$ as in Section~\ref{app:ad-setting}.
Assume $\E D^2<\infty$, and write $\mu(x):=\E[D\mid X=x]$ and $v(x):=\Var(D\mid X=x)$.
The varying-coefficient class is
\[
\Gamma_{\mathrm{VC}}(X):=\bigl\{h(d,x)=d\,a(x)+c(x):\ a\in L^2(P_X),\ h\in L^2(P_V)\bigr\},
\]
For a specified head, $m(O,h)=a(X)$ and $\E[m(O,h)]=\E[a(X)]$. If $v=0$ on a set of positive probability, different slopes can give the same $L^2$ prediction; statements about failure of continuity or unbounded loss below then concern these specified heads, before passing to equivalence classes.
Write $v_{\min}:=\operatorname*{ess\,inf}v(X)$, and read $1/v$ as $+\infty$ where $v=0$.

\begin{lemma}[Mean-square continuity for varying-coefficient heads]
\label{lem:ad-vc}
\begin{enumerate}[label=(\roman*),leftmargin=*]
\item For every $h=da+c$ in $\Gamma_{\mathrm{VC}}(X)$,
\[
\norm{h}_{P,2}^2=\E\bigl[v(X)\,a(X)^2\bigr]+\E\bigl[\{c(X)+\mu(X)a(X)\}^2\bigr].
\]
Consequently~\eqref{eq:ad-rp} holds on $\Gamma_{\mathrm{VC}}(X)$ if and only if $v_{\min}>0$, and then $C_m(\Gamma_{\mathrm{VC}}(X))^2=1/v_{\min}$.
\item Mean continuity~\eqref{eq:ad-mean} holds on $\Gamma_{\mathrm{VC}}(X)$ if and only if $\E[1/v(X)]<\infty$.
In that case the representer is
\[
\alpha_{\mathrm{VC}}(D,X)=\frac{D-\mu(X)}{v(X)}\in\bar\Gamma_{\mathrm{VC}}(X),
\qquad
\E\bigl[\alpha_{\mathrm{VC}}^2\bigr]=\E\bigl[1/v(X)\bigr],
\]
with $\alpha_{\mathrm{VC}}\in\Gamma_{\mathrm{VC}}(X)$ when $v_{\min}>0$, and the Riesz loss $Q(h)=\E[h^2]-2\E[m(O,h)]$ attains its minimum $-\E[1/v(X)]$ on $\bar\Gamma_{\mathrm{VC}}(X)$ at $\alpha_{\mathrm{VC}}$.
If $\E[1/v(X)]=\infty$, then $Q$ is unbounded below on $\Gamma_{\mathrm{VC}}(X)$.
\item If $v_{\min}>0$, then $\Gamma_{\mathrm{VC}}(X)$ is closed.
If moreover $Y\in L^2(P)$, the projection of $Y$ onto $\Gamma_{\mathrm{VC}}(X)$ has slope $a_{\mathrm{VC}}(X)=\mathrm{Cov}(D,Y\mid X)/v(X)$, and the restricted target $\theta_{\mathrm{VC}}:=\E[m(O,g_{\mathrm{VC}})]=\E[a_{\mathrm{VC}}(X)]$, where $g_{\mathrm{VC}}$ is the projection, is the average of the within-$X$ least-squares slopes,
\[
\theta_{\mathrm{VC}}=\E\Bigl[\frac{\mathrm{Cov}(D,Y\mid X)}{\Var(D\mid X)}\Bigr]=\E\bigl[Y\alpha_{\mathrm{VC}}(D,X)\bigr].
\]
\item If $X_\phi=x(Z,\phi)$, then $\Gamma_{\mathrm{VC}}(X_\phi)\subseteq\Gamma_{\mathrm{VC}}(Z)$, and $\operatorname*{ess\,inf}\Var(D\mid Z)\le\operatorname*{ess\,inf}\Var(D\mid X_\phi)$, so $C_m(\Gamma_{\mathrm{VC}}(X_\phi))\le C_m(\Gamma_{\mathrm{VC}}(Z))$.
\end{enumerate}
\end{lemma}

\begin{proof}
(i) Write $h=a(X)\{D-\mu(X)\}+\{c(X)+\mu(X)a(X)\}$.
Conditionally on $X$ the first term has mean zero and variance $v\,a^2$, and the second is constant, so the cross term vanishes and the identity follows.
It gives $\norm{h}_{P,2}^2\ge v_{\min}\E[a^2]=v_{\min}\norm{m(O,h)}_{P,2}^2$.
For the converse and for sharpness, fix $\varepsilon>0$ with $A_\varepsilon:=\{v(X)\le v_{\min}+\varepsilon\}$ of positive probability, and take $a=\1_{A_\varepsilon}$ and $c=-\mu\,\1_{A_\varepsilon}$; then $h=(D-\mu)\1_{A_\varepsilon}\in L^2(P_V)$, and, when the denominator is positive, $\norm{m(O,h)}_{P,2}^2/\norm{h}_{P,2}^2=\Pp(A_\varepsilon)/\E[v\1_{A_\varepsilon}]\ge1/(v_{\min}+\varepsilon)$. If the denominator is zero, $h=0$ in $L^2$ while its derivative has positive norm, so no finite bound is possible.
Letting $\varepsilon\downarrow0$ shows that no constant below $v_{\min}^{-1/2}$ works, and that none works when $v_{\min}=0$.

(ii) Let $a_n:=\min(1/v,n)$, so that $(D-\mu)a_n\in\Gamma_{\mathrm{VC}}(X)$.
Suppose $\E[1/v]<\infty$; then $v>0$ almost surely.
Conditionally on $X$, $\E[(D-\mu)^2(1/v-a_n)^2\mid X]=v\,(1/v-a_n)^2\le1/v$, and the left side tends to zero, so by dominated convergence $(D-\mu)a_n\to\alpha_{\mathrm{VC}}$ in $L^2(P_V)$; hence $\alpha_{\mathrm{VC}}\in\bar\Gamma_{\mathrm{VC}}(X)$ and $\E[\alpha_{\mathrm{VC}}^2]=\E[v/v^2]=\E[1/v]$.
If $v_{\min}>0$, then $1/v\le1/v_{\min}$ is in $L^2(P_X)$, and $\alpha_{\mathrm{VC}}=D\cdot(1/v)-\mu/v\in\Gamma_{\mathrm{VC}}(X)$.
For $h=da+c\in\Gamma_{\mathrm{VC}}(X)$ write $h=a(D-\mu)+(c+\mu a)$ as in~(i).
Then $\E\bigl[|a|(D-\mu)^2/v\bigr]=\E|a|<\infty$ and $\E\bigl[|c+\mu a|\,|D-\mu|/v\bigr]\le\norm{c+\mu a}_{P,2}\,(\E[1/v])^{1/2}<\infty$, so conditioning on $X$ gives
$\E[h\,\alpha_{\mathrm{VC}}]=\E\bigl[a\,\E\{(D-\mu)^2\mid X\}/v\bigr]+\E\bigl[(c+\mu a)\,\E\{D-\mu\mid X\}/v\bigr]=\E[a]$,
and $|\E[a]|\le\norm{h}_{P,2}(\E[1/v])^{1/2}$ by Cauchy--Schwarz; this is~\eqref{eq:ad-mean}.
Hence $Q(h)=\norm{h-\alpha_{\mathrm{VC}}}_{P,2}^2-\E[1/v]$ on $\bar\Gamma_{\mathrm{VC}}(X)$, with minimum $-\E[1/v]$ at $\alpha_{\mathrm{VC}}$.
Conversely, suppose $\E[1/v]=\infty$ and let $h_n:=(D-\mu)a_n$, with $a_n$ as above.
Since $v\,a_n^2\le a_n$, $\norm{h_n}_{P,2}^2=\E[v\,a_n^2]\le\E[a_n]$, while $\E[m(O,h_n)]=\E[a_n]\uparrow\infty$ by monotone convergence.
So $|\E[m(O,h_n)]|/\norm{h_n}_{P,2}\ge(\E[a_n])^{1/2}\to\infty$, which violates~\eqref{eq:ad-mean} (if $\norm{h_n}_{P,2}=0$ for some $n$, the functional is not even defined on $L^2$ classes). Moreover, $Q(h_n)=\E[v a_n^2]-2\E[a_n]\le-\E[a_n]\to-\infty$, including the case in which a denominator above vanishes.

(iii) Put $U:=(D-\mu(X))/\sqrt{v(X)}$, so $\E[U\mid X]=0$ and $\E[U^2\mid X]=1$.
The map $(t,b)\mapsto t(X)U+b(X)$ is an isometry from $L^2(P_X)\oplus L^2(P_X)$ into $L^2(P_V)$, since its squared norm is $\E[t^2+b^2]$; its range is therefore closed.
This range is exactly $\Gamma_{\mathrm{VC}}(X)$: for a head $da+c$, part~(i) gives $t=\sqrt v\,a\in L^2(P_X)$ and $b=c+\mu a\in L^2(P_X)$; conversely, $a=t/\sqrt v\in L^2(P_X)$ because $v\ge v_{\min}>0$, and $c=b-\mu a$ gives the required head.
No separate square-integrability condition on $c$ is needed.
For the projection, $\E[(Y-Da-c)^2]=\E\bigl[\E\{(Y-Da-c)^2\mid X\}\bigr]$, and for each value of $X$ the inner expectation is minimized by the least-squares coefficients $a_{\mathrm{VC}}=\mathrm{Cov}(D,Y\mid X)/v$ and $c=\E[Y\mid X]-\mu\,a_{\mathrm{VC}}$.
By the conditional Cauchy--Schwarz inequality $a_{\mathrm{VC}}^2\le\Var(Y\mid X)/v\le\Var(Y\mid X)/v_{\min}$, so $a_{\mathrm{VC}}\in L^2(P_X)$.
The pointwise minimizer $h^\ast:=a_{\mathrm{VC}}D+c$ also lies in $L^2(P_V)$: by~(i), with $c+\mu a_{\mathrm{VC}}=\E[Y\mid X]$, $\norm{h^\ast}_{P,2}^2=\E[\mathrm{Cov}(D,Y\mid X)^2/v]+\E[\E(Y\mid X)^2]\le\E[\Var(Y\mid X)]+\E[\E(Y\mid X)^2]=\E Y^2$.
So the pointwise minimizer is an element of $\Gamma_{\mathrm{VC}}(X)$; it is therefore the projection.
Then $\theta_{\mathrm{VC}}=\E[a_{\mathrm{VC}}]$, and $\E[Y\alpha_{\mathrm{VC}}]=\E[\mathrm{Cov}(D,Y\mid X)/v]$ by conditioning on $X$.

(iv) Functions of $X_\phi$ are functions of $Z$, which gives the inclusion.
By the law of total variance, $\Var(D\mid X_\phi)\ge\E[\Var(D\mid Z)\mid X_\phi]\ge\operatorname*{ess\,inf}\Var(D\mid Z)$.
\end{proof}

Part~(ii) explains why the population Riesz loss can become unstable as features become more informative about treatment: its minimum is $-\E[1/v]$, which diverges downward along any sequence for which $\E[1/v]\to\infty$.
Small variance on a shrinking set does not by itself imply this divergence.
\begin{corollary}[Cross-fitted inference with varying-coefficient heads]
\label{cor:ad-vc}
Suppose $\E D^2<\infty$ and $\operatorname*{ess\,inf}\Var(D\mid Z)>0$, and let $\mathcal H$ be a fixed closed linear subspace of $\Gamma_{\mathrm{VC}}(Z)$.
Examples are $\Gamma_{\mathrm{VC}}(X_\phi)$ for a fixed representation $\phi$, and the closed span of the heads $d\,a(x_j(Z))+c(x_j(Z))$, $j\in J(\phi)$, of a fixed menu $\phi$, including all heads reachable by its tuning rules. The common space $\mathcal H$ is fixed before the inference sample is drawn, even though the fitted heads are not.
Then~\eqref{eq:ad-rp} holds on $\mathcal H$ with $C_m\le(\operatorname*{ess\,inf}\Var(D\mid Z))^{-1/2}$, and with $C_m=(\operatorname*{ess\,inf}\Var(D\mid X_\phi))^{-1/2}$ if $\mathcal H=\Gamma_{\mathrm{VC}}(X_\phi)$.
Let $g_0$ be the projection of $Y$ onto $\mathcal H$, let $\alpha_0\in\mathcal H$ be the Riesz representer of $h\mapsto\E[m(O,h)]$ on $\mathcal H$, which is the minimizer of the Riesz loss over $\mathcal H$, and let $\theta_0:=\E[m(O,g_0)]=\E[Y\alpha_0]$.
Suppose the remaining conditions of Assumption~\ref{ass:pipeline-dml} hold: $Y,g_0,\alpha_0\in L^q(P)$ with $q>4$; $\sigma_0^2$ is bounded away from zero and infinity; the fitted heads $\hat g^{(-k)},\hat\alpha^{(-k)}$ lie in $\mathcal H$, are built from $I_k^c$ only, and have $O_\Pp(1)$ $q$th moments; and they are consistent in $L^2(P)$ with product error $o_\Pp(n^{-1/2})$.
Then the conclusions of Theorem~\ref{thm:pipeline-dml} hold: the expansion for $\hat\theta_0-\theta_0$, the normal limit, and the consistency of $\hat\sigma_0^2$.
For $\mathcal H=\Gamma_{\mathrm{VC}}(X)$, $g_0=a_{\mathrm{VC}}(X)\,D+c(X)$, $\alpha_0=\alpha_{\mathrm{VC}}$, and
\[
\theta_0=\E\Bigl[\frac{\mathrm{Cov}(D,Y\mid X)}{\Var(D\mid X)}\Bigr]=\E[Y\alpha_{\mathrm{VC}}].
\]
\end{corollary}

\begin{proof}
A head $h=d\,a(x_j(Z))+c(x_j(Z))$ in $L^2(P_W)$ lies in $\Gamma_{\mathrm{VC}}(Z)$, because the identity $\E[h^2\mid Z]=v\,a^2+(c+\mu a)^2$ of the proof of Lemma~\ref{lem:ad-vc}(i) holds for any measurable $a,c$ and gives $\E[a^2]\le\norm{h}_{P,2}^2/v_{\min}<\infty$.
The space $\Gamma_{\mathrm{VC}}(Z)$ is closed by Lemma~\ref{lem:ad-vc}(iii) with $X=Z$, and so is $\Gamma_{\mathrm{VC}}(X_\phi)$, since $\operatorname*{ess\,inf}\Var(D\mid X_\phi)\ge\operatorname*{ess\,inf}\Var(D\mid Z)>0$.
The constants are those of Lemma~\ref{lem:ad-vc}(i) and~(iv), restricted to $\mathcal H$.
On $\mathcal H$ the functional is $L^2$-continuous, so $\alpha_0$ exists by Lemma~\ref{lem:ad-closure}(i), and it is the projection of $\alpha_{\mathrm{VC}}(\cdot,Z)$ onto $\mathcal H$ by Lemma~\ref{lem:ad-closure}(ii).
When $D$ has an atom there is no representer $\alpha^\star$ in $L^2(P_W)$ (Lemma~\ref{lem:ad-no-rp}(ii)), so $\alpha_0$ is defined here by its Riesz property in $\mathcal H$, not as a projection of $\alpha^\star$.
Proposition~\ref{prop:pipeline-mixed-bias}(i) is then the defining property of $\alpha_0$, and part~(ii) follows from it by Steps~3 and~4 of the proof of the proposition, which use only that $g_0$ is the projection of $Y$ onto $\mathcal H$ and do not use~\eqref{eq:rr-long}.
Part~(iii) of the proposition is replaced by Lemma~\ref{lem:ad-closure}(iii), with $\Gamma_1=\mathcal H$ and $\Gamma_2=\Gamma_{\mathrm{VC}}(Z)$.
Theorem~\ref{thm:pipeline-dml} uses part~(ii) and the bound of Assumption~\ref{ass:pipeline-dml} on $\mathcal H$, and no other property of $\alpha^\star$.
The formulas for $\mathcal H=\Gamma_{\mathrm{VC}}(X)$ are Lemma~\ref{lem:ad-vc}(ii) and~(iii), and $\theta_0=\E[g_0\alpha_0]=\E[Y\alpha_0]$ because $\alpha_0\in\mathcal H$ and $g_0$ is the projection of $Y$.
\end{proof}

The corollary allows $D$ to have an atom at which the functional is evaluated, and its constant $C_m$ is sharp for $\mathcal H=\Gamma_{\mathrm{VC}}(X_\phi)$.
For a fixed $\phi$, the variance condition can be weakened to $\operatorname*{ess\,inf}\Var(D\mid X_\phi)>0$ with $\mathcal H=\Gamma_{\mathrm{VC}}(X_\phi)$, which can hold even when finer conditioning makes some conditional variances zero. All other moment, rate, and sampling conditions remain in force.
The cost is the target: the average within-$X$ least-squares slope.
It equals the average derivative of the regression if the regression is linear in $d$.
More generally, suppose $g(d,x)=\E[Y\mid D=d,X=x]$ has a locally absolutely continuous version on an interval containing the conditional support of $D$, and $\int |\partial_dg(t,x)|K_x(t)\,\mathrm dt<\infty$, where
\[
K_x(t):=\E[(D-\mu(x))\1\{D>t\}\mid X=x].
\]
Then
\[
\frac{\mathrm{Cov}(D,Y\mid X=x)}{v(x)}
=\int \partial_dg(t,x)\frac{K_x(t)}{v(x)}\,\mathrm dt,
\qquad K_x(t)\ge0,\qquad \int K_x(t)\,\mathrm dt=v(x).
\]
To see this, take conditionally independent copies $D,D'$ and write the covariance as $\tfrac12\E[(D-D')\{g(D,x)-g(D',x)\}\mid X=x]$.
The fundamental theorem of calculus and Fubini's theorem give the display; applying the same argument to $g(d,x)=d$ gives the integral of $K_x$.
The target is thus a weighted average of derivatives under these additional conditions, not generally their average under the observed treatment law.
The least-squares target itself does not require a differentiable regression version.
Proposition~\ref{prop:ad-where} lists, for the other results of the paper, the set on which each proof applies the bound; for varying-coefficient heads the uniform constant comes from $\Gamma_{\mathrm{VC}}(Z)$.

\subsection{Where the paper uses mean-square continuity}
\label{app:ad-where}

The following proposition localizes the continuity bounds and records how the resulting projection identities extend the inference arguments. Items (i)--(vi) identify the functions to which the existing proofs apply this bound; item~(vii) gives the consequences for the remaining inference arguments.

\begin{proposition}[The restricted reading of mean-square continuity]
\label{prop:ad-where}
Let the short spaces $\Gamma_\phi$ be closed linear subspaces of a common closed long space $\Gamma_L$.
Use the projections and restricted representers of Lemma~\ref{lem:ad-closure} in these spaces to define the short and long targets.
The results below remain valid if a global continuity assumption is replaced by the same bound on the indicated spans, subject to the following provisions and all other hypotheses of each result.
First, both fitted heads lie in the short space defining their target, so the mixed-bias identity follows from projection orthogonality and the restricted Riesz property.
Second, every use of a long--short identity involves nested spaces and a common functional; the long objects are the restricted ones on $\Gamma_L$.
Third, moment bounds previously obtained by conditional Jensen are imposed directly on the restricted projections, uniformly over any drifting representations or tuning-sample realizations under consideration.
Projection onto a general linear space need not preserve $L^q$ norms for $q>2$.
All constants and rates retain the uniformity required by the original result.
For fold-specific spaces, the conditions hold with the training sample fixed; pipeline inference uses a common fixed space as in Corollary~\ref{cor:ad-vc}.
Efficiency claims (Theorem~\ref{thm:strong-adaptive}) additionally need the restricted score to be the efficient influence function of the restricted target.
For the within-$X$ slope $\beta(x)=\mathrm{Cov}(D,Y\mid X=x)/v(x)$, write $\nu(x)=\E[Y\mid X=x]$ and $g_{\mathrm{VC}}=\nu+\beta(D-\mu)$.
Differentiating the conditional covariance and variance along a regular submodel gives the conditional contribution
\[
\frac{(D-\mu)(Y-\nu)-\beta(D-\mu)^2}{v}
=\alpha_{\mathrm{VC}}(Y-g_{\mathrm{VC}}).
\]
Thus, when pathwise differentiation is justified and this score is square integrable, the efficient influence function of $\theta_{\mathrm{VC}}=\E[\beta(X)]$ in the nonparametric model is
\[
\beta(X)-\theta_{\mathrm{VC}}+\alpha_{\mathrm{VC}}(Y-g_{\mathrm{VC}}),
\]
which is the restricted score of this paper.
\begin{enumerate}[label=(\roman*),leftmargin=*]
\item Assumption~\ref{ass:dml}(i) and Theorem~\ref{thm:dml-expansion}: the outcome heads $\hat g^{(-k)}_\phi$ and the target $g_\phi$, through $\hat g^{(-k)}_\phi-g_\phi$ in the remainder bound and $g_\phi$ in the bound $|\theta_\phi|\le C_m\norm{Y}_{P,2}$. With a restricted $\Gamma$, $g_\phi$ is the projection of $Y$ onto $\bar\Gamma$ rather than $\E[Y\mid W_\phi]$, and the orthogonality $\E[\delta\alpha\,(Y-g_\phi)]=0$ in Step~2 of the proof holds for $\delta\alpha\in\bar\Gamma$ only, so the Riesz heads must lie in $\bar\Gamma$ as well. This is the setting of Theorem~\ref{thm:pipeline-dml}, and Corollary~\ref{cor:ad-vc} states it for varying-coefficient heads.
\item Lemma~\ref{lem:score-conv} and Theorem~\ref{thm:strong-adaptive}: the difference $g^\star-g_{\phi_n}$ for every $n$.
\item Theorem~\ref{thm:pipeline-dml} and the cross-fold results of Appendix~\ref{app:cross-fold}: the deviations $\hat g^{(-k)}-g_0$ and $\hat\alpha^{(-k)}-\alpha_0$, which lie in $\cH$; Assumption~\ref{ass:pooled-star}(ii) already states the bound on $\cH$, and Propositions~\ref{prop:riesz-theorem32} and~\ref{prop:star-finite-riesz} state it on the span of the representer class.
\item Lemma~\ref{lem:pooled-margin}: the Riesz heads $\hat\alpha_{jk}$, their limits $\alpha_{n,j}^\ast$, and their differences.
\item Corollary~\ref{cor:joint}(a): the Riesz heads $\hat\alpha_k$.
\item Lemma~\ref{lem:dml-pipeline-bound-components} and Theorem~\ref{thm:bounds-pipeline}: the deviations $\delta g^{(-k)}$ and $\delta\alpha^{(-k)}$, which lie in $\cH$, and $m(O,\alpha_0)$ in the influence function of $\hat V_\alpha$, with $\alpha_0\in\cH$.
\item Propositions~\ref{prop:onesided}, \ref{prop:joint}, and~\ref{prop:sfu}, and Theorem~\ref{thm:pooled-joint}: these results use the preceding expansions, mixed-bias identity, and sensitivity bounds. Their score-consistency, moment, margin, and conditional-uniformity assumptions remain necessary. For the frozen union, require that its short space contain both component short spaces; this inclusion holds for full varying-coefficient spaces on the respective feature vectors.
\end{enumerate}
For the average derivative with varying-coefficient heads, (i), (iii)--(vii) hold under these provisions with $\Gamma_\phi=\Gamma_{\mathrm{VC}}(X_\phi)$ (or the closure of the span of the heads of a menu $\phi$), and the targets are the projections onto $\Gamma_\phi$, if $\E D^2<\infty$ and $\operatorname*{ess\,inf}\Var(D\mid Z)>0$, so that $\Gamma_{\mathrm{VC}}(Z)$ is closed.
The star and pooled results retain their separate sup-norm assumptions on the heads and on $g_0$, $\alpha_0$ (Assumption~\ref{ass:pooled-star}(i)). A sufficient condition for bounded heads is bounded $D$ together with uniformly bounded slope and intercept functions. Neither finite second moments nor the conditional-variance bound supplies these envelopes.
The space $\Gamma_{\mathrm{VC}}(Z)$ supplies a constant that is uniform over $\phi$, tuned ones included: $\Gamma_\phi\subseteq\Gamma_{\mathrm{VC}}(Z)$ and $C_m^2\le1/\operatorname*{ess\,inf}\Var(D\mid Z)$ by Lemma~\ref{lem:ad-vc}.
For a fixed $\phi$, $\Gamma_{\mathrm{VC}}(X_\phi)$ has the constant $C_m^2=1/\operatorname*{ess\,inf}\Var(D\mid X_\phi)$, which is no larger.
In Assumption~\ref{ass:pooled-star}(ii) and Propositions~\ref{prop:riesz-theorem32} and~\ref{prop:star-finite-riesz} the constant $C_m$ multiplies $\norm{h}_{P,2}^2$ (the bound reads $\E[m(O,h)^2]\le C_m\norm{h}_{P,2}^2$), so there $C_m=1/\operatorname*{ess\,inf}\Var(D\mid\cdot)$ and not its square root.
For item~(ii), take the long objects on $\Gamma_L$ and require the same continuity bound on every difference $g^\star-g_{\phi_n}$, together with the score-product and efficiency conditions of the original result.
This gives convergence to the restricted long score and inference for the restricted long target.
An interpretation as inference for an unrestricted average derivative requires a separate argument identifying that target and its efficient score with the restricted ones.
\end{proposition}

\begin{proof}
Each item is checked by reading the displayed steps of the proof in which $C_m$ appears.
(i) In the proof of Theorem~\ref{thm:dml-expansion} the bound is applied to $m(O,\delta g^{(-k)})$ with $\delta g^{(-k)}=\hat g^{(-k)}_\phi-g_\phi$, and to $m(O,g_\phi)$ in the uniform-integrability step for a drifting representation.
(ii) In the proof of Lemma~\ref{lem:score-conv} it is applied once, to $m(O,g^\star-g_{\phi_n})$, and Theorem~\ref{thm:strong-adaptive} uses it only through that lemma.
(iii) In Appendix~\ref{app:cross-fold} it is applied to $m(O,\delta g^{(-k)})$ and $m(O,\delta\alpha^{(-k)})$, and the proof notes that $\delta\alpha^{(-k)}\in\cH$ conditionally on the training sample; the same holds for $\delta g^{(-k)}$.
(iv) In the proof of Lemma~\ref{lem:pooled-margin} it is applied to $m(O,\hat\alpha_{jk})$ in the variance bound and to $m(O,\hat\alpha_{jk}-\alpha_{n,j}^\ast)$ in the bound on $|Q(\hat\alpha_{jk})-Q(\alpha_{n,j}^\ast)|$.
(v) In the proof of Corollary~\ref{cor:joint} it is applied to $m(O,\hat\alpha_k)$ in the second-moment bounds.
(vi) In the proof of Lemma~\ref{lem:dml-pipeline-bound-components} and Theorem~\ref{thm:bounds-pipeline} it is applied to $m(O,\delta\alpha^{(-k)})$ in the bound on $\widetilde\Delta_k$, to $m(O,\delta g^{(-k)})$ in the first component of the score difference, and to $m(O,\alpha_0)$, which enters the influence function.
In every case the restricted representers enter only through the Riesz property on the space that contains the heads, which Lemma~\ref{lem:ad-closure}(i) supplies, and through the identity, which Lemma~\ref{lem:ad-closure}(iii) supplies.
For the last statement, Lemma~\ref{lem:ad-vc}(iv) places every $\Gamma_{\mathrm{VC}}(X_\phi)$, and so every head and every projection onto $\cH_\phi$, inside $\Gamma_{\mathrm{VC}}(Z)$, which is closed under the conditions of Lemma~\ref{lem:ad-vc}(iii) with $X=Z$; Lemma~\ref{lem:ad-vc}(i) gives its constant.
Head limits, including triangular-array comparison functions such as $\alpha_{n,j}^\ast$, are assumed to lie in $\Gamma_{\mathrm{VC}}(Z)$. For a fixed $L^2$ limit this follows from closedness. Projecting a triangular-array comparison function preserves its $L^2$ approximation error, but need not preserve the other moment or margin conditions; those conditions must be checked for the chosen comparison functions.
For (vii), Proposition~\ref{prop:onesided} uses only the endpoint expansions and the sensitivity containment, so (vi) and Lemma~\ref{lem:ad-closure}(iii) give its restricted version.
For Proposition~\ref{prop:joint}, conditional on each training sample the two projection identities give~\eqref{eq:mixed-bias-k} in the corresponding short space. The rest of its proof uses only the assumed product rate and convergence of the scores to a common nonrandom sequence; conditional Chebyshev, the central limit theorem, and the variance argument are unchanged. This also supplies the expansion used in Corollary~\ref{cor:joint}.
The proof of Theorem~\ref{thm:pooled-joint} combines those expansions over finitely many recipes; its covariance, Gaussian-quantile, and selection arguments impose no further requirement on the ambient function space. Item~(iv) supplies its margin criterion.
For Proposition~\ref{prop:sfu}, condition on the independent tuning sample and apply (i), (vi), and Proposition~\ref{prop:onesided} to each resulting fixed short space. Direct uniform moment bounds replace the two conditional-Jensen steps in its proof; the subsequence argument then gives the same conditional and unconditional coverage.
Finally, the frozen-union inequalities follow from nested orthogonal projections, just as in Lemma~\ref{lem:ad-closure}. These arguments retain the original independent-observation sampling assumptions.
\end{proof}

\clearpage
\section{Implementation details for the demand illustration}
\label{app:empirical-details}

The main text compares estimates and sensitivity bounds across information sets and tuning approaches. This appendix records the implementation choices and intermediate results behind those comparisons.

\subsection{Encoder training and sample allocation}

The baseline uses frozen text and image embeddings, and in Approach~1 we train each of the fourteen learners, one outcome and one Riesz learner per information set, once per fold.
In the tuned runs we fine-tune every weight of EmbeddingGemma~2 (\texttt{google/embeddinggemma-2}, revision \texttt{914f7f8}) that an information set uses; the tabular block is never tuned. The text embedding uses the language model of $271$ million parameters, and the image embedding and the joint embedding of text and image use the language model together with the vision tower, $439$ million parameters in total. Each embedding averages the final hidden states over all tokens and is normalized to unit length, with $768$ dimensions; texts carry the prefix \texttt{task: clustering | query: } and are truncated at $512$ tokens.
In each fit we first train the heads on the frozen features and then the encoder and heads together (encoder learning rate $2\times10^{-5}$, bf16, gradient checkpointing, five epochs, early stopping), with batches formed by product.
The frozen baseline uses the embeddings of the untuned model.
In Approach~1 we tune a separate encoder for each nuisance in each training fold ($7\times5\times2=70$ fits).
In Approach~2 we tune one shared encoder by the joint loss~\eqref{eq:joint-loss} ($7\times5=35$ fits), with $V_Y$ and $V_\alpha$ computed from the untuned full representation on each training fold.
In Approach~3 we split the products $30/70$ into $\mathcal S_0$ ($2{,}317$ products) and $\mathcal S_1$ ($5{,}408$ products, $130{,}434$ changes), tune on $\mathcal S_0$ ($7\times2=14$ fits), freeze the union of the two selected encoders with the tabular block entering once, and cross-fit the heads on $\mathcal S_1$.

\subsection{Aggregation weights}

Star aggregation treats the two nuisances separately.
In the frozen baseline, the text-plus-tabular learner has the smallest pooled out-of-fold outcome loss and becomes the center, and the star step mixes in the text learner, so that $\hat g_{\mathrm{star}}=0.6179\,\hat g_{\mathrm{txt+tab}}+0.3821\,\hat g_{\mathrm{txt}}$.
For the representer, the text-plus-image learner is the center, mixed with the tabular learner, so that $\hat\alpha_{\mathrm{star}}=0.6350\,\hat\alpha_{\mathrm{txt+img}}+0.3650\,\hat\alpha_{\mathrm{tab}}$.
Both aggregates lie in the span $\cH$ of the seven pipelines, so the estimator built from them has the pipeline target $\theta_0$ as its estimand (Theorem~\ref{thm:pipeline-dml}).
Combining the two aggregates in~\eqref{eq:empirical-dr-score} gives $\hat\theta_{\mathrm{star}}=-1.142$, with standard error $0.017$ and 95\% confidence interval $[-1.176,-1.108]$.
With tuned encoders the text-plus-tabular set remains the outcome center and the full set becomes the representer center, $\hat g_{\mathrm{star}}=0.6402\,\hat g_{\mathrm{txt+tab}}+0.3598\,\hat g_{\mathrm{txt+img}}$ and $\hat\alpha_{\mathrm{star}}=0.6005\,\hat\alpha_{\mathrm{full}}+0.3995\,\hat\alpha_{\mathrm{img}}$, and the aggregate estimate is $-1.129$ against $-1.142$.
The out-of-fold $R^2$ of the outcome learners rises from between $0.120$ and $0.129$ to between $0.121$ and $0.135$, highest ($0.1348$) for the text-plus-tabular set.
A 1\% price increase lowers the inverse-rank demand index by about 1.1\%, which under the power law of Section~\ref{sec:empirical-specification} is a demand elasticity of about $-1.9$, or $-1.5$ to $-2.5$ over the band for $\kappa$.

\subsection{Deletion calibration}

We first calibrate the sensitivity parameters by the deletion exercise of Section~\ref{sec:calibration}, treating the full information set as long and text alone as short.
In the frozen baseline, deleting images and tabular attributes lowers the out-of-fold $R^2$ of the outcome learner only from $0.1288$ to $0.1284$, and the out-of-fold Riesz loss without the penalty falls from $-41.73$ to $-42.13$.
The gains of Section~\ref{sec:calibration} are $\widehat G_Y=(\widehat R^2_{\mathrm{full}}-\widehat R^2_{\mathrm{text}})/(1-\widehat R^2_{\mathrm{full}})\approx0.0005$ and $\widehat G_D=(\widehat{\mathcal L}_{\mathrm{text}}-\widehat{\mathcal L}_{\mathrm{full}})/(-\widehat{\mathcal L}_{\mathrm{text}})\approx-0.009$. The second gain is negative, so $\widehat C_D^{2}=\widehat G_D/(1-\widehat G_D)$ and the alignment $\hat\rho$ are undefined, and missing information as strong as the deleted modality ($k=1$) has $\widehat C_Y^{2}\approx0.0005$.
Table~\ref{tab:empirical-calibration} collects these numbers for the frozen baseline and for Approaches~1 and~2.
In the frozen baseline, images and tabular attributes add almost nothing to the prediction of the differenced outcome and raise the out-of-fold Riesz loss.
With tuned encoders in Approach~2, the outcome contrast grows to $\widehat C_Y^2=0.0038$, while the text-only Riesz loss is again lower ($-42.45$ against $-42.27$), so the benchmark and the alignment remain undefined.
In Approach~1 the full set improves both fits, with $\widehat G_Y\approx0.0017$ and $\widehat G_D\approx0.018$, so that $\widehat C_Y^{2}\approx0.0017$, $\widehat C_D^{2}\approx0.019$, and $\hat\rho=-0.01$; Section~\ref{sec:empirical-sensitivity} uses this benchmark.
In Approach~3, deleting within $\mathcal S_1$ the half of the union selected for the representer, the image block with the tabular attributes, gives $\widehat C_Y^{2}=0.0035$ and $\widehat G_D=-0.0004$, so $\widehat C_D^{2}$ is undefined; deleting the half selected for the outcome regression, the text block, gives $\widehat C_Y^{2}=0.0116$, $\widehat C_D^{2}=0.026$, and $\hat\rho=0.58$. On $\mathcal S_1$ the text-only representer attains a lower out-of-fold Riesz loss than the full one, so Approach~3 has no benchmark multiple.

\subsection{Independent selection and frozen-union inference}

Table~\ref{tab:empirical-route3} separates the selection stage, inference sample, and sensitivity diagnostics for Approach~3.

\begin{table}[!htbp]
\centering
\setstretch{1}
\caption{Independent selection and inference on the frozen union}
\label{tab:empirical-route3}
\small
\renewcommand{\arraystretch}{1.16}
\begin{tabular*}{\textwidth}{@{\extracolsep{\fill}}lr@{}}
\toprule
\multicolumn{2}{@{}l}{\textit{Panel A. Selection and sample allocation}}\\[3pt]
Tuning sample, products & $2{,}317$\\
Inference sample, products & $5{,}408$\\
Inference sample, weekly changes & $130{,}434$\\
Selected outcome encoder (validation $R^2$) & Text ($0.1196$)\\
Selected representer encoder (validation loss) & Image $+$ tabular ($-48.94$)\\
Selected representer, $\E[\beta^2]$ & $4{,}830$\\
Frozen-union features & $1{,}564$\\
\midrule
\multicolumn{2}{@{}l}{\textit{Panel B. Estimate and inference}}\\[3pt]
Estimate & $-1.146$\\
Clustered standard error & $0.022$\\
95\% confidence interval & $[-1.189,-1.103]$\\
Standard-error inflation, $\sqrt{n/n_1}$ & $1.20$\\
\midrule
\multicolumn{2}{@{}l}{\textit{Panel C. Sensitivity diagnostics}}\\[3pt]
Deleting the representer half: $\widehat C_Y^2$ and $\widehat C_D^2$ & $0.0035$ and ---\\
Deleting the outcome half: $\widehat C_Y^2$ and $\widehat C_D^2$ & $0.0116$ and $0.026$\\
Robustness value at $U=-0.605$, $\rho=1$ & $0.137$\\
Robustness value at $U=-0.605$, $\rho=0.2$ & $0.515$\\
Upper bound $U$, strong scenario & $-0.940$\\
\bottomrule
\end{tabular*}
\empnotes{Selection uses $\mathcal S_0$; cross-fitting and inference use the independent sample $\mathcal S_1$. The selected encoders are frozen before inference. The representer half consists of the image block and the tabular attributes, and the outcome half of the text block; the dash marks an undefined quantity. The strong scenario sets $C_Y^2=1-R^2_{\alpha^\star\sim\alpha_\phi}=0.050$ and $\rho=1$.}
\end{table}
\FloatBarrier

\subsection{Estimates on all weekly changes}
\label{app:empirical-allrows}

The estimation sample of Section~\ref{sec:empirical-specification} excludes weeks without a price change. With them, the panel has $626{,}919$ weekly changes for $8{,}642$ products. We keep the tuned encoders and $\lambda$ and refit the heads on all weekly changes with the same assignment of products to folds. Table~\ref{tab:empirical-allrows} reports both sets of estimates. On all weekly changes the estimates lie between $-1.150$ and $-1.116$, within $0.015$ of their values on the estimation sample.

\begin{table}[!htbp]
\centering
\setstretch{1}
\caption{Estimates on the estimation sample and on all weekly changes}
\label{tab:empirical-allrows}
\small
\renewcommand{\arraystretch}{1.13}
\begin{tabular*}{\textwidth}{@{\extracolsep{\fill}}lrr@{}}
\toprule
Target & Price changes only & All weekly changes\\
\midrule
Frozen, star & $-1.142$ ($0.017$) & $-1.127$ ($0.020$)\\
Approach 1, star & $-1.129$ ($0.018$) & $-1.127$ ($0.019$)\\
Approach 2, text $+$ tabular & $-1.129$ ($0.019$) & $-1.116$ ($0.021$)\\
Approach 3, union & $-1.146$ ($0.022$) & $-1.150$ ($0.023$)\\
\bottomrule
\end{tabular*}
\empnotes{Standard errors clustered by product are in parentheses. The estimation sample has $186{,}335$ weekly changes for $7{,}725$ products and the full panel $626{,}919$ weekly changes for $8{,}642$ products; for Approach~3 the inference sample $\mathcal S_1$ has $130{,}434$ and $440{,}077$ weekly changes.}
\end{table}

\FloatBarrier

\subsection{Text and image encoders with two folds}
\label{app:empirical-encoders}

We also estimated the frozen baseline and the three approaches with $K=2$ folds, once with EmbeddingGemma~2 and once with RoBERTa (\texttt{twitter-roberta-base-sentiment-latest}, $512$ tokens) for text and the Vision Transformer (\texttt{vit-base-patch16-224}, class token after the final layer normalization) for images, where the text-plus-image set concatenates the two embeddings. Both runs use the sample of Section~\ref{sec:empirical-specification} and the training settings of Section~\ref{sec:empirical-learners}, including the penalty. Table~\ref{tab:empirical-encoders} reports the four main estimates. Replacing RoBERTa and the Vision Transformer by EmbeddingGemma~2 moves each of them by at most $0.011$, or $0.6$ standard errors, and moving from $K=2$ to $K=5$ folds moves each EmbeddingGemma~2 estimate by at most $0.008$. In all three runs, text plus tabular attributes has the smallest pooled joint loss in Approach~2; with RoBERTa and the Vision Transformer, we select text plus tabular attributes for both nuisances in Approach~3.

\begin{table}[!htbp]
\centering
\setstretch{1}
\caption{Main estimates by encoder and number of folds}
\label{tab:empirical-encoders}
\small
\renewcommand{\arraystretch}{1.13}
\begin{tabular*}{\textwidth}{@{\extracolsep{\fill}}lrrr@{}}
\toprule
& RoBERTa and ViT & \multicolumn{2}{c}{EmbeddingGemma~2}\\
\cmidrule(lr){2-2}\cmidrule(lr){3-4}
Target & $K=2$ & $K=2$ & $K=5$\\
\midrule
Frozen, star & $-1.161$ ($0.018$) & $-1.150$ ($0.018$) & $-1.142$ ($0.017$)\\
Approach 1, star & $-1.140$ ($0.018$) & $-1.136$ ($0.018$) & $-1.129$ ($0.018$)\\
Approach 2, selected & $-1.147$ ($0.017$) & $-1.136$ ($0.019$) & $-1.129$ ($0.019$)\\
Approach 3, union & $-1.149$ ($0.023$) & $-1.150$ ($0.022$) & $-1.146$ ($0.022$)\\
\bottomrule
\end{tabular*}
\empnotes{Standard errors clustered by product are in parentheses. The last column repeats the estimates of Section~\ref{sec:empirical}.}
\end{table}

\FloatBarrier

\subsection{Fine-tuning diagnostics}
\label{app:empirical-finetune}

If no epoch of fine-tuning improves the validation loss, early stopping keeps the state after the first stage, in which only the heads are trained on the frozen features. This happens in $36$ of the $102$ fits that fine-tune an encoder, namely in none of the $30$ outcome fits and in $15$ of the $30$ representer fits of Approach~1, in $16$ of the $30$ fits of Approach~2, and in $2$ of the $6$ outcome fits and $3$ of the $6$ representer fits of Approach~3. The fine-tuned fits took $51.1$ GPU hours on NVIDIA H100 cards, $34.0$ for Approach~1, $13.7$ for Approach~2, and $3.4$ for Approach~3, with at most $70$ minutes for one fit; the frozen baseline took $0.3$ GPU hours.

\end{document}